\documentclass[physrev,aps,reprint,bibnotes,superscriptaddress]{revtex4-2} 

\usepackage[T1]{fontenc}
\usepackage[utf8]{inputenc}
\usepackage[USenglish]{babel}
\usepackage[a4paper,centering,hmargin=1.75cm,vmargin=2cm]{geometry} \usepackage{amsmath,amssymb,graphicx,bm,microtype} 
\usepackage{amsthm}
\usepackage[dvipsnames]{xcolor} 
\usepackage{booktabs} 
\usepackage{multirow} 
\usepackage{setspace}
\usepackage[colorlinks,allcolors=blue!50!black]{hyperref}
\usepackage[all]{hypcap}
\usepackage[capitalize]{cleveref} 
\usepackage{braket}
\usepackage[version=4]{mhchem} 
\usepackage{upgreek,textgreek} 
\usepackage{siunitx,textcomp} 
\DeclareSIUnit\bohr{\text {\ensuremath {a}}_{0}}
\DeclareSIUnit\hartree{\text {\ensuremath {E}}_{\mathrm {h}}}
\usepackage{tikz}
\usepackage{longtable}
\usetikzlibrary{quantikz2}

\newcommand{\norm}[1]{\left\lVert #1 \right\rVert}

\newcommand{\dyad}[2]{\ket{#1}\bra{#2}}

\DeclareMathOperator*{\argmin}{argmin}
\newcommand{\R}{\mathbf{R}}

\newtheorem*{theorem*}{Theorem}
\newtheorem{lemma}{Lemma}
\newtheorem{corollary}{Corollary}
\newtheorem{proposition}{Proposition}

\newcommand{\ceil}[1]{\left\lceil{#1}\right\rceil}
\newcommand{\bits}[1]{\left\lceil\log\left({#1}\right)\right\rceil}
\newcommand{\expval}[2]{\langle#2|#1|#2\rangle}

\newcommand{\bnu}{\boldsymbol{\nu}}

\newcommand{\doublezeta}{\zeta\zeta}

\newcommand{\cX}{\mathcal{X}}

\newcommand{\cO}{\mathcal{O}}
\newcommand{\cW}{\mathcal{W}}
\newcommand{\cA}{\mathcal{A}}

\newcommand{\cR}{\mathcal{R}}
\newcommand{\cM}{\mathcal{M}}

\newcommand{\cL}{\mathcal{L}}

\newcommand{\data}{\mathrm{data}}
\newcommand{\flag}{\mathrm{flag}}
\newcommand{\junk}{\mathrm{junk}}
\newcommand{\init}{\mathrm{init}}
\newcommand{\meas}{\mathrm{meas}}
\newcommand{\prop}{\mathrm{prop}}
\newcommand{\rot}{\mathrm{rot}}

\newcommand{\nuc}{\mathrm{nuc}}

\newcommand{\grid}{\mathrm{grid}}

\newcommand{\bond}{\mathrm{bond}}
\newcommand{\ADD}{\textsc{add}}
\newcommand{\SUB}{\textsc{sub}}
\newcommand{\SOS}{\mathrm{SoS}}
\newcommand{\QFT}{\textsc{qft}}
\newcommand{\QSP}{\mathrm{QSP}}

\newcommand{\QAE}{\mathrm{QAE}}
\newcommand{\est}{\mathrm{est}}
\newcommand{\PREP}{\textsc{prep}}
\newcommand{\proxPREP}{\widetilde{\PREP}}
\newcommand{\proxL}{\tilde{\mathcal{L}}}
\newcommand{\SEL}{\textsc{sel}}
\newcommand{\proxSEL}{\widetilde{\SEL}}
\newcommand{\block}[1]{\PREP_{#1}^{\dagger}\SEL_{#1}\PREP_{#1}}
\newcommand{\proxblock}[1]{\proxPREP_{#1}^{\dagger}\SEL_{#1}\proxPREP_{#1}}
\newcommand{\UNIF}{\textsc{unif}}
\newcommand{\COMP}{\textsc{comp}}
\newcommand{\QROM}{\textsc{qrom}}
\newcommand{\SWAP}{\textsc{swap}}

\newcommand{\AND}{\textsc{and}}
\newcommand{\OR}{\textsc{or}}

\newcommand{\tilf}{\tilde{f}}
\newcommand{\tildegre}{\tilde{d}}
\newcommand{\proxW}{\widetilde{\cW}}
\newcommand{\LCU}{\mathrm{LCU}}

\newcommand{\contr}[1]{\textrm{ctrl-}{#1}}
\newcommand{\out}{\mathrm{out}}
\newcommand{\anc}{\mathrm{anc}}
\newcommand{\eq}{\mathrm{eq}}
\newcommand{\Ps}{\mathrm{Ps}}
\newcommand{\Er}{\mathrm{Er}}

\newcommand{\mass}{\mathrm{mass}}

\newcommand{\Toff}{\mathrm{Toff}}
\newcommand{\Had}{\mathrm{Had}}

\DeclareSIUnit\angstrom{\text{\AA}}

\newcommand{\Ind}{U_{\Pi_{S}}}
\newcommand{\isp}{\mathrm{ISP}}

\newcommand{\ext}{\mathrm{ext}}

\newcommand{\ho}{\Omega_{i \nu}}
\newcommand{\sm}{\Phi_{i \mu}}
\newcommand{\hatsm}{\hat{\Phi}_{i \mu}}
\newcommand{\hatsmQ}{\hat{\Phi}_{i \mu, \bar{\mathbf{Q}}}}
\newcommand{\hatsmQAprx}{\tilde{\hat{\Phi}}_{i \mu, \bar{\mathbf{Q}}}}
\newcommand{\SmnormAprx}{\widetilde{\mathcal{N}}_{i \mu}^{(n)}}
\newcommand{\SmMom}{K_{i \mu}^{(n)}}
\newcommand{\HatSmAprx}{\tilde{\hat{\Phi}}_{i \mu}}
\newcommand{\SmEpsT}{\epsilon^{(n,t)}_{i \mu}}
\newcommand{\SmEpsP}{\epsilon^{(n,p)}_{i \mu}}
\newcommand{\SmEpsC}{\epsilon^{(n,c)}_{i \mu}}
\newcommand{\SmDelT}{\delta^{(n,t)}_{i \mu}}
\newcommand{\SmDelP}{\delta^{(n,p)}_{i \mu}}
\newcommand{\SmDelC}{\delta^{(n,c)}_{i \mu}}
\newcommand{\SmPolyDeg}{M_{i \mu}}

\newcommand{\SmPolyAprx}{\tilde{\hat{\Phi}}^{(\mathrm{poly})}_{i \mu}}
\newcommand{\BBar}{\bar{\mathcal{B}}_{\mathrm{ISP}}}
\newcommand{\NBar}{\bar{N}_{\mathrm{ISP}}}

\newcommand{\GTO}{g_{\ell \mu}}
\newcommand{\GTOMom}{\hat{g}_{\ell \mu}}
\newcommand{\GTOMomAprx}{\tilde{\hat{g}}_{\ell \mu}}
\newcommand{\GTONormAprx}{\widetilde{\mathcal{N}}_{\ell \mu}^{(e)}}

\newcommand{\MOMom}{\hat{\phi}_a}
\newcommand{\MOMomAprx}{\tilde{\hat{\phi}}_a}
\newcommand{\MOPolyAprx}{\tilde{\hat{\phi}}_a^{(\mathrm{poly})}}
\newcommand{\SMTensor}{c^{(n)}_{i \mu \nu}}
\newcommand{\NormT}{\mathcal{N}^{(t)}_{i \mu}}

\makeatletter
\def\l@subsubsection#1#2{}
\makeatother

\begin{document}
\title{End-to-End Simulation of Chemical Dynamics on a Quantum Computer}

\author{Elliot C. Eklund} 
\affiliation{School of Chemistry, University of Sydney, NSW 2006, Australia}

\author{Arkin Tikku}
\affiliation{School of Chemistry, University of Sydney, NSW 2006, Australia}

\author{Patrick Sinnott}
\affiliation{School of Chemistry, University of Sydney, NSW 2006, Australia}

\author{William~J.~Huggins}
\affiliation{Google Quantum AI, Venice, CA, United States}

\author{Guang Hao Low}
\affiliation{Google Quantum AI, Venice, CA, United States}

\author{Dominic W. Berry}
\affiliation{School of Mathematical and Physical Sciences, Macquarie University, NSW 2109, Australia}

\author{Ivan Kassal}
\affiliation{School of Chemistry, University of Sydney, NSW 2006, Australia}

\begin{abstract}
Simulations of chemical dynamics are a powerful means for understanding chemistry. However, classical computers struggle to simulate many chemical processes, especially non-adiabatic ones, where the Born-Oppenheimer approximation breaks down. Quantum computers could simulate quantum-chemical dynamics more efficiently than classical computers, but there is currently no complete quantum algorithm for calculating dynamical observables to within a known error. Here, we develop an efficient, end-to-end quantum algorithm for simulating chemical dynamics that avoids all uncontrolled approximations (including the Born-Oppenheimer approximation) and whose error is bounded subject to mild assumptions. To do so, we treat the nuclei and the electrons on an equal footing and simulate the full molecular wavefunction on a momentum-space grid in first quantization, including all algorithmic steps: initial-state preparation, time evolution using qubitization, and measurement of chemical observables such as reaction yields and rates. Our work gives the first algorithm for quantum simulation of chemistry whose end-to-end complexity achieves sublinear scaling in the size of the grid. We achieve this by developing an exponentially faster method for initial-state-preparation. Photochemistry is a likely early application of our algorithm and we estimate resources required for end-to-end simulations of non-adiabatic dynamics of atmospherically important molecules. Classically intractable photochemical computations could be performed using resources comparable to those required for other chemical applications of quantum computing.
\end{abstract}

\maketitle 

Chemical reactions are inherently dynamical processes, but accurate simulation of the quantum-mechanical motion of both
nuclei and electrons is among the most challenging
problems in computational chemistry. Non-relativistic chemistry is governed by the \textit{molecular Hamiltonian}, which describes electrons and nuclei interacting via the Coulomb potential. We refer to the combined electron-nuclear quantum dynamics under this Hamiltonian as \textit{full dynamics}. Full dynamics reproduces all chemical phenomena, but simulating it on classical computers, by propagating the time-dependent Schrödinger equation, is possible only for the smallest molecules~\cite{arteaga2024strong}.

Classical algorithms can simulate dynamics in many difficult chemical systems, but retain exponential worst-case scaling with molecule size. Methods such as multiconfigurational time-dependent Hartree (MCTDH)~\cite{worth2008using,wang2015multilayer} significantly reduce the simulation cost compared to full dynamics using basis-set contractions optimized for each system~\cite{meng2013multilayer,xie2015full,schulze2016multi}. Nevertheless, like all fully quantum treatments of chemical dynamics, MCTDH scales exponentially with system size in the worst case~\cite{cigrang2025roadmap,prlj2025bestpracticesnonadiabaticmolecular}. Particularly challenging are strongly non-adiabatic systems, which significantly depart from the Born-Oppenheimer (BO) approximation~\cite{domcke2004conical,cederbaum1981multimode}.

Quantum computers promise efficient simulations of full dynamics; by using the molecular Hamiltonian, they could achieve the same performance for adiabatic or non-adiabatic dynamics. An \textit{efficient} simulation is one whose cost scales polynomially with molecule size and the inverse error of the final result. There are three additional desirable features in a simulation algorithm.
First, an \textit{end-to-end} algorithm completes each of the three essential steps: initial-state preparation (loading a given initial state, such as the output of an electronic-structure calculation, onto a quantum computer), time evolution, and measurement of observables. 
Second, in a \textit{controlled-error} algorithm, the final error can be arbitrarily reduced using greater computational resources.
And third, in a \textit{bounded-error} algorithm, the final error can be rigorously bounded in terms of system parameters.

Efficient quantum simulation of dynamics~\cite{zalka1998simulating} was applied to chemical dynamics early on~\cite{kassal2008}. More recent work has used advances in Hamiltonian simulation~\cite{berry2015simulating,Low2017,Low2019hamiltonian,childs2021theory} to extend dynamical simulations to other fields, including purely electronic dynamics~\cite{su2021,von2021quantum}, nuclear dynamics on potential energy surfaces~\cite{ollitrault2020,miessen2023quantum}, and electron-nuclear dynamics in fusion~\cite{rubin2024}. Similar techniques have been applied to simulating chemical dynamics~\cite{kassal2008,Jornada2025,pocrinic2026}, but the existing algorithms do not have the other three desirable features.

Gaps in full-dynamics simulation exist in both initial-state preparation and time evolution. In state preparation, there is no bounded-error method for preparing the full molecular wavefunction (electrons and nuclei) in first quantization. Bounded-error methods exist for preparing the electronic wavefunction alone~\cite{huggins2024}, but those for nuclear-state preparation~\cite{Jornada2025} lack error bounds on crucial approximations and scale linearly (i.e., steeply) with the number of grid points because of the cost of loading classical data. 
In time evolution, bounded-error algorithms exist only for electronic dynamics with fixed nuclei~\cite{su2021,childs2022quantum,georges2025quantum}, when nuclei are described by pseudopotentials~\cite{berry2024quantum,Jornada2025} (an uncontrolled approximation that has not been validated for dynamics), or when the Coulomb potential is truncated~\cite{pocrinic2026}.

Here, we develop an efficient quantum algorithm for simulating full dynamics that has the three desirable features above. 
First, the algorithm is end-to-end, going from a given initial state to expectation values of observables. 
For initial-state preparation, we prepare both electronic and nuclear wavefunctions, as well as entangled ones, in first quantization with only logarithmic scaling with grid size, an exponential improvement over prior work~\cite{Jornada2025}. For the nuclei, state preparation involves a novel approach for general linear coordinate transformations that may be of independent interest. To simulate dynamics, we construct a block encoding of the molecular Hamiltonian and use it to simulate the dynamics using qubitization in a plane-wave basis~(PWB). We then integrate these protocols into an optimal measurement algorithm for computing chemically relevant observables~\cite{cigrang2025roadmap}.
Second, the algorithm is controlled-error, as it makes no uncontrolled approximations such as the BO approximation or pseudopotentials.  
Third, it is bounded-error, subject to two minor approximations in \cref{sect:grid} that do not dominate the errors in the practice. 
As a result, we can provide constant-factor qubit and gate costs for the entire algorithm. 

We estimate the resources required by our algorithm to simulate the light-induced dissociation of small molecules in the atmosphere, one of the most important, but challenging, problems in non-adiabatic dynamics. A key quantity of interest is the quantum yield, the probability that a molecule fragments into smaller constituents after a given time. The molecules we consider require 3000--8000 logical qubits and \num{e13}--\num{e16} Toffoli gates to calculate the quantum yield up to an error $\epsilon=0.095$. These are the first end-to-end, bounded-error resource estimates for calculating dynamical observables and show chemical dynamics as a promising application of early quantum computers.

\tableofcontents

\newpage

\section{Algorithm overview}

Given input data that specifies the initial molecular state, our algorithm (\cref{fig:overview}) efficiently computes an observable of interest to within a specified error tolerance after evolving the initial state under the molecular Hamiltonian. It has all of the desired properties introduced above, being full-dynamics, controllable- and bounded-error, and end-to-end.

\textit{Full dynamics.} Time evolution occurs under the molecular Hamiltonian
\begin{equation} \label{eq:mol-ham}
    H = T + V,
\end{equation}
which is the sum of the total kinetic energy of all particles and their Coulomb interactions,
\begin{align} \label{eq:ham_terms}
    T &= - \sum_{i=1}^{\eta} \frac{\nabla_{i}^2}{2 m_{i}}, \quad V = \sum_{i < j}^{\eta} \frac{\zeta_{i} \zeta_{j}}{\left\|\mathbf{R}_{i}-\mathbf{R}_{j}\right\|_2}, 
\end{align}
where $\eta$ is the total number of particles (the sum of the numbers of electrons $\eta_{e}$ and nuclei $\eta_{n}$) and $m_{j}$, $\zeta_{j}$, and $\mathbf{R}_{j}$ are the mass, electric charge, and position of the $j$th particle, respectively. We work in atomic units ($\hbar = e = m_e = 4\pi\epsilon_0 = 1$), $\Vert \cdot \Vert_k$ is either the $\ell^k$ norm or the Schatten $k$-norm depending on context, and logarithms are in base~2 except as otherwise noted.

\textit{Controllable and bounded errors.} The error of the observable that is calculated after time evolution is fully controllable, i.e., given adequate computational resources, it can be systematically reduced arbitrarily close to zero. In addition, subject to mild assumptions about grid parameters in \cref{sect:grid}, the final error is bounded, i.e., it can be rigorously computed from the simulation parameters.

\textit{End-to-end.} The algorithm's four steps cover all stages of dynamics simulation, shown in \cref{fig:overview}.

\begin{figure*}[tb]
\centering
    \includegraphics[width=0.85\textwidth]{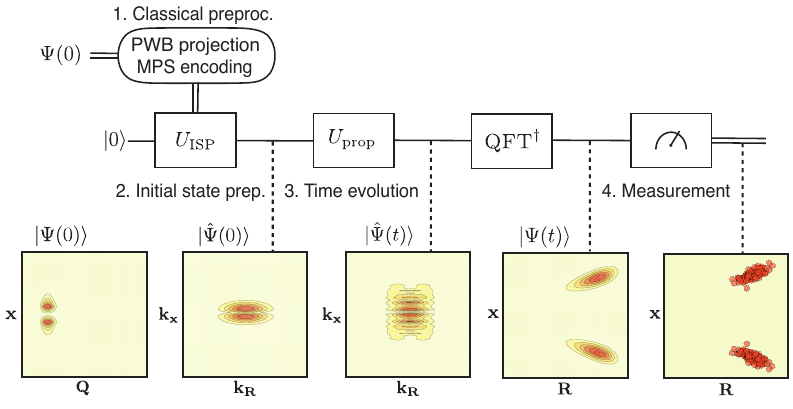} \label{fig:overview}
    \caption{Chemical dynamics algorithm. Given a specification of the initial state in the occupation-number basis (ONB), the algorithm returns an estimate of an observable. Each step is followed by an illustration of the wavefunction at that point. 
    \textit{Step 1: Classical preprocessing} projects the initial state onto a plane-wave basis (PWB) and finds a matrix-product-state (MPS) encoding it.
    \textit{Step 2: Initial-state preparation} converts the MPS into $|\hat{\Psi}(0)\rangle$ on the quantum computer in a PWB representation, and a normal-mode transformation $\textsc{nct}$ is applied to the nuclear states. 
    \textit{Step 3: Time evolution} evolves $|\hat{\Psi}(0)\rangle$ under the molecular Hamiltonian into $|\hat{\Psi}(t)\rangle$. $\QFT^{\dagger}$ is applied if the observable of interest is most easily measured in position space, yielding $|\Psi(t)\rangle$. 
    \textit{Step 4: Measurement} of observable $O$ on $|\Psi(t)\rangle$ gives outcomes that can be averaged to yield an expectation value. Not shown is quantum amplitude estimation, used to reduce the number of measurements. } 
    \label{fig:e2e}
\end{figure*}

\textit{0. Input Data} (\cref{sect:input_data}).
Our algorithm requires classical input data that specifies the initial molecular state. Any molecular state can be expressed as
\begin{equation} \label{eq:non_BO_kspace}
   \Psi(t=0) = \sum_{I}\sum_{J}C_{IJ} \psi^{(e)}_{I} \psi^{(n)}_{J},
\end{equation}
where $\{\psi^{(e)}_{I}\}$ and $\{\psi^{(n)}_J\}$ are electronic and nuclear configurations, with amplitude $\{C_{IJ}\}$. Then, the required input data includes $\{C_{IJ}\}$ and the specifications of $\{\psi^{(e)}_{I}\}$ and $\{\psi^{(n)}_J\}$ (e.g., the underlying basis sets). These requirements are met by output from standard electronic- and vibrational-structure calculations.

\textit{1. Classical Preprocessing (\cref{sect:preproc})}. We first preprocess $\Psi(t=0)$ to yield an approximate representation that can be efficiently prepared on a quantum computer. The electronic state is preprocessed by projecting the molecular orbitals~(MOs) onto a PWB, then encoding as a matrix product state (MPS)~\cite{huggins2024}. We preprocess the nuclear state by extending this approach to efficiently project the single-modals (SMs) onto the PWB and encode them as MPSs. The resulting state is $|\hat{\Psi}(t=0)\rangle_{\mathrm{MPS}}$, where a hat denotes a PWB state.

\textit{2. Initial State Preparation (\cref{sect:isp})}. $U_\mathrm{ISP}$ prepares the approximation $|\tilde{\hat{\Psi}}(t=0)\rangle_{\mathrm{MPS}}$ to $|\hat{\Psi}(t=0)\rangle_{\mathrm{MPS}}$ on the quantum computer in four steps.
First, $|\hat{\Psi}(t=0)\rangle_{\mathrm{MPS}}$ is prepared in the MO and SM bases (MOB and SMB, respectively) using the sum-of-Slaters (SoSlat) algorithm~\cite{fomichev2024}. 
Second, for the electronic state, the MO labels are antisymmetrized to ensure fermionic statistics~\cite{berry2018}.
Third, the state is approximately transformed from the MOB/SMB to a PWB using the unitary synthesis approach for preparing MPSs~\cite{huggins2024}. 
Finally, because the Coulomb operator to be implemented in time evolution is naturally expressed in Cartesian coordinates, the nuclear state is transformed from normal coordinates to Cartesian coordinates. \Cref{fig:coords} illustrates the various bases and how they relate to each other.

\textit{3. Time Evolution (\cref{sect:sim})}. $U_{\mathrm{prop}}$ evolves $|\hat{\Psi}(0)\rangle$ under $H$ for time $t$ to give the final state
\begin{equation} \label{eq:time_evol}
|\hat{\Psi}(t)\rangle = e^{-i H t} |\hat{\Psi}(0)\rangle.
\end{equation}
We simulate the time evolution using qubitization~\cite{Low2019hamiltonian}. To do so, we first project $H$ onto the PWB using the Galerkin representation (see \cref{app:Galerkin_review}), 
\begin{align}\label{eq:Galerkin_Ham}
T&=\sum_{j=1}^{\eta} \sum_{\mathbf{p} \in G} \frac{\|\mathbf{k_p}\|_2^2}{2 m_{j}} \ket{\mathbf{p}}\bra{\mathbf{p}}_{j}\\ 
V&=\frac{2 \pi}{L^3} \sum_{i \neq j=1}^{\eta} \sum_{\mathbf{p}, \mathbf{q} \in G} \sum_{\substack{\bnu \in G_0 \\ 
(\mathbf{p}+\bnu) \in G \\
(\mathbf{q}-\bnu) \in G}} \frac{\zeta_{i} \zeta_{j}}{\|\mathbf{k_\nu}\|_2^2} \nonumber\\
&\quad\,\quad\,\quad\,\quad\,\quad\,\quad\,\quad\,\quad\, \times \ket{\mathbf{p}+\bnu}\bra{\mathbf{p}}_{i}\ket{\mathbf{q}-\bnu}\bra{\mathbf{q}}_{j},
\end{align}
where $L$ is the length of the computational cell along any dimension, the momenta are 
\begin{equation}
    \mathbf{k}_{\mathbf{p}} =\frac{2\pi}{L} \mathbf{p} \;\text{ for }\;\mathbf{p}\in G \label{eq:k_mu},
\end{equation}
and the simulation grid is $G= \left[-(N-1)/2, (N-1)/2\right]^{3} \cap \mathbb{Z}^3$, where $N$ is the number of plane waves along each dimension. Finally, $G_0=G\backslash\{(0,0,0)\}$ excludes the singular zero mode~\cite{martin2004}. 
The projected Hamiltonian is decomposed into a linear combination of unitaries~(LCU), which allows us to block-encode it and perform time-evolution via quantum signal processing~(QSP) \cite{Low2017}. 

\textit{4. Measurement (\cref{sect:measure})}. Any chemical observable $O$ can be extracted from $|\hat{\Psi}(t)\rangle$. To do so, we transform the wavefunction from the PWB into the eigenbasis of the observable. For many chemical observables, such as bond lengths and reaction yields, this means transforming $|\hat{\Psi}(t)\rangle$ to the position-space $\ket{\Psi(t)}$ by applying an inverse QFT. For greatest asymptotic efficiency in estimating $\langle O \rangle$, we use an amplitude-estimation approach~\cite{Rall_2020}. For target error $\epsilon$ in $\langle O \rangle$, amplitude estimation has the asymptotically optical query complexity $\cO(\lambda_{O}/\epsilon)$ for $\lambda_{O}\geq \norm{O}$. For many important observables, such as yields, $\lambda_O =1$, ensuring a small cost.

\textit{Costs (\cref{sect:errors}).}
Ours is the first end-to-end quantum-simulation algorithm for chemistry whose resource costs can be stated in terms of only the system parameters $\eta$, $t$, and $\epsilon$. Overall, it requires 
$\widetilde{\cO}\left(\eta^3t/\epsilon^{4/3}+\eta^2t/\epsilon^{5/3}\right)$ gates, with constant factors given below.

\section{Input data}
\label{sect:input_data}

We begin with conventional approaches to specify $\Psi$ given in \cref{eq:non_BO_kspace}. In many areas of chemistry, dynamical simulations are challenging because correlation between the nuclei and electrons develops during dynamics, but is often negligible at $t=0$. In this case, $\Psi(t=0)$ is well described by a separable state
\begin{equation} \label{eq:psi_bo}
    \Psi = \Psi^{(n)} \Psi^{(e)},
\end{equation}
where $\Psi^{(e)}$ and $\Psi^{(n)}$ are the electronic and nuclear wavefunctions, respectively, and we drop the explicit reference to $t=0$. In the separable setting, the electronic and nuclear states can be determined separately, as discussed in \cref{sect:input_data_elec,sect:input_data_nuc}, respectively. In \cref{sect:input_data_non_sep}, we consider non-separable initial states.

\begin{figure*}[t!] 
	\centering
	\includegraphics[width= 0.86\textwidth]{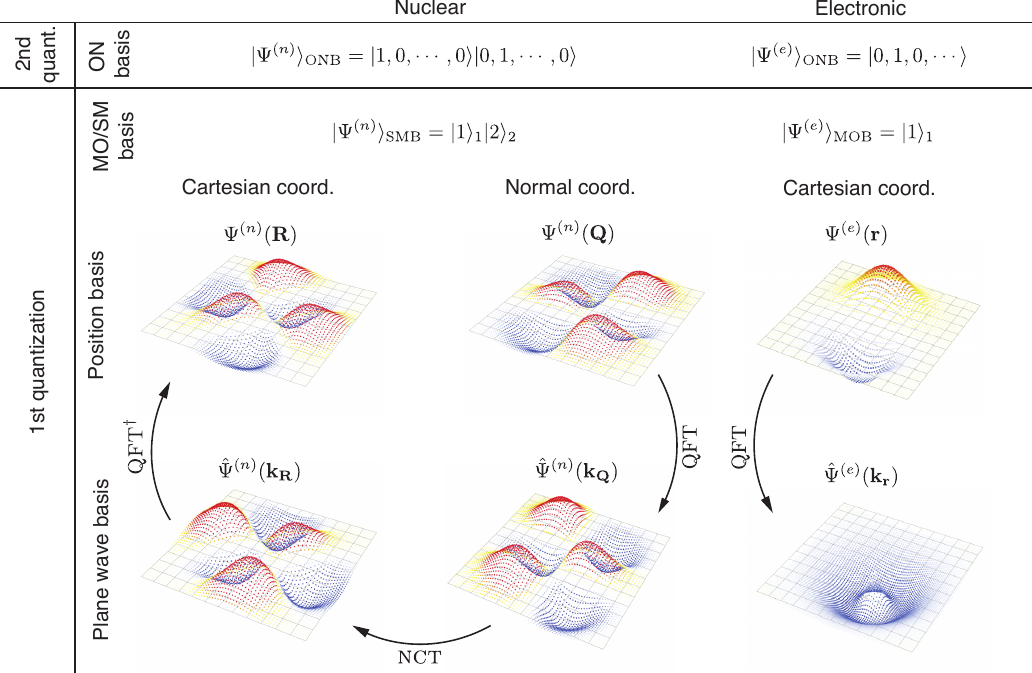}
	\caption{Wavefunction representations (bases and coordinate systems) used in the algorithm.
    \textbf{Top row:} example nuclear and electronic wavefunctions in the occupation number (ON) basis.
    \textbf{Bottom row:} the same wavefunctions in first-quantized representations, including the molecular orbital (MO) and single-modal (SM) bases. The grid bases are related by the normal-mode transformation $\textsc{nct}$ or quantum Fourier transform QFT. Hats indicate functions in plane-wave bases.}
	\label{fig:coords}
\end{figure*}

\subsection{Electronic-state data} \label{sect:input_data_elec}
The electronic state $\Psi^{(e)}$ is specified using classical methods for solving the time-independent electronic Schr\"{o}dinger equation. For fixed nuclear geometry $\mathbf R$, $H$ reduces to the electronic Hamiltonian~\cite{szabo11996, helgaker2012},
\begin{multline}\label{eq:Elec_Ham}
    H^{(e)}(\mathbf{R}) = -\sum_{i=1}^{\eta_e} \frac{\nabla_i^2}{2} + 
\sum_{i < j}^{\eta_e} \frac{1}{\left\|\mathbf{r}_i-\mathbf{r}_j\right\|_2}  \\
 -\sum_{i=1}^{\eta_e} \sum_{\ell=1}^{\eta_n} \frac{\zeta_{\ell}}{\left\|\mathbf{R}_{\ell}-\mathbf{r}_i\right\|_2}.
\end{multline}
Its eigenstates obey
\begin{equation} \label{eq:elec_schrodinger}
    H^{(e)}(\mathbf{R})\Psi^{(e)}(\mathbf{x};\mathbf{R}) = \mathcal{E}(\mathbf{R})\Psi^{(e)}(\mathbf{x};\mathbf{R}),
\end{equation}
where $\mathcal{E}(\mathbf{R})$ is the electronic energy of the molecule given nuclei at $\mathbf{R}$. We use a semicolon to indicate that $\Psi_e$ depends parametrically on $\mathbf{R}$.

In the Hartree-Fock (HF) approximation, the electronic state is a single Slater determinant,
\begin{equation} \label{eq:slat_dt}
\psi^{(e)}_{I}(\mathbf{x}) = \frac{1}{\sqrt{\eta_{e}!}}\left|\begin{array}{cccc}
\chi_a\left(\mathbf{x}_1\right) & \chi_b\left(\mathbf{x}_1\right) & \cdots & \chi_{c}\left(\mathbf{x}_1\right) \\
\chi_a\left(\mathbf{x}_2\right) & \chi_b\left(\mathbf{x}_2\right) & \cdots & \chi_{c}\left(\mathbf{x}_2\right) \\
\vdots & \vdots & \ddots & \vdots \\
\chi_a\left(\mathbf{x}_{\eta_{e}}\right) & \chi_b\left(\mathbf{x}_{\eta_{e}}\right) & \cdots & \chi_{c}\left(\mathbf{x}_{\eta_{e}}\right)
\end{array}\right|,
\end{equation}
where $\chi_{a}(\mathbf{x}_i)$ is the $a$th molecular orbital (MO) occupied by the $i$th electron, $\mathbf{x}_i=(\mathbf{r}_i,\omega_i)$ is a vector of the electron's position $\mathbf{r}_i$ and its spin variable $\omega_i$, and we suppress the explicit reference to $\mathbf{R}$. Each MO is a product of a spatial orbital and a spin function,
\begin{equation}
    \chi_a(\mathbf{x}_i) = \phi_a(\mathbf{r}_i)\sigma_a(\omega_i), \,\,\, \mathrm{where} \,\,\, \sigma(\omega) \in \{\alpha(\omega),\beta(\omega)\}.
\end{equation}
We refer to $\psi^{(e)}_I(\mathbf{x})$ as an electronic configuration and index it using its occupation number vector $I = (n_1,n_2,\ldots,n_{N_{\mathrm{MOB}}})$, where $n_a \in \{0,1\}$ indicates if the $a$th orbital is occupied and $N_{\mathrm{MOB}}$ is the number of MOs.

Each spatial orbital can be expanded as a linear combination of $N_g$ atom-centered Gaussians~\cite{szabo11996, helgaker2012},
\begin{equation} \label{eq:gauss_expand}
    \phi_a(\mathbf{r}) = \sum_{\ell=1}^{\eta_n}\sum_{\mu=1}^{N_g} c_{a \ell \mu}^{(e)} g_{\ell\mu}(\mathbf{r}),
\end{equation}
where the Gaussian primitives are defined by
\begin{equation} \label{eq:gauss_prim}
g_{\ell\mu}(\mathbf{r}) \propto (x-X_\ell)^{l_\mu} (y -Y_\ell)^{m_\mu} (z-Z_\ell)^{n_\mu} e^{-\gamma_\mu\Vert \mathbf{r} - \mathbf{R}_\ell\Vert^2},
\end{equation}
$\gamma_\mu$ is a width, $l_\mu,m_\mu,n_\mu$ are integers, $x, y, z$ are components of $\mathbf{r}$, and $X_\ell, Y_\ell, Z_\ell$ are components of $\mathbf{R}_\ell$. 

Higher-accuracy electronic-structure methods go beyond HF, typically by expressing the electronic wavefunction as a linear combination of configurations~\cite{szabo11996, helgaker2012},
\begin{equation} \label{eq:post-hf}
    \Psi^{(e)}(\mathbf{x}) = \sum_I C_I^{(e)} \psi_I^{(e)}(\mathbf{x}).
\end{equation}
Post-HF methods differ in how they select the set of configurations. For example, configuration interaction singles and doubles (CISD) only includes singly and doubly excited configurations. Post-HF methods converge to the solution of the Schrödinger equation (within the basis set) in the limit of full configuration interaction (FCI), which includes all configurations. Overall, an electronic state is specified by the $C_{I}^{(e)}$, $c_{a\ell\mu}^{(e)}$ and $g_{\ell \mu}$.

\subsection{Nuclear-state data} 
\label{sect:input_data_nuc}
For a separable state, the nuclear wavefunction $\Psi^{(n)}$ is found by solving the time-independent nuclear Schr\"{o}dinger equation,
\begin{equation} \label{eq:nuc_schro}
H^{(n)}_a\Psi^{(n)}(\mathbf{R}) = E \Psi^{(n)}(\mathbf{R}),
\end{equation}
where $H^{(n)}_a$ is the nuclear Hamiltonian corresponding to the $a$th electronic excited state. That is~\cite{szabo11996},
\begin{equation} 
    H^{(n)}_a = - \sum_{\ell=1}^{\eta_n}\frac{\nabla_\ell^2}{2m_\ell}  + \mathcal{V}_a(\mathbf{R}),
\end{equation}
where the potential energy surface (PES) of the $a$th excited state is
\begin{equation}
    \mathcal{V}_{a}(\mathbf{R}) = \mathcal{E}_{a}(\mathbf{R}) + \sum_{\ell < k}^{\eta_n} \frac{\zeta_{\ell} \zeta_k}{\left\|\mathbf{R}_{\ell}-\mathbf{R}_k\right\|_2},
\end{equation}
and $\mathcal{E}_{a}(\mathbf{R})$ is the electronic energy from \cref{eq:elec_schrodinger}.

For chemically relevant, low-energy nuclear states, $\mathcal{E}_{a}(\mathbf{R})$ only needs to be I computed for small displacements about nuclear equilibrium positions, which can be done accurately using a variety of methods~\cite{bowman1979,carter1997, carter1997_2, seidler2007,neff2009,christoffel1982,thompson1980}.

A convenient coordinate system for expressing $\mathcal{V}_{a}(\mathbf{R})$ is the normal coordinates. Expanding $\mathcal{V}_{a}(\mathbf{R})$ about the equilibrium geometry $\mathbf{R}^{(0)}$ to second order in terms of the mass-weighted nuclear displacements $\bar{\mathbf{D}} = \mathbf{M}^{1/2}\mathbf{D}$, with $\mathbf{D} = \mathbf{R} - \mathbf{R}^{(0)}$, and 
\begin{multline}
    \mathbf{M} =\operatorname{diag}(m_{1}, m_{1}, m_{1}, m_{2}, m_{2}, m_{2},\ldots, \\ m_{\eta_n}, m_{\eta_n}, m_{\eta_n})
\end{multline}
yields the quadratic Hamiltonian
\begin{equation}
    H^{(n,2)}_a=-\frac{1}{2}\nabla^2_{\bar{\mathbf{D}}}  +\frac{1}{2} \bar{\mathbf{D}}^{\top} \mathbf{He}_a\bar{\mathbf{D}},
\end{equation}
where $\nabla^{2}_{\bar{\mathbf{D}}}$ is the Laplacian with respect to $\bar{\mathbf{D}}$ and the Hessian is defined by its elements
\begin{equation}
    [\mathbf{He}_a]_{ij} = \frac{\partial^2 \mathcal{V}_a(\bar{\mathbf{D}})} {\partial \bar{\mathbf{D}}_i\partial \bar{\mathbf{D}}_j}\Bigr|_{\bar{\mathbf{D}}=\mathbf{0}}.
\end{equation}
Defining the normal coordinates $\mathbf{Q} = \mathbf{O}\bar{\mathbf{D}}$ as those that diagonalize the Hessian (with orthogonal matrix $\mathbf{O}$  such that $\mathbf{O} \mathbf{He}_a\mathbf{O}^{\top} = \mathbf{W}_a$ for $\mathbf{W}_a=\operatorname{diag}(\omega_{a1}^2, \ldots \omega_{a 3 \eta_n}^2)$),  we arrive at
\begin{equation} \label{eq:hnuc_second}
    H^{(n,2)}_a =  - \frac{1}{2} \nabla_{\mathbf{Q}}^2 + \frac{1}{2} \mathbf{Q}^{\top}\mathbf{W}_a \mathbf{Q}.
\end{equation}
Thus, in normal coordinates, $H^{(n,2)}_a$ is a sum of $3\eta_n$ independent harmonic oscillators with frequencies $\omega_{i}$.  

The eigenstates of $H^{(n,2)}_a$ are products of $3\eta_n$ individual states, each of which is a function of a single normal coordinate. For a non-linear molecule, $3\eta_n - 6$ of the normal coordinates ($3\eta_n -5 $ for linear molecules) have non-zero eigenvalues and correspond to vibrational modes of the molecule. The $\nu$th eigenstate of the $i$th mode is the harmonic-oscillator eigenstate~\cite{tannor2007,wilson1980}
\begin{equation} \label{eq:ho_eig}
    \ho (Q)
    = \frac{1}{\sqrt{2^{\nu}\nu!}}\left(\frac{\omega_i}{\pi}\right)^{\!1/4}
      H_{\nu}\big(\sqrt{\omega_i}Q\big)\,e^{-\omega_i Q^2/2},
\end{equation}
where $H_{\nu}(x)$ are Hermite polynomials. We call $\Omega_{i\nu}$ Hermite-Gaussian primitives because they play a similar role to Gaussian primitives in the previous section. From now on, we omit the dependence on $a$, with it being understood that $\ho $ and $\omega_i$ are defined with respect to a given PES.

The remaining six coordinates (five for linear molecules) have eigenvalues equal to zero and correspond to translations and rotations \cite{bunker2005, wilson1980}:
\begin{align} \label{eq:trans_rot_coords}
    &Q_{3\eta_n-5}=\mathcal T_x,\ \ Q_{3\eta_n-4}=\mathcal T_y,\ \ Q_{3\eta_n-3}=\mathcal T_z,\nonumber\\
    &Q_{3\eta_n-2}=\mathcal R_x,\ \ Q_{3\eta_n-1}=\mathcal R_y,\ \ Q_{3\eta_n}=\mathcal R_z,
\end{align}
where $\mathcal{T}_\tau$ and $\mathcal{R}_\tau$ are the translational and rotational coordinates along or about the axis $\tau$, respectively.
We denote translational and rotational states as $\Phi^{(\mathrm{tr})}_\tau(\mathcal{T}_\tau)$ and $\Phi^{(\mathrm{rot})}_\tau(\mathcal{R}_\tau)$, respectively. Formally, these are plane-wave solutions to the free-particle Hamiltonian. However, to ensure a localized centre of mass, we instead prepare them as Gaussian wavepackets, 
\begin{align} \label{eq:trans_rot_wvfxn}
    \Phi^{(\mathrm{tr})}_\tau(\mathcal{T}_\tau) &=\left(\Gamma_\tau/\pi\right)^{1/4}
    e^{- \Gamma_\tau  \mathcal{T}_\tau^2/2}, \nonumber \\
    \Phi^{(\mathrm{rot})}_\tau(\mathcal{R}_\tau) &= \left(\Upsilon_\tau/\pi\right)^{1/4} e^{- \Upsilon_\tau \mathcal{R}_\tau^2/2},
\end{align}
where $\Gamma_\tau$ and $\Upsilon_\tau$ set the widths of the Gaussians. The rotational coordinates only transform as rotations for infinitesimal displacements~\cite{bunker2005, wilson1980}. Therefore, $\Upsilon_{\tau}$ should be large enough that $\Phi^{(\mathrm{rot})}_\tau(\mathcal{R}_\tau)$ has support only over a sufficiently small range of $\mathcal{R}_{\tau}$.

The total nuclear state consisting of vibrational, translational, and rotational states is therefore
\begin{equation} \label{eq:nuc_st_approx}
    \Psi^{(n)}_{\mathrm{VHP}}(\mathbf{Q}) = \left(\prod_{i=1}^{N_{\mathrm{vib}}} \ho(Q) \right)\psi^{(\mathrm{tr})} \psi^{(\mathrm{rot})},
\end{equation}
where $N_{\mathrm{vib}} = 3\eta_n - 6$, and we define the translational and rotational states as
\begingroup
\allowdisplaybreaks
\begin{align}
\psi^{(\mathrm{tr})} &= \prod_{\tau \in\{x, y, z\}} \Phi^{(\mathrm{tr})}_\tau\left(\mathcal{T}_\tau\right) \\
    \psi^{(\mathrm{rot})} &= \prod_{\tau \in\{x, y, z\}}\Phi^{(\mathrm{rot})}_\tau\left(\mathcal{R}_\tau\right).
\end{align}
\endgroup
We refer to $\Psi^{(n)}_{\mathrm{VHP}}(\mathbf{Q})$ as the vibrational Hartree product (VHP) and the individual states that comprise it as single-modals (SMs). 

By treating the vibrational and rotational states as uncoupled, \cref{eq:nuc_st_approx} implicitly assumes the rigid-rotor approximation. This approximation is justified if the coupling between vibrational and rotational motion is initially small \cite{bunker2005}, which is typically the case for molecules at or near the rovibrational ground state. The translational state exactly separates from the vibrational and rotational states because there is no external field \cite{bunker2005}.

\renewcommand{\arraystretch}{1.5} 
\begin{table*}
\centering
    \begin{tabular}{l@{\hskip 12pt}c@{\hskip 18pt}c}
    \toprule
        & \multicolumn{2}{c}{State description} \\
        \cmidrule(l){2-3}
        Layer & Electronic & Nuclear \\
        \midrule
        Molecular & \multicolumn{2}{c}{ $\Psi = \sum_{IJ}C_{IJ} \psi^{(e)}_I \psi^{(n)}_J $}\\
        Configuration
            & $\psi^{(e)}_I = \det(\chi_a(\mathbf{x}_1)\cdots\chi_c(\mathbf{x}_{\eta_e}))$
            & $\psi_J^{(n)} = \Phi_{1\mu}\cdots\Phi_{3\eta_n \tau}$   \\
        MO/SM
            & $\chi_a(\mathbf{x}) = \sigma_a(\omega)\sum_{\ell,\mu} c_{a\ell\mu}^{(e)} g_{\ell\mu}(\mathbf{r})$
            & $\Phi_{i\mu}(Q) = \sum_{\nu} c_{i\mu\nu}^{(n)} \Omega_{i\nu}(Q)$ \\
        \bottomrule
    \end{tabular}
    \caption{States considered as input data decomposed into layers. The top layer are molecular states which consist of electronic and nuclear degrees of freedom. The middle layer shows electronic and nuclear configurations decomposed as MOs or SMs. The final layer shows MOs and SMs decomposed as Gaussian or Hermite-Gaussian primitives.}
    \label{tab:input_data}
\end{table*}

For states where anharmonic corrections to $\mathcal{V}_{a}$ are important, higher levels of theory are required to improve upon VHP. The next level of accuracy is the VSCF approach~\cite{bowman1979, christoffel1982, carter1997}, in which $H^{(n)}$ is treated using a mean-field approximation, similar to the HF approach in electronic structure. Here, each normal coordinate experiences the average interaction of all other normal coordinates. Like the VHP state, the VSCF state is also a product state,
\begin{equation} \label{eq:vscf_wfxn}
    \Psi^{(n)}_{\mathrm{VSCF}}(\mathbf{Q}) = \psi_{J}^{(\mathrm{vib})} \psi^{(\mathrm{tr})} \psi^{(\mathrm{rot})},
\end{equation}
where
\begin{equation}
    \psi_{J}^{(\mathrm{vib})} = \prod_{i=1}^{N_{\mathrm{vib}}} \sm(Q),
\end{equation}
and $J$ is a collective index given by
\begin{multline}
  J=  ((n_{11},n_{12},\dots, n_{{1N_{\mathrm{SMB}}}}), (n_{21},n_{22},\dots, n_{2N_{\mathrm{SMB}}}), \\
  \ldots (n_{N_{\mathrm{vib}}1}, n_{N_{\mathrm{vib}}2}, \ldots, n_{N_{\mathrm{vib}} N_{\mathrm{SMB}}})),
\end{multline}
with $N_\mathrm{SMB}$ being the size of the SM basis set. Here, $n_{i\mu} = 1$ if the $\mu$th state of the $i$th mode is occupied, and $0$ otherwise. Hence, $J$ defines an occupation number vector for the vibrational states. We refer to each $|\psi_J^{(\mathrm{vib})}\rangle$ as a vibrational configuration. In VSCF, vibrational SMs are linear combinations of Hermite-Gaussian primitives,
\begin{equation} \label{eq:single_mode_vib}
    \sm(Q) = \sum_{\nu = 0}^{N_{hg}-1} c_{i\mu\nu}^{(n)}\ho(Q),
\end{equation}
where $N_{hg}$ is the number of Hermite-Gaussian primitives per normal mode, and $c_{i\mu\nu}^{(n)}$ are expansion coefficients. 

Occasionally, coupling between normal coordinates can be important in the initial state. Capturing this effect requires post-VSCF (pVSCF) methods, such as vibrational configuration interaction~\cite{thompson1980, christoffel1982, neff2009} and vibrational coupled cluster~\cite{seidler2007}. We express a pVSCF nuclear wavefunction as
\begin{equation} \label{eq:psi_nuc}
    \Psi^{(n)}_{\mathrm{pVSCF}}(\mathbf{Q}) = \sum_J C_J^{(n)}\psi_J^{(n)},
\end{equation}
where we define nuclear configurations as
\begin{equation}
\psi^{(n)}_J = \psi^{(\mathrm{vib})}_J \psi^{(\mathrm{tr})} \psi^{(\mathrm{rot})}.
\end{equation}
VHP, VSCF, and pVSCF are all valid nuclear states to be used in \cref{eq:psi_bo}. A nuclear state is fully specified by $C^{(n)}_J$, $c^{(n)}_{i\mu\nu}$, and $\Omega_{\nu_i} $, as well as $\Gamma_{\tau} $ and $\Upsilon_{\tau} $.

In treating nuclear states, we ignore nuclear spin because most chemical processes do not depend on nuclear spin. However, it could easily be included if desired, with identical nuclei being either symmetrized or antisymmetrized depending on their statistics.

\subsection{Non-separable state data} \label{sect:input_data_non_sep}
More general techniques can be used to find a non-separable state of the form in \cref{eq:non_BO_kspace} that is a solution to the full time-independent molecular Schr\"{o}dinger equation, $H \Psi = E\Psi$, where $H$ is given in \cref{eq:mol-ham}. 

We consider $\Psi$ as being decomposed into three layers of abstraction, shown in \cref{tab:input_data}. At the molecular layer, the wavefunction depends on all coordinates of both the electrons and nuclei and is expanded as a linear combination of products of many-body electronic and nuclear configurations. At the configuration layer, each electronic configuration is an antisymmetric product (Slater determinant) of single-body electronic MOs, while each nuclear configuration is a product of single-body SMs. Finally, at the MO/SM layer, MOs and SMs can be expanded, respectively, in finite sets of atom-centered Gaussian primitives or Hermite-Gaussian primitives.

For a non-separable state, the expansion coefficients $C_{IJ}$ are found by treating the nuclear and electronic components simultaneously. One such approach is the nuclear-electronic orbital (NEO) method~\cite{schiffer_2021}, which treats some low-mass nuclei---typically hydrogens---on the same footing as the electrons. The NEO wavefunction has the same form as \cref{eq:non_BO_kspace}, except that the nuclear configurations are in terms of Cartesian coordinates. In our approach, NEO would be used to treat all nuclei quantum mechanically, with the wavefunction expressed in normal coordinates. 

The parallelism between the descriptions of electronic and nuclear configurations in \cref{tab:input_data} means that methods for preparing electronic states on a quantum computer can also be used to prepare nuclear states with only minor modifications.

\section{Classical preprocessing} \label{sect:preproc}
Here, we show that MOs and SMs that comprise $\Psi$, defined in \cref{eq:non_BO_kspace} and specified by the input data, can be efficiently projected onto a PWB and subsequently encoded as an MPS, a form suitable for preparation on a quantum computer. We refer to this sequence of steps as preprocessing. This result extends to SMs the approach of Huggins et al.~\cite{huggins2024}, who showed that MOs can be efficiently projected onto a PWB and then encoded as an MPS. Compared with Grover–Rudolph state-preparation methods~\cite{grover_2002}, this approach avoids potentially costly coherent integration over basis functions.

For both electronic and nuclear states, we only need to preprocess the MOs/SMs. These preprocessed MOs/SMS will be used to construct configurations and linear combinations of configurations in \cref{sect:isp}. Because of this, preprocessing does not distinguish between separable and non-separable initial states. In either case, preprocessing is performed the same, and the distinction between separable and non-separable initial states realized during initial state preparation, \cref{sect:isp}.

\subsection{MO preprocessing} \label{sect:preproc_mo}
An MO can be efficiently projected onto a PWB and encoded as an MPS~\cite{huggins2024}. In what follows, we focus on the spatial orbitals because we show in \cref{sect:isp_elec} that preparing spin functions is trivial. We assume that a spatial orbital $\phi_a(\mathbf{r})$ is a linear combination of $N_g$ atom-centred Gaussian primitives, obtained via canonical orthogonalization (Appendix D of \cite{huggins2024}). 

Projecting $\phi_a$ onto a PWB amounts to projecting its constituent Gaussian primitives:
\begin{equation} \label{eq:3d-guass-proj}
    \GTOMom(\mathbf{r}) = \frac{1}{\GTONormAprx} \sum_{\mathbf{k} \in \mathbb{K}^3}  \left( \int_{\mathbb{R}^3}  \varphi_{\mathbf{k}}(\mathbf{r}' ) \GTO(\mathbf{r}') \,d\mathbf{r}' \right) \varphi_{\mathbf{k}}(\mathbf{r}),
\end{equation}
where $\GTO(\mathbf{r})$ is defined in \cref{eq:gauss_expand}, $\GTONormAprx$ are normalization constants, and the plane waves are 
\begin{equation} \label{eq:plane_wave}
\varphi_{\mathbf{k}}(\mathbf{r})= e^{-i k_x x} e^{-i k_y y} e^{-i k_z z} / L^{3 / 2},
\end{equation}
where $\mathbf{k} = (k_x,k_y,k_z)$ is the momentum vector, and $L^3$ is the simulation volume. We use a hat to denote states in a PWB (e.g., $\GTOMom$). The simulation grid along a given dimension in a PWB is defined as 
\begin{equation} \label{eq:K_lattice}
    \mathbb{K} = \left\{\frac{2 \pi p}{L}: p \in \mathbb{Z} \right\}.
\end{equation}
In \cref{eq:3d-guass-proj}, the integration bounds have been extended from $\pm L$ to $\pm \infty$ so that the integral becomes a Fourier transform. This approximation is justified if $\GTO(\mathbf{r})$ decay sufficiently before $\pm L$. We denote $\phi_a$ projected onto the PWB and expanded in terms of $\GTOMom$ as $\MOMom$.

To make the sum in \cref{eq:3d-guass-proj} tractable, we replace the sum over $\mathbb{K}$ by a sum over the finite set
\begin{equation} \label{eq:mom-cut-set}
\mathbb{K}_{\mathrm{cut}}=\{k \in \mathbb{K}:|k| \leq K^{(e)}_a\},
\end{equation}
where $K^{(e)}_a$ is a momentum cutoff. We denote the approximate Gaussian primitives projected onto this truncated PWB as $\GTOMomAprx$, and MOs expanded in terms of $\GTOMomAprx$ as $\MOMomAprx$. 

Lemma 3 of \cite{huggins2024} demonstrates that there exists a function $\MOPolyAprx$ that approximates $\MOMomAprx$ by expanding $\GTOMomAprx$ over a finite set of polynomials. The error between $\MOMom$ and $\MOPolyAprx$ is defined as 
\begin{equation}
    \epsilon^{(e,c)}_a = D(\MOMom, \MOPolyAprx),
\end{equation}
where, for pure states $|f\rangle$ and $|g\rangle$, the trace distance is
\begin{equation}
    D(|f\rangle,|g\rangle) = \sqrt{1- |\langle f | g \rangle|^2}.
\end{equation}
Here, $|f\rangle$ and $|g\rangle$ may both be either discrete or continuous pure states.

For a given $\delta^{(e,c)}_a\in ( 0,1)$, this error is bounded by~\cite{huggins2024}
\begin{equation}
\label{eq:elec_mps_bound}
    \epsilon^{(e,c)}_a \leq \delta^{(e,c)}_a,
\end{equation}
provided that the momentum cutoff satisfies
\begin{multline} \label{eq:Kelec}
K^{(e)}_a = 2 \sqrt{2 \gamma_{\mathrm{max}}} \\ \times \sqrt{2 \ln \Bigg(\frac{288 \sqrt{3} N_g}{\big(\delta^{(e,c)}_a\big)^4 \Sigma^2}\Bigg)+l_{\mathrm{max}} \ln (4 l_{\mathrm{max}})+\ln (45)},
\end{multline}
$\gamma_{\max } =\max_{\mu} \gamma_\mu,$ $l_{\mathrm{max}} = \max_{\mu} \{l_\mu, m_\mu, n_\mu\}$,
and $\Sigma$ is the eigenvalue cutoff used during canonical orthogonalization \cite{huggins2024}.
Asymptotically, 
\begin{equation} \label{eq:Kelec_asymptotic}
K^{(e)}_{a} = \mathcal{O}\bigg(\sqrt{\ln(1/\epsilon_{a}^{(e,c)})}\bigg).
\end{equation}

We can express $\MOPolyAprx$ as a quantum state by mapping $\varphi_\mathbf{k}$ to the computational basis:
\begin{equation}
    \varphi_\mathbf{k}(\mathbf{r}) \rightarrow \left|\tfrac{k_x L }{2 \pi} \right\rangle \left|\tfrac{k_y L }{2 \pi} \right\rangle \left|\tfrac{k_z L }{2 \pi} \right\rangle.
\end{equation}
Using $N^{(e)}_a = 2\lceil L K^{(e)}_a/2\pi \rceil + 1$ and $\mathbb{G}^{(e)}_a =[-N^{(e)}_a,  N^{(e)}_a ]^{3} \cap \mathbb{Z}^3$ we get 
\begin{equation}
    |\MOPolyAprx\rangle = \sum_{\mathbf{n}\in \mathbb{G}_{a}^{(e)}} \MOPolyAprx(\tfrac{2\pi\mathbf{n}}{L}) \left|n_1\right\rangle \left|n_2 \right\rangle \left|n_3 \right\rangle,
\end{equation}
where $\mathbf{n} = (n_1, n_2, n_3)$.

Finally, $|\MOPolyAprx\rangle$ can be expressed exactly as a $\lceil\log N_a^{(e)}\rceil$-qubit MPS, $|\MOMom\rangle_{\mathrm{MPS}}$ \cite{huggins2024}. An MPS $|\kappa\rangle$ has the form
\begin{equation} \label{eq:mps}
|\kappa \rangle = \sum_{s \in\{0,1\}^n,\{\alpha\}} A_{\alpha_1}^{s_1} A_{\alpha_1 \alpha_2}^{s_2} \cdots A_{\alpha_{n-1}}^{s_n}\left|s_1 s_2 \cdots s_n\right\rangle,
\end{equation}
where $A_{\alpha_i}^{s_{j}}$ and $A_{\alpha_i \alpha_{i+1}}^{s_{i+1}}$ are rank-2 and rank-3 tensors, respectively, and $\alpha_i$ are virtual indices. The dimension of $\alpha_i$ is called the bond dimension. For a given $\delta^{(e,c)}$, the maximum dimension $M_a$ of $|\MOMom\rangle_{\mathrm{MPS}}$ obeys \cite{huggins2024} 
\begin{multline} \label{eq:M_elec}
M_a \leq 8 e^2 N_g\bigg(2 \ln \bigg(\frac{288 \sqrt{3} N_g}{\big(\delta^{(e,c)}_a\big)^4 \Sigma^2}\bigg) \\ +l_{\mathrm{max}} \ln (4 l_{\mathrm{max}})+4\bigg),
\end{multline}
which implies that $|\MOMom\rangle_{\mathrm{MPS}}$ can be efficiently encoded. Additionally, numerical evidence from \cite{huggins2024} suggests that the scaling of the bond dimension with $N_g$ may be sublinear in practice for typical molecular systems. 

\subsection{SM preprocessing} \label{sect:preproc_sm}
We extend the MO preprocessing approach and prove that SMs can also be efficiently projected onto a PWB and encoded as an MPS. A given SM $\sm(Q)$ consists of a linear combination of $N_{hg}$ Hermite-Gaussian primitives, \cref{eq:ho_eig}. For notational simplicity, here $\sm(Q)$ represents vibrational, rotational, and translational states. In the case of translational and rotational states, $\mu$ is dropped, $i$ is substituted with $\tau$, and $Q$ is replaced with the appropriate coordinate in \cref{eq:trans_rot_coords}.

The projection of a given $\sm$ is
\begin{equation} \label{eq:project_pwb}
\hatsm(Q) =\frac{1}{\SmnormAprx} \sum_{k \in \mathbb{K}}\left(\int_{\mathbb{R}} \varphi_{k}^*(Q') \sm(Q') d Q'\right) \varphi_{k}(Q),
\end{equation}
where
\begin{equation}
    \varphi_k(Q) = \frac{e^{-ikQ}}{L^{1/2}} 
\end{equation}
is a 1D plane wave and $\SmnormAprx$ is a normalization constant. As in \cref{eq:3d-guass-proj}, extending the integration bounds to $\pm \infty$ is justified assuming $\sm$ is negligible beyond $\pm L$.

We approximate $\hatsm$ as a finite sum by replacing $\mathbb{K}$ with $\mathbb{K}_{\mathrm{cut}}$ of \cref{eq:mom-cut-set}, but with $K^{(e)}_a$ replaced with the nuclear momentum cutoff $\SmMom$. We denote these approximate states by $\HatSmAprx$ and define the error $\SmEpsT$ from truncating $\mathbb{K}$ to $\mathbb{K}_{\mathrm{cut}}$ as 
\begin{equation} \label{eq:nuc_trunc_error}
\SmEpsT = D(\hatsm, \HatSmAprx).
\end{equation}
According to \cref{lem:pwb} in \cref{app:pwb_proj_sm}, for a given $\SmDelT\in(0,1)$, the truncation error is bounded by 
\begin{equation}
    \SmEpsT \leq \SmDelT
\end{equation}
provided each $\SmnormAprx \geq 2/3$ and $\SmMom$ satisfies 
\begin{multline} \label{eq:K_nuc}
    \SmMom = \sqrt{2\omega}_i \\ \times \sqrt{2\ln \Big(\big(\SmDelT\big)^{-1}\Big) + \ln\left(\tfrac{1}{\sqrt{2}} + \tfrac{2\sqrt{\pi}}{L \sqrt \omega_i}\right) + F(N_{\mathrm{hg}})}
\end{multline}
and $F(N_{\mathrm{hg}}) = (N_{\mathrm{hg}}-1)\ln(4(N_{\mathrm{hg}} -1)) + 5.4$. Asymptotically,
\begin{equation}
    K_{i\mu}^{(n)}=\cO\Big(\sqrt{\ln(1/\epsilon_{i\mu}^{(n,t)})}\Big).
    \label{eq:Knuc_asymptotic}
\end{equation}

\Cref{lem:mps} in \cref{app:mps_encod_sm} proves the existence of $\SmPolyAprx$, a degree-$(\SmPolyDeg-1)$ polynomial approximation to $\HatSmAprx$, such that for a given $\SmDelP \in (0,1)$, the error
\begin{equation} \label{eq:nuc_poly_approx}
    \SmEpsP = D\left(\HatSmAprx, \SmPolyAprx \right)
\end{equation}
is bounded by 
\begin{equation}
    \SmEpsP \leq \SmDelP,
\end{equation}
provided the conditions in \cref{eq:Nuc_M_cond} are satisfied.

The combined error of approximating $\hatsm$ by projecting onto a truncated PWB and approximating this projection with a polynomial is 
\begin{equation} \label{eq:nuc_eps_classical}
    \SmEpsC =  D\big(\hatsm, \SmPolyAprx\big).
\end{equation}
We show in \cref{app:mps_encod_sm} that for a given $\SmDelC\in(0,1)$, setting $\SmDelP =\SmDelT = \SmDelC/2$ and using the triangle inequality leads to the upper bound
\begin{equation}
    \SmEpsC \leq \SmDelC,
\end{equation}
provided that 
\begin{equation} \label{eq:m_nuc_bound}
    \SmPolyDeg \geq \frac{e^2}{\omega_i}\big(\SmMom\big)^2
\end{equation}
and \cref{eq:K_nuc} is satisfied. In this case, \cref{eq:m_nuc_bound} replaces the conditions on $\SmPolyDeg$ given in \cref{eq:Nuc_M_cond}.

In the same way that we expressed $\MOPolyAprx$ as a quantum state, we can express $\HatSmAprx^{(\mathrm{poly})}$ as a quantum state by mapping $\varphi_k$ to the computational basis to get
\begin{equation}
    |\SmPolyAprx\rangle = \sum_{n \in \mathbb{G}_{i\mu}^{(n)}} \SmPolyAprx\bigg(\frac{2\pi n}{L}\bigg)\left| n \right\rangle,
    \label{eq:psi_poly_mps_nuc}
\end{equation}
where $\mathbb{G}_{i \mu}^{(n)} = [-N_{i \mu}^{(n)}, N_{i \mu}^{(n)}] \cap \mathbb{Z}$
and $N_{i \mu}^{(n)} = 2\lceil LK_{i\mu}^{(n)} /2\pi\rceil +1$. In \cref{app:twos_complement_mps}, we prove that, using a two's-complement representation to encode the computational basis in \cref{eq:psi_poly_mps_nuc}, $|\SmPolyAprx\rangle$ can be expressed exactly as a $\lceil\log N_{i\mu}^{(n)}\rceil$-qubit MPS with a maximum bond dimension of $2\SmPolyDeg + 4$. We denote the MPS representation as $|\hatsm\rangle_{\mathrm{MPS}}$ and note that $N_{\mathrm{SMB}}$ is the only term that scales superlogarithmically in $\SmPolyDeg$, which establishes that $\hatsm$ has an efficient MPS encoding.

SMs have the property that, under a coordinate rescaling, the bounds on $\SmMom$ and $\SmPolyDeg$ are adjusted by rescaling $\omega_i$. To see this, let $\bar{Q}= d^{-1/2} Q$, where $d >0$ and $\sm'$ denote $\sm$ in terms of $\bar{Q}$,
\begin{equation}
   \sm'(\bar{Q}) = \sm(Q) =  \sm(\sqrt{d} \bar{Q}).
\end{equation}
Because $\sm$ is a linear combination of Hermite-Gaussian primitives, $\sm(\sqrt{d} \bar{Q})$ is given by substituting $\omega_i \rightarrow \omega_i d$ in \cref{eq:ho_eig}. Thus, the bounds in \cref{lem:pwb,lem:mps} apply to $\sm'(\bar{Q})$ after replacing every instance of $\omega_i$ by $\omega_i d$.

\section{Initial-state preparation}
\label{sect:isp}

This section constructs a unitary $U_{\mathrm{ISP}}$ that prepares the initial state in the PWB,
\begin{equation}
    U_{\mathrm{ISP}}|0\rangle = |\hat{\Psi}\rangle,
\end{equation}
where $|\hat{\Psi}\rangle$ is composed of MOs and SMs projected onto the PWB as described in \cref{sect:preproc}. For separable states, the ISP unitary is the product
\begin{equation} \label{eq:Uinit_sep}
    U_{\mathrm{ISP}}^{(\mathrm{sep})} = U_{\mathrm{ISP}}^{(e)} U_{\mathrm{ISP}}^{(n)},
\end{equation}
of electronic and nuclear ISP unitaries,
\begin{equation}
    U_{\mathrm{ISP}}^{(e)}|0\rangle = |\hat{\Psi}^{(e)}\rangle \quad\, \mathrm{and}\quad\, U_{\mathrm{ISP}}^{(n)}|0\rangle = |\hat{\Psi}^{(n)}\rangle.
\end{equation}

In \cref{sect:grid}, we explain how we construct a common simulation grid to represent the electrons and nuclei in first quantization, which determines the memory requirements to store the molecular wavefunction.

In \cref{sect:isp_elec}, we construct $U_{\mathrm{ISP}}^{(e)}$ using the unitary-synthesis approach~\cite{huggins2024}. For multi-configurational states, we first prepare the second-quantized representation of the state using a combination of arbitrary state preparation \cite{low2024} and the sum-of-Slaters (SoSlat) method~\cite{fomichev2024}, then use this state to prepare $|\hat{\Psi}^{(e)}\rangle$.

For $U_{\mathrm{ISP}}^{(n)}$, we develop a new state-preparation approach in \cref{sect:isp_nuc} which solves the challenge of converting normal-coordinate expressions for initial nuclear states to a Cartesian representation suitable for time evolution. To prepare the normal-coordinate representation, we adapt the unitary-synthesis approach~\cite{huggins2024} that we use for electronic states. We then transform the wavefunction into Cartesian coordinates using a sequence of shears, reflections, and $\pm \pi/2$-rotations. 

We construct separable electronic and nuclear states, $U_{\mathrm{ISP}}^{(\mathrm{sep})}$, in \cref{sect:isp_sep}. 

Our approach to preparing separable states can be extended to prepare non-separable states with only minor modifications. In \cref{sect:Non-BO state prep}, we show that the cost to prepare non-separable states is comparable to the cost of preparing separable ones. 

In each section, we summarize the ancilla-qubit and Toffoli-gate resources needed to implement each ISP unitary. We also give the error that each unitary introduces. Finally, in \cref{sect:asymp_isp}, we determine the overall asymptotic complexity of ISP.

\subsection{Common grid}
\label{sect:grid}

We prepare every particle on the same cubic simulation grid, which means that the momentum cutoff $K_{\mathrm{max}}$ and the momentum grid spacing $\Delta=2\pi/L$ are the same for every electron and nucleus along each of the three Cartesian coordinates. 
We set $K_{\mathrm{max}}$ to be the largest among the momentum cutoffs derived in \cref{sect:preproc}, 
\begin{equation}
    K_{\mathrm{max}} = \max  \Big\{ \max_{a} K^{(e)}_{a} , \max_{i,\mu} K^{(n)}_{i \mu} \Big\},\label{eq:cutoff_prop}
\end{equation}
which ensures that all error bounds from \cref{sect:preproc} are simultaneously satisfied. For a given truncation error $\epsilon^{(t)}$, $K_{\mathrm{max}}$ inherits the scaling of  $K_{a}^{(e)}$ and $K_{i\mu}^{(n)}$, given in \cref{eq:Kelec_asymptotic,eq:Knuc_asymptotic}, i.e.,
\begin{equation}
    K_{\mathrm{max}} = \cO\Big(\sqrt{\log(1/\epsilon^{(t)})}\Big).
    \label{eq:k_max_asymp}
\end{equation}

We set $\Delta$ based on error bounds for performing normal-coordinate transformations. In \cref{eq:L_asymptotic}, we find
\begin{equation} 
    \Delta = \cO(\epsilon_{\mathrm{\mathrm{LCT}}}/\eta_n^2),
    \label{eq:delta_asymp}
\end{equation}
where $\epsilon_{\mathrm{LCT}}$ is the target error for a linear coordinate transformation, with the precise values given by the procedure in \cref{sect:LCT}.

Our choices of $K_{\mathrm{max}}$ and $\Delta$ involve the only approximations in our algorithm whose errors we do not bound, even though both are controllable. We do not bound them because of the difficulty of finding fully general bounds and because, in practice, the errors incurred are significantly smaller than the ones we do bound.

The first approximation is that $K_{\mathrm{max}}$ suffices to represent $|\hat{\Psi}(t)\rangle$ at all times. \Cref{eq:cutoff_prop} guarantees that our encoding of $|\hat{\Psi}(0)\rangle$ is accurate at $t=0$, but the momentum cutoff may become insufficient at longer times if $|\hat{\Psi}(t)\rangle$ explores higher-momentum regions of the Hilbert space. 
In general, it is expected that $K_{\max}=\cO(\epsilon^{-1/3})$ is needed to resolve arbitrary molecular wavefunctions~\cite{gruneis2013explicitly,hattig2011explicit,Harl2008,shepherd2012convergence,helgaker1997basis_b,klopper1995ab2,Halkier1998b}.
However, because we consider low-energy dynamics, we assume that $K_{\mathrm{max}}$ is large enough for this problem to be unimportant.
In the examples in \cref{sect:photochemistry}, the problem does not arise because we impose small ISP error targets, resulting in large $K_{\mathrm{max}}$ (i.e., small real-space grid spacing $\Delta L = \pi/K_{\mathrm{max}}$). As shown in \cref{table_resource}, ISP requires $\Delta L\sim\SI{0.01}{\bohr}$, in agreement with what is generally used in the field~\cite{ksenia2018}.
In any event, the approximation is controllable, and the error can be reduced arbitrarily by increasing $K_{\mathrm{max}}$ beyond \cref{eq:cutoff_prop}.

The second approximation arises from our use of the PWB, which implies periodic boundary conditions, resulting in spurious interactions between particles in the molecule and those in its periodic images. Therefore, we assume that $L$ is large enough that these interactions are negligible. In practice, as shown in \cref{sect:photochemistry}, for the normal-mode transformation to meet our strict error bounds already requires large $L$, typically hundreds of times the size of the molecule (see \cref{table_resource}), making the periodic interactions small. For neutral molecules, the interaction is dominated by the dipole-dipole term, which scales as $\mathcal{O}(L^{-3})$, meaning that an $L$ that is 100 times the size of the molecule results in a relative energy error of \num{e-6}, significantly smaller than other errors we consider. If that were not satisfactory, this approximation is also controllable, with the error arbitrarily reducible by increasing $L$, at polynomial cost.

Overall, $K_{\mathrm{max}}$ and $\Delta$ determine the qubit costs to store the molecular wavefunction. The number of plane waves per particle per spatial dimension must be at least
\begin{equation}
    \bar{N}= 2\left\lceil\frac{K_{\mathrm{max}}}{\Delta} \right\rceil + 1 = \mathcal{O}\bigg( \frac{\eta_n^2\sqrt{\log(1/\epsilon^{(t)})}}{\epsilon_{\mathrm{LCT}}}\bigg).
    \label{eq:N}
\end{equation}
Because the number of qubits must be an integer, we represent the grid using
\begin{equation}
    n_p = \lceil\log_2 \bar{N} \rceil
\end{equation}
qubits per particle per dimension. This increases the number of plane waves to 
\begin{equation}
    N = 2^{n_p} -1,
    \label{eq:N_defined}
\end{equation}
where one plane wave is removed in the signed-magnitude representation, which requires odd $N$.
The increase from \cref{eq:N} implies that $K_\mathrm{max}$, $\Delta$, or both must change. We keep $K_\mathrm{max}$ fixed and update $\Delta$ to $\Delta = 2 K_{\mathrm{max}}/(N-1)$.

The total number of qubits needed to represent $|\hat{\Psi}\rangle$ is therefore
\begin{equation}
    C_\data = 3\eta n_p + \eta_e,\label{eq:cdata}
\end{equation}
of which $3 n_p$ represent the spatial wavefunction of each of $\eta$ particles and an additional qubit represents the spin of each of $\eta_e$ electrons. The integer grid points are expressed using the signed-magnitude representation.

\subsection{Electronic ISP} \label{sect:isp_elec}
Given an initial electronic state $|\Psi^{(e)}\rangle$, specified as in \cref{sect:input_data_elec}, with MOs preprocessed as in \cref{sect:preproc_mo}, we can construct a unitary $\tilde{U}_{\mathrm{ISP}}^{(e)}$ that prepares the approximate state $|\tilde{\hat{\Psi}}^{(e)}\rangle_{\mathrm{MPS}}$ on a quantum computer \cite{huggins2024}. Our approach consists of five main steps. 

\subsubsection{Arbitrary state preparation}
Let $ \textsc{asp}^{(e)}$ be a unitary that implements arbitrary state preparation such that
\begin{equation} \label{eq:asp_psi_e_exact}
    \textsc{asp}^{(e)}|0\rangle = |\psi\rangle =\sum_{i=0}^{D^{(e)}-1} C^{(e)}_i|i\rangle,
\end{equation}
where $C_{i}^{(e)}$ is the $i$th coefficient of the set $\{C_I^{e)} \}$, $D^{(e)}$ is the number of electronic configurations, and $|0\rangle$ is the all-zeros state.  We implement an approximation $\widetilde{\textsc{asp}}^{(e)}$ to $\textsc{asp}^{(e)}$ using the method in \cite{low2024}, which prepares
\begin{equation} \label{eq:asp_psi_e_approx}
    \widetilde{\textsc{asp}}^{(e)}|0\rangle = |\tilde{\psi}\rangle=\sum_{i=0}^{ D^{(e)} -1} \widetilde{C}^{(e)}_i|i\rangle,
\end{equation}
where $\widetilde{C}^{(e)}_i$ is an approximation to $C^{(e)}_i$. The resulting state can be represented using $\lceil\log D^{(e)} \rceil$ qubits. We give the number of ancilla qubits and Toffoli gates used by a subroutine to implement a given unitary by the functions $C_{\mathrm{anc}}(\cdot)$ and $C_{\mathrm{Toff}}(\cdot)$, respectively. The costs to implement $\widetilde{\textsc{asp}}^{(e)}$ are \cite{huggins2024}
\begin{align}
     C_{\mathrm{anc}}(\widetilde{\textsc{asp}}^{(e)}) & = \tfrac12 {\log(D^{(e)})} + \frac{D^{(e)}b_{\textsc{asp}}^{(e)}}{4(b_{\textsc{asp}}^{(e)}+1)^{1/2}} \nonumber \\
     &\quad + \tfrac{1}{2}\log(b_{\textsc{asp}}^{(e)} + 1) 
     + 3b_{\textsc{asp}}^{(e)} - 4   \label{eq:canc_asp} \\
    C_{\mathrm{Toff}}(\widetilde{\textsc{asp}}^{(e)}) & = 2^{5/2}(1 + \sqrt{2}) D^{(e)}(b_{\textsc{asp}}^{(e)} + 1)^{1/2} \nonumber \\ 
    & \quad + 2 \log(D^{(e)})(b_{\textsc{asp}}^{(e)} - 4), \label{eq:ctoff_asp}
\end{align}
where $b_{\textsc{asp}}^{(e)}$ is the number of precision qubits used by the rotation multiplexor in the ASP algorithm~\cite{huggins2024}.

\subsubsection{ONB preparation}
The next step is to convert $|\tilde{\psi}\rangle$ into the ONB. Let 
\begin{equation} \label{eq:psi_2nd}
|\tilde{\Psi}^{(e)}\rangle_{\mathrm{ONB}} =\sum_{I} \widetilde{C}_{I}^{(e)}|I\rangle,
\end{equation}
be an approximate ONB representation of $|\Psi^{(e)}\rangle$, where the occupation number vector $I$ is defined in \cref{sect:input_data_elec} and each $|I\rangle$ encodes a state in the ONB. We prepare $|\Psi^{(e)}\rangle_{\mathrm{ONB}}$ using the SoSlat algorithm~\cite{fomichev2024}. Let $\mathrm{SoSlat}^{(e)}$ be a unitary that implements the SoSlat algorithm,
\begin{equation}
    \mathrm{SoSlat}^{(e)}|\tilde{\psi}\rangle = |\tilde{\Psi}^{(e)}\rangle_{\mathrm{ONB}}.
\end{equation}
The resulting state is represented using $N_{\mathrm{MOB}}$ qubits. To prepare $|\tilde{\Psi}^{(e)}\rangle_{\mathrm{ONB}}$ using $\mathrm{SoSlat}^{(e)}$, we require \cite{fomichev2024}
\begin{align} \label{eq:prep_2nd_t}
     C_{\mathrm{anc}}(\mathrm{SoSlat}^{(e)})  &= 5\log D^{(e)} - 3\\
    C_{\mathrm{Toff}}(\mathrm{SoSlat}^{(e)}) &\leq D^{(e)}(2 \log D^{(e)} + 3).
\end{align}

\subsubsection{ONB to MOB transformation}
Next, we convert $|\tilde{\Psi}^{(e)}\rangle_{\mathrm{ONB}}$ into the MO basis (MOB) representation $|\tilde{\Psi}^{(e)}\rangle_{\mathrm{MOB}}$ using the unitary $\textsc{onb2mo}$ that performs
\begin{equation}\label{eq:onb2mob}
    \textsc{onb2mob}|n_1,n_2\dots\rangle = |a\rangle_1|b \rangle_2 \dots |c\rangle_{\eta_e}.
\end{equation}
In the MOB, each electron has its own register of size $\lceil\log N_{\mathrm{MOB}} \rceil$, so that the register $|a\rangle_{i}$ indicates that the $i$th electron occupies the $a$th MO. Then, $|\tilde{\Psi}^{(e)}\rangle_{\mathrm{MOB}}$ is expressed as 
\begin{equation} \label{eq:elec_mo_basis}
    |\tilde{\Psi}^{(e)}\rangle_{\mathrm{MOB}} = \sum_{a,b,\dots ,c=1}^{N_{\mathrm{MOB}}}\widetilde{C}_{ab\ldots c}^{(e)}|a\rangle_1|b\rangle_2\ldots|c\rangle_{\eta_e},
\end{equation}
where $\widetilde{C}_{ab\ldots c}^{(e)}$ are the components of $\widetilde{C}_{I}^{(e)}$ in the MOB. Representing $|\tilde{\Psi}^{(e)}\rangle_{\mathrm{MOB}}$ uses $\eta_e \lceil \log(N_{\mathrm{MOB}}) \rceil$ qubits.

$\textsc{onb2mob}$ can be implemented using the approach in Appendix G of \cite{Babbush_2023}, which requires
\begin{align}\label{eq:cost_onb2mob}
C_{\mathrm{anc}}(\textsc{onb2mob}) &=N_{\mathrm{MOB}} + 3\lceil \log(\eta_e + 1) \rceil \\
C_{\mathrm{Toff}}(\textsc{onb2mob}) &= N_{\mathrm{MOB}} (2 \eta_e  + \lceil \log (\eta_e + 1) \rceil  \nonumber\\
&\quad + \eta_e \lceil \mathrm{log}(N_{\mathrm{MOB}})\rceil -4 ).
\end{align}

\subsubsection{Antisymmetrization}
We then define $\textsc{asym}$ as the unitary that antisymmetrizes the overall electronic state,
\begin{multline}
\textsc{asym}|a\rangle_1|b\rangle_2 \cdots |c\rangle_{\eta_e} = \\\frac{1}{\sqrt{\eta_e!}} \sum_{\pi \in S_{\eta_e}}(-1)^{|\pi|}|\pi(a)\rangle_1|\pi(b)\rangle_2 \cdots|\pi(c)\rangle_{\eta_e},
\end{multline}
where $S_{\eta_e}$ is the set of all permutations $\pi$ of MO labels. The resulting state is stored using $n_p + 1$ qubits per electron, of which 1 qubit encodes the electron's spin, and the remaining $n_p$ qubits encode the $\mathrm{MOB}$.

After applying $\textsc{asym}$, the state is
\begin{multline}
    |\tilde{\Psi}^{(e)}\rangle_{\textsc{asym}} =
\frac{1}{\sqrt{\eta_e!}} \sum_{a,b,\dots, c=1}^{N_{\mathrm{MOB}}} \sum_{\pi \in S_{\eta_e}}(-1)^{|\pi|} \widetilde{C}_{ab\dots c}^{(e)} \times\\ |\pi(a)\rangle_1|\pi(b)\rangle_2 \cdots|\pi(c)\rangle_{\eta_e}.
\end{multline}

Using the method based on sorting networks~\cite{berry2018}, which improves on previous antisymmetrization algorithms~\cite{abrams_1997}, $\textsc{asym}$ can be implemented using
\begin{align}\label{eq:Antisym}
    C_{\mathrm{anc}}(\textsc{asym}) &=\eta_e \log  \eta_e+\tfrac14 {\bar{n}_p}{\log \bar{n}_p(1+\log \bar{n}_p)}  \nonumber\\
    &\quad +2 ( \eta_e-1)\\
    C_{\mathrm{Toff}}(\textsc{asym}) &= 2(\eta_e-1)(\log \bar{n}_p + 1) + {}\nonumber\\
     \tfrac14 \bar{n}_p  \log \bar{n}_p & (1 + \log  \bar{n}_p ) (6\log \bar{n}_p + n_p + 1),
\end{align}
where $\bar{n}_p = 2^{\lceil \log  (n_p+1) \rceil}$.

\subsubsection{MOB to PWB transformation}
The final step is to transform each register from the MOB to the PWB. This transformation is achieved by implementing the unitary \cite{huggins2024}   
\begin{equation} \label{eq:reflect}
    W^{(e)} = |0\rangle \langle 1| \otimes V^{(e)}_{\mathrm{MPS}}  + |1\rangle \langle 0| \otimes \big(V^{(e)}_{\mathrm{MPS}}\big)^{\dagger},
\end{equation}
where
\begin{equation}
    V^{(e)}_{\mathrm{MPS}} = \sum_{a=1}^{N_{\mathrm{MOB}}}|\sigma_a\rangle|\hat{\phi}_{a}\rangle_{\mathrm{MPS}}\,\langle a-1|.
\end{equation}
Here, $|\sigma_a\rangle$ encodes the spin state of the $a$th MO, and $|\hat{\phi}_{a}\rangle_{\mathrm{MPS}}$ are the MPS encodings of \cref{sect:preproc_mo}. The spin states are mapped to qubits by
\begin{equation}
    |\alpha\rangle \rightarrow |0\rangle, \quad
    |\beta\rangle \rightarrow |1\rangle.
\end{equation}
When $W^{(e)}$ operates on an electron's data register with an ancilla set to $|1\rangle$, it performs the transformation 
\begin{equation}
    W^{(e)}|1\rangle|a-1\rangle = |1\rangle|\sigma_a\rangle|\MOMom\rangle_{\mathrm{MPS}},
\end{equation}
which gives the $a$th spin orbital in a PWB representation. To transform the entire electronic state, we apply $(W^{(e)})^{\otimes\eta_e}$, which applies $W^{(e)}$ to all $\eta_e$ registers. 

We can implement an approximation $\widetilde{W}^{(e)}$ to $W^{(e)}$ using the unitary synthesis approach~\cite{low2024, kliuchnikov2013}, where $\widetilde{W}^{(e)}$ is given as a product of reflections about the states
\begin{equation}
\left|\tilde{w}_a^{(e)}\right\rangle  =\frac{1}{\sqrt{2}}\left(|1\rangle|a-1\rangle-|0\rangle |\sigma_a\rangle \ket{\MOMomAprx}_{\mathrm{MPS}}\right)
\end{equation}
where $\ket{\MOMomAprx}_{\mathrm{MPS}}$ is an approximation to $\ket{\hat{\phi}_{a}}_{\mathrm{MPS}}$. To prepare $|\tilde{w}_a^{(e)}\rangle$, we need a method to prepare the spin state and spatial orbital. The spin state is trivially prepared by applying a $\textsc{not}$ gate to the spin register to prepare $|\beta\rangle$, and doing nothing to prepare $|\alpha\rangle$. The spatial orbital is prepared using the MPS state preparation approach~\cite{fomichev2024}.
Then, $\widetilde{W}^{(e)}$ is given by
\begin{equation} \label{eq:prod_ref}
\widetilde{W}^{(e)} =\prod_{a=1}^{N_{\mathrm{MOB}}} \big( I-2|\tilde{w}_a^{(e)}\langle \tilde{w}_a^{(e)}| \big).
\end{equation}
Explicit circuits to implement \cref{eq:prod_ref} are given in \cite{huggins2024}.

To implement $(\widetilde{W}^{(e)})^{\otimes\eta_e}$ (including the cost to prepare $|\tilde{w}_a^{(e)}\rangle$), we have
\begin{multline}
    C_{\mathrm{anc}}((\widetilde{W}^{(e)})^{\otimes\eta_e}) <  N_{\mathrm{MOB}} n_p+ \tfrac{1}{2}\log(m_{ai_\mathrm{max}}^{(e)})  \\ +\frac{2b_{\mathrm{rot}}^{(e)}(m_{ai_\mathrm{max}}^{(e)})^{1/2}}{(b_{\mathrm{rot}}^{(e)}+1)^{1/2}} + \tfrac{1}{2} \log (b_{\mathrm{rot}}^{(e)}+1)+3 b_{\mathrm{rot}}^{(e)},
\end{multline}
where $m_{ai_\mathrm{max}}^{(e)} = \max_{1\leq i \leq n_p} m_{ai}^{(e)}$,  $m_{ai}^{(e)}$ is the bond dimension of the $i$th tensor of the MPS encoding of the $a$th MO, and $b_{\mathrm{rot}}^{(e)}$ is the number of precision qubits used by the rotation multiplexor subroutine during unitary synthesis \cite{huggins2024}. The number of Toffoli gates is bounded by~\cite{huggins2024}
\begin{multline}\label{eq:Elec_Gates}
C_{\mathrm{Toff}}((\widetilde{W}^{(e)})^{\otimes\eta_e})<\eta_e N_{\mathrm{MOB}} n_p \\+2 \eta_e \sum_{a=1}^{N_{\mathrm{MOB}}} \sum_{i=1}^{n_p}\Big( 32(1+\sqrt{2})(b_{\mathrm{rot}}^{(e)}+1)^{\frac{1}{2}} m_{ai}^{(e)} \big(\bar{m}_{ai}^{(e)}\big)^{\frac{1}{2}}\\
+(8 b_{\mathrm{rot}}^{(e)}-15) m_{ai}^{(e)} \log \big(2 \bar{m}_{ai}^{(e)}\big)\Big),
\end{multline}
where, $\bar{m}_{ai}^{(e)}$ is
\begin{equation}
\bar{m}_{ai}^{(e)} = \max \bigg\{ 2^{\left\lceil\log m_{a(i-1)}^{(e)}\right\rceil}, 2^{\left\lceil\log m_{ai}^{(e)}\right\rceil} \bigg\},
\end{equation}

\subsubsection{Error analysis}
The unitary that prepares  $|\tilde{\hat{\Psi}}^{(e)}\rangle_{\mathrm{MPS}}$ is
\begin{equation}
    \tilde{U}_{\mathrm{ISP}}^{(e)}  = (\widetilde{W}^{(e)})^{\otimes3\eta_n} \textsc{asym} \; \textsc{onb2mob} \; \mathrm{SoSlat}^{(e)} \widetilde{\textsc{asp}} \label{eq:approx_U_init_e_total}.
\end{equation}
The error $\epsilon_{\mathrm{ISP}}^{(e)}$ of preparing $|\tilde{\hat{\Psi}}^{(e)}\rangle_{\mathrm{MPS}}$ consists of two terms~\cite{huggins2024},
\begin{equation} \label{eq:elec_error_bound}
    \epsilon^{(e)}_{\mathrm{ISP}} \leq  \epsilon_{\textsc{asp}}^{(e)} + 2^{3/2} \eta_e \sum_{a=1}^{N_{\mathrm{MOB}}} \epsilon_a^{(e)}.
\end{equation}
The first term is the ASP error
\begin{equation}
   \epsilon_{\textsc{asp}}^{(e)} = D( |\psi\rangle ,| \tilde{\psi}\rangle),
\end{equation}
with $|\psi\rangle$ and $|\tilde{\psi}\rangle$ defined in \cref{eq:asp_psi_e_exact,eq:asp_psi_e_approx}, which is bounded by \cite{huggins2024}
\begin{equation}
    \epsilon_{\textsc{asp}}^{(e)} \leq 2\pi \,2^{-b^{(e)}_{\textsc{asp}}}  \log(D^{(e)}).
\end{equation}
The second is the error from $(\widetilde{W}^{(e)})^{\otimes3\eta_n}$, 
\begin{align} \label{eq:mo_mps_bound_final}
    \epsilon^{(e)}_a &= D\big(|\MOMom\rangle, |\MOMomAprx\rangle_{\mathrm{MPS}}\big) \leq \epsilon_a^{(e,c)} + \epsilon_a^{(e,q)},
\end{align}
where $\epsilon_a^{(e,c)}$ results from MO preprocessing as given in \cref{sect:preproc_mo} and $\epsilon_a^{(e,q)}$ results from unitary synthesis~\cite{huggins2024},
\begin{align} \label{eq:rot_error_bound}
   \epsilon^{(e,q)}_a &= D\big(|\MOMom\rangle_{\mathrm{MPS}}, |\MOMomAprx\rangle_{\mathrm{MPS}} \big)\\
 &<  2^{\frac{7}{2}- b_{\mathrm{rot}}^{(e)} } \sum_{i=1}^{n_p} m_{ai}^{(e)} \log \big(2 \bar{m}_{ai}^{(e)}\big).
\end{align}

If $|\hat{\Psi}^{(e)}\rangle$ is a single Slater determinant and the overlap between $|\MOMom\rangle_{\mathrm{MPS}}$ and $|\MOMomAprx\rangle_{\mathrm{MPS}}$ for $a\neq b$ is small, then $\epsilon^{(e)}_{\mathrm{ISP}}$ is upper bounded by \cite{huggins2024}
\begin{equation} \label{eq:approx_bound}
    \epsilon^{(e)}_{\mathrm{ISP}} \leq \epsilon_{\textsc{asp}}^{(e)} + \sum_{a=1}^{N_{\mathrm{MOB}}} \epsilon^{(e)}_a.
\end{equation}

\subsection{Nuclear ISP} \label{sect:isp_nuc}
We give a method to prepare $|\hat{\Psi}^{(n)}\rangle$ in the Cartesian PWB. We show that $|\hat{\Psi}^{(n)}\rangle$ can be prepared as a linear combination of nuclear configurations in terms of rescaled normal coordinates, then transformed into the Cartesian PWB on the quantum computer. This is useful because linear combinations of configurations are straightforward to prepare using the unitary synthesis approach used in \cref{sect:isp_elec}. 

In our nuclear ISP, the size of the simulation grid changes twice before we arrive at the final grid of size $N=2^{n_p}-1$ per dimension (specified in \cref{sect:grid}). Initially, we prepare $|\hat{\Psi}^{(n)}\rangle$ as uncoupled normal coordinates on a grid of size $N_{\mathrm{ISP}} = 2^{n_{\mathrm{ISP}}}$, where\begin{equation}
    n_{\mathrm{ISP}} =\Big\lceil \log \big( \max_{i,\mu} 2\lceil K_{i\mu}^{(n)}/\Delta  \rceil + 1 \big)\Big\rceil
\end{equation}
is set by the precision of the nuclear pre-processing (\cref{eq:K_nuc}). When converting this state to Cartesian coordinates using the normal-mode transformation, numerical reasons given in \cref{sect:LCT} require a larger grid of $\bar{N}_{\mathrm{ISP}} = 2^{\bar{n}_{\mathrm{ISP}}}$ states, where $\bar{n}_{\mathrm{ISP}} = n_{\mathrm{ISP}} + n_{\mathrm{pad}}$, where the number $n_{\mathrm{pad}}$ of padding qubits is given in \cref{sect:LCT}. To accommodate this larger grid, we pre-allocate $\bar{n}_{\mathrm{ISP}}$ qubits per dimension at the outset of the algorithm, and maintain them throughout. However, after the normal-mode transformation, there is little amplitude in the exterior of the grid, and we show that it can be neglected in subsequent time-evolution with minimal error. This trimming of the exterior grid produces the final grid of size $N=2^{n_p}-1$. 

\subsubsection{Constructing \texorpdfstring{$U_{\mathrm{ISP}}^{(n)}$}{U init n}}
\label{sect:ISP-constructingU}
$|\hat{\Psi}^{(n)}\rangle$ can be prepared in the Cartesian PWB by
\begin{equation}
    U_{\mathrm{ISP}}^{(n)} = \textsc{tc2sm}^{\otimes 3\eta_{n}}_{\bar{n}_{\isp}} \cdot\textsc{nct} \cdot U_{\mathrm{PWB}}, \label{eq:U_ISP_def_n}
\end{equation}
where $ U_{\mathrm{PWB}}$ prepares a linear combinations of configurations in terms of rescaled normal coordinates in the PWB, $\textsc{nct}$ implements a normal-coordinate transformation into Cartesian coordinates, and $\textsc{tc2sm}^{\otimes 3\eta_{n}}_{\bar{n}_{\isp}}$  converts integer grid points from the two's-complement representation used in state preparation to the signed-magnitude representation used in time propagation. 

This decomposition is obtained by rewriting $\Psi^{(n)}(\mathbf{Q})$ using the fact that $\mathbf{Q} = \mathbf{O}\mathbf{M}^{1/2}(\mathbf{R} - \mathbf{R}^{(0)})$,
\begin{equation} 
    \Psi^{(n)}(\mathbf{Q}) = \Psi^{(n)}(\mathbf{O} \mathbf{M}^{1/2}\mathbf{R} - \mathbf{O} \mathbf{M}^{1/2}\mathbf{R}^{(0)}).
\end{equation}
Using an RQ decomposition \cite{golub2013}, we rewrite $\mathbf{O} \mathbf{M}^{1/2}$ as
\begin{equation} \label{eq:RQ_decomp}
    \mathbf{O} \mathbf{M}^{1/2} = \mathbf{d} \mathbf{S}_{\mathrm{U}} \bar{\mathbf{O}} = \mathbf{d} \mathbf{A},
\end{equation}
where $\mathbf{A} = \mathbf{S}_{\mathrm{U}} \bar{\mathbf{O}}$, $\mathbf{d}$ is a $3\eta_n \times 3\eta_n $ diagonal matrix with real, positive values and the same units as $\mathbf{M}^{1/2}$, $\mathbf{S}_{\mathrm{U}} $ is a $3\eta_n \times 3\eta_n $ shear (specifically, a unit upper-triangular matrix), and $\bar{\mathbf{O}}$ is a $3\eta_n \times 3\eta_n $ orthogonal matrix. \Cref{eq:RQ_decomp} is the standard RQ decomposition, except that we factored out a diagonal matrix $\mathbf{d}$ so that $\mathbf{S}_{\mathrm{U}}$ has diagonal elements equal to one. We denote the rescaled normal coordinates as $\bar{\mathbf{Q}} = \mathbf{d}^{-1} \mathbf{Q} $ and note that
\begin{equation} \label{eq:Qbar_def}
    \bar{\mathbf{Q}} = \mathbf{A} \mathbf{R} - \bar{\mathbf{R}}^{(0)}, \quad\, \mathrm{where} \quad\, \bar{\mathbf{R}}^{(0)} = \mathbf{A}\mathbf{R}^{(0)},
\end{equation}
and $\bar{\mathbf{Q}}$, $\mathbf{R}$, and $\bar{\mathbf{R}}^{(0)}$ all have units of length.

We use $\Psi^{(n)}_{\bar{\mathbf{Q}}}$ to denote the nuclear wavefunction given in terms of $\bar{\mathbf{Q}}$ such that 
\begin{equation}\label{eq:PsiQBar}
    \Psi^{(n)}_{\bar{\mathbf{Q}}}(\bar{\mathbf{Q}}) = \Psi^{(n)}(\mathbf{Q}) = \mathcal{N}_{\mathbf{d}} \Psi^{(n)}(\mathbf{d}\bar{\mathbf{Q}})
\end{equation}
where $\mathcal{N}_{\mathbf{d}}$ is a normalization constant given by $\mathcal{N}_{\mathbf{d}} = (\prod_{i} d_{ii})^{1/2}$. Importantly, we see that because $\Psi^{(n)}(\mathbf{Q})$ is expressed as a linear combination of nuclear configurations, $\Psi^{(n)}_{\bar{\mathbf{Q}}}(\bar{\mathbf{Q}})$ is also expressed as a linear combination of nuclear configurations. We also define $\Psi_{\mathbf{R}}^{(n)}$ as the nuclear wavefunction given in terms of $\mathbf{R}$,
\begin{equation} \label{eq:Psi_R_prop}
    \Psi_{\mathbf{R}}^{(n)}(\mathbf{R}) = \Psi^{(n)}_{\bar{\mathbf{Q}}}(\mathbf{A}\mathbf{R} - \bar{\mathbf{R}}^{(0)}),
\end{equation}
which is normalized because $\det(\mathbf{A}) =  1$.

Next, we project both sides of \cref{eq:Psi_R_prop} onto the PWB. We use $\hat{\Psi}_{\mathbf{R}}^{(n)}(\mathbf{k_R})$ to denote $\Psi_{\mathbf{R}}^{(n)}(\mathbf{R})$ in the PWB, where $\mathbf{k_R}$ is the conjugate variable to $\mathbf{R}$. To find the PWB representation of $\hat{\Psi}^{(n)}_{\bar{\mathbf{Q}}}(\mathbf{A}\mathbf{R} - \bar{\mathbf{R}}^{(0)})$, we take the Fourier transform of $\Psi^{(n)}_{\bar{\mathbf{Q}}}$. Using properties of affine transformations under Fourier transforms \cite{bracewell_1993}, we have
\begin{align}
      \hat{\Psi}_{\mathbf{R}}^{(n)}(\mathbf{k_R})
      &= \int d\mathbf{R} \, \varphi_{\mathbf{k}_{\mathbf{R}}} \Psi^{(n)}_{\bar{\mathbf{Q}}}(\mathbf{A}\mathbf{R} - \bar{\mathbf{R}}^{(0)})  \\
      & = e^{-i\mathbf{k_R}^{\top} \mathbf{A}^{-1} \bar{\mathbf{R}}^{(0)} } \hat{\Psi}^{(n)}_{\bar{\mathbf{Q}}}(\mathbf{A}^{- \top} \mathbf{k_R}) \\
      &=  e^{-i  \mathbf{k_R}^{\top} \mathbf{R}^{(0)}}\hat{\Psi}^{(n)}_{\bar{\mathbf{Q}}}(\mathbf{A}^{- \top} \mathbf{k_R}),
\end{align}
where $\mathbf{A}^{- \top} = (\mathbf{A}^{-1})^{\top}$, and
\begin{equation}
    \varphi_{\mathbf{k}_{\mathbf{R}}} =\frac{e^{-i \mathbf{k}_{\mathbf{R}}^{\top} \mathbf{R}} }{ L^{3\eta_n/2}}.
\end{equation}
After mapping $\Psi_{\mathbf{R}}^{(n)}$ to the computational basis as in \cref{sect:preproc}, we get
\begin{equation}
    |\hat{\Psi}_{\mathbf{R}}^{(n)}\rangle = \sum_{\mathbf{n}\in \mathbb{Z}^{3\eta_n}}  e^{-i \Delta \mathbf{n}^{\top} \mathbf{R}^{(0)} } \, \hat{\Psi}^{(n)}_{\bar{\mathbf{Q}}}(\Delta  \mathbf{A}^{- \top}\mathbf{n})|\mathbf{n}\rangle,
\end{equation}
where $\Delta = 2 \pi/L$ and $\mathbf{k}_{\mathbf{R}}=\Delta \mathbf{n}$.

We can now more precisely define the unitaries in \cref{eq:U_ISP_def_n}. $U_{\mathrm{PWB}}$ prepares a normal-coordinate PWB representation of the nuclear state,
\begin{equation}
U_{\mathrm{PWB}}|0\rangle = |\hat{\Psi}^{(n)}_{\bar{\mathbf{Q}}}\rangle  = \sum_{\mathbf{n}\in \mathbb{Z}^{3\eta_n}}  \hat{\Psi}^{(n)}_{\bar{\mathbf{Q}}}(\Delta  \mathbf{n})|\mathbf{n}\rangle.
\end{equation}
Then, $\textsc{nct}$ transforms the normal coordinates to Cartesian coordinates, taking $|\hat{\Psi}^{(n)}_{\bar{\mathbf{Q}}}\rangle$ and preparing $|\hat{\Psi}_{\mathbf{R}}^{(n)}\rangle$,
\begin{equation} \label{eq:Unm_def}
    \textsc{nct} |\hat{\Psi}^{(n)}_{\bar{\mathbf{Q}}}\rangle =  \sum_{\mathbf{n}\in \mathbb{Z}^{3\eta_n}} e^{-i \Delta \mathbf{n}^{\top} \mathbf{R}^{(0)} } \hat{\Psi}^{(n)}_{\bar{\mathbf{Q}}}(\Delta \mathbf{A}^{- \top} \mathbf{n})| \mathbf{n}\rangle.
\end{equation}
We further decompose $\textsc{nct} = \textsc{pk} \cdot U_{\mathbf{A}^{- \top}}$, where $U_{\mathbf{A}^{- \top}}$ linearly transforms the wavefunction coordinates and the phase kickback $\textsc{pk}$ applies the phase $e^{-i \Delta \mathbf{n}^{\top} \mathbf{R}^{(0)}}$. Overall,
\begin{equation}
    U_{\mathrm{ISP}}^{(n)} = \textsc{tc2sm}^{\otimes 3\eta_{n}}_{\bar{n}_{\isp}} \cdot\textsc{pk} \cdot U_{\mathbf{A}^{- \top}} \cdot U_{\mathrm{PWB}}.
\end{equation}

\subsubsection{Implementing \texorpdfstring{$U_{\mathrm{PWB}}$}{UPWB}}
\label{sect:UPWB}
$U_{\mathrm{PWB}}$ prepares $|\hat{\Psi}^{(n)}_{\bar{\mathbf{Q}}}\rangle$ in three steps,
\begin{equation}
    U_{\mathrm{PWB}} = W^{(n)} \cdot \textsc{onb2smb} \cdot \textsc{asp}^{(n)}.
\end{equation}
An arbitrary state in the ONB is prepared by $\textsc{asp}^{(n)}$, which is first transformed to the SMB by $\textsc{onb2smb}$ and then to the PWB by $W^{(n)}$. Because $|\hat{\Psi}^{(n)}_{\bar{\mathbf{Q}}}\rangle$ consists of a linear combination of configurations, we can prepare an approximation $|\tilde{\hat{\Psi}}_{\bar{\mathbf{Q}}}^{(n)}\rangle_{\mathrm{MPS}}$ using the approach from \cref{sect:isp_elec} that prepares an approximation to $|\hat{\Psi}^{(e)}\rangle$. 

As we did for electronic ISP, we begin by using arbitrary state preparation to prepare the state
\begin{equation} \label{eq:nuc_asp}
    \textsc{asp}^{(n)}|0\rangle = |\psi\rangle = \sum_{i=0}^{D^{(n)}-1} C_i^{(n)}|i\rangle,
\end{equation}
where $\textsc{asp}^{(n)}$ is a unitary that implements arbitrary state preparation, $C_i^{(n)}$ is the $i$th coefficient of the set $C_J^{(n)}$, and $D^{(n)}$ is the number of nuclear configurations. Using the method described in \cite{low2024}, we prepare the approximate state 
\begin{equation} \label{eq:nuc_asp_approx}
    \widetilde{\textsc{asp}}^{(n)}|0\rangle = |\tilde{\psi}\rangle = \sum_{i=0}^{D^{(n)}-1} \widetilde{C}_i^{(n)}|i\rangle,
\end{equation}
where $\widetilde{C}_i^{(n)}$ is an approximation to $C_i^{(n)}$ produced by $\widetilde{\textsc{asp}}^{(n)}$. The resources needed for $\widetilde{\textsc{asp}}^{(n)}$ are the same as in \cref{eq:canc_asp,eq:ctoff_asp}, but with $D^{(e)}$ replaced by $D^{(n)}$. Representing $|\tilde{\psi}\rangle$ uses $\lceil \log(D^{(n)}) \rceil$ qubits.

Let $|\tilde{\Psi}^{(n)}_{\bar{\mathbf{Q}}}\rangle_{\mathrm{ONB}}$ denote the approximate ONB representation of $|\Psi^{(n)}_{\bar{\mathbf{Q}}}\rangle$,
\begin{equation} \label{eq:nuc_ONV}
    |\tilde{\Psi}^{(n)}_{\bar{\mathbf{Q}}}\rangle_{\mathrm{ONB}} = \sum_{J}\widetilde{C}_J^{(n)} |J \rangle,
\end{equation}
where the index $J$ is defined in \cref{sect:input_data_nuc}. Because $\Psi^{(n)}_{\bar{\mathbf{Q}}}$ is related to $\Psi^{(n)}$ by a rescaling of $\mathbf{Q}$ and a normalization constant, the two states have the same coefficients $C_J$.

We prepare $|\widetilde{\Psi}^{(n)}_{\bar{\mathbf{Q}}}\rangle_{\mathrm{ONB}}$ using the SoSlat method as in \cref{sect:isp_elec}:
\begin{equation}
    \mathrm{SoSlat}^{(n)}|\tilde{\psi}\rangle = |\tilde{\Psi}^{(n)}_{\bar{\mathbf{Q}}}\rangle_{\mathrm{ONB}}.
\end{equation}
Representing this state uses $N_{\mathrm{vib}} N_{\mathrm{SMB}}$ qubits. The resources needed to implement $\mathrm{SoSlat}^{(n)}$ are:
\begin{align} 
    C_{\mathrm{anc}}(\mathrm{SoSlat}^{(n)}) &= 5\log D^{(n)} -3\\
    C_{\mathrm{Toff}}(\mathrm{SoSlat}^{(n)}) &\leq D^{(n)}(2 \log D^{(n)} + 3).
\end{align} 

Next, we transform the ONB to the SMB using a unitary $\textsc{onb2smb}$ that acts on ONB states as
\begin{equation}
    \textsc{onb2smb}|J\rangle= |\mu\rangle_1 |\nu \rangle_2 \dots |\tau \rangle_{N_{\mathrm{vib}}},
\end{equation}
where the right-hand side is given in the SMB and $|\mu\rangle_i$ indicates, in binary, that the $i$th normal coordinate is in the $\mu$th SM. The SMB consists of $N_{\mathrm{vib}}$ registers, each containing $\lceil \log N_{\mathrm{SMB}} \rceil $ qubits. After the transformation, the vibrational state is 
\begin{equation} \label{eq:nuc_sos}
    |\tilde{\Psi}^{(n)}_{\bar{\mathbf{Q}}}\rangle_{\mathrm{SMB}} \!=\! 
    \sum_{\mu,\nu,\dots, \tau =0}^{N_{\mathrm{SMB}}-1}\widetilde{C}_{\mu \nu \ldots \tau}^{(n)} |\mu\rangle_1 |\nu\rangle_2 \dots |\tau\rangle_{N_{\mathrm{vib}}}.
\end{equation}
Representing $|\tilde{\Psi}^{(n)}_{\bar{\mathbf{Q}}}\rangle_{\mathrm{SMB}}$ uses $N_{\mathrm{vib}} \lceil \log N_{\mathrm{SMB}} \rceil$ qubits.

We can implement $\textsc{onb2smb}$  using the approach used in \cref{sect:isp_elec} to implement $\textsc{onb2mob}$, requiring
\begin{align}\label{eq:cost_onv2sm}
    C_{\mathrm{anc}}(\textsc{onb2smb}) &= N_{\mathrm{SMB}} + 3\\
    C_{\mathrm{Toff}}(\textsc{onb2smb}) &= N_{\mathrm{vib}} N_{\mathrm{SMB}} \left( \lceil \log N_{\mathrm{SMB}} \rceil -2 \right).
\end{align}

The final step is to transform from the SMB to the normal-coordinate PWB. We let $\Phi_{i\mu,\bar{\mathbf{Q}}}$ denote the SMs corresponding to $\Psi^{(n)}_{\bar{\mathbf{Q}}}$ and $|\hat{\Phi}_{i\mu,\bar{\mathbf{Q}}} \rangle_{\mathrm{MPS}}$ the corresponding preprocessed state. As discussed in \cref{sect:preproc_sm}, given $\sm$, $\Phi_{i\mu,\bar{\mathbf{Q}}}$ is easily found by performing the substitution $\omega_i \rightarrow d_{ii}^2 \omega_i$ for all Hermite-Gaussian primitives, where, $d_{ij}$ are the elements of $\mathbf{d}$. The resource requirements for preprocessing an SM with rescaled coordinates are described at the end of \cref{sect:preproc_sm}.   

To convert from the SMB to the MPS encoding of the PWB, we construct the unitary
\begin{equation} \label{eq:vib_unitary}
    W_i^{(n)} = |0\rangle\langle1|\otimes V_i^{(n)} + |1\rangle \langle0 | \otimes \big(V_i^{(n)}\big)^{\dagger}
\end{equation}
for the $i$th normal coordinate, where 
\begin{equation}
    V_i^{(n)} = \sum_{\mu = 0}^{N_{\mathrm{SMB}}-1} |\hat{\Phi}_{i\mu,\bar{\mathbf{Q}}} \rangle_{\mathrm{MPS}} \langle \mu|_i
\end{equation}
acts on $n_{\mathrm{ISP}}$ qubits. 
When $i=\{3\eta_n-5, 3\eta_n-4, \ldots, 3\eta_n\}$, $\hat{\Phi}_{i\mu,\bar{\mathbf{Q}}}$ denotes a rotational or translational state, in which case $N_{\mathrm{SMB}}=1$. We implement an approximation $\widetilde{W}^{(n)}_{i}$ to $W^{(n)}_{i}$ by defining $|\tilde{w}^{(n)}_{\mu}\rangle_i$ given by
\begin{equation} \label{eq:vib_synth_approx}
\big|\tilde{w}^{(n)}_{\mu}\big\rangle_i  =\frac{1}{\sqrt{2}}\big(|1\rangle|\mu\rangle-|0\rangle |\tilde{\hat{\Phi}}_{i\mu,\bar{\mathbf{Q}}} \rangle_{\mathrm{MPS}} \big),
\end{equation}
where $|\tilde{\hat{\Phi}}_{i\mu,\bar{\mathbf{Q}}} \rangle_{\mathrm{MPS}}$ is an approximation to $|\hat{\Phi}_{i\mu,\bar{\mathbf{Q}}} \rangle_{\mathrm{MPS}}$ and is prepared using the approach from \cite{fomichev2024}. Then, $\widetilde{W}^{ (n)}_{i}$ can be synthesized as a product of reflections using the approach from \cref{sect:isp_elec},
\begin{equation} \label{eq:vib_synthesizer}
\widetilde{W}_i^{(n)}=\prod_{\mu=0}^{N_{\mathrm{SMB}}-1} \big(I-2\big|\tilde{w}^{(n)}_{\mu}\big\rangle\big\langle \tilde{w}^{(n)}_{\mu}\big|_i\big).
\end{equation}
We define
\begin{equation} \label{eq:vib_unit_synth}
    \widetilde{W}^{(n)} = \bigotimes_{i=1}^{3\eta_n}\widetilde{W}_i^{(n)}.
\end{equation}
Then, $|\tilde{\hat{\Psi}}_{\bar{\mathbf{Q}}}^{(n)}\rangle_{\mathrm{MPS}}$ is produced by performing
\begin{multline}
    \widetilde{W}^{(n)} |1\rangle|\tilde{\Psi}^{(n)}_{\bar{\mathbf{Q}}}\rangle_{\mathrm{SMB}} = |0\rangle|\tilde{\hat{\Psi}}_{\bar{\mathbf{Q}}}^{(n)}\rangle_{\mathrm{MPS}} =\\
    |0\rangle \!\sum_{\mu, \ldots, \tau=0}^{N_{\mathrm{SMB}}-1} \! \widetilde{C}_{\mu \ldots \tau}^{(n)} |\tilde{\hat{\Phi}}_{1\mu, \bar{\mathbf{Q}}} \rangle_{\mathrm{MPS}} \cdots |\tilde{\hat{\Phi}}_{N_{\mathrm{vib}}\tau, \bar{\mathbf{Q}}} \rangle_{\mathrm{MPS}}  . 
\end{multline}
The ancilla qubit can be reused throughout by resetting it to $|1\rangle$ after each $\widetilde{W}_i^{(n)}$ is applied. 

Implementing $\widetilde{W}^{(n)}$ requires 
\begin{align} \label{eq:W_n_data}
    C_{\mathrm{anc}}(\widetilde{W}^{(n)}) &\leq \tfrac12  \log(m_{i \mu j_\mathrm{max}}^{(n)}) +\frac{2 b_{\mathrm{rot}}^{(n)} (m_{i \mu j_\mathrm{max}}^{(n)})^{1/2} }{(b_{\mathrm{rot}}^{(n)}+1)^{1/2}} \nonumber \\
    & \qquad +\tfrac{1}{2} \log (b_{\mathrm{rot}}^{(n)}+1)+3 b_{\mathrm{rot}}^{(n)}
\end{align}
ancilla qubits, where $b_{\mathrm{rot}}^{(n)}$ is the number of precision qubits, $m_{i \mu j}^{(n)}$ is the bond dimension of the $j$th tensor of the $\mu$th state of the $i$th normal coordinate, and 
\begin{equation}
    m_{i \mu j_\mathrm{max}}^{(n)} = \max_{1\leq j \leq n_p} m_{i \mu j}^{(n)}.
\end{equation}
We assume that $C_{\mathrm{anc}}(\widetilde{W}^{(n)})>n_p-2$, so that the ancillas can be reused when implementing the multi-controlled-\textsc{not} gates of the reflection operators. Further, we reuse the ancillas for each $\widetilde{W}_i^{(n)}$, meaning that the number of ancillas is independent of system size.

The Toffoli cost to implement each $\widetilde{W}_i^{(n)}$ is \cite{huggins2024}
\begin{multline}\label{eq:one_vib_to_PWB_gates}
C_{\mathrm{Toff}}(\widetilde{W}_i^{(n)}) < N_{\mathrm{SMB}}n_{\mathrm{ISP}}  \\
+ 2\sum_{\mu=0}^{N_{\mathrm{SMB}}-1}\sum_{j=1}^{n_{\mathrm{ISP}}}\bigg(32(1+\sqrt{2})(b_{\mathrm{rot}}^{(n)}+1)^{1/2} m_{i \mu j}^{(n)} \bigl(\bar{m}_{i \mu j}^{(n)} \bigr)^{1/2} \\
+(8 b_{\mathrm{rot}}^{(n)} - 15) m_{i \mu j}^{(n)} \log \bigl(2 \bar{m}_{i \mu j}^{(n)}\bigr) \bigg),
\end{multline}
where
\begin{equation}
\bar{m}_{i\mu j}^{(n)} = \max \bigg\{  2^{\left\lceil\log m_{i \mu (j-1)}^{(n)}\right\rceil}, 2^{\left\lceil\log m_{i \mu j}^{(n)}\right\rceil} \bigg\}.
\end{equation} 
The total cost to implement $\widetilde{W}^{(n)}$ is then
\begin{equation}\label{eq:full_vib_to_PWB_gates}
C_{\mathrm{Toff}}(\widetilde{W}^{(n)}) = \sum_{i=1}^{3\eta_n} C_{\mathrm{Toff}}(\widetilde{W}_i^{(n)}).
\end{equation}

\begin{figure*} \label{fig:shear}
\centering
    \includegraphics[width=0.9\textwidth]{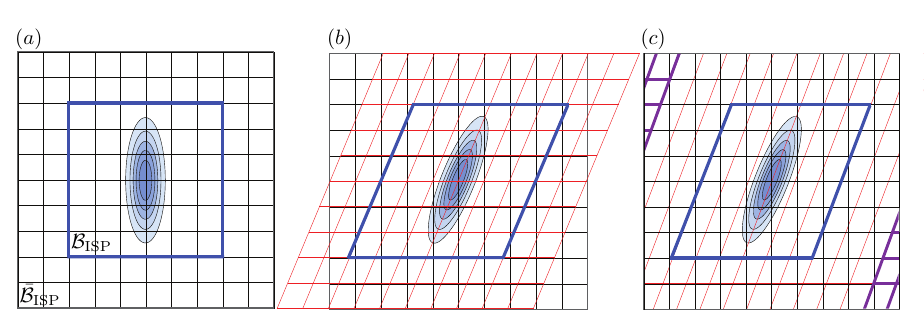}
    \vspace{-3mm}
    \caption{Eliminating modular-addition error by padding the grid $\mathcal{B}_{\mathrm{ISP}}$. 
    \textbf{(a)}~Computational basis states $|\mathbf{n}\rangle$ are represented by grid points. Wavefunction $|g\rangle$ (blue contours) is supported by grid points in $\mathcal{B}_{\mathrm{ISP}}$ (blue square). Padding qubits expand the grid to $\bar{\mathcal{B}}_{\mathrm{ISP}}$. The padded state $|g\rangle$ has non-zero amplitude within $\bar{\mathcal{B}}_{\mathrm{ISP}}$ and is zero outside it. 
    \textbf{(b)}~Applying the shear $U_{\mathbf{S}}$ to $|g\rangle$, with transformed grid points shown in red. We choose $n_{\mathrm{pad}}$ sufficiently large that grid points in $\mathcal{B}_{\mathrm{ISP}}$ remain within $\bar{\mathcal{B}}_{\mathrm{ISP}}$ after the transformation. 
    \textbf{(c)}~Performing $\mathrm{cmod}\, \bar{N}_{\mathrm{ISP}}/2$ leaves grid points in the blue rhombus unchanged while those outside it are wrapped around to the opposite edge of the grid (purple regions). Because these grid points have zero amplitude, their wrapping around does not contribute an error.}
\end{figure*}

\subsubsection{Coordinate transformation \texorpdfstring{$U_{\mathbf{A}^{- \top}}$}{UAT}}
\label{sect:LCT}

Implementing $U_{\mathbf{A}^{- \top}}$ allows us to perform normal-coordinate transformations which enable a wide variety of nuclear initial states to be prepared. Previous algorithms for normal-coordinate transformations are inefficient, scaling as $\mathcal{O}(N)$, and lack error bounds~\cite{Jornada2025}. We develop a new, efficient algorithm that performs linear coordinate transformations (LCT) on a quantum computer with Toffoli cost $\mathcal{O}(\eta_n^2 \log^2 N)$.

Our LCT algorithm, $U_{\mathbf{A}^{- \top}}^{\textsc{lct}}$, implements arbitrary invertible transformations $\mathbf{y} = \mathbf{T}\mathbf{x}$, where, for nuclear ISP, we set $\mathbf{T}=\mathbf{A}^{- \top}$. It is based on a subroutine for performing shearing transformations~\cite{kitaev2009,bauer2021}, along with reflections and rotations by $\pm \pi/2$. We assume that $\lvert\det(\mathbf{T})\rvert=1$; if not, we first express $\mathbf{T}$ as in~\cref{eq:RQ_decomp} and rescale the coordinates accordingly. 

For the important case where the initial nuclear state is a multivariate Gaussian, we provide two approaches to implement $U_{\mathbf{A}^{- \top}}$ and prove error bounds on both. The first approach uses our LCT algorithm, with the error $\epsilon_{\mathrm{LCT}}$ from approximately preparing $U_{\mathbf{A}^{- \top}}^{\textsc{lct}}$ being $\epsilon_{\mathrm{LCT}} = \mathcal{O}(\eta_n^2 \Delta)$. Alternatively, a multivariate Gaussian can be prepared using a single-shear coordinate transformation (SSCT), which requires a single application of the shearing-transformation subroutine~\cite{kitaev2009,bauer2021}. We denote this approach by $U_{\mathbf{A}^{- \top}}^{\textsc{ssct}}$ and prove that the error from approximating $U_{\mathbf{A}^{- \top}}^{\textsc{ssct}}$ is only $\epsilon_{\mathrm{SSCT}} = \mathcal{O}(\eta_n^{1/2} \Delta)$.
For both approaches, we prove a lower bound on $n_{\mathrm{pad}}$ that ensures the nuclear wavefunction remains on the grid while $U_{\mathbf{A}^{- \top}}$ is implemented.

\paragraph{Linear coordinate transformation.}
Our goal is, given a state 
\begin{equation} \label{eq:generic_state}
     |g\rangle =\frac{1}{\mathcal{N}_{\mathrm{ISP}}} \sum_{\mathbf{n}\in \bar{\mathcal{B}}_{\mathrm{ISP}}} g(\Delta \mathbf{n})|\mathbf{n}\rangle,
\end{equation}
where, $\bar{\mathcal{B}}_{\mathrm{ISP}} = [- \frac{1}{2}\bar{N}_{\mathrm{ISP}},\frac{1}{2}\bar{N}_{\mathrm{ISP}}-1]^{3\eta_n}\cap \mathbb{Z}^{3\eta_n}$, $\mathcal{N}_{\mathrm{ISP}}$ is a normalization constant and $g(\Delta \mathbf{n}) = 0$ if $\Vert\mathbf{n}\Vert_{\infty} > N_{\mathrm{ISP}}/2$, to transforms $|g\rangle$ into 
\begin{equation} \label{eq:generic_coord_trans}
\frac{1}{\mathcal{N}_{\mathrm{ISP}}} \sum_{\mathbf{n} \in \bar{\mathcal{B}}_{\mathrm{ISP}}} g(\Delta \mathbf{T}\mathbf{n})|\mathbf{n}\rangle,
\end{equation}
where the normalization is preserved because $\lvert\det(\mathbf{T})\rvert=1$. Rather than transform the argument of the amplitudes, $g(\Delta \mathbf{n}) \rightarrow g(\Delta \mathbf{Tn})$, our approach is, equivalently, to construct a unitary that approximately transforms each computational basis state, $\ket{\mathbf{n}} \rightarrow \ket{\mathbf{T}^{-1}\mathbf{n}}$, as accurately as possible when working on a grid. For a matrix $\mathbf{C}$, we use the notation
\begin{equation}
    |\mathbf{C}\mathbf{n}\rangle = | [\mathbf{C}\mathbf{n}]_1\rangle \dots | [\mathbf{C}\mathbf{n}]_{3\eta_n}\rangle,
\end{equation}
where $[\cdot]_i$ is the $i$th element of the enclosed vector.

To begin, we perform a QL decomposition to get $\mathbf{T}^{-1} =  \mathbf{X} \mathbf{L}$, where $\mathbf{X}$ is an orthogonal matrix and $\mathbf{L}$ is a lower shear (unit lower-triangular matrix),
\begin{equation}
\mathbf{L}=\left(\begin{array}{ccccc}
1 & 0 & \cdots & 0 \\
b_{21} & 1 &  \cdots & 0 \\
\vdots & \vdots & \ddots & \vdots \\
b_{3 \eta_n 1} & b_{3 \eta_n 2} & \cdots & 1
\end{array}\right),
\end{equation}
with $b_{ij}\in \mathbb{R}$. 
In principle, any variation of the QR decomposition can be used; we use the QL decomposition to align with our eventual implementation of $U_{\mathbf{A}^{- \top}}$. We also require $\det(\mathbf{X})=1$; if it does not, we first multiply it by the reflection 
$\mathbf{Y} = \mathrm{diag}(\mathrm{det}(\mathbf{X}),1,1,\dots,1)$.

We then decompose $\mathbf{X}$ into a product of $N_{G} = 3\eta_n(3\eta_n-1)/2$ Givens rotations~\cite{golub2013},  
\begin{equation} 
    \mathbf{X} = \prod_{(i,j)\in \mathcal{G}} \mathbf{G}(i,j,\theta_{ij}),
\end{equation}
where $\mathcal{G}$ is a set containing $N_G$ pairs of indices and $\mathbf{G}(i,j,\theta_{ij})$ is a Givens rotation by angle $\theta_{ij}\in [-\pi,\pi)$ in a plane spanned by the indices $i$ and $j$,
\begin{equation}\label{eq:givens}
\mathbf{G}(i,j,\theta_{ij})=
{\renewcommand{\arraystretch}{0.85}
\begin{pmatrix}
1 & \cdots & 0 & \cdots & 0 & \cdots & 0 \\
\smash{\vdots} & \ddots & \smash{\vdots} & & \smash{\vdots} & & \smash{\vdots} \\
0 & \cdots & c_{ii} & \cdots & s_{ij} & \cdots & 0 \\
\smash{\vdots} & & \smash{\vdots} & \ddots & \smash{\vdots} & & \smash{\vdots} \\
0 & \cdots & -s_{ji} & \cdots & c_{jj} & \cdots & 0 \\
\smash{\vdots} & & \smash{\vdots} & & \smash{\vdots} & \ddots & \smash{\vdots} \\
0 & \cdots & 0 & \cdots & 0 & \cdots & 1
\end{pmatrix}},
\end{equation}
with $c_{ij}=\mathrm{cos}(\theta_{ij})$, $s_{ij} = \mathrm{sin}(\theta_{ij})$, and $\theta_{ij} \in [-\pi, \pi)$, which are found by performing a Givens decomposition on a classical computer~\cite{golub2013}. 

Next, each Givens rotation is decomposed into a product of three shears. To avoid complications that arise when $|\theta_{ij}| > \pi/2$, we first reduce the angle of the Givens rotations to $\varphi_{ij}\in[-\tfrac{\pi}{2}, \tfrac{\pi}{2})$ by writing
\begin{equation}
    \mathbf{G}(i,j,\theta_{ij}) =  \mathbf{G}(i,j, \varphi_{ij} ) \mathbf{J}(i,j, \theta_{ij}),
\end{equation}
where $\varphi_{ij} = \theta_{ij} - \tfrac{\pi}{2} h_{ij} \operatorname{sgn}\theta_{ij}$, $h_{ij} = H(|\theta_{ij}| -\tfrac{\pi}{2})$, $H(\cdot)$ is the Heaviside function, and $\mathbf{J}(i,j,\theta_{ij}) = \mathbf{G}(i,j,\tfrac{\pi}{2}\operatorname{sgn}\theta_{ij})^{h_{ij}}$.
Thus, $\mathbf{J}(i,j, \theta_{ij})$ is a rotation by $\pm \pi/2$ in the $ij$ plane when $h_{ij}=1$ and the identity otherwise. 

Then, $\mathbf{G}(i,j,\varphi_{ij})$ can be safely decomposed into a product of three 2D shears \cite{toffoli1997},
\begin{equation} \label{eq:shearing}
    \mathbf{G}(i,j,\varphi_{ij}) = \mathbf{S}_1(i,j,\varphi_{ij}) \mathbf{S}_2(i,j,\varphi_{ij}) \mathbf{S}_1(i,j,\varphi_{ij}),
\end{equation}
where
\begin{align}
    \mathbf{S}_1(i,j,\varphi_{ij})& = \mathbf{I} + t_{ij} \mathbf{E}(i,j), \quad t_{ij} = \mathrm{tan}(\varphi_{ij}/2) \label{eq:s1_def}   \\
    \mathbf{S}_2(i,j,\varphi_{ij}) &= \mathbf{I} - s_{ji} \mathbf{E}(i,j), \quad s_{ij} = \sin(\varphi_{ij}) \label{eq:s2_def}.
\end{align}
Here, $\mathbf{I}$ is the $3\eta_n$-dimensional identity matrix, and $\mathbf{E}(i,j)$ has only one non-zero matrix element, $E_{kl}(i,j)= \delta_{ik} \delta_{kl}$. We refer to $\mathbf{S}_1$ and $\mathbf{S}_2$ as 2D shears because they operate within 2D planes spanned by the indices $i$ and $j$, as opposed to $\mathbf{L}$, which operates on all $3\eta_n$ coordinates. 

The LCT algorithm is thus the sequence of unitaries
\begin{multline}
    U_{\mathbf{A}^{- \top}}^{\textsc{lct}} = U_{\mathbf{Y}} \bigg(\prod_{(i,j) \in \mathcal{G}} U_{\mathbf{S}_1(i,j,\varphi_{ij}) } \\
     \times U_{\mathbf{S}_2(i,j,\varphi_{ij})} U_{\mathbf{S}_1(i,j,\varphi_{ij})} U_{ \mathbf{J}(i,j,\theta_{ij})} \bigg) U_{\mathbf{L}},
    \label{eq:U_AT_exact}
\end{multline}
where $U_{\mathbf{Y}}$ and $U_{ \mathbf{J}(i,j,\theta_{ij})}$ implement $\mathbf{Y}$ and $\mathbf{J}(i,j,\theta_{ij})$,
\begin{align}
    U_\mathbf{Y}|\mathbf{n}\rangle &= |\mathbf{Y}\mathbf{n}\rangle \\
     U_{\mathbf{J}(i,j, \theta_{ij})}|\mathbf{n}\rangle &= |\mathbf{J}(i,j,\theta_{ij})\,\mathbf{n}\rangle,
\end{align}
and remaining unitaries implement shears.

Shears are performed in-place using an iterative algorithm~\cite{bauer2021}. In-place means that the data register of the initial state is directly transformed into the final state without requiring the initial state to be copied. A lower shear with elements $b_{ij}$ is implemented by performing
\begin{equation}
    |n_i\rangle \leftarrow \bigg|\sum_{j=1}^{i}b_{ij} n_j \bigg\rangle
    \label{eq:simple_shear_lower}
\end{equation}
first for $i=3\eta_n$ and then sequentially down to $i=1$. In doing so, the $i$th register is overwritten at the $i$th step. For an upper shear, we iteratively perform 
\begin{equation}
    |n_i\rangle \leftarrow \bigg|\sum_{j=i}^{3\eta_n }b_{ij} n_j \bigg\rangle
    \label{eq:simple_shear},
\end{equation}
starting with $i=1$ and going up to $i= 3\eta_n$.

A general shear $\mathbf{S}$ cannot be exactly implemented on an integer lattice using the approach above because $\mathbf{S}\mathbf{n}$ is not necessarily a grid point of $\bar{\mathcal{B}}_{\mathrm{ISP}}$, since it can be outside of $\bar{\mathcal{B}}_{\mathrm{ISP}}$ or not integer valued. Instead, we implement an approximation to $\mathbf{S}$,
\begin{equation}
    \tilde{U}_{\mathbf{S}}|\mathbf{n}\rangle = |S(\mathbf{n})\rangle,
    \label{eq:approx_shear_general}
\end{equation}
where $S:\bar{\mathcal{B}}_{\mathrm{ISP}}\rightarrow \bar{\mathcal{B}}_{\mathrm{ISP}}$ is a bijection. When $\mathbf{S}$ is a lower shear, the components of the transformed state $|S(\mathbf{n})\rangle$ are given by iteratively performing
\begin{equation} \label{eq:approx_SL}
    |n_i\rangle \leftarrow  \left| R\Bigg( \Bigg(\sum_{j=1}^{i} b_{ij} n_j \Bigg) \, \operatorname{cmod}\, \bar{N}_{\mathrm{ISP}}/2 \Bigg) \right\rangle,
\end{equation}
where the centered modulo function is defined by
\begin{equation}
    y \operatorname{cmod} M = ((y+M) \bmod (2 M))-M,
\end{equation}
and $R(\cdot)$ rounds each vector component to the nearest integer. The transformation starts with $i=3\eta_n$, and $i$ is decremented after each iteration until $i=1$. For an upper shear, the transformation iterates from $i=1$ to $i =3\eta_n$. At each iteration, \cref{eq:approx_SL} is performed, except that the summation runs from $j=i$ to $3\eta_n$.

Importantly, $S$ is invertible. For a lower shear, the state $|S^{-1}(\mathbf{n})\rangle$ is given iteratively, incrementing $i$ from 1 to $3\eta_n$. For $i=1$, $|n_1\rangle \leftarrow |n_1\rangle$. Thereafter, 
\begin{equation} \label{eq:n_i_expression}
    |n_i\rangle \leftarrow \left|n_i \oplus R\Bigg( \Bigg(-\sum_{j=1}^{i-1}b_{ij}n_j\Bigg) \, \operatorname{cmod}\, \bar{N}_{\mathrm{ISP}}/2 \Bigg) \right\rangle,
\end{equation}
where $a_i \oplus a_j = (a_i \oplus a_j) \operatorname{cmod} \bar{N}_{\mathrm{ISP}}/2$. The inverse of an upper shear is also given by an iterative procedure, running from $i=3\eta_n$ down to $i=1$. For $i=3\eta_n$, $|n_{3\eta_n}\rangle \leftarrow |n_{3\eta_n}\rangle$. The following iterations perform \cref{eq:n_i_expression}, except that the summation runs from $j=i+1$ to $3\eta_n$ for an upper shear.

Therefore, we approximate $U_{\mathbf{A}^{- \top}}^{\textsc{lct}}$ as
\begin{multline}
        \tilde{U}_{\mathbf{A}^{- \top}}^{\textsc{lct}} = U_{\mathbf{Y}} \bigg(\prod_{(i,j) \in \mathcal{G}} \tilde{U}_{\mathbf{S}_1(i,j,\varphi_{ij}) } \\
    \times \tilde{U}_{\mathbf{S}_2(i,j,\varphi_{ij})} \tilde{U}_{\mathbf{S}_1(i,j,\varphi_{ij})} U_{ \mathbf{J}(i,j,\theta_{ij})} \bigg) \tilde{U}_{\mathbf{L}},
    \label{eq:ortho_seq_shear}
\end{multline}
where each shearing unitary $U_{\mathbf{S}}$ is replaced by the appropriate $\tilde{U}_{\mathbf{S}}$. Using $\tilde{U}_{\mathbf{S}}$ deals with both problems with shears on a grid, ensuring that $S(\mathbf{n}) \in \bar{\mathcal{B}}_{\mathrm{ISP}}$. First, $\mathrm{cmod}$ wraps points that would be mapped off the simulation grid back onto its opposite edge. Second, rounding ensures that each transformed grid point is an integer vector. An explicit algorithm to implement $\tilde{U}_{\mathbf{A}^{- \top}}^{\textsc{lct}}$ is given in \cref{app:implementing_A}.

At the beginning of \cref{sect:isp_nuc}, we padded the grid $\mathcal{B}_{\mathrm{ISP}} = [- {N_{\mathrm{ISP}}}/{2},{N_{\mathrm{ISP}}}/{2}-1]^{3\eta_n}\cap \mathbb{Z}^{3\eta_n}$ with additional qubits to create a larger grid $\bar{\mathcal{B}}_{\mathrm{ISP}}$. This padding was necessary because modular addition can introduce errors by mapping vectors with $\Vert\mathbf{S}\mathbf{n}\Vert_{\infty} > N_{\mathrm{ISP}}/2$ to the opposite edge of the grid (see \cref{fig:shear}). Therefore, we choose $\bar{\mathcal{B}}_{\mathrm{ISP}}$ to be large enough that all grid points $\Vert\mathbf{n}\Vert_{\infty} < N_{\mathrm{ISP}}/2$ remain within $\bar{\mathcal{B}}_{\mathrm{ISP}}$ during every step of $\tilde{U}_{\mathbf{A}^{- \top}}^{\textsc{lct}}$. Because $g(\Delta \mathbf{n}) = 0$ if $\Vert\mathbf{n}\Vert_{\infty} > N_{\mathrm{ISP}}/2$, no error is introduced if these points are mapped to the opposite edge. 
In \cref{app:nm_algo_grid_pad}, we show that this requirement is met if
\begin{equation}
   n_{\mathrm{pad}} \geq \left\lceil \log \left(1.619 \sqrt{3 \eta_n} ( N_{\mathrm{ISP}}\Vert \mathbf{L}\Vert_{\infty} +\beta)+1\right) \right \rceil - n_{\mathrm{ISP}},
   \label{eq:n_pad_bound_LCT}
\end{equation}
where $\beta = 18 \eta_n^2 - 6 \eta_n +1$ is the total number of unitaries in \cref{eq:U_AT_exact}, not counting $U_{\mathbf
L}$. 

On $\bar{\mathcal{B}}_{\mathrm{ISP}}$, each approximate shear $\tilde{U}_{\mathbf{S}}$ acts with a controllable error that can be decreased by reducing $\Delta$. We define the ideal state transformed by $\mathbf{S}$ as
\begin{equation} \label{eq:amp_shear}
   |g_{\mathrm{ideal}}\rangle = \frac{1}{\mathcal{N}_{\mathrm{ISP}}}\sum_{\mathbf{n} \in \bar{\mathcal{B}}_{\mathrm{ISP}}} g(\Delta  \mathbf{S}^{-1} \mathbf{n})|\mathbf{n}\rangle,
\end{equation}
and write the approximate state as
\begin{align} \label{eq:amp_shear_2}
    \tilde{U}_{\mathbf{S}}|g\rangle& = \frac{1}{\mathcal{N}_{\mathrm{ISP}}}\sum_{\mathbf{m} \in \bar{\mathcal{B}}_{\mathrm{ISP}}} g(\Delta  \mathbf{m})|S(\mathbf{m})\rangle  \\ 
    &= \frac{1}{\mathcal{N}_{\mathrm{ISP}}}\sum_{\mathbf{n} \in \bar{\mathcal{B}}_{\mathrm{ISP}}} g(\Delta  S^{-1}(\mathbf{n}))|\mathbf{n}\rangle, \label{eq:inverse_S}
\end{align}
where $\mathbf{n} = S(\mathbf{m})$. In \cref{app:lct_ssct_error}, we show that $S^{-1}(\mathbf{n}) = \mathbf{S}^{-1}\mathbf{n} -\mathbf{S}^{-1}\boldsymbol{\delta}$, where $\boldsymbol{\delta}$ is a vector of residual components that are rounded off in \cref{eq:n_i_expression}. Therefore, the argument of $g$ of the approximate state differs from that of the ideal state by $\Delta \mathbf{S}^{-1}\boldsymbol{\delta}$, meaning that the approximate state approaches the ideal state as $\Delta\to 0$. 

Therefore, the LCT procedure determines the grid spacing $\Delta$.
The error from implementing $\tilde{U}_{\mathbf{A}^{- \top}}^{\textsc{lct}}$ is
\begin{equation}
    \epsilon_{\mathrm{LCT}} = D(\tilde{U}_{\mathbf{A}^{- \top}}^{\textsc{lct}} |g\rangle,U_{\mathbf{A}^{- \top}}^{\textsc{lct}}|g\rangle).
\end{equation}
When $|g\rangle$ is a separable multivariate Gaussian, we upper bound $\epsilon_{\mathrm{LCT}}$ in \cref{app:lct_ssct_error}, which is asymptotically 
\begin{equation}
    \epsilon_{\mathrm{LCT}} = \mathcal{O}(\eta_n^2 \Delta).
    \label{eq:eps_LCT}
\end{equation}
\Cref{eq:eps_LCT} therefore determines $\Delta$ as
\begin{equation}
    \Delta = \mathcal{O}(\epsilon_{\mathrm{LCT}}/\eta_n^2).
    \label{eq:L_asymptotic}
\end{equation}

We show in \cref{app:implementing_A} that the resources needed to implement $\tilde{U}_{\mathbf{A}^{- \top}}^{\textsc{lct}}$ are 
\begin{align} 
     C_{\mathrm{anc}}(\tilde{U}_{\mathbf{A}^{- \top}}^{\textsc{lct}}) &= 4\bar{n}_{\mathrm{ISP}} -3 \label{eq:LCT_resources_b}\\
    C_{\mathrm{Toff}}(\tilde{U}_{\mathbf{A}^{- \top}}^{\textsc{lct}}) &= \tfrac{9}{2}\eta_n^2(8 \bar{n}_{\mathrm{ISP}}^2 + 39\bar{n}_{\mathrm{ISP}} - 8)  - \bar{n}_{\mathrm{ISP}} \nonumber\\
    &\quad - \tfrac{3}{2}\eta_n(8\bar{n}_{\mathrm{ISP}}^2 + 35\bar{n}_{\mathrm{ISP}} -8) \label{eq:LCT_resources_c}.
\end{align}

In the specific case of $\textbf{T}=\textbf{A}^{-\top}$, we can give expressions for the shears above in terms of the matrices introduced in \cref{sect:ISP-constructingU}. 
Because $\mathbf{A} = \mathbf{S}_{\mathrm{U}} \bar{\mathbf{O}}$, we have $\mathbf{A}^{\top} = \bar{\mathbf{O}}^{\top}\mathbf{S}_{\mathrm{L}} $, where $\mathbf{S}_{\mathrm{L}} = \mathbf{S}_{\mathrm{U}}^{\top}$ and $\bar{\mathbf{O}}^{\top}$ is an orthogonal transformation. $\tilde{U}_{\mathbf{A}^{- \top}}^{\textsc{lct}}$ is then given by setting $\mathbf{L}= \mathbf{S}_{\mathrm{L}}$ and $\mathbf{X} =\bar{\mathbf{O}}^{\top}$.

\paragraph{Single-shear coordinate transformation.}
\label{sect:SSCT}
Multivariate Gaussian initial states can also be prepared using an SSCT~\cite{kitaev2009,bauer2021}. 

For a multivariate Gaussian wavefunction,
\begin{equation}
    \hat{\Psi}_{\bar{\mathbf{Q}}}^{(n)}(\Delta \mathbf{A}^{- \top} \mathbf{n}) \propto e^{- \Delta^2 \mathbf{n}^{\top} 
    \boldsymbol{\Lambda} \mathbf{n} /2 },
\end{equation}
where $\boldsymbol{\Lambda} = \mathbf{A}^{-1} \mathbf{W}^{-1} \mathbf{A}^{- \top}$ and $\mathbf{W} = \mathrm{diag}(\omega_1^2, \ldots, \omega_{3\eta_n}^2)$, the Cholesky decomposition $\boldsymbol{\Lambda}=\mathbf{L} \mathbf{D}_{\mathrm{Ch}} \mathbf{L}^{\top}$ gives a shear $\mathbf{L}$ and a diagonal matrix $\mathbf{D}_{\mathrm{Ch}}$. This yields
\begin{equation}
    \hat{\Psi}_{\bar{\mathbf{Q}}}^{(n)}(\Delta \mathbf{A}^{- \top} \mathbf{n}) = \hat{\Psi}_{\bar{\mathbf{Q}}}^{(n)'}(\Delta \mathbf{L}^{\top} \mathbf{n}),
\end{equation}
where $\hat{\Psi}_{\bar{\mathbf{Q}}}^{(n)'}(\Delta\mathbf{n}) \propto e^{- \Delta \mathbf{n}^{\top} \mathbf{D}_{\mathrm{Ch}} \mathbf{n}}$ is a product state of $3\eta_n$ independent Gaussians. 

An approximation to $\hat{\Psi}_{\bar{\mathbf{Q}}}^{(n)'}(\Delta \mathbf{L}^{\top} \mathbf{n})$ can be prepared using a single application of $\tilde{U}_{\mathbf{S}}$ given in \cref{eq:approx_shear_general}. First, the state 
\begin{equation}
    |\hat{\Psi}^{(n)'}_{\bar{\mathbf{Q}}}\rangle = \sum_{\mathbf{n} \in \mathcal{B}_{\mathrm{ISP}}} \hat{\Psi}^{(n)'}_{\bar{\mathbf{Q}}}(\Delta \mathbf{n}) | \mathbf{n}\rangle
\end{equation}
is prepared on the quantum computer. Then, the transformation $|\mathbf{n}\rangle \rightarrow |\mathbf{L}^{-\top}\mathbf{n}\rangle= \tilde{U}_{\mathbf{S}}|\mathbf{n}\rangle$ is achieved with an application of $\tilde{U}_{\mathbf{S}}$ with $\mathbf{S}=\mathbf{L}^{- \top}$.
Therefore, the approximate SSCT unitary is simply $\tilde{U}_{\mathbf{A}^{-\top}}^{\textsc{ssct}} = \tilde{U}_{\mathbf{S}}$.

When using the SSCT approach, the number of padding qubits needed is (see \cref{app:nm_algo_grid_pad})
\begin{equation} \label{eq:npad_shear}
    n_{\mathrm{pad}} \geq \lceil\log(N_{\mathrm{ISP}}\Vert \mathbf{\mathbf{L}^{-\top}} \Vert_{\infty} + 1)\rceil.
\end{equation}
The resources needed to implement $\tilde{U}_{\mathbf{A}^{-\top}}^{\textsc{ssct}}$ are
\begin{align}
    C_{\anc}(\tilde{U}_{\mathbf{A}^{-\top}}^{\textsc{ssct}}) = & 4\bar{n}_{\mathrm{ISP}} -3 \\
    C_{\Toff}(\tilde{U}_{\mathbf{A}^{-\top}}^{\textsc{ssct}}) = &9\eta_n^2(\bar{n}^2_{\mathrm{ISP}} + 4\bar{n}_{\mathrm{ISP}} -1 )  \\ 
    &- 3\eta_n(\bar{n}_{\mathrm{ISP}}^2 + 2\bar{n}_{\mathrm{ISP}} -1) - 2\bar{n}_{\mathrm{ISP}}.
\end{align}

The SSCT error
\begin{equation}
    \epsilon_{\mathrm{SSCT}} = D(\tilde{U}_{\mathbf{A}^{-\top}}^{\textsc{ssct}}|g\rangle, U_{\mathbf{A}^{-\top}}^{\textsc{ssct}}|g\rangle)
\end{equation}
results from approximating a single shearing transformation. \Cref{lem:shear_error} immediately establishes
\begin{equation}
    \epsilon_{\mathrm{SSCT}} \leq   2^{1/2} \Big(1 -  e^{- \Delta^2 3 \eta_n \lambda_{\mathrm{max}}(\boldsymbol{\Lambda}')  }  \Big)^{1/2} = \mathcal{O}(\eta_n^{1/2} \Delta),
\end{equation}
where $\lambda_{\mathrm{max}}(\cdot)$ is the largest eigenvalue of the input matrix, $\boldsymbol{\Lambda}' = \mathbf{S}^{- \top} \boldsymbol{\Sigma}' \mathbf{S}^{-1}$, and $\boldsymbol{\Sigma}'$ is a diagonal matrix with positive entries that defines the Gaussian state. Therefore, the SSCT method allows the use of a larger $\Delta$ than the LCT method,
\begin{equation}
    \Delta = \mathcal{O}\big(\epsilon_{\mathrm{SSCT}}/\eta_n^{1/2}\big).
\end{equation}

\subsubsection{Phase kickback \textsc{pk}}\label{sect:isp_nuc_nm}
The action of the phase-kickback unitary  $\textsc{pk}$ is \cite{cleve_1998}
\begin{equation} \label{eq:U_A}
    \textsc{pk}|\mathbf{n}\rangle =  e^{-i \Delta \mathbf{n}^{\top} \mathbf{R}^{(0)} } |\mathbf{n}\rangle.
\end{equation}
To perform an approximate PK, denoted $\widetilde{\textsc{pk}}$, we use the implementation~\cite{nam2020approximate} in which a repeat-until-success (RUS) circuit synthesizes the $Z$ rotations~\cite{bocharov_2015}. A circuit for $\widetilde{\textsc{pk}}$ is provided in \cref{app:nm_algo_pkb}, along with a resource analysis that shows that
\begin{align}
C_{\mathrm{anc}}(\widetilde{\textsc{pk}})&=4 b_{\mathrm {grad}}+2\bar{n}_{\text {ISP }}-1, \\
C_{\mathrm{Toff}}(\widetilde{\textsc{pk}})=&3\eta_n (4 \bar{n}_{\mathrm{ISP}} b_{\mathrm{grad}} + b_{\mathrm{grad}} - 2\bar{n}_{\mathrm{ISP}}) \nonumber\\ &+  b_{\mathrm{grad}} (1.149 (\log(b_{\mathrm{grad}}\epsilon_{\mathrm{PK}}^{-1})) + 9.2)/4 \label{eq:cost_PK}
\end{align}
where $b_{\mathrm{grad}}$ is the number of qubits used to construct a phase-gradient state and $\epsilon_{\mathrm{PK}}$ is the target PK error.

\subsubsection{Integer conversion}
\label{subsubsec:tc2sm}

We convert from the two's-complement representation to the signed-magnitude representation, which is used during time evolution. In the signed-magnitude representation, we do not use the plane wave corresponding to $-0$, giving a simulation grid with one fewer grid point per dimension, $\bar{\mathcal{B}}_{\mathrm{ISP}}\rightarrow \left[-\tfrac12(\bar{N}_{\mathrm{ISP}}-1), \tfrac12(\bar{N}_{\mathrm{ISP}}-1)\right]^{3} \cap \mathbb{Z}^{3}$.

We implement the conversion using the unitary $\textsc{tc2sm}_{n}$ which acts on $n$-bit integers following the standard two-step procedure.
First, if the leading bit is 1, we invert the remaining bits using $n-1$ \textsc{cnot}s, which require no Toffoli gates. Second, we add the leading bit to the remaining bits, discarding the carry-out bit. This addition can be achieved using a ripple-carry adder with $n-2$ Toffoli gates and $n-2$ ancillas.
For all $3\eta_n$ registers, this requires
\begin{align}
    C_{\mathrm{anc}}(\textsc{tc2sm}_{\bar{n}_{\mathrm{ISP}}}^{\otimes 3\eta_{n}})&= \bar{n}_{\mathrm{ISP}}-2\\
    C_{\mathrm{Toff}}(\textsc{tc2sm}_{\bar{n}_{\mathrm{ISP}}}^{\otimes 3\eta_{n}}) &= 3 \eta_n (\bar{n}_{\mathrm{ISP}}-2). \label{eq:cost_tc2sm} 
\end{align}

\subsubsection{Grid trimming}
\label{subsubsec:ISP_trimming}
To reduce the cost of subsequent time evolution, we trim the prepared state on the padded grid by removing grid points on the exterior of the grid, which have negligible or zero amplitude. Doing so reduces the padded grid $\BBar$ to the interior grid $G = [- \tfrac12(N-1),\tfrac12(N-1)]^{3\eta_n} \cap \mathbb{Z}^{3\eta_n}$, whose size is set by the grid size $N$ determined in \cref{sect:grid}. $G$ is the grid that time evolution is performed on.

Trimming could be done by discarding the qubits that represent the exterior of the grid or, in architectures where there is a cost to mid-circuit measurement and re-initialization, these qubits can be kept, but not acted on by the subsequent time-evolution subroutines. We follow the latter approach in our resource estimates.

For any state $|\hat{\Psi}\rangle$, the error that results from this approximation is 
\begin{equation}
    \epsilon_{\mathrm{trim}} = D(|\hat{\Psi}\rangle_{\mathrm{int}}, |\hat{\Psi}\rangle) = \sqrt{1 - {\mathcal{N}_{\mathrm{int}}^2}}\label{eq:epsilon-trim},
\end{equation}
where $|\hat{\Psi}_{\mathrm{int}}\rangle = \Pi_{\mathrm{int}}|\hat{\Psi}\rangle/\mathcal{N}_{\mathrm{int}}$, $\Pi_{\mathrm{int}} = I^{(e)}\otimes \sum_{\mathbf{n}\in\mathcal{B}_{\mathrm{int}}}|\mathbf{n}\rangle \langle \mathbf{n}|$
is the projection onto the interior grid, and the normalization is $\mathcal{N}_{\mathrm{int}} = \sqrt{|\langle \hat{\Psi}|\Pi_{\mathrm{int}} |\hat{\Psi}\rangle|}$. 

In resource estimation, we estimate $\epsilon_{\mathrm{trim}}$ using a Monte Carlo calculation to estimate $p=\mathcal{N}^2_{\mathrm{int}}$, which is the probability that a grid point $\mathbf{n}\in \bar{\mathcal{B}}_{\mathrm{ISP}}$ sampled from the distribution $|\hat{\Psi}_{\mathbf{n}}|^2$ is within $\mathcal{B}_{\mathrm{int}}$. We obtain the estimate
$\hat{p} = {N_{\mathrm{int}}}/{N_{\mathrm{MC}}}$
by numerically sampling $N_{\mathrm{MC}}$ points from this distribution and determining the number $N_{\mathrm{int}}$ within $\mathcal{B}_{\mathrm{int}}$.

For the molecules in \cref{sect:photochemistry}, $\hat{p}$ is very close to~1. To estimate the uncertainty in $\hat{p}$, we use a one-sided confidence interval, i.e., we are $1 - \alpha$ confident that $p \geq \alpha^{1/N_{\mathrm{MC}}}$~\cite{owen_2013}.
This in turn allows us to be $1 - \alpha$ confident that 
\begin{equation}
    \epsilon_{\mathrm{trim}} \leq \sqrt{1 - \alpha^{1/N_{\mathrm{MC}}}}. \label{eq:MC-error}
\end{equation}

\subsubsection{Error analysis} \label{sect:isp_nuc_error}
Overall, our nuclear-ISP unitary is
\begin{multline}
    \tilde{U}_{\mathrm{ISP}}^{(n)}  = \textsc{tc2sm}^{\otimes 3\eta_{n}}_{\bar{n}_{\isp}}\;\widetilde{\textsc{pk}} \; \tilde{U}_{\mathbf{A}^{- \top}}  \widetilde{W}^{(n)} \\
    \cdot \textsc{onb2smb} \; \mathrm{SoSlat}^{(n)}  \widetilde{\textsc{asp}}^{(n)}. \label{eq:u_init_n_approx}
\end{multline}
In \cref{app:isp_error_breakdown}, we show that the error $\epsilon_{\mathrm{ISP}}^{(n)}$ from the approximations in $\tilde{U}_{\mathrm{ISP}}^{(n)}$ is bounded by a sum of seven contributions, 
\begin{multline}\label{eq:epsilon_init_n}
        \epsilon_{\mathrm{ISP}}^{(n)}  \leq \epsilon_{\textsc{asp}}^{(n)} + 2^{3/2}\sum_{i=1}^{3\eta_n} \sum_{\mu=0}^{N_{\mathrm{SMB}}-1} (\epsilon^{(n,c)}_{i \mu} + \epsilon^{(n,q)}_{i \mu})\\ + \epsilon_{\mathrm{shear}} + \epsilon_{\mathrm{ortho}} + \epsilon_{\textsc{pk}} + \epsilon_{\mathrm{trim}}.
\end{multline} 

The first error contribution is the arbitrary-state-preparation error
\begin{equation}
    \epsilon_{\textsc{asp}}^{(n)} = D(|\psi\rangle, |\tilde{\psi}\rangle),
\end{equation}
where $|\psi\rangle$ and $|\tilde{\psi}\rangle$ are defined in \cref{eq:nuc_asp,eq:nuc_asp_approx}. It is bounded by \cite{huggins2024}
\begin{equation}
    \epsilon_{\textsc{asp}}^{(n)} \leq 2\pi \, 2^{-b^{(n)}_{\textsc{asp}}} \log(D^{(n)}).
\end{equation}

The next two contributions, $\epsilon_{i\mu}^{(n,c)}$ and $\epsilon_{i\mu}^{(n,q)}$, result from $\widetilde{W}^{(n)}$ and are analogous to  $\epsilon_{i\mu}^{(e,c)}$ and $\epsilon_{i\mu}^{(e,q)}$ because they result from projecting each SM onto a PWB and then encoding it as an MPS. $\epsilon_{i\mu}^{(n,c)}$ is defined in \cref{eq:nuc_eps_classical} and $ \epsilon^{(n,q)}_{i \mu}$ is
\begin{equation}
    \epsilon^{(n,q)}_{i \mu} = D\big(|\hatsmQ\rangle _{\mathrm{MPS}}, |\hatsmQAprx\rangle_{\mathrm{MPS}} \big).
\end{equation}
In \cref{app:isp_error_breakdown}, we show that
\begin{equation}
    \epsilon^{(n,q)}_{i \mu} < 2^{\frac{7}{2}-b_{\mathrm{rot}}^{(n)}
    }\sum_{i=j}^{n_{\mathrm{ISP}}}m_{i\mu j}^{(n)}\log(2 \bar{m}_{i \mu j}^{(n)} ) .
\end{equation}

The errors $\epsilon_{\mathrm{shear}}$ and $\epsilon_{\mathrm{ortho}}$ come from approximately performing a coordinate transformation using $\tilde{U}_{\mathbf{A}^{- \top}}^{\textsc{lct}}$. We define the shearing error as
\begin{align}
    \epsilon_{\mathrm{shear}} &= D\Big( U_{\mathbf{S}_{\mathrm{L}}}|\tilde{\hat{\Psi}}^{(n)}_{\bar{\mathbf{Q}}}\rangle, \, \tilde{U}_{\mathbf{S}_{\mathrm{L}}}|\tilde{\hat{\Psi}}^{(n)}_{\bar{\mathbf{Q}}}\rangle\Big),
\end{align}
where $|\tilde{\hat{\Psi}}_{\bar{\mathbf{Q}}}^{(n)}\rangle$ is the state that results from projecting $|\hat{\Psi}_{\bar{\mathbf{Q}}}^{(n)}\rangle$ onto a finite PWB (see \cref{app:lct_ssct_error}).
In \cref{app:lct_ssct_error}, we prove \cref{lem:shear_error}, which shows that for a Gaussian initial state, 
\begin{equation} \label{eq:eps_shear_bound}
     \epsilon_{\mathrm{shear}} \leq   2^{1/2} \Big(1 -  e^{- \Delta^2 3\eta_n \lambda_{\mathrm{max}}(\boldsymbol{\Lambda}')  }  \Big)^{1/2}.
\end{equation}
The orthogonal transformation error is 
\begin{equation}
     \epsilon_{\mathrm{ortho}} = D\Big(U_{\bar{\mathbf{O}}^{\top}} U_{\mathbf{S}_{\mathrm{L}}}|\tilde{\hat{\Psi}}^{(n)}_{\bar{\mathbf{Q}}}\rangle,\, \tilde{U}_{\bar{\mathbf{O}}^{\top}} U_{\mathbf{S}_{\mathrm{L}}}|\tilde{\hat{\Psi}}^{(n)}_{\bar{\mathbf{Q}}}\rangle\Big).
\end{equation}
We show in \cref{app:lct_ssct_error} that
\begin{equation} \label{eq:eps_ortho_bound}
    \epsilon_{\mathrm{ortho}} \leq\sum_{l=1}^{\beta} \epsilon^{(\mathrm{2D-shear})}_l,
\end{equation}
where $\epsilon^{(\mathrm{2D-shear})}_l$ is the error of performing the $l$th unitary in the decomposition of $\tilde{U}_{\bar{\mathbf{O}}^{\top}}$, defined in \cref{eq:2d_shear_eps_def}. We call this error the 2D-shearing error because it is non-zero only for 2D shears, since reflections and $\pm \pi/2$-rotations can be implemented exactly. \Cref{lem:ortho} of \cref{app:lct_ssct_error} shows that for a Gaussian initial state,
\begin{equation}
    \epsilon_{P}^{(\mathrm{2D-shear})} \leq  2^{1/2}\big(1  -  e^{- \Delta^2 [\boldsymbol{\Lambda'}_P]_{k_P,k_P}}\big)^{1/2},
\end{equation}
where $k_P = \mathrm{index}(\mathbf{F}_P)$, $\mathbf{F}_P$ represents the $P$th sequential unitary in the decomposition of $\tilde{U}_{\bar{\mathbf{O}}^{\top}}$,  
\begin{equation}
    \mathrm{index}(\mathbf{F}_P) = \begin{cases}
        i, & \mathrm{if}\,\,\mathbf{F}_P = \mathbf{S}_1(i,j,\varphi_{ij})\\
        j, & \mathrm{if}\,\,\mathbf{F}_P = \mathbf{S}_2(i,j,\varphi_{ij}),
    \end{cases}
\end{equation}
$\boldsymbol{\Lambda}_P'  = (\bar{\mathcal{F}}_P)^{\top} \mathbf{S}_L^{- \top} \boldsymbol{\Sigma}' \mathbf{S}_L^{-1} \bar{\mathcal{F}}_P$, and $ \bar{\mathcal{F}}_P = \prod_{l=P}^{0}\mathbf{F}_l^{-1}$.

The phase-kickback error $\epsilon_{\textsc{pk}}$, given in \cref{eq:eps_pk}, can be made very small because the cost of doing so, $C_{\mathrm{Toff}}(\widetilde{\textsc{pk}}) \sim b_{\mathrm{grad}}\log(\epsilon_{\mathrm{PK}}^{-1})$, is logarithmic in $\epsilon_{\mathrm{PK}}^{-1}$.

Finally, the trimming error obeys \cref{eq:epsilon-trim},
\begin{equation}
    \epsilon_{\mathrm{trim}} \leq \sqrt{1 - {\mathcal{N}_{\mathrm{int}}^2}}.
\end{equation}

\subsection{Separable ISP} \label{sect:isp_sep}
Overall, we can prepare any separable initial state to within a controllable error by approximating $U_{\mathrm{ISP}}^{(\mathrm{sep})}$ as
\begin{equation}
\tilde{U}_{\mathrm{ISP}}^{(\mathrm{sep})} = \tilde{U}_{\mathrm{ISP}}^{(e)} \tilde{U}_{\mathrm{ISP}}^{(n)}.
\end{equation}
The resources needed are 
\begin{align}
         C_{\mathrm{anc}}(\tilde{U}_{\mathrm{ISP}}^{(\mathrm{sep})}) & = \mathrm{max}\big\{  C_{\mathrm{anc}}(\tilde{U}_{\mathrm{ISP}}^{(e)}) ,\, C_{\mathrm{anc}}(\tilde{U}_{\mathrm{ISP}}^{(n)})\big\}\\
         C_{\mathrm{Toff}}(\tilde{U}_{\mathrm{ISP}}^{(\mathrm{sep})})  & = C_{\mathrm{Toff}}(\tilde{U}_{\mathrm{ISP}}^{(e)}) +     C_{\mathrm{Toff}}(\tilde{U}_{\mathrm{ISP}}^{(n)}),
\end{align}
where $C_{\mathrm{anc}}$ is the maximum total of all ancilla qubits used at any step of the ISP. The components are
\begin{align}
    C_{\mathrm{anc}}(\tilde{U}_{\mathrm{ISP}}^{(e)}) & = \mathrm{max}\big\{ C_{\mathrm{anc}}(\widetilde{\textsc{asp}}^{(e)}), C_{\mathrm{anc}}(\mathrm{SoSlat}^{(e)}), \nonumber \\ 
    &\qquad C_{\mathrm{anc}}(\textsc{onb2mob}), C_{\mathrm{anc}}(\textsc{asym}), \nonumber \\
    &\qquad C_{\mathrm{anc}}((\widetilde{W}^{(e)})^{\otimes \eta_e})\big\} \\
    C_{\mathrm{Toff}}(\tilde{U}_{\mathrm{ISP}}^{(e)}) & =  C_{\mathrm{Toff}}(\widetilde{\textsc{asp}}^{(e)}) + C_{\mathrm{Toff}}(\mathrm{SoSlat}^{(e)})  \nonumber \\
    & \qquad {} + C_{\mathrm{Toff}}(\textsc{onb2mob}) + C_{\mathrm{Toff}}(\textsc{asym}) \nonumber \\
    & \qquad {} + C_{\mathrm{Toff}}((\widetilde{W}^{(e)})^{\otimes \eta_e} ), \\
    C_{\mathrm{anc}}(\tilde{U}_{\mathrm{ISP}}^{(n)}) & =  3\eta_n n_{\mathrm{ext}} + \mathrm{max}\big\{ C_{\mathrm{anc}}(\widetilde{\textsc{asp}}^{(n)}), \nonumber \\
     &\qquad C_{\mathrm{anc}}(\mathrm{SoSlat}^{(n)}) , C_{\mathrm{anc}} (\textsc{onb2smb}), \nonumber \\
     &\qquad C_{\mathrm{anc}}(\widetilde{W}^{(n)}), C_{\mathrm{anc}}(\tilde{U}_{\mathbf{A}^{- \top}}), \nonumber  \\
     &\qquad C_{\mathrm{anc}}(\widetilde{\textsc{pk}}), C_{\mathrm{anc}}(\textsc{tc2sm})\big\} \\
     C_{\mathrm{Toff}}(\tilde{U}_{\mathrm{ISP}}^{(n)})  & = C_{\Toff}(\widetilde{\textsc{asp}}^{(n)}) +  C_{\mathrm{Toff}}(\mathrm{SoSlat}^{(n)})  \nonumber \\ 
     & \qquad {} +  C_{\mathrm{Toff}}(\textsc{onb2smb}) + C_{\mathrm{Toff}}(\widetilde{W}^{(n)}) \nonumber \\
     & \qquad {} + C_{\mathrm{Toff}}(\tilde{U}_{\mathbf{A}^{- \top}}) C_{\mathrm{Toff}}(\widetilde{\textsc{pk}}) \nonumber \\
     & \qquad {}+ C_{\mathrm{Toff}}(\textsc{tc2sm}^{\otimes 3\eta_{n}}_{\bar{n}_{\isp}}), \label{eq:cost_toff_isp_n}
\end{align}
where $n_{\mathrm{ext}} = \bar{n}_{\mathrm{ISP}} - n_p$.

The error from approximately preparing a separable initial state with $\tilde{U}_{\mathrm{ISP}}^{(\mathrm{sep})}$ is 
\begin{equation}
    \epsilon_{\mathrm{ISP}}^{(\mathrm{sep})} = D \big( |\hat{\Psi}^{(e)} \rangle_{\mathrm{MPS}} |\hat{\Psi}^{(n)} \rangle_{\mathrm{MPS}}, |\tilde{\hat{\Psi}}^{(e)} \rangle_{\mathrm{MPS}} |\tilde{\hat{\Psi}}^{(n)} \rangle_{\mathrm{MPS}}   \big).
\end{equation}
Using the triangle inequality, 
\begin{equation}
    \epsilon_{\mathrm{ISP}}^{(\mathrm{sep})} \leq \epsilon_{\mathrm{ISP}}^{(e)} + \epsilon_{\mathrm{ISP}}^{(n)}, 
    \label{eq:error_isp_sep}
\end{equation}
which is bounded by the sum of \cref{eq:approx_bound,eq:epsilon_init_n}.

\subsection{Non-separable ISP}
\label{sect:Non-BO state prep}
Non-separable states can be prepared by extending the approach for separable states. The main difference is that, instead of initially preparing a product of ONB electronic and nuclear states, we prepare a single ONB state with entangled electronic and nuclear states.   

It is helpful to express $\Psi$ using a collective index $K = (I,J)$ formed by the concatenation of $I$ and $J$, the collective indices that specify electronic and nuclear configurations, respectively. $K$ therefore contains $N_{\mathrm{MOB}} + N_{\mathrm{vib}} N_{\mathrm{SMB}}$ elements. In terms of $K$, \cref{eq:non_BO_kspace} is 
\begin{equation}
    \Psi = \sum_K C_{K}^{(en)} \psi_K^{(en)},
\end{equation}
where $\psi_K^{(en)} = \psi_I^{(e)} \psi_J^{(n)}$.
We assume that the set of $C_K^{(en)}$ contains $D^{(en)}$ non-zero coefficients. 

To begin, we use arbitrary state preparation to prepare
\begin{equation}
    \textsc{asp}^{(en)}|0\rangle = |\psi\rangle = \sum_{i}C_i^{(en)}|i\rangle,
\end{equation}
where 
$C_i^{(en)}$ is the $i$th non-zero coefficient of the set $\{C_K^{(en)} \}$, and 
$\textsc{asp}^{(en)}$ is the arbitrary-state-preparation unitary. As for separable ISP, we implement an approximate unitary $\widetilde{\textsc{asp}}^{(en)}$ such that
\begin{equation}
    \widetilde{\textsc{asp}}^{(en)}|0\rangle = |\tilde{\psi}\rangle = \sum_{i}\widetilde{C}_i^{(en)}|i\rangle,
\end{equation}
where $\widetilde{C}_i^{(en)}$ is an approximation to $C_i^{(en)}$. Using $b^{(en)}_{\textsc{asp}}$ precision qubits, the resources needed to implement $\widetilde{\textsc{asp}}^{(en)}$ are the same as in \cref{eq:canc_asp,eq:ctoff_asp} with $D^{(e)}$ and $b^{(e)}_{\textsc{asp}}$ replaced by $D^{(en)}$ and $b^{(en)}_{\textsc{asp}}$, respectively. 

Next, we convert to an approximate ONB representation of $|\Psi\rangle$,
\begin{align}
    |\tilde{\Psi}\rangle_{\mathrm{ONB}} & = \sum_K \widetilde{C}_{K}^{(en)} |K \rangle \\
    & =  \sum_{I,J} \widetilde{C}_{IJ} |I \rangle |J\rangle.
\end{align}
We prepare $ |\tilde{\Psi}\rangle_{\mathrm{ONB}}$ using the SoSlat unitary $\mathrm{SoSlat}^{(en)}$,
\begin{equation}
    \mathrm{SoSlat}^{(en)}|\tilde{\psi}\rangle = |\tilde{\Psi}\rangle_{\mathrm{ONB}}
\end{equation}
using 
\begin{align} 
    C_{\mathrm{anc}}(\mathrm{SoSlat}^{(en)}) &= 5\log D^{(en)} -3\\
    C_{\mathrm{Toff}}(\mathrm{SoSlat}^{(en)}) &\leq D^{(en)}(2 \log D^{(en)} + 3).
\end{align} 

From here, the remaining ISP steps are the same as for preparing the electronic and nuclear states separably. The non-separable ISP unitary is 
\begin{equation} \label{eq:non_sep_uinit_approx}
    \tilde{U}_{\mathrm{ISP}} = (\tilde{U}_{\mathrm{ISP}}^{(e)'} \otimes \tilde{U}_{\mathrm{ISP}}^{(n)'}) \big(\mathrm{SoSlat}^{(en)}  \widetilde{\textsc{asp}}^{(en)}\big),
\end{equation}
where
\begin{align}
    \tilde{U}_{\mathrm{ISP}}^{(e)'} &= (\widetilde{W}^{(e)})^{\otimes 3\eta_e} \textsc{asym}\; \textsc{onb2mob} \label{eq:UISP_e_prime}\\
    \tilde{U}_{\mathrm{ISP}}^{(n)'} & = \textsc{tc2sm}_{\bar{n}_{\mathrm{ISP}}}^{\otimes 3\eta_n}\, \widetilde{\textsc{pk}} \, \tilde{U}_{\mathbf{A}^{-\top}} \, \widetilde{W}^{(n)} \textsc{onb2smb} \label{eq:UISP_n_prime}.
\end{align}

The resources needed to implement $\tilde{U}_{\mathrm{ISP}}$ are
\begin{align}
     C_{\mathrm{anc}}(\tilde{U}_{\mathrm{ISP}}) & = \mathrm{max}\big\{ C_{\mathrm{anc}}(\tilde{U}^{(e)'}),  C_{\mathrm{anc}}(\tilde{U}^{(n)'}), \nonumber \\ 
     &\qquad C_{\mathrm{anc}}(\widetilde{\textsc{asp}}^{(en)}),C_{\mathrm{anc}}(\mathrm{SoSlat}^{(en)})\big\} \label{eq:anc_cost_isp_full}  \\
     C_{\mathrm{Toff}}(\tilde{U}_{\mathrm{ISP}}) & =   C_{\mathrm{Toff}}(\tilde{U}_{\mathrm{ISP}}^{(e)'}) + C_{\mathrm{Toff}}(\tilde{U}_{\mathrm{ISP}}^{(n)'})  \nonumber \\
     &\qquad +C_{\mathrm{Toff}}(\widetilde{\textsc{asp}}^{(en)}) + C_{\mathrm{Toff}}(\mathrm{SoSlat}^{(en)}). \label{eq:toff_cost_isp_full}
\end{align}
Therefore, preparing a non-separable state costs more than a separable one by only the small additive term $C_{\mathrm{Toff}}(\widetilde{\textsc{asp}}^{(en)}) + C_{\mathrm{Toff}}(\mathrm{SoSlat}^{(en)})$.

The error from using $\tilde{U}_{\mathrm{ISP}}$ consists of seven terms. In \cref{app:isp_error_breakdown}, we show that 
\begin{align} 
    \epsilon_{\mathrm{ISP}} & \leq  \epsilon_{\textsc{asp}}^{(en)} + \epsilon_{\mathrm{trim}} +  2^{3/2} \eta_e \sum_{a=1}^{N_{\mathrm{MOB}}}\epsilon_{a}^{(e)}  \nonumber \\
    & \quad + 2^{3/2}\sum_{i=1}^{3\eta_n} \sum_{\mu=0}^{N_{\mathrm{SMB}}-1} \big(\epsilon^{(n,c)}_{i \mu} +\epsilon^{(n,q)}_{i \mu} \big)  \nonumber \\
    & \quad + \sum_{I,J}|C_{IJ}| \big( \epsilon_{\mathrm{shear}}' + \epsilon'_{\mathrm{ortho}} + \epsilon_{\textsc{pk}}' \big),
    \label{eq:error_budget_non_sep}
\end{align}
where $\epsilon_{\textsc{asp}}^{(en)} = D(\textsc{asp}^{(en)}|0\rangle, \widetilde{\textsc{asp}}^{(en)}|0\rangle)$
and $\epsilon'_{\mathrm{shear}}$, $\epsilon'_{\mathrm{ortho}}$, and  $\epsilon'_{\textsc{pk}}$ are defined in \cref{eq:eps_shear_nuc_prime,eq:eps_ortho_nuc_prime,eq:eps_pk_nuc_prime}. These error terms are similar to their unprimed counterparts ($\epsilon_{\mathrm{shear}}$, $\epsilon_{\mathrm{ortho}}$, and $\epsilon_{\textsc{pk}}$, but are given in terms of the two-norm. Nevertheless, they obey the same upper bounds as their the unprimed counterparts.

\subsection{Asymptotic cost of ISP} \label{sect:asymp_isp}
Overall, ISP scales quadratically in system size, up to logarithmic terms. Its asymptotic time complexity for both separable and non-separable states is
\begin{equation}\label{eq:asympt_Toff_cost_ISP}
     C_{\Toff}(U_{\mathrm{ISP}}) =   \cO\Big( \eta^2 \log^2 \big( \eta^2 \epsilon^{-1} \sqrt{\log(1/\epsilon)} \big)\Big) ,
\end{equation}
which arises from the most expensive subroutine, $\tilde{U}_{\mathbf{A}^{- \top}}^{\textsc{lct}}$. The scaling follows from \cref{eq:LCT_resources_c}, which shows that
$C_{\Toff}(\tilde{U}_{\mathbf{A}^{- \top}}^{\textsc{lct}}) = \mathcal{O}(\eta_n^2 \bar{n}_{\mathrm{ISP}}^2) = \mathcal{O}(\eta^2 \bar{n}_{\mathrm{ISP}}^2)$, with $\bar{n}_{\mathrm{ISP}} = n_{\mathrm{pad}} + n_{\mathrm{ISP}}$. Then, from \cref{eq:n_pad_bound_LCT} and the fact that $N_{\mathrm{ISP}} \Vert \mathbf{L} \Vert_{\infty} \gg \beta$ asymptotically, we have $n_{\mathrm{pad}} = \mathcal{O}(\log(\eta) + n_{\mathrm{ISP}})$. Then, using $n_{\mathrm{ISP}} = \log(N_{\mathrm{ISP}})$ and $N_{\mathrm{ISP}}\propto \bar{N}$, \cref{eq:N} establishes $n_{\mathrm{ISP}} = \mathcal{O}(\log( \eta^2 \sqrt{\log(1/\epsilon)}/\epsilon))$. This asymptotic analysis treats the following quantities as constants because they do not depend on $\eta_e$, $\eta_n$, $N$, or $L$: $D^{(e)}$, $D^{(n)}$, $D^{(en)}$, $N_{\mathrm{MOB}}$, $N_{\mathrm{SMB}}$, $\gamma_{\mathrm{max}}$, $l_{\mathrm{max}}$, $N_{g}$, $\Sigma$, $\omega_i$, $\Vert\mathbf{L}\Vert_{\infty}$, and $\sum_{I,J}|C_{IJ}|$. We also assume that all errors go as $\cO(\epsilon)$, where $\epsilon$ is the total error bound on the final observable. 

Separable and non-separable ISP have the same asymptotic cost because the two differ only in the dominant term $\tilde{U}_{\mathbf{A}^{- \top}}^{\textsc{lct}}$ for non-separable ISP being multiplied by $\sum_{I,J}|C_{IJ}|$, which we assume to be constant.

Our ISP routine is exponentially faster in grid size $N$ than prior work. Because it is quadratic in $\bar{n}_{\mathrm{ISP}}$, our approach scales as $\mathcal{O}(\log^2 N)$. By contrast, previous work performed the coordinate transformation on a classical computer, using classical data loading with cost $\cO(\eta N)$ to transfer the result onto a quantum computer~\cite{Jornada2025}. Our approach instead performs the coordinate transformation coherently on the quantum computer, avoiding the costly data loading.

\section{Time evolution by qubitization} \label{sect:sim}

Our algorithm continues by time-evolving the initial state for a specified time $t$ under the molecular Hamiltonian $H$. Because $H$ does not include terms that depend on the spin of the electrons, the spin registers are suppressed in this section and we concern ourselves only with the registers that represent the simulation grid.

We use the qubitization-based approach to time-evolution by first block-encoding $H$ and then performing QSP to synthesize the propagator $e^{-iHt}$. To this end, we explicitly block encode the full molecular Hamiltonian of \cref{eq:mol-ham} in a way that can be understood as an extension of the block-encoding of the electronic Hamiltonian~\cite{su2021}.

Qubitization is a method to construct the qubiterate, a unitary walk operator that one can use to synthesize functions of the Hamiltonian---including the propagator---using QSP~\cite{Low2017,Low2019hamiltonian}.
The starting point for the approach is to write $H$ as a linear combination of unitaries (LCU),
\begin{equation} \label{eq:LCU}
    H=\sum_{\ell=1}^{N_{\LCU}} \alpha_{\ell} H_{\ell}, \quad \lambda_H=\sum_{\ell=1}^{N_{\LCU}} \alpha_\ell,
\end{equation}
where the terms $H_\ell$ are self-inverse unitaries and the coefficients $\alpha_{\ell}$ are positive. Using this decomposition, $H$ is embedded into a larger unitary $U_{H}$ whose upper-left sub-block contains the normalized version of $H$,
\begin{equation}\label{eq:UH}
U_H=\left(\begin{array}{cc}
H /\lambda_H & \cdot \\
\cdot & \cdot
\end{array}\right).
\end{equation}
To construct such a unitary block-encoding, qubitization prescribes two oracles, $\PREP_H$ and $\SEL_H$. $\PREP_H$ prepares a state that encodes the LCU coefficients $\alpha_{\ell}$ on an ancillary register of size $\bits{N_{\LCU}}$,
\begin{equation}\label{eq:Prep}
\PREP_H|0\rangle=\sum_{\ell=1}^{L} \sqrt{\frac{\alpha_{\ell}}{\lambda_{H}}}|\ell\rangle,
\end{equation}
where $\ket{0}$ is the all-zeros state.
$\SEL_H$ selectively applies the unitary $H_\ell$ on the data register if the ancilla register is in state $\ket{\ell}$, i.e.,
\begin{equation}
\SEL_H|\ell\rangle|\psi\rangle=|\ell\rangle \otimes H_{\ell} |\psi\rangle.
\end{equation}
These oracles combine to give the block-encoding $U_{H}$,
\begin{equation}\label{eq:PSP}
U_H=\PREP_H  \SEL_H  \PREP_H^{\dagger}.
\end{equation}
With this block-encoding, we construct the qubiterate 
\begin{equation}\label{Qubit_Walk}
\cW=(2|0\rangle\langle 0|-I)  \PREP_H^{\dagger}  \SEL_H  \PREP_H.
\end{equation}

To perform time-evolution, we need to transform $\cW$ to the propagator $e^{-iHt}$. This can be done using QSP, which, given a real function $f$ and a unitary $e^{iO}$ generated by a Hermitian operator $O$, constructs the unitary $e^{if(O)}$ ~\cite{Low2017,Low2019hamiltonian}. Since $\cW = \exp(-i\,\mathrm{arccos}(H/\lambda_H))$, it is generated by $\mathrm{arccos}(H/\lambda_H)$, meaning that we can synthesize the propagator $e^{-iHt}$ by performing QSP on the qubiterate $\cW$ with $f= \lambda_{H}t \cos(\cdot)$. This is done by repeatedly performing controlled applications of $\cW$ and its inverse interleaved with single-qubit rotations applied to the control qubit. The number of repetitions, conventionally referred to as queries, required to approximate the propagator to precision $\epsilon$ for a system with LCU-norm $\lambda_{H}$ evolving for time $t$ is~\cite{Low2017}
\begin{equation} \label{eq:qubitz_query}
\cO\left(\lambda_{H} t+\frac{\log(1/\epsilon)}{\log\log(1/\epsilon)}\right).
\end{equation}
The Toffoli complexity $C_{\Toff}$ for time-evolution is then given by this number of repetitions multiplied by the Toffoli complexity of synthesizing the qubiterate $\cW$. 

For the first-quantized electronic Hamiltonian in a PWB, the qubiterate can be constructed using $\cO(\eta_{e}n_{p})$ Toffoli gates~\cite{su2021}. In the following sections, we use the same techniques to construct the qubiterate for the first-quantized molecular Hamiltonian in the PWB and show that it can be done using $\cO(\eta n_{p})$ Toffoli gates. As $\eta_{e}>\eta_{n}$ for most molecules, this represents only a modest increase of the Toffoli complexity compared to the electronic Hamiltonian. However, the number $n_{p}$ of plane-waves per particle per dimension required to resolve the molecular wavefunction to a given resolution is higher than for the electronic wavefunction, which further increases the Toffoli complexity.

\subsection{LCU decomposition of the molecular Hamiltonian}
We now give an LCU decomposition for the molecular Hamiltonian and construct $\PREP_{H}$ and $\SEL_{H}$ oracles to block-encode it.
We order the particles by index $j$, so that electrons are labelled by $j\in \{1,2,\ldots,\eta_{e}\}$ and the nuclei by $j\in \{\eta_{e}+1,\ldots,(\eta_{e}+\eta_{n})\}$. 
We rewrite the kinetic and potential terms $T$ and $V$ in terms of a basis of unitary and Hermitian operators $H_{\ell_{T}}$ and $H_{\ell_{V}}$ with positive real coefficients $\alpha_{\ell_{T}}$ and $\alpha_{\ell_{V}}$ such that $T=\sum_{\ell_{T}} \alpha_{\ell_{T}} H_{\ell_{T}}$ and $V=\sum_{\ell_{V}} \alpha_{\ell_{V}} H_{\ell_{V}}$. Explicitly,
\begin{align}
T&=\sum_{j=1}^{\eta} \sum_{w\in \{x,y,z\}}\sum_{r,s=0}^{n_{p}-2}\sum_{b\in \{0,1\}} \frac{\pi^2 2^{r+s}}{\Omega^{2/3}m_{j}} H_{(b,j,w,r,s)}\label{eq:LCU_decomp_kin}\\
V&=\sum_{\bnu\in G_{0}}\sum_{i\neq j=1}^{\eta}\sum_{b\in\{0,1\}}\frac{\pi |\zeta_{i}||\zeta_{j}|}{\Omega \norm{\mathbf{k}_{\bnu}}^{2}_2}H_{(b,i,j,\bnu)},\label{eq:LCU_decomp_pot}
\end{align}
where $\Omega=L^3$ is the real-space simulation-grid volume, $\mathbf{k}_{\bnu}$ is defined in \cref{eq:k_mu}, $G_0=G \setminus \{(0,0,0)\}$ is the basis set of plane waves $G$ without the singular zero mode, and $|\zeta _{j}|$ is the absolute value of the electric charge $\zeta_{j}$ of the $j$th particle (for an electron, $\zeta_{j}=-1$). The self-inverse unitaries are 
\begin{align}\label{eq:H_l_T}
    H_{(b,j,w,r,s)}&=\sum_{\mathbf{p}\in G} (-1)^{b\left(p_{w,r} p_{w,s}\oplus 1\right)}\ket{\mathbf{p}}\bra{\mathbf{p}}_{j}\\
    H_{(b,i,j,\bnu)}&=\sum_{\mathbf{p},\mathbf{q}\in G}\left(-1\right)^{f(\mathbf{p},\mathbf{q},\boldsymbol{\nu},i,j)} \nonumber\\
    &\quad\,\quad\,\quad\, \times\ket{\mathbf{p}+\boldsymbol{\nu}}\bra{\mathbf{p}}_{i} \ket{\mathbf{q}-\boldsymbol{\nu}}\bra{\mathbf{q}}_{j},\label{eq:H_l_V}
\end{align}
where $p_{w,r}$ denotes the $r$th bit of the binary representation of the $w$-component of the grid point $\mathbf{p}$ (likewise for $p_{w,s}$). Negative terms in $V$, caused by electrostatic attraction, are accounted for by the Boolean function 
\begin{multline}\label{eq:Boolfunc_sel}
 f(\mathbf{p},\mathbf{q},\boldsymbol{\nu},i,j)= b\left([(\mathbf{p}+\boldsymbol{\nu})\notin G]\lor[(\mathbf{q}-\boldsymbol{\nu}\notin G)]\right)\\
 \oplus \left([i< \eta_{e}+1]\oplus[j<\eta_{e}+1]\right).   
\end{multline}
The LCU norms of the kinetic and potential terms are
\begin{align}
    \lambda_{T}&=\frac{6\pi^2}{\Omega^{2/3}} (2^{n_{p}-1}-1)^2\lambda_{m}=\cO\left(\frac{\eta N_{\grid}^{2/3}}{\Omega^{2/3}}\right)\label{eq:lambda_T}\\
    \lambda_{V}&=\frac{\sum_{i\neq j=1}^{\eta} |\zeta_{i}||\zeta_{j}|} {2\pi \Omega^{1/3}}\lambda_{\bnu}=\cO\left(\frac{\eta^2 N_{\grid}^{1/3}}{\Omega^{1/3}}\right)
    \label{eq:lambda_V}
\end{align}
where $\lambda_{m}=\sum_{j=1}^{\eta} m_{j}^{-1}$ and $\lambda_{\bnu}=\sum_{\boldsymbol{\nu} \in G_0} \|\boldsymbol{\nu}\|^{-2}_2$.

\subsection{Overview of block encoding}\label{sec:block-encoding_overview}
We construct the block-encoding for the molecular Hamiltonian by block-encoding the kinetic and potential terms separately and taking a linear combination of them. For this purpose, we define $\PREP$ and $\SEL$ routines for the kinetic and potential terms separately, such that 
\begin{align}\label{eq:blockenc_T}
\bra{0} \block{T}\ket{0}=\frac{T}{\lambda_T},\\\label{eq:blockenc_V}
\bra{0} \block{V}\ket{0}=\frac{V}{\lambda_V}.
\end{align}
These routines are defined via the LCU decompositions for the kinetic term $T$ and the potential term $V$ in \cref{eq:LCU_decomp_kin} and \cref{eq:LCU_decomp_pot}, respectively. The state preparation routines are therefore
\begin{align}\label{eq:PREP_T}
    \PREP_{T}\ket{0}&=\!\sum_{b,j,w,r,s} \sqrt{\frac{2^{r+s}}{6\lambda_{m} \left(2^{n_{p}-1}-1\right)m_{j}}} \ket{b,j,w,r,s}\\
    \PREP_{V}\ket{0}&=\sum_{\bnu, i\neq j} \sqrt{\frac{|\zeta_{i}||\zeta_{j}|}{\norm{\bnu}_{2}\lambda_{\bnu}\sum_{k\neq l=1}^{\eta} |\zeta_{k}||\zeta_{l}|}}\ket{\bnu,i,j}\label{eq:PREP_V}, 
\end{align}
where the summation ranges are as given in \cref{eq:LCU_decomp_kin,eq:LCU_decomp_pot} and $\ket{0}$ denotes the all-zeros state on the registers supporting the states in \cref{eq:PREP_T,eq:PREP_V}.

The routines $\SEL_{T}$ and $\SEL_{V}$ are determined by \cref{eq:H_l_T,eq:H_l_V} to act as
\begin{align}
\nonumber \SEL_{T}&\ket{b}\ket{j}\ket{w}\ket{r}\ket{s}\ket{\mathbf{p}}_{j}=\\ 
\label{eq:SEL_T_main}& (-1)^{b\left(p_{w,r} p_{w,s}\oplus 1\right)} \ket{b}\ket{j}\ket{w}\ket{r}\ket{s}\ket{\mathbf{p}}_{j}\\
\nonumber\SEL_{V}& \ket{b}\ket{\bnu}\ket{i}\ket{j}\ket{\mathbf{p}}_{i}\ket{\mathbf{q}}_{j}=\\&(-1)^{f(\mathbf{p},\mathbf{q},\boldsymbol{\nu},i,j)}\ket{b}\ket{\bnu}\ket{i}\ket{j}
\ket{\mathbf{p}+\bnu}_{i}\ket{\mathbf{q}-\bnu}_{j}\label{eq:SEL_V},
\end{align}
where $\SEL_{T}$ acts on the ancilla register prepared in the state $\PREP_{T}\ket{0}$ as well as the data register and $\SEL_{V}$ acts on the register prepared in the state $\PREP_{V}\ket{0}$ and the data register. 

To construct the block-encoding for the full Hamiltonian $H=T+V$, we combine the routines in \cref{eq:PREP_T,eq:PREP_V,eq:SEL_T_main,eq:SEL_V} to construct $\PREP_{H}$ and $\SEL_{H}$. More precisely, 
\begin{align}\label{eq:PREP_H}
\PREP_{H}\ket{0}&= \ket{\cL_{\theta}}_{a}\otimes \PREP_T \ket{0} \otimes \PREP_{V}\ket{0}.
\end{align}
Here, the additional ancilla qubit $a$ is prepared as
\begin{align}\label{eq:control_state}
\ket{\cL_{\theta}}_{a}=\cos(\theta)\ket{0}_{a}+\sin(\theta)\ket{1}_{a}
\end{align}
with 
$\theta=\arccos\left(\sqrt{\frac{\lambda_T}{\lambda_T+\lambda_V}}\right)=\arcsin\left(\sqrt{\frac{\lambda_V}{\lambda_T+\lambda_V}}\right)$
chosen so that the Hamiltonian is block-encoded with the correct LCU norm $\lambda_{H}=\lambda_{T}+\lambda_{V}$. The corresponding $\SEL_{H}$ operator uses register $a$ as a control qubit to select between applying $\SEL_{T}$ and $\SEL_{V}$, i.e.,
\begin{align}
    \SEL_H&=
|0\rangle\langle 0|_{a} \otimes \SEL_T+|1\rangle\langle 1|_{a} \otimes \SEL_{V}.
\end{align}
The block-encoding $\PREP_{H}^{\dagger}\SEL_{H}\PREP_{H}$ then satisfies 
\begin{align}
    H=\lambda_{H}(\bra{0}\otimes I)\,\PREP_{H}^{\dagger}\SEL_{H}\PREP_{H}\, (\ket{0}\otimes I).
    \label{eq:block_simple}
\end{align}

The subroutines $\PREP_{T}$ and $\PREP_{V}$ can be decomposed into more elementary protocols, as described in \cref{subsubsec: time-evo_PREP_T,subsec:time-evo_PREP_V} and summarized here. 

The state $\PREP_{T}\ket{0}$ given in \cref{eq:PREP_T} can be decomposed into the product state
\begin{align}
    \PREP_{T}\ket{0}=\ket{+}_{b}\otimes \ket{\cL_{m}}_{c}\otimes   \ket{\cL_{w}}_{d} \otimes \ket{\cL_{r,s}}_{e,f} 
    \label{eq:prept-decomp}
\end{align}
where the mass state
\begin{align}
\ket{\cL_{m}}=\frac{1}{\sqrt{\sum_{j=1}^{\eta}m_{j}^{-1}}}\sum_{j=1}^{\eta_{}}\frac{1}{\sqrt{m_{j}}}\ket{j}_{c} 
\end{align}
encodes the inverse masses $m_{j}^{-1}$ of the particles. Next, we have the uniform superposition over three basis states,
\begin{align}
    \ket{\cL_{w}}=\frac{1}{\sqrt{3}} \sum_{w=0}^2|w\rangle_{d},
\end{align}
which acts as an index register allowing $\SEL_{T}$ to iterate over the three spatial components of any single-particle momentum register. Finally, the state 
\begin{align}\label{eq:state_rs}
\ket{\cL_{r,s}}=\frac{1}{2^{n_p-1}-1} \sum_{r, s=0}^{n_p-2} 2^{(r+s) / 2}\ket{r}_{e}\ket{s}_{f}\end{align}
allows $\SEL_{T}$ to iterate over the $r$th and $s$th bits of the $w$th spatial component of the momentum register of a particle. 
Because $\PREP_{T}\ket{0}$ is a product state, we prepare each factor independently, with detailed procedures and costs given in \cref{subsubsec: time-evo_PREP_T}.

We also write $\PREP_{V}\ket{0}$ as the product state 
\begin{align}
\PREP_{V}\ket{0}=\ket{\cL_{\bnu}}_{h}\otimes \ket{\cL_{\zeta\zeta}}_{\ell,m},
\label{eq:prepv-decomp}
\end{align}
where the momentum state
\begin{align}
    \ket{\cL_{\bnu}}_{h}=\frac{1}{\sqrt{\lambda_{\bnu}}}\sum_{\bnu \in G_0} \frac{1}{\|\bnu\|_2}\ket{\bnu}_{h},
\end{align}
encodes the inverse momentum of a particle~\cite{su2021}. Similarly, the Coulomb state 
\begin{align}\label{eq:Coulomb_state}
    \ket{\cL_{\zeta\zeta}}_{\ell,m}=\sum_{i\neq j=1}^{\eta}\frac{\sqrt{|\zeta_{i}||\zeta_{j}|}}{\sqrt{\sum_{i\neq j=1}^{\eta}|\zeta_{i}||\zeta_{j}|}}\ket{i,j}_{\ell,m}
\end{align}
encodes the absolute values of the charge pairs occurring in the Coulomb potential (\cref{eq:LCU_decomp_pot}).

In practice, we can reduce resource costs by replacing $\PREP_{V}$ with the approximate version 
\begin{align}\label{eq:proxPREP_V_tensor}
\proxPREP_{V}\ket{0}_{g,h,k,\ell,m}=\ket{\tilde{\cL}_{\bnu}}_{g,h}\otimes \ket{\tilde{\cL}_{\zeta\zeta}}_{k,\ell,m}
\end{align}
that is the product of approximations $\ket{\tilde{\cL}_{\bnu}}$ and $\ket{\tilde{\cL}_{\zeta\zeta}}$ to the momentum and Coulomb states, respectively. For more details, see \cref{subsec:time-evo_PREP_V}, where we show that this state can be prepared with a success probability $p_{V}$ that is lower-bounded by a constant close to $1/8$, with the actual value typically closer to $1/4$ for most molecules.

The approximate state preparation also necessitates a modified implementation of $\SEL_{H}$. This implementation is controlled on ancilla qubits in the $\proxPREP_{H}$ register that flag successful state preparation and can be written as 
\begin{align}
    \widetilde{\SEL}_{H}=\sum_{x}\dyad{x}{x}\otimes \SEL_{x},
\end{align}
where bit string $x$ iterates over flag-qubit values and $\SEL_{x}\in\{\SEL_{T}, \SEL_{V}, I\}$, as detailed in \cref{subsec:SEL_H_implement}.

By combining $\proxPREP_{H}$ with $\proxSEL_{H}$, we arrive at an approximate block encoding $\proxblock{H}$ of $H$,
\begin{align}
     \big\|H-\tilde{\lambda}_{H}\bra{0}\proxblock{H}\ket{0}\big\|\leq \epsilon_{H},
 \end{align}
where the error $\epsilon_{H}$ is defined in \cref{eq:epsilon_H}.
In \cref{app:SEL_H_implement}, we show that the LCU norm $\tilde{\lambda}_{H}$ of our block-encoded Hamiltonian is
\begin{align}\label{eq:norm_LCU}
\tilde{\lambda}_{H}=\frac{1}{P_{\eq}}\max\left\{\lambda_{T}+\lambda_{V},\frac{\lambda_{V}}{p_{V}}\right\},
\end{align}
where $p_{V}$ is the success probability of $\proxPREP_{V}$ and 
$P_{\eq}=\Ps(3,8)\Ps(\eta,b_{r})\Ps(\eta,b_{r})^2$ is a product of success probabilities for the preparation of uniform superposition states that occur for both $\PREP_{T}$ and $\proxPREP_{V}$. Each of these probabilities is very close to $1$, with the analytic expression provided in \cite{su2021}.

\subsection{\texorpdfstring{$\PREP_{T}$}{PREPT}}
\label{subsubsec: time-evo_PREP_T}
Implementing $\PREP_{T}$ for the block-encoding of the kinetic term $T$ requires the preparation of the three states in \cref{eq:prept-decomp}: the mass state $\ket{\cL_{m}}$, the uniform-superposition state $\ket{\cL_{w}}$, and the state $\ket{\cL_{r,s}}$. Our main focus in this section is to describe and cost the preparation of $\ket{\cL_{m}}$, as the preparation protocols for the other two states are the same as for the known block-encoding of the electronic Hamiltonian~\cite{su2021}. By contrast, the preparation of $\ket{\cL_{m}}$ is only necessary for the block-encoding of the molecular Hamiltonian. 

To encode the masses of the particles, we prepare the mass state 
\begin{align}
    \ket{\cL_{m}}_{c}= \frac{1}{\sqrt{\sum_{j=1}^{\eta}m_{j}^{-1}}}\sum_{j=1}^{\eta} \frac{1}{\sqrt{m_{j}}}\ket{j}_{c}.
\end{align}
To do so, we use coherent alias sampling with standard $\QROM$, a technique to prepare an arbitrary quantum state, with $L$ known coefficients, up to error at most $\epsilon$, in time $\cO\left(L+\log(1/\epsilon)\right)$~\cite{babbush2018}. We obtain the approximate mass state (entangled with a junk register of ancillas)
\begin{align}
    \ket{\tilde{\cL}_{m}}_{c,\anc}= \frac{1}{\sqrt{\sum_{j=1}^{\eta}\tilde{m}_{j}^{-1}}}\sum_{j=1}^{\eta} \frac{1}{\sqrt{\tilde{m}_{j}}}\ket{j}_{c}\ket{\junk_{j}}_{\anc}\label{eq:mass_state}
\end{align}
by using alias sampling with $L=\eta$ and precision $2^{-\mu_{T}}\eta^{-1}$, which guarantees an accuracy
\begin{align}
\max_{j\in \{1,2,...,\eta\}}\left|\frac{1}{m_{j}}-\frac{1}{\tilde{m}_{j}}\right|\leq \frac{2^{-\mu_{T}}}{\eta}\Bigg(\sum_{j=1}^{\eta}\frac{1}{m_{j}}\Bigg),
\end{align} 
where $\mu_{T}$ is the number of bits of precision used to approximate the coefficients of the state.

Alias sampling consists of several subroutines, whose Toffoli costs are~\cite{babbush2018}
\begingroup
\allowdisplaybreaks
\begin{align}
C_{\Toff}(\UNIF(\eta,b_{r}))&=3n_{\eta}+2b_{r}-9\\
C_{\Toff}(\QROM_{\eta})&= \eta-2\\
C_{\Toff}(\COMP_{\mu_{T}})&=\mu_{T}\\
C_{\Toff}(\textsc{swap}(\ket{\ell}\leftrightarrow\ket{\textrm{alt}_{\ell}}))&=n_{\eta},
\end{align}
\endgroup
where $n_{\eta}=\lceil\log(\eta)\rceil$. For the uniform state preparation $\UNIF(\eta, b_{r})$, we set the ancilla-qubit rotation parameter $b_{r}=8$, which ensures a sufficiently high probability of success~\cite{su2021}.
The Toffoli cost of preparing $\ket{\tilde{\cL}_{m}}$ is then
\begin{align}
C_{\Toff}(\ket{\tilde{\cL}_{m}})&= \eta+\mu_{T}+ 4n_{\eta}+5.
\label{eq:C_Toff_m}
\end{align}
In \cref{app: bound_T}, we show that if the block-encoded kinetic energy $\tilde{T}=\lambda_{T}\bra{0}\block{T}\ket{0}$ needs to be accurate to $\|T-\tilde{T}\| \leq \epsilon_{T}$,
the number of qubits needed to represent the coefficients is at least 
\begin{align}
\mu_{T}=\left\lceil\log\left(\lambda_{T}/\epsilon_{T}\right)\right\rceil. 
\end{align}
Therefore, the asymptotic Toffoli cost of the mass-state preparation is $\cO\left(\eta+\left\lceil\log\left(\lambda_{T}/\epsilon_{T}\right)\right\rceil\right)$. 

The ancilla cost of preparing $\ket{\tilde{\cL}_{m}}$,
\begin{align}
C_{\anc}(\ket{\tilde{\cL}_{m}})=3n_{\eta}+3\mu_{T}+2,\label{eq:anc_cost_mass_state}
\end{align}
is the sum of the $n_{\eta}$ qubits in register $c$, the $n_{\eta}+2\mu_{T}+1$ junk qubits in the $\anc$ register in \cref{eq:mass_state}, and the ancillas used within the subroutines,
\begin{align}
C_{\anc}(\UNIF(\eta,b_{r}))&=2\\
C_{\anc}(\QROM_{\eta})&= n_{\eta}\\
C_{\anc}(\COMP_{\mu_{T}})&=\mu_{T}-1.
\end{align}

When unpreparing $\ket{\tilde{\cL}_{m}}$, we use measurement-based uncomputation~\cite{Berry_2019} to uncompute the $\QROM$ because it is cheaper than uncomputing via standard inversion. It involves measuring the junk registers in \cref{eq:mass_state} in the Hadamard basis and performing adaptive phase corrections. This erasure process has a Toffoli cost of 
\begin{align}
C_{\Toff}(\QROM^{\dagger})=\Er(\eta)=\min_{k\in \mathbb{N}} \left(\left\lceil\frac{\eta}{2^k}\right\rceil+2^k\right),
\end{align}
which is approximately $2\sqrt{\eta}$ and therefore quadratically smaller than the cost to apply $\QROM$~\cite{Berry_2019}. It also needs 
\begin{align}
C_{\anc}(\QROM^{\dagger})=
n_{\Er}=\argmin_{k\in \mathbb{N}} \left(\left\lceil\frac{\eta}{2^k}\right\rceil+2^k\right)
\end{align}
ancilla qubits.
The inequality test $\COMP_{\mu_{T}}$ can be uncomputed using only Cliffords by keeping around the flag register storing the test outcome, on which we control the SWAPs. We thus only have to consider the additional Toffoli cost to invert the controlled SWAPs and the uniform state preparation, which we do via standard inversions. Hence, with slight abuse of notation, the costs to unprepare $\ket{\tilde{\cL}_{m}}$ are
\begin{align}
C_{\anc}(\bra{\tilde{\cL}_{m}})&=n_{\Er}\label{eq:anc_cost_unprep_mass}\\
 C_{\Toff}(\bra{\tilde{\cL}_{m}})&= \Er(\eta)+4n_{\eta}+7.\label{eq:Toff_cost_unprep_mass}
\end{align}

The second state we need to prepare is the uniform superposition $\ket{\cL_{w}}$. We do so using the uniform-state-preparation $\UNIF(3,8)$, whose costs are~\cite{su2021}
\begin{align}
    C_{\anc}(\ket{\cL_{w}})&= 4\\ 
    C_{\Toff}(\ket{\cL_{w}})&=13. \label{eq:Toff_cost_w}
\end{align}

Finally, we review the preparation of $\ket{\cL_{r,s}}$~\cite{su2021}. As shown in \cref{eq:state_rs}, $\ket{\cL_{r,s}}$ is a product state of the form 
\begin{align}
    \ket{\cL_{r,s}}_{e,f}= \ket{\cL_{r}}_{e}\otimes \ket{\cL_{s}}_{f},
\end{align}
with both factors of the form
\begin{align}
\ket{\cL_{x}}&=\frac{1}{\sqrt{2^{n_{p}-1}-1}}\sum_{x=0}^{n_{p}-2}\sqrt{2^{x}}\ket{x},
\end{align}
where $x\in\{r,s\}$. We prepare both states using the protocol in \cite{su2021}, whose Toffoli cost is $C_{\Toff}(\ket{\cL_{x}})=n_{p}-2$, and therefore
\begin{align}
C_{\Toff}\left(\ket{\cL_{r,s}}\right)=2(n_{p}-2) \label{eq:Toff_cost_rs}.
\end{align}
Each state is supported on $n_{p}$ qubits. In addition, one ancilla qubit is prepared in the state $R_{z}(\pi/4)\ket{+}$ to be catalytically used to perform controlled Hadamards needed to prepare $\ket{\cL_{r,s}}$~\cite{su2021}, so that
\begin{align}
 C_{\anc}(\ket{\cL_{r,s}})=2n_p+1.
\end{align}

The full Toffoli cost of $\PREP_{T}$ is
\begin{align}
C_{\Toff}(\PREP_{T}) &=C_{\Toff}(\ket{\tilde{\cL}_{m}})+C_{\Toff}(\ket{\cL_{w}}) \nonumber \\ & \quad +C_{\Toff}(\ket{\cL_{r,s}})
\end{align}
so that, from \cref{eq:C_Toff_m,eq:Toff_cost_w,eq:Toff_cost_rs}, it follows that
\begin{align}\label{eq:PREP_T_Toff_cost}
C_{\Toff}\left(\PREP_{T}\right)=\eta+\mu_{T}+ 4n_{\eta}+ 2n_{p}+14.
\end{align}
For the unpreparation, we get a reduced cost of 
\begin{equation}\label{eq:UNPREP_T_Toff_cost}
C_{\Toff}\bigl(\PREP_{T}^{\dagger}\bigr)=    \Er(\eta)+4n_{\eta}+2n_{p}+16.
\end{equation}
The mass-state preparation dominates the Toffoli complexity of the $\PREP_{T}$ routine with an $\cO(\eta)$ scaling.

The total number of ancilla qubits used is
\begin{align}
C_{\anc}(\proxPREP_{T})&=C_{\anc}(\ket{+})+C_{\anc}(\ket{\cL_w})\nonumber\\
&\qquad+C_{\anc}(\ket{\cL_{r,s}})+C_{\anc}(\ket{\tilde{\cL}_{m}})\nonumber\\
&=3n_{\eta}+3\mu_{T}+2n_{p}+8.
\label{eq:anc_cost_PREPT}
\end{align}

\subsection{\texorpdfstring{$\PREP_{V}$}{PREPV}}\label{subsec:time-evo_PREP_V}

Implementing $\PREP_{V}$ requires the preparation of the two states in \cref{eq:prepv-decomp}: the momentum state $\ket{\cL_{\bnu}}$ that encodes the inverse momenta of the particles and the Coulomb state $\ket{\cL_{\doublezeta}}$ that encodes the charge pairs in the Coulomb potential. Our main focus in this section is the preparation of $\ket{\proxL_{\doublezeta}}$, since the preparation routine for $\ket{\proxL_{\nu}}$ has been described previously~\cite{Babbush_2019,su2021}. 

The exact momentum state 
\begin{align}
    \ket{\cL_{\bnu}}=\frac{1}{\lambda_{\bnu}}\sum_{\bnu\in G_{0}} \frac{1}{\norm{\bnu}_2}\ket{\bnu},
\end{align}
with normalization $\lambda_{\bnu}=\sum_{\boldsymbol{\nu} \in G_0} \|\boldsymbol{\nu}\|_2^{-2}$, encodes the inverse momentum of each particle. It is approximated by the easier-to-prepare state
\begin{align}
    &\ket{\proxL_{\bnu}}_{g,h,\anc}=\frac{1}{\sqrt{\cM2^{n_{p}+1}}}\sum_{\mu=2}^{n_{p}+1}\sum_{\bnu\in B_{\mu}}    \nonumber\\
    &\sum_{m=0}^{\left\lceil\mathcal{M}\left(2^{\mu-2} /\|\nu\|_{2}\right)^{2}\right\rceil-1}\frac{1}{2^\mu}\ket{0}_{g}\ket{\bnu}_{h}\ket{\mu,m}_{\anc}+\sqrt{1-p_{\bnu}} \ket{\perp},
\end{align}
where $\cM=2^{n_{\cM}}$, with $n_{\cM}$ being a control parameter setting the precision, and $B_{\mu}=\{\bnu \in G_{0}\mid (\norm{\bnu}_{\infty}< 2^{\mu-1})\wedge(\bigvee_{w\in\{x,y,z\}}(|\bnu_{w}|\geq 2^{\mu-2}))\}$. The qubit $g$ flags successful preparation of $\ket{\tilde{\cL}_{\bnu}}$ with  probability~\cite{su2021}
\begin{align}\label{eq:p_bnu_def}
p_{\bnu}=\sum_{\mu=2}^{n_{p}+1}\sum_{\bnu\in B_{\mu}}\frac{\lceil\cM (2^{\mu-2}/\norm{\bnu}_2)^2\rceil}{\cM 2^{2\mu}2^{n_{p}+1}}\approx1/4.
\end{align} 

The cost to prepare and unprepare $\ket{\proxL_{\bnu}}_{g,h}$ is~\cite{su2021}
\begin{align}
C_{\anc}(\ket{\proxL_{\bnu}})+C_{\anc}(\bra{\proxL_{\bnu}})&=3n_{p}^2+10n_{p}+5n_{\cM}\nonumber\\
&\quad{} +4n_{\cM}n_{p}+10\label{eq:anc_cost_L_nu}\\
C_{\Toff}(\ket{\proxL_{\bnu}})+C_{\Toff}(\bra{\proxL_{\bnu}})&=3n_{p}^{2}+15n_{p}\nonumber\\
&\quad{} +4n_{\cM}(n_{p}+1)-7.
\label{eq:Toff_cost_L_nu}
\end{align}
The parameter $n_{\cM}$ is set by the error $\epsilon_{V}=\| V-\tilde{V}\|_{\infty}$ on the block-encoded potential energy $\tilde{V}$. This error is due to the imprecision in the amplitudes of $\ket{\proxL_{\bnu}}$. We show in \cref{app:bound_V} that
\begin{align}
n_{\cM}&=\left\lceil\log\left(\lambda_{V}r_{\bnu}/\epsilon_{V} \right)\right\rceil,
\end{align}
where
\begin{align}\label{eq:ratio_bnu}
r_{\bnu}=\frac{4}{\lambda_{\bnu}}\left(7\times2^{n_{p}+1}-9n_{p}-11-3\times 2^{-n_{p}}\right).    
\end{align}
We also show that $r_{\bnu}\leq 12$ in \cref{lem:r_bnu_bound}, meaning that $n_{\cM}=\cO\left(\left\lceil\log\left(\lambda_{V}/\epsilon_{V}\right)\right\rceil\right)$.

Our preparation of the Coulomb state $\ket{\cL_{\doublezeta}}$ extends the preparation routine for a similar state used for the block-encoding of the simpler, purely electronic Hamiltonian~\cite{su2021}. That block-encoding requires the preparation of the state 
\begin{align}
    \frac{1}{\sqrt{\sum_{l=1}^{\eta_{n}}\zeta_{l}}}\sum_{l=1}^{\eta_{n}} \sqrt{\zeta_{l}}\ket{l},
    \label{eq:su-charge-state}
\end{align}
which encodes the electric charges $\zeta_{l}$ of only the nuclei of the system. The preparation of this state in~\cite{su2021} uses QROM, but simplifies the preparation by using the fact that $\zeta_{l}$ are integers.

The Coulomb state is 
\begin{align}
    \ket{\cL_{\zeta\zeta}}=\sum_{i\neq j=1}^{\eta}\frac{\sqrt{|\zeta_{i}||\zeta_{j}|}}{\sqrt{\sum_{i\neq j=1}^{\eta}|\zeta_{i}||\zeta_{j}|}}\ket{i,j}.
\end{align}
The approximate state includes a flag register and is 
\begin{align}
\ket{\tilde{\cL}_{\zeta\zeta}}_{k,\ell,m}=\sqrt{p_{\zeta}}\ket{0}_{k}\ket{\cL_{\zeta\zeta}}_{\ell,m}+\sqrt{1-p_{\zeta}}\ket{1}_{k}\ket{\cL_{\zeta\zeta}^{\perp}}_{\ell,m},
\end{align}
with overlap 
\begin{align}\label{eq:def_p_zeta}
p_{\zeta}=\frac{\sum_{i\neq j=1}^{\eta} |\zeta_{i}||\zeta_{j}| }{\sum_{i,j=1}^{\eta} |\zeta_{i}||\zeta_{j}| }
\end{align}
with $\ket{\cL_{\zeta\zeta}}_{\ell,m}$ and with the orthogonal state
\begin{align}\label{eq:Coulomb_perp_state}
\ket{\cL_{\zeta\zeta}^{\perp}}_{\ell,m}&=\sum_{j=1}^{\eta}\frac{\sqrt{|\zeta_{j}||\zeta_{j}|}}{\sqrt{\sum_{ j=1}^{\eta}\zeta^2_{j}}}\ket{j,j}_{\ell,m}.
\end{align}

The backbone of our technique for preparing $\ket{\proxL_{\doublezeta}}$ is the charge state 
\begin{align}
\ket{\zeta}=\frac{1}{\sqrt{\sum_{j=1}^{\eta}|\zeta_{j}|}}\sum_{j=1}^{\eta}\sqrt{|\zeta_{j}|}\ket{j},
\label{eq:charge_state_nuc_elec}
\end{align}
which encodes the absolute values of the electric charges $\zeta_{j}$ of both electrons and nuclei. We prepare it using the procedure in~\cite{su2021}, originally used to prepare the nuclei-only state in \cref{eq:su-charge-state}, but now applied to \cref{eq:charge_state_nuc_elec}, as detailed in \cref{app:Coulomb_state_prep}.

We then prepare $\ket{\proxL_{\doublezeta}}$ by performing an inequality test on two independent copies of $\ket{\zeta}$ and storing the outcome of the test in the flag register $k$ in $\ket{\proxL_{\doublezeta}}$, as detailed in \cref{app:Coulomb_state_prep}. 
Therefore,
\begin{align}\label{eq:comp_prep}\ket{\proxL_{\doublezeta}}&=\COMP_{k,\ell,m}^{\neq}\left(\ket{0}_{k}\otimes \ket{\zeta}_{\ell}\otimes \ket{\zeta}_{m}\right),
\end{align}
where $\COMP_{k,\ell,m}^{\neq}$ is an inequality test between registers $\ell$ and $m$ that writes the result into the flag register $k$. 

The Toffoli cost of this procedure is
\begin{multline}
C_{\Toff}\big(\ket{\proxL_{\doublezeta}}\big)= C_{\Toff}(C\COMP_{k,\ell,m}^{\neq})\\+ 2C_{\Toff}\left(\ket{\zeta}\right).
\end{multline}
Since an $n$-qubit inequality test costs $n-1$ Toffoli gates,
\begin{align}
    C_{\Toff}(\COMP_{k,\ell,m}^{\neq})=n_{\eta}-1,
\end{align}
where $n_{\eta}=\lceil\log(\eta)\rceil$.
The Toffoli cost of preparing $\ket{\zeta}$ is
\begin{multline}
C_{\Toff}(\ket{\zeta})=C_{\Toff}(\UNIF(2\eta_{e},r))\\+C_{\Toff}(\QROM_{2\eta_{e}}),
\end{multline}
where the subscript on $\QROM$ denotes the number of coefficients it outputs. Specifically \cite{su2021,babbush2018}, 
\begin{align}
C_{\Toff}(\UNIF(2\eta_{e},r))&=3\lceil\log(2\eta_{e})\rceil+2b_{r}-9\\
C_{\Toff}(\QROM_{2\eta_{e}})&=2\eta_{e}-1,
\end{align}
giving a total cost of 
\begin{align}
C_{\anc}(\ket{\proxL_{\doublezeta}})
&=6\bits{2\eta_{e}}+3n_{\eta}+4\label{anc_cost_doublezeta}\\
C_{\Toff}(\ket{\proxL_{\doublezeta}})
&=4\eta_{e}+n_{\eta}+6\lceil\log(2\eta_{e})\rceil+4 b_{r}-21.
\label{eq:Toff_cost_doublezeta}
\end{align}

To reduce the cost of unpreparation, we again use measurement-based uncomputation to erase the $\QROM$, as we did in \cref{subsubsec: time-evo_PREP_T} for the mass state. Unpreparing the $\QROM$ by erasing costs $\Er(2\eta_{e})$ Toffoli gates, so that the total cost of unpreparing $\ket{\proxL_{\doublezeta}}$ is at most 
\begin{align}
C_{\Toff}(\bra{\proxL_{\doublezeta}})=2\Er(2\eta_{e})+6\lceil\log(2\eta_{e})\rceil+4b_{r}-18.
\end{align}

The total Toffoli cost to prepare and unprepare is then $\cO(\eta_{e})$ and therefore slightly cheaper than the preparation of the mass state, which is $\cO(\eta)$.

Our procedure prepares $\ket{\cL_{\doublezeta}}$ with probability $p_{\zeta}$, which is at least $1/2$. Because $\sum_{i=1}^{\eta}|\zeta_{i}|^2=\eta_{e}+\sum_{j=\eta_{e}+1}^{\eta_{n}}\zeta_{j}^2$ and, for charge-neutral systems, $\sum_{j=1}^{\eta_{}}|\zeta_{j}|=2\eta_{e}$, we can write
\begingroup
\allowdisplaybreaks
\begin{align}
p_{\zeta}&=\big|\langle\tilde\cL_{\doublezeta}|\cL_{\doublezeta}\rangle\big|
    =1-\frac{\sum_{i=1}^{\eta}|\zeta_{i}|^2}{\left(\sum_{i=1}^{\eta}|\zeta_{i}|\right)^2}\\
    &=1-\frac{\eta_{e}+\sum_{j=\eta_{e}+1}^{\eta_{n}}\zeta_{j}^2}{4\eta_{e}^2}\\
    &\ge 1-\frac{\eta_{e}+\left(\sum_{j=\eta_{e}+1}^{\eta_{n}}\zeta_{j}\right)^2}{4\eta_{e}^2}\\
    &=\frac34-\frac{1}{4\eta_{e}} \ge \frac12,
    \label{eq:p_zeta_lowerbound}
\end{align}
\endgroup
with equality holding for a one-electron system such as hydrogen. For large molecules ($\eta_{e}\rightarrow\infty$), the lower bound on $p_\zeta$ improves to $3/4$, providing us with a rigorous performance guarantee for our state preparation. 

The joint state $\ket{\proxL_{\bnu}}\otimes\ket{\proxL_{\doublezeta}}$ can be written as   
\begin{align}\label{eq:proxPREP_V}
\proxPREP_{V}\ket{0}=\sqrt{p_{\bnu}p_{\zeta}} \PREP_{V}\ket{0}+\sqrt{1-p_{\bnu}p_{\zeta}}\ket{\perp},
\end{align}
which has overlap $p_{\bnu}p_{\zeta}$ with $\PREP_{V}\ket{0}$. Here $\ket{\perp}$ is orthogonal to $\PREP_{V}\ket{0}$.
\Cref{eq:p_bnu_def,eq:p_zeta_lowerbound} show that $p_{\bnu}p_{\zeta}$ is lower-bounded by a constant close to $1/8$, with the actual value typically closer to $1/4$ for most molecules. Our approximate protocol $\proxPREP_{V}$ therefore has an appreciable probability of success. 

The total Toffoli cost of this protocol is 
\begin{align}\label{eq:PREP_V_Toff_cost}
C_{\Toff}\bigl(\proxPREP_{V}\bigr)&=C_{\Toff}(\ket{\proxL_{\doublezeta}})+C_{\Toff}(\ket{\proxL_{\bnu}})\\
&=4\eta_{e}+n_{\eta}+6\lceil\log(2\eta_{e})\rceil+4b_{r}-24\nonumber\\
&\quad +3n_{p}^2+11 n_{p}+4n_{\cM}(n_{p}+1),
\end{align}
while for the unpreparation we obtain 
\begin{align}\label{eq:UNPREP_V_Toff_cost}
C_{\Toff}\bigl(\proxPREP_{V}^{\dagger}\bigr)&=C_{\Toff}(\bra{\proxL_{\doublezeta}})+C_{\Toff}(\bra{\proxL_{\bnu}})\\
&=n_{\eta}+2\Er(2\eta_{e})+6\lceil\log(2\eta_{e})\rceil\nonumber\\&\quad +4b_{r}-19+4(n_{p}-1).
\end{align}

The ancilla cost is
\begin{align}
C_{\anc}(\proxPREP_{V})&=C_{\anc}(\ket{\tilde{\cL}_{\bnu}})+C_{\anc}(\ket{\tilde{\cL}_{\doublezeta}})\nonumber\\
&=3n_{p}^2+10n_{p}+6\bits{2\eta_{e}}+3n_{\eta}\nonumber\\
&\qquad+5n_{\cM}+4n_{\cM}n_{p}+14.
\label{eq:anc_cost_PREPV}
\end{align}

\subsection{\texorpdfstring{$\PREP_{H}$}{PREPH}} \label{subsec: time-evo_PREP_H}

Given the preceding implementations of $\proxPREP_{T}$ and $\proxPREP_{V}$, we construct $\proxPREP_{H}$ as 
\begin{align}
\proxPREP_{H}\ket{0}&= \ket{\proxL_{\theta}}_{a}\otimes \proxPREP_T \ket{0} \otimes \proxPREP_{V}\ket{0}.
\label{eq:proxPREP_H}
\end{align}
Here, the additional ancilla qubit $a$ is prepared as
\begin{align}
\ket{\proxL_{\theta}}_{a}=\cos(\tilde{\theta})\ket{0}_{a}+\sin(\tilde{\theta})\ket{1}_{a},
\end{align}
where $\theta$ denotes an $n_{\theta}$-bit approximation of $\theta$, such that $\lvert\theta-\tilde{\theta}\rvert\leq 2^{-n_{\theta}}$. Preparing and unpreparing this single-qubit state has a cost of
\begin{align}
C_{\anc}\bigl(R_{z}(\tilde{\theta})_{a}\bigr)&= n_{\theta}\label{eq:anc_cost_Rz}\\
C_{\Toff}\bigl(R_{z}(\tilde{\theta})_{a}\bigr)&=n_{\theta}-3,
\label{eq:R_z_Toff_cost}
\end{align}
where the ancillas store the phase gradient state \cite{su2021} that is used to catalytically synthesize the single-qubit rotation. For the unpreparation, we run the circuit in reverse, so that  $C_{\Toff}\bigl(R^{\dagger}_{z}(\tilde{\theta})_{a}\bigr)=C_{\Toff}\bigl(R_{z}(\tilde{\theta})_{a}\bigr)$ and likewise for the ancilla cost.  

Therefore, the cost of $\proxPREP_{H}$ is 
\begin{align}
C_{\anc}(\proxPREP_{H})=C_{\anc}(\proxPREP_{T})+C_{\anc}(\proxPREP_{V})\nonumber \\ +C_{\anc}\bigl(R_{z}(\tilde{\theta})_{a}\bigr)\label{eq:anc_cost_Prep_H}\\
C_{\Toff}(\proxPREP_{H})=C_{\Toff}(\proxPREP_{T})+C_{\Toff}(\proxPREP_{V})\nonumber \\ +C_{\Toff}\bigl(R_{z}(\tilde{\theta})_{a}\bigr)\label{eq:Toff_cost_Prep_H},
\end{align}
which can be converted into a constant-factor cost using \cref{eq:anc_cost_Rz,eq:anc_cost_PREPT,eq:anc_cost_PREPV} for the ancillas and \cref{eq:PREP_T_Toff_cost,eq:PREP_V_Toff_cost,eq:R_z_Toff_cost} for the Toffoli costs.
Correspondingly, for the unpreparation, we have 
\begin{multline}
C_{\Toff}\bigl(\proxPREP_{H}^{\dagger}\bigr)=C_{\Toff}\bigl(\proxPREP_{T}^{\dagger}\bigr)+C_{\Toff}\bigl(\proxPREP_{V}^{\dagger}\bigr) \\  +C_{\Toff}\big(R_{z}(\tilde{\theta})^{\dagger}_{a}\big), \label{eq:Toff_cost_UNPREP_H}
\end{multline}
which can be converted into constant factor costs using \cref{eq:UNPREP_T_Toff_cost,eq:UNPREP_V_Toff_cost}.

Asymptotically, the Toffoli cost of $\PREP_{H}$ is dominated by the preparation of the momentum state, 
\begin{align}
   C_{\Toff}(\PREP_{H})=\cO(n_{p}^2)=\cO\left((\log N)^2\right).
\end{align}

\subsection{\texorpdfstring{$\SEL_{T}$}{SELT}} \label{subsec: time-evo_SEL_T}
Our implementation of $\SEL_{T}$ follows the approach in~\cite{su2021}. The map
\begin{multline}
\SEL_{T}: \ket{b}_{b}\ket{j}_{c}\ket{w}_{d}\ket{r}_{e}\ket{s}_{f}\ket{\mathbf{p}_{j}}\\ 
\rightarrow (-1)^{b\cdot\left(p_{w,r}\cdot p_{w,s}\oplus 1\right)} \ket{b}_{b}\ket{j}_{c}\ket{w}_{d}\ket{r}_{e}\ket{s}_{f}\ket{\mathbf{p}_{j}},
\end{multline}
adds a phase to the data register controlled on ancilla registers $c,d,e,f$.

To reduce the need to control many operations off the LCU ancilla registers, we first swap the momentum register of the $j$th particle into an ancilla register of $3n_{p}$ qubits controlled on register $c$ in \cref{eq:mass_state} being in state $\ket{j}_{c}$. The Toffoli cost of the controlled $\SWAP$s is $3\eta n_{p}$ and the cost of the controls, which we implement via unary iteration \cite{babbush2018}, is $\eta-2$. However, we assign these costs to $\SEL_{H}$ later in \cref{subsec:SEL_H_implement}.

To apply $\SEL_{T}$, we copy the $w$th momentum component of particle $j$ that is stored in the wavefunction data register, up to the sign bit, into an ancilla register of size $n_{p}-1$. To do so, we control the copying off the register $d$ storing $\ket{w}$, which can be done with $3(n_{p}-1)$ Toffoli gates. Controlling off registers $e$ and $f$, we then copy the $r$th and the $s$th bits into two ancilla qubits, with a Toffoli cost of $n_{p}-1$ each. Then, controlled on these two and register $a$ in \cref{eq:proxPREP_H}, we apply a $(-Z)$-gate on register $b$ prepared in the state $\ket{+}$ as shown in \cref{eq:prept-decomp}. This costs one Toffoli gate. Within $\SEL_{H}$ (see \cref{subsec:SEL_H_implement}), the application of $\SEL_{T}$ is controlled on an ancilla qubit. This control is added to the already doubly controlled $(-Z)$-gate described above, which adds one additional Toffoli gate to the cost, giving
\begin{align}
C_{\anc}(\contr{\SEL_{T}})&=n_{p}+2\label{eq:anc_cost_SEL_T}\\
C_{\Toff}(\contr{\SEL_{T}})&=5(n_{p}-1)+2\label{eq:Toff_cost_SEL_T},
\end{align}
where the ancilla cost excludes the unary-iteration ancillas~\cite{babbush2018} used to control the swaps of the momenta. We include this additional ancilla cost later in \cref{subsec:SEL_H_implement}. There, we also account for the cost of inverting the control-$\SWAP$s between the data and ancilla registers, which we perform after $\SEL_{T}$.

\subsection{\texorpdfstring{$\SEL_{V}$}{SELV}} \label{subsubsec: time-evo_SEL_V}
We implement 
\begin{multline}
\SEL_{V}:\ket{b}_{b}\ket{\bnu}_{h}\ket{i}_{\ell}\ket{j}_{m}\ket{\mathbf{p}_{i}}\ket{\mathbf{q}_{j}}\rightarrow\\(-1)^{f(\mathbf{p},\mathbf{q},\boldsymbol{\nu},i,j)}\ket{b}\ket{\bnu}_{h}\ket{i}_{\ell}\ket{j}_{m}
\ket{\mathbf{p}_{i}+\bnu}\ket{\mathbf{q}_{j}-\bnu},
\end{multline}
with the Boolean function
\begin{multline}
    f(\mathbf{p},\mathbf{q},\boldsymbol{\nu},i,j)= b\left([(\mathbf{p}+\boldsymbol{\nu})\notin G]\lor[(\mathbf{q}-\boldsymbol{\nu})\notin G]\right) \\ \oplus \left([i< \eta_{e}+1]\oplus[j<\eta_{e}+1]\right).   
\end{multline}
We perform the arithmetic on the momenta as follows. First, controlled on register $\ell$ in \cref{eq:Coulomb_state}, we $\SWAP$ data register $\ket{\mathbf{p}_{i}}$ into an ancilla register of size $3n_{p}$, add $\bnu$ to it, and controlled-$\SWAP$ it back. We then subtract $\bnu$ from $\ket{\mathbf{q}_{j}}$ using the same procedure, but controlled on register $m$ in \cref{eq:Coulomb_state}. Since both have to be swapped back and forth, and since $i$ and $j$ range over all particles, the Toffoli cost of the controlled $\SWAP$s is $12\eta n_{p}$. However, we assign these costs to $\SEL_{H}$ later in \cref{subsec:SEL_H_implement}. To perform $\SEL_{V}$, we apply both a controlled phase and controlled additions on the data register. 

To apply the controlled phase $(-1)^{f(\mathbf{p},\mathbf{q},\boldsymbol{\nu},i,j)}$, apart from the last two terms that apply phases for electron charges, we follow the procedure for implementing the corresponding function in~\cite{su2021}. We perform the controlled phasing for the electronic charges via a control-$Z$ onto the register $b$ in state $\ket{+}_{b}$ as part of the unary iteration on the Coulomb state, which determines which registers are to be swapped into the ancilla register. Since controlled-$Z$ is a Clifford gate, there is no Toffoli cost for this modification.

The controlled addition of the momenta $\bnu$ into the momentum registers $\ket{\mathbf{p}_{i}}$ and $\ket{\mathbf{q}_{j}}$ is controlled by the ancilla register that the control circuit in $\SEL_{H}$ outputs on (see \cref{subsec:SEL_H_implement}). This requires one qubit fewer than in~\cite{su2021}. To add $\bnu$ into the momentum register, we first copy one component of $\bnu$ into an ancilla register of size $n_{p}$. For the addition of the momenta, we require $n_{p}+1$ ancilla qubits. In addition, we need 2 ancillas to deal with arithmetic overflow for every spatial component of the momentum, totaling 6 qubits. 

Overall,
\begin{align}
C_{\anc}\left(\contr{\SEL_{V}}\right)&=2n_{p} + 7\label{eq: SEL_V_anc_cost}\\
C_{\Toff}(\contr{\SEL_{V}})&=24 n_{p}\label{eq: SEL_V_Toff_cost}. 
\end{align}

\subsection{\texorpdfstring{$\SEL_{H}$}{SELH}}\label{subsec:SEL_H_implement}
Given the preceding implementations of $\SEL_{T}$ and $\SEL_{V}$, we implement $\SEL_{H}$ so that it selectively applies either $\SEL_{T}$, $\SEL_{V}$, or the identity, controlled on the values of a set of flag registers that indicate successful state preparation in the $\PREP_{H}$ register. 
The general form of $\SEL_{H}$ is 
\begin{multline}\label{eq:SEL_H_definition}
\sum_{x_{\eq}, x_{a},x_{g,k}}\ket{x_{\eq}, x_{a},x_{g,k}}\bra{x_{\eq}, x_{a},x_{g,k}}_{a,g,k} \\ \otimes \SEL_{(x_{\eq},x_{a},x_{g,k})},
\end{multline}
where $\SEL_{(x_{\eq},x_{a},x_{g,k})}\in\{\SEL_{T}, \SEL_{V}, I\}$, which is controlled by the values of the flag registers $(x_{\eq}, x_{a},x_{g,k})\in\{0,1\}^{3}$. Here, $\ket{x_{\eq}}$ is a register that flags the joint success of all uniform state preparations in $\PREP_{H}$, register $\ket{x_{g,k}}$ flags the success branch of $\proxPREP_{V}$, and register $\ket{x}_{a}$ is given in \cref{eq:control_state}. 

There are two possible choices of $\SEL_{(x_{\eq},x_{a},x_{gk})}$ for a given triple $(x_{\eq},x_{a},x_{g,k})$ of flag values, as detailed in \cref{app:SEL_H_implement}. For both choices, we can implement the associated controls to apply either $\SEL_{T}$ or $\SEL_{V}$ via one additional ancilla and one multi-qubit Toffoli gate and Clifford gates. In summary, the control circuit checks the success of the preparation of the uniform states in $\proxPREP_{H}$, the state of register $a$, and the success of $\proxPREP_{V}$ via controls and anti-controls on the appropriate flag registers. Conditioned on the success of these state preparations and on $\ket{x_{a}}$, we flip a qubit that controls the application of $\SEL_V$ and anti-controls the application of $\SEL_{T}$. This circuit costs at most $6$ Toffoli gates.  

We noted above the need to precede the applications of both $\SEL_{T}$ and $\SEL_{V}$ by controlled-swapping of either one or two momentum registers into ancillas and then back again, which has cumulative Toffoli cost $18\eta n_{p}$ for the controlled $\SWAP$s and $6\eta-12$ for the unary iteration that implements the controls on the index register of the particles~\cite{su2021}. The total Toffoli cost is therefore
\begin{align}
C_{\Toff}(\SEL_{H})&=6(C_{\Toff}(\textsc{cswaps})+C_{\Toff}\left(\mathrm{UnaryIt}\right))\nonumber\\&\quad +C_{\Toff}(\mathrm{controls})+C_{\Toff}(\textrm{ctrl-}\SEL_{T})\nonumber\\
&\quad +C_{\Toff}\left(\textrm{ctrl-}\SEL_{V}\right),
\label{eq:Toff_cost_SEL_H}
\end{align}
where $C_{\Toff}(\textsc{cswaps})=3\eta n_{p}$, $C_{\Toff}(\mathrm{UnaryIt})=\eta-2$, and  $C_{\Toff}(\mathrm{controls})=6$ for the controls in \cref{eq:SEL_H_definition} to implement $\SEL_{H}$. Therefore, the total cost of $\SEL_{H}$ is
\begin{align}
C_{\Toff}(\SEL_{H})=18\eta n_{p}+6\eta+5n_{p}+24n_{p}-9.
\label{eq:Toff_cost_SEL_H_const_factor}
\end{align}
The cost of the controlled version of $\SEL_{H}$ is
\begin{align}
C_{\Toff}(\textrm{ctrl-}\SEL_{H})=  C_{\Toff}(\SEL_{H})+1
\label{eq:ctrl_SEL_H_cost}
\end{align}
because we only need one more control on the multi-qubit Toffoli gate in the control circuit of $\SEL_{H}$ that conditions on the success of the uniform state preparations.

In terms of ancilla costs, we need $3n_{p}$ ancillas into which the value of the the $j$th momentum register is swapped, $n_{\eta}-1$ ancillas for the unary iteration, 5 ancillas for the controls in \cref{eq:SEL_H_definition}, as well as ancillas for $\contr{\SEL_{T}}$ and $\contr{\SEL_{V}}$,
\begin{align}
C_{\anc}(\SEL_{H})&=\max\{C_{\anc}(\contr{\SEL_{T}}),C_{\anc}(\contr{\SEL_{V}})\}\nonumber\\
&\quad{} +C_{\anc}(\textsc{cswaps})+C_{\anc}(\mathrm{UnaryIt})\nonumber\\
&\quad{}+C_{\anc}(\mathrm{controls})\\
&=5n_{p}+n_{\eta}+11. \label{eq:Anc_cost_SEL_H}
\end{align}

\subsection{Block-encoding \texorpdfstring{$U_{H}$}{UH}}\label{subsec:U_H}
We are now in position to assemble the block-encoding
\begin{align}
U_{H}=\proxPREP_{H}^{\dagger}\SEL_{H}\proxPREP_{H},
\label{eq:block_encoding_U_H}
\end{align}
which requires 
\begin{align}
C_{\anc}(U_{H})&=C_{\anc}(\proxPREP_{H})+C_{\anc}(\SEL_{H})\label{eq:anc_cost_U_H}\\
C_{\Toff}(U_{H})&=C_{\Toff}(\proxPREP_{H})+C_{\Toff}(\SEL_{H})\nonumber\\
&\qquad {} + C_{\Toff}(\proxPREP^{\dagger}_{H})\label{eq:Toff_cost_U_H},
\end{align}
with the ancilla subcosts given in \cref{eq:anc_cost_Prep_H,eq:Anc_cost_SEL_H} and the Toffoli subcosts given in \cref{eq:Toff_cost_Prep_H,eq:Toff_cost_SEL_H,eq:Toff_cost_UNPREP_H}, respectively. The expansion of $C_{\anc}(U_{H})$ into its components is summarized in \cref{app:qubit_cost}.

$U_{H}$ forms a $(\tilde{\lambda}_{H}, q_{H}, \epsilon_{H})$ block encoding with the three parameters given as follows. 
First, the LCU-norm is
\begin{align}
\tilde{\lambda}_{H}=\frac{1}{P_{\eq}}\max\left\{\lambda_{T}+\lambda_{V},\frac{\lambda_{V}}{p_{\bnu}p_{\zeta}}\right\},
\label{eq:lambda_H_tilde}
\end{align}
which we derive in \cref{app:SEL_H_implement}. Here, $\lambda_{T}$ and $\lambda_{V}$ are given in \cref{eq:lambda_T,eq:lambda_V} and the success probabilities $p_{\bnu}$ and $p_{\zeta}$ of preparing $\ket{\cL_{\nu}}$ and $\ket{\cL_{\doublezeta}}$ are defined in \cref{eq:p_bnu_def,eq:def_p_zeta}. $P_{\eq}$ is the probability of all uniform state preparations in $\proxPREP_{H}$ succeeding. These occur independently as subroutines in the preparations of $\ket{\cL_{w}}, \ket{\cL_{m}}$ and $\ket{\cL_{\doublezeta}}$. Therefore, $P_{\eq}=\Ps(3,8)\Ps(\eta,b_{r})\Ps(\eta,b_{r})^2$, where $\Ps(n,b_{r})$ for $n,b_{r}\in \mathbb{N}$ is the success probability of preparing a uniform superposition over $n$ basis states using rotations with $b_{r}$-bit precision~\cite{su2021}. For $b_{r}=8$, this probability is very close to one and is given in \cref{eq:prob_success_unif}. 
Second, the size of the block-encoding ancilla register is $q_{H}=C_{\anc}(U_{H})$.
And third, the block-encoding error is
\begin{align}
    \big\|H-\tilde{H}\big\|_{\infty}\leq \epsilon_{H},
\label{eq:Ham_blockenc_definition}
\end{align}
for the approximately encoded Hamiltonian $\tilde{H}=\tilde{\lambda}_{H}\bra{0^{q_{H}}}\proxPREP_{H}^{\dagger}\SEL_{H}\proxPREP_{H}\ket{0^{q_{H}}}$, as can be obtained from the definition of $\proxPREP_{H}$ in \cref{eq:proxPREP_H}.

The block-encoding error $\epsilon_{H}$ constitutes the error budget available for the block encoding. As we show in \cref{app:bound_H_blockenc}, it decomposes as 
\begin{align} \epsilon_{H}=\epsilon_{T}+\epsilon_{V}+\epsilon_{\theta}.
\label{eq:epsilon_H}
\end{align}
Here, the block-encoding errors of the kinetic and potential terms are defined via 
\begin{align}
    \| T-\tilde{T}\|_{\infty} &\leq \epsilon_{T}\label{eq:epsilon_T}\\
    \| V-\tilde{V}\|_{\infty} &\leq \epsilon_{V}\label{eq:epsilon_V},
\end{align}
where $\tilde{T}$ and $\tilde{V}$ are defined analogously to $\tilde{H}$ using the respective $\proxPREP$ and $\SEL$ routines. Finally, $\epsilon_{\theta}$ bounds the error in the block-encoding of $H$ due to the preparation of  $\ket{\proxL_{\theta}}$, which is a $\Delta\theta$-approximation of $\ket{\cL_{\theta}}=R_{z}(\theta)\ket{+}$, where 
\begin{align}
    \Delta\theta=|\theta-\tilde{\theta}|\leq2^{-n_{\theta}}
\end{align}
is the error incurred by performing the $n_{\theta}$-bit approximation $R_{z}(\tilde{\theta})$ of $R_{z}(\theta)$. In \cref{app:bound_H_blockenc}, we show that
\begin{align}   
\epsilon_{\theta}\leq 2\tilde{\lambda}_{H}\Delta\theta.
\label{eq:delta_theta}
\end{align}

Distributing the error budget set by $\epsilon_{H}$ across $\epsilon_{\theta}, \epsilon_{T}$, and $\epsilon_{V}$ in accordance with \cref{eq:epsilon_H} results in explicit parameter values
\begin{align}
    n_{\theta}&=\lceil\log(\tilde{\lambda}_{H}/\epsilon_{\theta})\rceil\\    
    \mu_{T}&=\ceil{\log(\lambda_{T}/\epsilon_{T})}\label{eq:mu_T}\\ 
    n_{\cM}&=\ceil{\log(\lambda_{V}r_{\bnu}/\epsilon_{V})},
\end{align}
which are used to give the costs for $\proxPREP_{H}$ in \cref{eq:Toff_cost_Prep_H}. Because $r_{\bnu}\leq 12$ (see \cref{app:bound_V}), $n_{\cM}=\cO\left(\bits{\lambda_{V}/\epsilon_{V}}\right)$.

\subsection{Qubiterate \texorpdfstring{$\cW$}{W}}\label{subsec:qubiterate}

Performing time evolution using QSP will require a qubiterate that we construct using the block-encoding $U_{H}$ and a reflection operator $\cR_{0}^{(\cW)}=I-2\dyad{0}{0}$, giving
\begin{align}
\widetilde{\cW}=\cR_{0}^{(\cW)}U_{H}=e^{-i\arccos(\tilde{H}/\tilde{\lambda}_{H})},
\label{eq:approx_qubiterate}
\end{align}
where $\cR_{0}^{(\cW)}$ acts only on a subset of ancilla qubits, as described in \cref{app:reflection_cost}.
We will also require $\widetilde{\cW}$ controlled off an ancilla qubit, which is done by adding controls for both $U_{H}$ and $\cR_{0}^{(\cW)}$. 

To perform $\textrm{ctrl-}\cR_{0}^{(\cW)}$, we apply a multi-controlled-$Z$ gate with anti-controls on the output register of $\proxPREP_{H}^{\dagger}$ and a target on the qubit that controls the application of $\widetilde{\cW}$. The Toffoli cost of implementing this gate using measurement-based uncomputation is $C_{\Toff}\big(\contr{\cR_{0}^{(\cW)}}\big)=\out(\PREP_{H}^{\dagger})-1$ and the number of required fresh ancillas is $\out(\PREP_{H}^{\dagger})-2$, where 
\begin{equation}\label{eq:Toff_cost_qubiterate_refl}
\out(\PREP_{H}^{\dagger})=n_{\eta}+6n_{p}+n_{\cM}+2\bits{2\eta_{e}}+11
\end{equation}
is the size of the output register of $\proxPREP_{H}^{\dagger}$, as derived in \cref{app:reflection_cost}.

The cost of applying $\contr{\proxW}$ is then
\begingroup
\allowdisplaybreaks
\begin{align}
C_{\anc}(\contr{\proxW})&=\max\{C_{\anc}(U_{H}), 2(\out(\PREP_{H}^{\dagger})-1)\}\label{eq:anc_cost_ctrl_W}\\
C_{\Toff}(\contr{\proxW})&= C_{\Toff}\left(\textrm{ctrl-}U_{H}\right) +C_{\Toff}(\textrm{ctrl-}\cR_{0}^{(\cW)})\label{eq:Toff_cost_ctrl_W}\\
C_{\Toff}(\textrm{ctrl-}U_{H})&=C_{\Toff}(\proxPREP_{H})+C_{\Toff}(\contr{\SEL_{H}})\nonumber\\
&\qquad {} +C_{\Toff}(\proxPREP^{\dagger}_{H})\label{eq:Toff_cost_ctrl_U_H},
\end{align}
\endgroup
where $C_{\anc}(U_{H})$ is given in \cref{eq:anc_cost_U_H} and where the controls on $\proxPREP$ and $\proxPREP^{\dagger}$ may be dropped in \cref{eq:Toff_cost_ctrl_U_H}.

\subsection{QSP}\label{subsec: time-evo_summary}
Finally, we construct $\tilde{U}_{\prop}$ by applying bidirectional QSP~\cite{BerryPRA24} to $\contr{\proxW}$, see \cref{fig:QSP_circuit}. 
The cost of doing so is
\begin{align}
C_{\anc}(\tilde{U}_{\mathrm{prop}})&=1+C_{\anc}(\tilde{R}_{i}(\tilde{\phi}_{i},\tilde{\gamma}_{i}))+C_{\anc}(\contr{\proxW})\label{eq:anc_cost_U_prop}\\
C_{\anc}(\tilde{R}_{i}(\tilde{\phi}_{i},\tilde{\gamma}_{i}))&=1\\
C_{\Toff}(\tilde{U}_{\mathrm{prop}})&=2C_{\Toff}(\textrm{ctrl-}\widetilde{\cW})+ \tilde{d}\, C_{\Toff}(\widetilde{\cW})\nonumber\\
&\quad {}+(\tilde{d}+1)C_{\Toff}(\tilde{R}_{i}(\tilde{\phi}_{j},\tilde{\gamma}_{j}))\label{eq:Toff_cost_U_prop}\\
C_{\Toff}(\tilde{R}(\tilde{\phi}_{i},\tilde{\gamma}_{i}))&= \tfrac{1}{2}( 0.56 \log
(1/\epsilon_{\rot}) + 5.3),\label{eq:Toff_cost_QSP_rotations}\\
\tilde{d}&=\tilde{\lambda}_{H}t+\log\left(1/\epsilon_{\tilde{d}}\right),\label{eq:degree_QSP}
\end{align}
where $\tilde{d}$ is the number of calls needed to $\proxW$ and where the ancilla and Toffoli costs of $\contr{\proxW}$ are given in \cref{eq:anc_cost_ctrl_W,eq:Toff_cost_ctrl_W}, respectively. 
$\tilde{R}(\tilde{\phi}_{i},\tilde{\gamma}_{i})$ are QSP rotations performed on the QSP control ancilla, as shown in \cref{fig:QSP_circuit}. 
The factor of $1/2$ in \cref{eq:Toff_cost_QSP_rotations} converts the known $T$-gate cost in parentheses~\cite{Kliuchnikov_2023} to Toffoli cost~\cite{su2021}.

The unitary $\tilde{U}_{\prop}$ approximates the true time evolution $U_{\prop}=e^{-iHt}$ up to error 
\begin{align}
    \bigl\| U_{\prop}-\tilde{U}_{\mathrm{prop}} \bigr\|_{\infty}\leq \epsilon_{\prop} \label{eq:epsilon_prop}.
\end{align}
This error is comprised of two contributions: the error $\epsilon_{H}$ in the Hamiltonian block encoding and the error $\epsilon_{\QSP}$ in the systhesis of the time evolution by QSP. Specifically, we show in \cref{app:bound_time_evo} that
\begin{align}
    \epsilon_{\prop} &\leq \bigl\| e^{-iHt}-e^{-i\tilde{H}t} \bigr\|_{\infty}+\bigl\| e^{-i\tilde{H}t}-\tilde{U}_{\prop} \bigr\|_{\infty}\\
    &\leq\epsilon_{H}t+\epsilon_{\QSP},
\end{align}
where 
$\epsilon_{H}\geq\Vert H-\tilde{H}\Vert_{\infty}$ is the Hamiltonian block-encoding error and the QSP error $\epsilon_{\QSP}$ is defined by 
\begin{align}
\big\Vert e^{-i\tilde{H}t}-\tilde{U}_{\prop}\big\Vert_{\infty}\leq \epsilon_{\QSP}.
\end{align}

The QSP error arises from the QSP implementation of the time evolution as the approximate $\tilde{U}_{\prop}=\tilde{f}_{\tilde{d}}(\proxW)$, where $\tilde{f}_{\tilde{d}}(\proxW)$ is a truncated Jacobi-Anger expansion of the function $f=\tilde{\lambda}_{H}t\cos(\cdot)$. As detailed in \cref{app:bound_time_evo}, two sources of error result in
\begin{align}
    \epsilon_{\QSP} &\le  \big\Vert e^{-i\tilde{H}t}-f_{\tilde{d}}(\proxW)\big\Vert_{\infty} + \big\Vert f_{\tilde{d}}(\proxW)-\tilde{U}_{\prop}\big\Vert_{\infty} \nonumber\\
    &\le \epsilon_{\tilde{d}}+(\tilde{d}+1)(\epsilon_{\rot}+\epsilon_{\phi}+\epsilon_{\gamma})\label{eq:epsilon_QSP}.
\end{align}
The first term accounts for the error from truncating the Jacobi-Anger expansion,
\begin{align}
    \big\Vert e^{-i\tilde{H}t}-f_{\tilde{d}}(\proxW)\big\Vert_{\infty}\leq \epsilon_{\tilde{d}}.
\end{align}
The second term arises because QSP implements $f_{\tildegre}$ only approximately as $\tilde{f}_{\tildegre}$, since both the synthesis of the rotations $R_{j}$ and the classical precomputations of the QSP angles $(\gamma_{j},\phi_{j})$ have finite precision, giving
\begin{align}
\big\Vert f_{\tilde{d}}(\proxW)-\tilde{f}_{\tilde{d}}(\proxW)\big\Vert_{\infty}\leq (\tilde{d}+1)(\epsilon_{\rot}+\epsilon_{\phi}+\epsilon_{\gamma}),
\end{align}
where $\epsilon_{\rot}$ is the error for synthesizing each single-qubit rotation $R(\phi_{i},\gamma_{i})$ and $\epsilon_{\phi}$ and $\epsilon_{\gamma}$ are the errors of each classically computed angle $\phi_{i}$ and $\gamma_{i}$, respectively.

The error bounds above allow us to budget $\epsilon_{H}$ and $\epsilon_{\QSP}$ using the constraints provided by 
the total time-evolution error $\epsilon_{\mathrm{prop}}$, the simulation time $t$, and the LCU-norm $\tilde{\lambda}_{H}$. The budget allocated to $\epsilon_{\QSP}$ then constrains the component errors $\epsilon_{\tildegre}$, $\epsilon_{\phi}$, $\epsilon_{\gamma}$, and $\epsilon_{\rot}$.

\begin{figure*}
\resizebox{.999\linewidth}{!}{
\begin{quantikz}[row sep={1cm, between origins},column sep=0.4cm]
\lstick{$\ket{+}_{\QSP}$}&&\ctrl{1}&\gate[1]{\tilde{R}_{0}(\tilde{\theta}_{0},\tilde{\phi}_{0})}&\octrl{1}&&\octrl{1}&\gate[1]{\tilde{R}_{1}(\tilde{\theta}_{1},\tilde{\phi}_{1})}& \ \ldots\ & \octrl{1} & &\octrl{1}&\gate[1]{\tilde{R}_{\tilde{d}}(\tilde{\theta}_{\tilde{d}},\tilde{\phi}_{\tilde{d}})}&\octrl{1}&\\
\lstick{$\ket{0}_{\PREP_{H}}$}&\qwbundle{}&\gate[2]{\proxW}&&\gate[1]{\cR_{0}^{(\cW)}}&\gate[2]{\proxW}&\gate[1]{\cR_{0}^{(\cW)}}&&\ \ldots\ & \gate[1]{\cR_{0}^{(\cW)}} & \gate[2]{\proxW}&\gate[1]{\cR_{0}^{(\cW)}}&&\gate[2]{\proxW^{\dagger}} &\\
\lstick{$|\tilde{\hat{\Psi}}(0)\rangle$}&\qwbundle{}&&&&&&&\ \ldots\ & & & & & &
\end{quantikz}
}
\caption{Circuit for the propagator $\tilde{U}_{\prop}$ via bidirectional QSP. }
\label{fig:QSP_circuit}
\end{figure*}
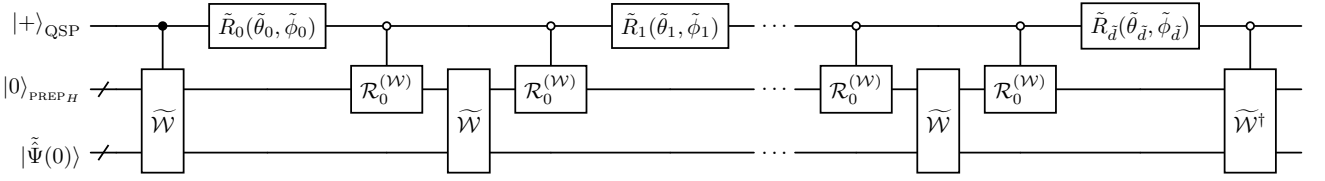

\subsection{Asymptotic cost of time evolution}

In summary, the asymptotic Toffoli cost of $\contr{\proxW}$ is 
\begin{align} C_{\Toff}\bigl(\contr{\proxW}\bigr)=\widetilde{\cO}(\eta),
\end{align}
with the exact expression given in \cref{tab:Gate_costs}.
The time evolution $\tilde{U}_{\mathrm{prop}}$ includes the additional QSP overhead, and its cost (\cref{eq:Toff_cost_U_prop}) has the asymptotic form
\begin{align}\label{eq:asympt_Toff_cost_U_prop}
 C_{\Toff}(\tilde{U}_{\mathrm{prop}})  =\widetilde{\cO}\left(\eta^2 \frac{N_{\grid}^{2/3}}{\Omega^{2/3}}t+\eta^3 \frac{N_{\grid}^{1/3}}{\Omega^{1/3}}t\right),
\end{align}
which is cubic in the number of particles $\eta$ and sublinear in the grid size $N_\mathrm{grid}$. Equivalently, the cost is quadratic in the momentum cutoff $K_{\max}=(N_{\grid}/\Omega)^{1/3}$.

The ancilla costs $C_{\anc}(\tilde{U}_{\prop})$ for time evolution are summarized in \cref{app:qubit_cost}.

\section{Measuring observables} \label{sect:measure}

The final step in our algorithm is to extract chemically relevant information from the time-evolved molecular wavefunction $\ket{\tilde{\hat{\Psi}}(t)}$. This information is obtained in the form of expectation values of observables such as bond lengths, reaction yields, or reaction rates. If we do not have access to the PWB representation of the observable $O$, we can first transform $\ket{\tilde{\hat{\Psi}}(t)}$ into the eigenbasis of $O$ before extracting the expectation value. To do so, we find $\ket{\tilde{\hat{\Psi}}(t)}= \tilde{U}_{B}\ket{\tilde{\hat{\Psi}}(t)}$, where $\tilde{U}_{B}$ is an $\epsilon_{B}$-approximation to the basis transformation $U_{B}$ from the PWB to the eigenbasis of $O$. In the chemically important case of $O$ being a real-space observable, such as a bond length or reaction yield, the basis-change is a QFT, $\ket{\tilde{\Psi}(t)}=\QFT ^{\dagger}\ket{\tilde{\hat{\Psi}}(t)}$. 
The objective is then to estimate the expectation value
\begin{align}
\langle{O}\rangle = \bra{\Psi(t)}O\ket{\Psi(t)}
\end{align}
of the observable $O$ to additive error $\epsilon$. This can be done by estimating $\bra{\tilde{\Psi}(t)}O\ket{\tilde{\Psi}(t)}$ using either of two schemes. 

The first, shot-noise-limited scheme is to directly sample bit-strings, representing eigenstates of $O$, from $\cO(\sigma^2/\epsilon^{2})$ copies of $\ket{\tilde{\Psi}(t)}$, where $\sigma^2=\langle{O^2}\rangle-\langle{O}\rangle^2$ is the variance of the observable. If $\sigma^2$ is unknown, one may use $\norm{O}^2$ as an upper bound. Then, $\langle{O}\rangle$ is estimated as the sample mean of the set of expectation values $\left\{\bra{x_{i}}O\ket{x_{i}}\right\}$, where $\ket{x_{i}}$ are the bit-strings sampled from $\ket{\tilde{\Psi}(t)}$.

The second scheme is to use quantum amplitude estimation (QAE) \cite{Brassard_2002}, a Heisenberg-limited approach. Doing so assumes access to an approximate unitary block-encoding $\tilde{U}_{O}$ of $O$, with rescaling factor $\lambda_{O}$. Using this block-encoding and unitaries to obtain $\ket{\tilde{\Psi}(t)}$, we construct a Grover-like walk iterate and estimate $\bra{\tilde{\Psi}(t)}O\ket{\tilde{\Psi}(t)}$ up to error $\epsilon$ using quantum phase estimation with $O\left(\lambda_{O}/\epsilon\right)$ queries to the iterate's controlled version, applied to $\ket{\tilde{\Psi}(t)}$.

Depending on $\epsilon$, $\sigma^2$, and the cost of block-encoding the observable $O$, either scheme may be more efficient. For low-precision calculations, sampling may be adequate. By contrast, for high-precision calculations, QAE asymptotically provides a quadratic improvement in the precision. The exact cross-over point in $\epsilon$ for which QAE is more efficient depends on the type of dynamics and the nature of the observable.
The choice of scheme therefore depends on a tradeoff between query complexity and cost per query for each scheme. The quadratically better query complexity for QAE comes at the cost of having to construct the iterate, which we now discuss.  

\subsection{Quantum amplitude estimation}
To estimate $\langle{O}\rangle$ via QAE, we use the scheme of~\cite{Rall_2020}. It requires a $(\lambda_{O},q_{O},\epsilon_{O})$-block-encoding of $O$, i.e., a unitary $\tilde{U}_O$ such that
\begin{align}
    \norm{O-\lambda_{O}\left(\bra{0^{q_{O}}}\otimes I\right) \tilde{U}_{O}\left(\ket{0^{q_{O}}}\otimes I\right) }_{\infty} \leq \epsilon_{O},
\end{align}
where $\lambda_{O}\geq \norm{O}_{\infty}$ is a rescaling factor and $q_{O}$ is the number of block-encoding ancillas, which depends on the observable and choice of block-encoding. 
Therefore,
\begin{align}
\left|\langle O \rangle -\lambda_{O}
\left(\bra{0^{q_{O}}}\otimes \bra{\tilde{\Psi}(t)}\right) \tilde{U}_{O}\left(\ket{0^{q_{O}}}\otimes \ket{\tilde{\Psi}(t)}\right)\right|\leq \epsilon_{O},
\end{align}
meaning we can estimate $\bra{\tilde{\Psi}(t)}O\ket{\tilde{\Psi}(t)}$ to within additive error $\epsilon_{O}$ by computing the overlap between $\tilde{U}_{O}\ket{0^{q_{O}}}\otimes \ket{\tilde{\Psi}(t)}$ and $\ket{0^{q_{O}}}\otimes \ket{\tilde{\Psi}(t)}$ and multiplying it by $\lambda_{O}$.

Estimating the overlap between two states is a standard QAE task~\cite{Brassard_2002}. QAE involves performing quantum phase estimation (QPE) on the walk unitary $\widetilde{W}= -\tilde{\cR}_{2}\tilde{\cR}_{1}$, where, in our case, the reflections are
\begin{align}
\cR_{1}&= I-2 \ket{\tilde{\Psi}(t)}\bra{\tilde{\Psi}(t)}\otimes \ket{0^{q_{O}}}\bra{0^{q_{O}}} \\ 
\cR_{2}&=I-2 \tilde{U}_{O}\left(\ket{\tilde{\Psi}(t)}\bra{\tilde{\Psi}(t)}\otimes \ket{0^{q_{O}}}\bra{0^{q_{O}}}\right) \tilde{U}_{O}^{\dagger}.
\end{align}
Given the block-encoding $\tilde{U}_{O}$ and the state $\ket{\tilde{\Psi}(t)}$ which is prepared by the unitary 
\begin{equation}
\tilde{U}=\tilde{U}_{B} \tilde{U}_{\prop}\tilde{U}_{\mathrm{ISP}},
\end{equation}
we construct
\begin{align}\label{eq:QAE_iterate}
\widetilde{W}=\underbrace{(\tilde{U}_{O}\tilde{U}\cR_{0}^{(2)}\tilde{U}^{\dagger}\tilde{U}_{O}^{\dagger})}_{\cR_{2}}\underbrace{(\tilde{U}\cR_{0}^{(1)}\tilde{U}^{\dagger})}_{\cR_{1}},
\end{align}
where $\cR_{0}^{(1)}=2\dyad{0}{0}-I$ is a reflection around the all-zeros state on the output register of $\tilde{U}^{\dagger}$ and $\cR_{0}^{(2)}$ is a reflection around the all-zeros state on the output register of $\tilde{U}^{\dagger}\tilde{U}_{O}^{\dagger}$. For a projector observable, $O^2=O$, the reflection $\cR_{1}$ can be implemented more efficiently as a $(-Z)$ gate controlled on the value of the ancilla qubit of the block-encoding of the observable (the flag qubit in \cref{fig:walk_operator_circuit}), which has no Toffoli cost.

Within the two-dimensional subspace spanned by
\begin{align}
\ket{\omega}&=\ket{\tilde{\Psi}(t)}\otimes\ket{0}\\
\ket{\omega^{\perp}}&=\frac{\tilde{U}_{O}\ket{\omega}-(1/\lambda_{O})O\ket{\omega}}{\sqrt{1-\langle O^2\rangle/\lambda_{O}^2}},
\end{align}
$\widetilde{W}$ acts as the rotation $e^{i\tilde{\theta} Y}$ and has eigenvectors 
\begin{align}
\ket{\omega^{\pm}}&=(\ket{\omega}\pm i\ket{\omega^{\perp}})/\sqrt{2}
\end{align}
with associated
eigenvalues $e^{\pm2i\tilde{\theta}}$, where 
\begin{equation}
    \tilde{\theta}=\arccos(\expval{\tilde{O}}{\tilde{\Psi}(t)}/\lambda_{O}),
\end{equation}
with $\tilde{O}=\lambda_{O}\left(\bra{0^{k}}\otimes I\right) \tilde{U}_{O}\left(\ket{0^{k}}\otimes I\right)$. By estimating the phase $\tilde{\theta}$ to precision $\epsilon_{\QAE}/\lambda_{0}$, we obtain an estimate $\widetilde{\langle \tilde{O}\rangle}$ of the expectation value that satisfies 
\begin{align}
    |\expval{O}{\tilde{\Psi}(t)}-\widetilde{\langle \tilde{O}\rangle}|\leq \epsilon_{\QAE}+\epsilon_{O}=\epsilon_{\meas}.
\end{align}
Using phase doubling~\cite{babbush2018}, this can be done using $2^{s-1}=\lambda_{O}/(2\epsilon_{\QAE})$ calls to the multiplexed walk operator 
\begin{align}
    \dyad{0}{0}\otimes \widetilde{W}^{\dagger}+\dyad{1}{1}\otimes \widetilde{W},
\end{align}
within QPE, where $s=\bits{\lambda_{O}/\epsilon_{\QAE}}$ is the size of the QPE ancilla register. 

The ancilla cost of QAE is then
\begingroup
\allowdisplaybreaks
\begin{align}
C_{\anc}(\textrm{QAE})&=s+C_{\anc}(\widetilde{W})\label{eq:anc_cost_QAE}\\
C_{\anc}(\widetilde{W})&=C_{\anc}(\cR_{2})\label{eq:anc_cost_W}\\
C_{\anc}(\cR_{2})&=\max\{C_{\anc}(U_O \tilde{U}),C_{\anc}(\cR_{0}^{(2)})\}\label{eq:anc_cost_R2}\\
C_{\anc}(\tilde{U}_{O}\tilde{U})&=C_{\anc}(\tilde{U}_{O})+C_{\anc}(\tilde{U})\\
C_{\anc}(\tilde{U})&=\max\{C_{\anc}(\tilde{U}_{\isp})-3\eta_{n}n_{\ext},\nonumber\\
&\qquad C_{\anc}(\tilde{U}_{\prop})\} + 3\eta_{n}n_{\ext}\\
C_{\anc}(\cR_{0}^{(2)})&=3(\eta_{e}n_{p}+\eta_{n}\bar{n}_{\isp})+q'_{O}-2,
\end{align}
\endgroup
where $q'_{O}\leq q_{O}$ is the size of the subset of observable block-encoding ancillas that needs to be reflected on and where $n_{\ext}=\bar{n}_{\isp}-n_{p}$. $C_{\anc}(\tilde{U}_{\isp})$ and $C_{\anc}(\tilde{U}_{\prop})$ are given in \cref{eq:anc_cost_isp_full,eq:anc_cost_U_prop}. $C_{\anc}(\tilde{U}_{O})$ depends on the chosen block-encoding and is therefore not broken down further. 
The Toffoli cost is 
\begin{align}
C_{\Toff}(\textrm{QAE})&=\frac{\lambda_{O}}{2\epsilon}C_{\Toff}(\widetilde{W})\label{eq:Toff_cost_QAE}\\
C_{\Toff}(\widetilde{W})&=C_{\Toff}(\cR_{2})+ C_{\Toff}(\cR_{1})\label{eq:Toff_cost_AA_iterate}\\
C_{\Toff}(\cR_{2})&=2(C_{\Toff}(\tilde{U}_{O})+C_{\Toff}(\tilde{U}))\nonumber\\
&\qquad +C_{\Toff}(\cR_{0}^{(2)})\\
C_{\Toff}(\cR_{0}^{(2)})&= 3(\eta n_{p}+\eta_{n}n_{\ext})+q'_{O}-1\\
C_{\Toff}(\cR_{1})&=2C_{\Toff}(\tilde{U})+C_{\Toff}(\cR_{0}^{(1)})\\
C_{\Toff}(\cR_{0}^{(1)})&= 3(\eta n_{p}+\eta_{n}n_{\ext})-1\\
C_{\Toff}(\tilde{U})&=C_{\Toff}(\tilde{U}_{B})+C_{\Toff}(\tilde{U}_{\prop})\nonumber\\
&\qquad +C_{\Toff}(\tilde{U}_{\isp})\label{eq:Toff_cost_U},
\end{align}
where $C_{\Toff}(\tilde{U}_{\isp})$ and $C_{\Toff}(\tilde{U}_{\prop})$ are given in \cref{eq:toff_cost_isp_full,eq:Toff_cost_U_prop}. In the above equations, we used $C_{\Toff}(\tilde{U})=C_{\Toff}(\tilde{U}^{\dagger})$ and assumed we uncompute $\tilde{U}_{O}$ via its standard inversion, so that $C_{\Toff}(\tilde{U}_{O})=C_{\Toff}(\tilde{U}_{O}^{\dagger})$.
The cost of $\tilde{U}_{B}$ depends on the observable and is assumed to be efficient, which holds for the $\QFT$. 
For a projector observable, $O^2=O$, one may simplify the construction of $\cR_{1}$, as we will see in \cref{sect:yield_protocol}.

\begin{figure*}
\resizebox{.999\linewidth}{!}{
\begin{quantikz}[row sep={1cm, between origins},column sep=0.4cm]
&&\ctrl{1}&\\
\lstick[1]{ext}&\qwbundle{}&\gate[3]{\widetilde{W}_{\Pi_{S}} \vphantom{U_{\Pi_{S}}^{\dagger}}}&\\
\lstick[1]{data}&\qwbundle{}&&\\
\lstick[1]{flag}&&&
\end{quantikz}
\,=\,
\begin{quantikz}[row sep={1cm, between origins},column sep=0.4cm]
&\ctrl{3}&&&&&\ctrl{1}&&&&&\\
&&&&&\gate[2]{\tilde{U}_{\mathrm{ISP}}^{\dagger}}&\gate[3]{\cR_{0}^{(\QAE)}\vphantom{U_{\Pi_{S}}^{\dagger}}}&\gate[2]{\tilde{U}_{\mathrm{ISP}}^{\vphantom{\dagger}}}&&&&\\
& &\gate[2]{U_{\Pi_{S}}^{\dagger}}&\gate[1]{\widetilde{\QFT}^{\vphantom{\dagger}}_{\vphantom{prop}}}&\gate[1]{\tilde{U}_{\mathrm{prop}}^{\dagger}}&&&&\gate[1]{\tilde{U}_{\mathrm{prop}}^{\vphantom{\dagger}}}&\gate[1]{\widetilde{\QFT}^{\dagger}_{\vphantom{prop}}}&\gate[2]{U_{\Pi_{S}}^{\phantom{\dagger}}}&\\
&\octrl{0}&&&&&&&&&&
\end{quantikz}
}
\caption{Circuit for the controlled iterate $\widetilde{W}_{\Pi_{S}}$ on which we perform QAE to estimate a quantum yield with respect to the real-space final state $\ket{\tilde{\Psi}(t)}$. The controls on $\tilde{U}_{\mathrm{prop}}$, $\tilde{U}_{\mathrm{ISP}}$, $\widetilde{\QFT}$ and $U_{\Pi_{S}}$ are dropped as depicted. Ancilla qubits are not shown, except the QAE ancilla ``flag'' and the exterior grid ancillas ``ext'', which continue throughout the circuit.}
    \label{fig:walk_operator_circuit}
\end{figure*}
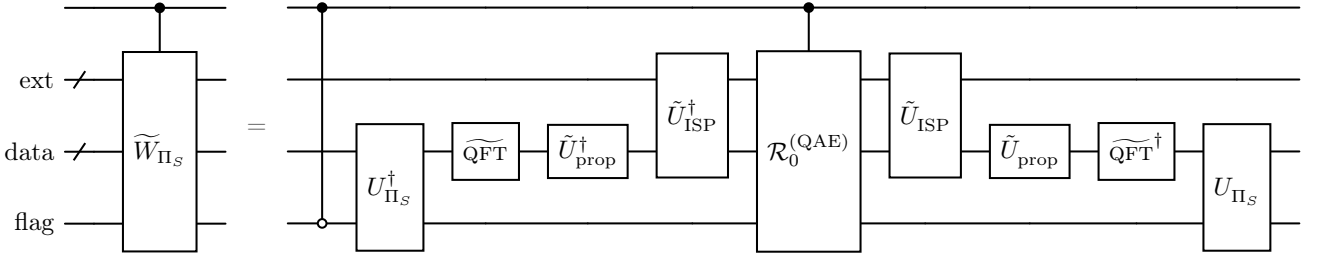

\begin{figure*}
\resizebox{.999\linewidth}{!}{
\begin{quantikz}[row sep={0.7cm, between origins}, column sep=0.5cm,,wire types={n,n,n,n,q,q,n}]
& & & & & \lstick[4]{$\ket{0^s}_{\mathrm{QAE}}$} & \gate[wires=4]{\Had^{\otimes s}} \setwiretype{q} & & & & & \ \ldots \  & \octrl{6} & &\octrl{6}  &\gate[4]{\widetilde{\QFT}^{\dagger}} & \meter{} \\ 
& & & & & \rstick{\hspace{2mm}$\vdots$}  \\
& & & & & & \setwiretype{q} & & \octrl{4} & &\octrl{4} & \ \ldots \   & & & & &  \meter{} \\
& & & & & & \setwiretype{q} & \ctrl{1} & & & &\ \ldots \   & & & & &  \meter{} \\[0.5cm]
\lstick[1]{$\ket{0}_{\mathrm{ext}}^{\otimes 3\eta_n n_\mathrm{ext}}$} & \qwbundle{}&  \gate[wires=2]{\tilde{U}_{\mathrm{ISP}}} & & & &  & \gate[3]{\widetilde{W}_{\Pi_{S}}\vphantom{\widetilde{W}_{\Pi_{S}}^{2^{s-2}}}} & & \gate[3]{\widetilde{W}_{\Pi_{S}}\vphantom{\widetilde{W}_{\Pi_{S}}^{2^{s-2}}}} & & \ \ldots \  & & \gate[3]{\widetilde{W}_{\Pi_{S}}^{2^{s-2}}} & & \\ 
\lstick[1]{$\ket{0}_{\mathrm{data}}^{\otimes C_\mathrm{data}}$} & \qwbundle{} & & \gate[1]{\tilde{U}_{\mathrm{prop}}^{\vphantom{\dagger}}} & \gate[1]{\widetilde{\QFT}^\dagger_{\vphantom{prop}}} & \gate[2]{U_{\Pi_{S}}} & & & & & &\ \ldots \ & & & &  \\ 
\lstick{$\ket{1}_{\mathrm{flag}}$} &\setwiretype{q} & & & & & & &\octrl{0}  &  &\octrl{0} & \ \ldots \   &\octrl{0} & & \octrl{0} &
\end{quantikz}
}
\caption{Circuit for the dynamics algorithm applied to yield estimation, composed of initial state preparation $\tilde{U}_{\mathrm{ISP}}$, dynamics $\tilde{U}_{\mathrm{prop}}$, and amplitude estimation via the QAE iterate $\widetilde{W}_{\Pi_{S}}$, implemented using QPE with phase doubling. Ancilla qubits are not shown, except the QAE ancilla ``flag'' and the exterior grid ancillas ``ext'', which continue throughout the circuit.}
    \label{fig:algo_circuit}
\end{figure*}
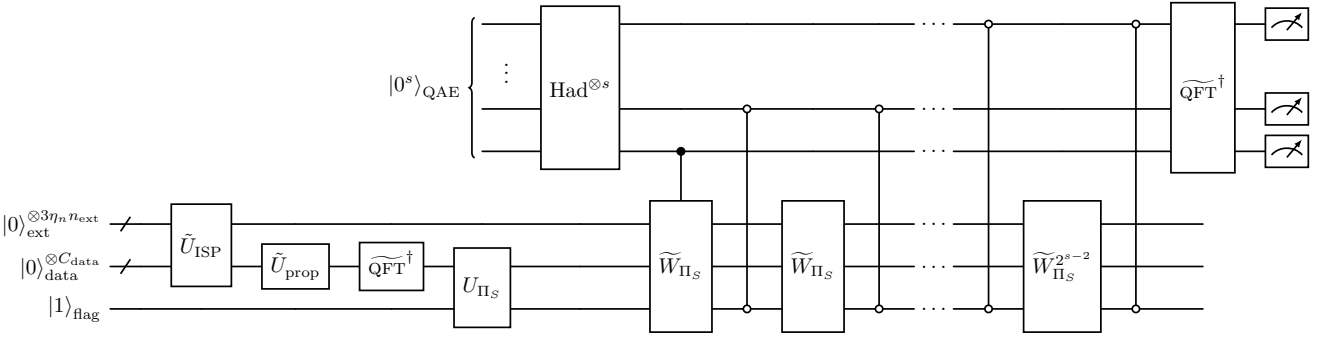

\subsection{Example: Reaction yields and rates}\label{sect:yield_protocol}
Key chemical observables are reaction (or product) yields, which are the probabilities that particular products are created in a reaction at a given time. They determine whether chemical reactions occur or not and, if computed as functions of time, how fast. This is quantified by reaction rates, which can be estimated by numerical differentiation of yields as functions of time or from fits to a time-series of yields.

Here, we show how to efficiently estimate product yields using our QAE scheme. We construct an explicit and cheap block-encoding of the yield observable with a Toffoli cost of $\cO(n_{p}^2)$ that can be exactly synthesized using standard quantum arithmetic circuits.

Chemical reaction products can be identified through reaction (or product) channels~\cite{tannor2007}, which are non-overlapping regions that cover the space of nuclear coordinates~$\mathbf{R}$, as in \cref{fig:reactions}.
A reaction channel $S$ can be assigned the projector observable
\begin{align}
    \Pi_{S}=\sum_{\mathbf{R}\in S} \ket{\mathbf{R}}\bra{\mathbf{R}}.
\end{align}
Because reaction channels partition the space of nuclear positions, these projectors form an orthogonal and complete set, $\Pi_{S}\Pi_{S'}=\delta_{SS'}\Pi_{S}$ and $\sum_{S} \Pi_{S}=I$. 
For the real-space wavefunction $\ket{\tilde{\Psi}(t)}=\QFT^{\dagger}\ket{\tilde{\hat{\Psi}}(t)}$, the yield for channel $S$ is the probability
\begin{align}
\langle \Pi_{S} \rangle = \bra{\tilde{\Psi}(t)}\Pi_{S}\ket{\tilde{\Psi}(t)} 
\end{align}
that the dynamics yields the associated outcome upon measuring the positions of the nuclei. The completeness of the projectors ensures that $\sum_S \langle \Pi_{S} \rangle=1$.  

We consider reaction channels to be determined by cutoffs in pairwise internuclear distances. This approach loses no generality because, especially at long times, chemical products can be determined by the distances between nuclei. We define reaction channel $S$ as a region of real-space nuclear coordinates~$\mathbf{R}$ determined by $B_j \le \eta_n(\eta_n-1)/2$ constraints $C_{j,k}(\mathbf{R})$ (with $1\le k \le B_j$) on distances between pairs of nuclei, which must be either smaller or larger than a cutoff $b_{j,k}$,
\begin{equation}
    C_{j,k}(\mathbf{R}) \in \{X_{j,k}(\mathbf{R}), \neg X_{j,k}(\mathbf{R})\},
\end{equation}
where
\begin{equation}
X_{j,k}(\mathbf{R}) = \mathbb{B}\left(\norm{\R_{\alpha_{j,k}}-\R_{\beta_{j,k}}}_2 > b_{j,k}\right),
\end{equation}
with $\mathbb{B}$ denoting the boolean truth value. That is, $C_{j,k}(\mathbf{R}) = X_{j,k}(\mathbf{R})$ if nuclei $\alpha_{j,k}$ and $\beta_{j,k}$ are required to be further than $b_{j,k}$ apart (corresponding to unbonded atoms) and $C_{j,k}(\mathbf{R}) = \neg X_{j,k}(\mathbf{R})$ otherwise (corresponding to bonded atoms). Reaction channel $S$ is then 
\begin{equation}\label{eq:react_chann_def}
    S=\{\mathbf{R} \mid I(\mathbf{R},S) = 1 \},
\end{equation}
where the indicator function $I(\mathbf{R},S)$ indicates whether $\mathbf{R}$ is in $S$,
\begin{equation}
    I(\mathbf{R},S)=\bigwedge_{k \le B_j} C_{j,k}(\mathbf{R}) = 
    \begin{cases}
        1, \text{ if } \mathbf{R} \in S, \\
        0, \text{ if } \mathbf{R} \notin S.
    \end{cases}
\end{equation}

To estimate the yield using our QAE measurement protocol, we define a $(1,1,0)$-block-encoding $U_{\Pi_{S}}$ of the observable $\Pi_{S}$ such that 
\begin{align}
    \Pi_{S}=\left(\bra{0}\otimes I\right)U_{\Pi_{S}}\left(\ket{0}\otimes I\right). 
\end{align}
This block-encoding is of the form
\begin{equation}
    U_{\Pi_{S}} = \sum_{\mathbf{R}} \ket{\mathbf{R}}\bra{\mathbf{R}}_{\data}\otimes X^{I(\mathbf{R},S)\oplus 1}_{\flag}, 
\end{equation}
where $X$ is the Pauli-$X$ gate and where $\oplus$ is logical \textsc{xor}.
$I(\mathbf{R},S)$ can be evaluated efficiently, since only a polynomial number of bond distances need to be checked ($B_{j}=\cO(\eta_n^2)$) and each inequality test takes time polynomial in $n_{p}$. $U_{\Pi_{S}}$ can then be exactly synthesized by a quantum circuit that computes the function 
\begin{align}
    \mathbb{B}\Bigl(\sum_{w\in\{x,y,z\}}(R_{\alpha_{j,k},w}-R_{\beta_{j,k},w})^2 > b_{j,k}^2\Bigr)
\end{align}
for relevant nuclear pairs and stores the outcome in the flag register. We do so using a quantum-arithmetic circuit that, in order, component-wise subtracts the position coordinates, sums the squares of the differences, and tests the inequality against the classical constant $b_{j,k}^2$. The flag register is then flipped when the conditions are not satisfied, which we do via a \textsc{not} gate controlled on the outcomes of all the inequality tests. We give an explicit construction in \cref{app:indicator_circuit}, which has cost 
\begin{align}
C_{\anc}(\Ind)&=3n_{p}^2-n_{p}\label{eq:anc_cost_U_Pi}\\
C_{\mathrm{Toff}}(\Ind)&=3B_j(2 n_p^2 + 2n_p-3)+3 n_{\nuc} (n_{p}-2)-1\label{eq:Toff_cost_U_Pi}
\end{align}
for checking the $B_{j}$ nuclear pairs, involving $n_{\nuc}$ nuclei. In the worst case, $B_{j}=\cO(\eta_{n}^2)$, which would give $C_{\Toff}(\Ind)=\cO(\eta_{n}^{2}n_{p}^2)$. In practice, $B_{j}$ is small and independent of system size, giving a small Toffoli cost of $C_{\Toff}(\Ind)=\cO(n_{p}^2)$. Hence, the cost of $\Ind$ is usually negligible compared to other parts of the algorithm.

Given the above $(1,1,0)$-block-encoding $U_{\Pi_{S}}$, we construct the QAE iterate $\widetilde{W}_{\Pi_{S}}$ for the special case of a projector observable such as a yield, where the subspace defined by $\Pi_{S}$ can be flagged using the block-encoding ancilla of $U_{\Pi_{S}}$. Using this flag qubit enables a simpler construction of the reflection $\cR_{1}$ as 
\begin{align}
\cR_{1}&= I-2I_{\data}\otimes \ket{0}\bra{0}_{\flag}=I\otimes (-Z). 
\end{align}
Then, given the unitary 
\begin{align}\label{eq:U}
\tilde{U}=\tilde{\QFT}^{\dagger} \tilde{U}_{\prop}\tilde{U}_{\mathrm{ISP}},
\end{align}
we construct the QAE iterate for yield estimation as
\begin{align}
\widetilde{W}_{\Pi_{S}}=(\Ind \tilde{U}\cR_{0}^{(\QAE)}\tilde{U}^{\dagger}\Ind^{\dagger})(I\otimes (-Z)),
\end{align}
where $\cR_{0}^{(\QAE)}= 2\ket{0}\bra{0}_{}\otimes \ket{0}\bra{0}_{\flag}-I$ is the reflection around the all-zeros state on the output register of $U_{\Pi_{S}}^{\dagger}\tilde{U}^{\dagger}$.
The circuit implementation of a controlled version of $\widetilde{W}_{\Pi_{S}}$ is shown in \cref{fig:walk_operator_circuit}. 

Given access to the state $\ket{\tilde{\Psi}(t)}$, the ancilla cost of estimating $\bra{\tilde{\Psi}|(t)} \Pi_{S} \ket{\tilde{\Psi}(t)}$ up to additive error $\epsilon$ is 
\begin{align}
C_{\anc}(\QAE(\Pi_{S}))&=s+C_{\anc}(\widetilde{W}_{\Pi_{S}})\\
C_{\anc}(\widetilde{W}_{\Pi_{S}})&=1+3\eta_{n}n_{\ext} + \max\{C_{\anc}(U_{\Pi_{S}})-1, \nonumber\\
&\qquad{} C_{\anc}(\tilde{U}_{\prop}), C_{\anc}(\tilde{U}_{\isp})-3\eta_{n}n_{\ext},\nonumber\\
&\qquad{} C_{\anc}(\cR_{0}^{(\QAE)})\}\\
C_{\anc}(\cR_{0}^{(\QAE)})&=3(\eta_{e}n_{p}+\eta_{n}\bar{n}_{\isp})-1,
\end{align}
with $C_{\anc}(U_{\Pi_{S}})$ given in \cref{eq:anc_cost_U_Pi}, $C_{\anc}(\tilde{U}_{\prop})$ given in \cref{eq:anc_cost_U_prop} and $C_{\anc}(\tilde{U}_{\isp})$ given in \cref{eq:anc_cost_isp_full}. 
In $C_{\anc}(\widetilde{W}_{\Pi_{S}})$, the first two terms---$1+3\eta_{n}n_{\ext}$---account for the QPE flag ancilla and the exterior grid ancillas, which continue throughout the circuit.

The Toffoli costs are 
\begin{align}
C_{\Toff}(\textrm{QAE}(\Pi_{S}))&= C_{\Toff}(\widetilde{W}_{\Pi_{S}})/2\epsilon \label{eq:Toff_cost_yield_estimation}\\
C_{\Toff}(\widetilde{W}_{\Pi_{S}})&=2(C_{\Toff}(U_{\Pi_{S}})+C_{\Toff}(\tilde{U}))\nonumber\\
&\qquad {} +C_{\Toff}(\cR_{0}^{(\QAE)})\\
C_{\Toff}(\tilde{U})&=C_{\Toff}(\widetilde{\QFT})+C_{\Toff}(\tilde{U}_{\prop})\nonumber\\
&\qquad {} +C_{\Toff}(\tilde{U}_{\isp})\\
C_{\Toff}(\cR_{0}^{(\QAE)})&=3\eta_{e}n_{p}+3\eta_{n}\bar{n}_{\isp},\label{eq:Toff_cost_R0_QAE}
\end{align}
where $C_{\Toff}(U_{\Pi_{S}})$ is given in~\cref{eq:Toff_cost_U_Pi} and~\cite{nam2020approximate}
\begin{multline}\label{eq:Toff_cost_QFT}
C_{\Toff}(\widetilde{\QFT}) = 4 n(\log (n/\epsilon_{\QFT})-2)\\+0.6 \log \left(n\log\left({n/\epsilon_{\QFT}}\right)/\epsilon_{\QFT}\right).
\end{multline}
This protocol, shown in \cref{fig:algo_circuit}, computes the yield of a single reaction channel. To compute all yields, one can either run a separate computation for each channel or use more sophisticated methods for extracting multiple expectation values~\cite{Huggins_2022}, which offer a quadratic improvement of the cost in the number of observables compared to running QAE for each channel separately.

\renewcommand{\arraystretch}{1.5}
\begin{table*}
	\centering
	\begin{longtable}
	    {p{1.7cm}p{5.5cm}p{9.4cm}}
		\toprule
		Subroutine & Description & Cost (Toffoli gates) \\
		\midrule
        $\widetilde{\textsc{asp}}^{(x)}$ & Prepare arbitrary state of $D^{(x)}$ computational basis states, \cref{eq:ctoff_asp} &   $2^{5/2}(1 + \sqrt{2}) D^{(x)}(b_{\textsc{asp}}^{(x)} + 1)^{1/2} + 2 \log(D^{(x)})(b_{\textsc{asp}}^{(x)} - 4)$ \\ 
		$\mathrm{SoSlat}^{(x)}$    & Prepare linear combination of $D^{(x)}$ configurations in ONB, \cref{eq:prep_2nd_t}   & $D^{(x)}(2 \log D^{(x)} + 3)$ \\
        \textsc{onb2mob} & Transform from ONB to MOB, \cref{eq:cost_onb2mob} &  $N_{\mathrm{MOB}} (2 \eta_e  + \lceil \log (\eta_e + 1) \rceil + 
        \eta_e \lceil \mathrm{log}(N_{\mathrm{MOB}})\rceil -4 )$
        \\
        \textsc{asym}     & Antisymmetrize electrons, \cref{eq:Antisym}   & $2(\eta_e-1)(\log \bar{n}_p + 1) + \bar{n}_p \log \bar{n}_p (1 + \log \bar{n}_p )(6\log \bar{n}_p + n_p + 1)/4$ \\
        $(\widetilde{W}^{(e)})^{\otimes\eta_e}$     & Transform electronic states from MOB to PWB, \cref{eq:Elec_Gates} &  $\eta_e N_{\mathrm{MOB}} n_p +2 \eta_e \sum_{a=1}^{N_{\mathrm{MOB}}} \sum_{i=1}^{n_p} \big( 32(1+\sqrt{2})(b_{\mathrm{rot}}^{(e)}+1)^{1/2}\cdot {m}_{ai}^{(e)} (\bar{{m}}_{ai}^{(e)})^{1/2}
        +(8 b_{\mathrm{rot}}^{(e)}-15) {m}_{ai}^{(e)} \log (2 \bar{{m}}_{ai}^{(e)}) \big)$\\
            \textsc{onb2smb} & Transform from ONB to SMB, \cref{eq:cost_onv2sm} &  $N_{\mathrm{vib}} N_{\mathrm{SMB}} ( \lceil \log N_{\mathrm{SMB}} \rceil -2 )$\\
        $\widetilde{W}^{(n)}$  & Transform nuclear states from MOB to PWB, \cref{eq:one_vib_to_PWB_gates,eq:full_vib_to_PWB_gates} & $\sum_{i=1}^{3\eta_n}\big(N_{\mathrm{SMB}} n_{\mathrm{ISP}} +2\sum_{\mu=0}^{N_{\mathrm{SMB}}-1}\sum_{j=1}^{n_{\mathrm{ISP}}} \big( 32(1+\sqrt{2})(b_{\mathrm{rot}}^{(n)}+1)^{1/2}\cdot m_{i \mu j}^{(n)} (\bar{m}_{i \mu j}^{(n)} )^{1/2} 
        +(8 b_{\mathrm{rot}}^{(n)}-15) m_{i \mu j}^{(n)} \log (2 \bar{m}_{i \mu j}^{(n)})\big)\big)$ \\
        $\tilde{U}_{\mathbf{A}^{- \top}}^{\textsc{lct}}$    & Lin.\ coord.\ transformation, \cref{eq:LCT_resources_c}   & $\frac{9}{2}\eta_n^2(8 \bar{n}_{\mathrm{ISP}}^2 + 39\bar{n}_{\mathrm{ISP}} - 8)
    - \frac{3}{2}\eta_n(8\bar{n}_{\mathrm{ISP}}^2 + 35\bar{n}_{\mathrm{ISP}} -8) - \bar{n}_{\mathrm{ISP}}  $\\
        $\widetilde{\textsc{pk}}$      & Phase kickback, \cref{eq:cost_PK}   & $3\eta_n (4 \bar{n}_{\mathrm{ISP}} b_{\mathrm{grad}} + b_{\mathrm{grad}} - 2\bar{n}_{\mathrm{ISP}}) +  b_{\mathrm{grad}} (1.149 (\log(b_{\mathrm{grad}}\epsilon_{\mathrm{PK}}^{-1})) + 9.2)/4$ \\
        $\textsc{tc2sm}^{\otimes 3\eta_{n}}_{\bar{n}_{\mathrm{ISP}}}$    & Convert from two’s
        complement to signed magnitude, \cref{eq:cost_tc2sm} & $3 \eta_n (\bar{n}_{\mathrm{ISP}}-2)$\\
        $\proxPREP_H$      & Prepare state encoding Hamiltonian coefficents, \cref{eq:Toff_cost_Prep_H}      & $n_{\theta}+\eta+4\eta_{e}+\mu_{T}+ 5n_{\eta}+6\lceil\log(2\eta_{e})\rceil+4b_{r}+3n_{p}^2+13n_{p}+4n_{\cM}(n_{p}+1) -6$ \\
        $\proxPREP^{\dagger}_H$      & Invert $\proxPREP_H$, \cref{eq:Toff_cost_UNPREP_H}     & $   n_{\theta}+\Er(\eta)+\mu_{T}+5n_{\eta}+6n_{p}+2\Er(2\eta_{e})+6\lceil\log(2\eta_{e})\rceil+4b_{r}-9  $ \\
        $\contr{\SEL_{H}}$      & Controlled application of Hamiltonian terms, \cref{eq:Toff_cost_SEL_H,eq:ctrl_SEL_H_cost}   & $18\eta n_{p}+6\eta+5n_{p}+24n_{p}-8$  \\
        ctrl-$\cR_{0}^{(\cW)}$      & Controlled reflection around all-zeroes on $\proxPREP_{H}^{\dagger}$ register, \cref{eq:Toff_cost_qubiterate_refl}   & $n_{\eta}+6n_{p}
        +n_{\cM}+2\bits{2\eta_{e}}+10$  \\
        $R(\phi_{j},\gamma_{j})$    & QSP single-qubit rotations, \cref{eq:Toff_cost_QSP_rotations}   & $ \tfrac{1}{2}( 0.56 \log
(1/\epsilon_{\rot}) + 5.3)$ \\
        $U_{\Pi_{S}}$    & Indicator function for yields, \cref{eq:Toff_cost_U_Pi}    & $3B_j(2 n_p^2 + 2n_p-3)+3n_{\nuc}(n_{p}-2)-1 $     \\
        $\cR_{0}^{(\QAE)}$& Controlled reflection around all-zeroes
        on wavefunction register, \cref{eq:Toff_cost_R0_QAE}& $3\eta_{e}n_{p}+3\eta_{n}\bar{n}_{\isp}$ \\
        $\QFT$     & QFT on $n$ qubits, \cref{eq:Toff_cost_QFT}& $4 n(\log (n/\epsilon_{\QFT})-2)+0.6 \log (n\log({n/\epsilon_{\QFT}})/\epsilon_{\QFT})$ 
        \\
        \midrule
        $\tilde{U}_{\mathrm{ISP}}^{(e)'}$      & Prepare electrons in MPS,  \cref{eq:UISP_e_prime}  & $C_{\mathrm{Toff}}(\textsc{onb2mob})+C_{\mathrm{Toff}} (\textsc{asym})+C_{\mathrm{Toff}}  ((\widetilde{W}^{(e)})^{\otimes\eta_e})$ \\
        $\tilde{U}^{(n)'}_{\mathrm{ISP}}$    & Prepare nuclei in MPS and transform to Cartesian coordinates, \cref{eq:UISP_n_prime}  & $  C_{\mathrm{Toff}}(\textsc{onb2smb})+C_{\mathrm{Toff}}(\widetilde{W}^{(n)}) +C_{\mathrm{Toff}}(\tilde{U}_{\mathbf{A}^{- \top}}^{\textsc{lct}})+C_{\mathrm{Toff}}(\widetilde{\textsc{pk}}) + C_{\mathrm{Toff}}(\textsc{tc2sm}^{\otimes 3\eta_{n}}_{\bar{n}_{\mathrm{ISP}}}) $     \\
         $\tilde{U}_{\mathrm{ISP}}$    & Prepare initial state,~\cref{eq:toff_cost_isp_full} & $C_{\mathrm{Toff}}(\tilde{U}_{\mathrm{ISP}}^{(e)'}) + C_{\mathrm{Toff}}(\tilde{U}_{\mathrm{ISP}}^{(n)'}) +  \nonumber
     C_{\mathrm{Toff}}(\widetilde{\textsc{asp}}^{(en)}) + C_{\mathrm{Toff}}(\mathrm{SoSlat}^{(en)}) $     \\
        ctrl-$\proxW$   & Controlled qubiterate, \cref{eq:Toff_cost_ctrl_W}  & $C_{\Toff}(\contr{\SEL_{H}})+C_{\Toff}(\proxPREP_{H})+C_{\Toff}(\proxPREP_{H}^{\dagger})+C_{\Toff}(\contr{\cR_{0}^{(\cW)}})$    \\
	$\tilde{U}_{\mathrm{prop}}$     & Time evolution under $H$, \cref{eq:Toff_cost_U_prop}       &     $ \tilde{d} C_{\Toff}(\textrm{ctrl-}\proxW)+(\tilde{d}+1)C_{\Toff}(R(\phi_{j},\gamma_{j}))$      \\
        $ \textrm{QAE} $      & Quantum amplitude estimation for yields, \cref{eq:Toff_cost_yield_estimation}    &$\lambda_{O}\epsilon^{-1}\big(C_{\Toff}(\cR_{0}^{(\QAE)})+2C_{\Toff}(\tilde{U}_{\mathrm{ISP}})+2C_{\Toff}(\tilde{U}_{\mathrm{prop}})
    +2C_{\Toff}(U_{\Pi_{S}})+2C_{\Toff}(\QFT)\big)$ \\
        \midrule
        $\cA$     & Calculate quantum yield, \cref{eq:full_cost_algorithm}   & $C_{\Toff}(\tilde{U}_{\mathrm{ISP}})+C_{\Toff}(\tilde{U}_{\mathrm{prop}}) + C_{\Toff}(\QFT) + C_{\Toff}(\textrm{QAE})$\\
		\bottomrule
	\end{longtable}
	\caption{
	Constant-factor Toffoli costs for all subroutines of the dynamics algorithm $\mathcal{A}$ when used to evaluate a quantum yield.}
	\label{tab:Gate_costs}
\end{table*}

\section{Total cost and error budget}\label{sect:errors}

We have presented complete descriptions of all stages of our algorithm from ISP to dynamics to measurement. For each stage, we have bounded the relevant sources of error and given asymptotic as well as constant-factor costs in terms of the relevant parameters. Here, we briefly summarize these costs and give a bound on the error of the expectation value of the final observable, $\epsilon=|\widetilde{\langle \tilde{O}\rangle}-\langle O \rangle|$, in terms of errors for all stages. 

\subsection{Algorithm costs}
The space costs of the algorithm consist of the number of logical qubits to store the molecular wavefunction and of the number of ancilla qubits. The memory cost of storing the wavefunction is 
\begin{align}
    C_{\data}(\cA)&=3\eta n_{p}+\eta_{e},
\end{align}
as given in \cref{eq:cdata}. The ancilla cost is the maximum number of ancillas required at any stage of the algorithm, 
\begin{align}
C_{\anc}(\cA)= C_{\anc}(\mathrm{QAE}),
\end{align}
where $ C_{\anc}(\mathrm{QAE})$ is given in \cref{eq:anc_cost_QAE}.
The time cost is given by the number of Toffoli gates required,
\begin{align} \label{eq:full_cost_algorithm}
    C_{\Toff}(\cA)&=C_{\Toff}(\tilde{U})+C_{\Toff}(\mathrm{QAE}),
\end{align}
where $C_{\Toff}(\tilde{U})$ and $C_{\Toff}(\mathrm{QAE})$ are given in \cref{eq:Toff_cost_U,eq:Toff_cost_QAE}. The constant-factor expansion of $C_{\Toff}(\cA)$ in terms of the algorithmic subroutines is given in \cref{tab:Gate_costs}.

\subsection{Error budget}

In our algorithm, the precise computational cost can be given for every step, given a desired error $\epsilon$ in $\langle O\rangle=\bra{\Psi(t)}O\ket{\Psi(t)}$, which depends on the accumulated errors in all stages of the algorithm. In \cref{app:yield_error}, we show that our algorithm can estimate $\langle O\rangle$ of any observable with a $\left(\lambda_{O}, q_{O}, \epsilon_{O}\right)$-block encoding, where  $\lambda_{O}\geq\norm{O}_{\infty}$, to an error
\begin{align}\label{eq:error_split}
    \epsilon= 2\lambda_{O}(\epsilon_{\mathrm{ISP}}+\epsilon_{\mathrm{prop}}+\epsilon_{B})+\epsilon_\mathrm{meas},
\end{align}
where $\epsilon_{\mathrm{ISP}}$, $\epsilon_{\mathrm{prop}}$, $\epsilon_{B}$, and $\epsilon_{\mathrm{meas}}$ are the subroutine errors of, respectively, ISP, time evolution, basis transformation, and measurement.

Given a target $\epsilon$, to calculate the total cost of the algorithm using \cref{tab:Gate_costs}, one must assign an appropriate error budget, i.e., choose values of $\epsilon_{\mathrm{ISP}}$, $\epsilon_{\mathrm{prop}}$, $\epsilon_{B}$, and $\epsilon_{\mathrm{meas}}$ that satisfy
\begin{align}\label{eq:expect_val_error_bound}
    2\lambda_{O}(\epsilon_{\mathrm{ISP}}+\epsilon_{\mathrm{prop}}+\epsilon_{B})+\epsilon_\mathrm{meas}\le \epsilon.
\end{align}
To minimize the cost of the algorithm, it is in general advisable to allocate most of the error budget to $\epsilon_\mathrm{meas}$, because of the query complexity of the QAE. Moreover, the costs of $U_{\mathrm{ISP}}$ and $U_{\mathrm{prop}}$ are only logarithmically sensitive to their errors, see \cref{eq:asympt_Toff_cost_ISP,eq:asympt_Toff_cost_U_prop}. In general, the allocation of the error budget is an optimization problem whose more detailed solution we leave for future work.

\subsection{Overall asymptotic scaling}

The asymptotic Toffoli cost of our algorithm to estimate the expectation value of a unit-norm observable (such as a yield) is 
\begin{equation}
C_{\Toff}(\cA)=\widetilde{\cO}\left(
( \lambda_{H}t\eta
+ \eta_{n}^2
+ \eta_e N_{\mathrm{MOB}}
+ \eta_n N_{\mathrm{SMB}}
)/\epsilon\right),
\end{equation}
where (see \cref{eq:lambda_T,eq:lambda_V})
\begin{equation}
    \lambda_{H}=\cO\left( \eta \frac{N_{\grid}^{2/3}}{\Omega^{2/3}}+\eta^{2} \frac{N_{\grid}^{1/3}}{\Omega^{1/3}}  \right)
\end{equation} 
and the factors $N_{\mathrm{MOB}}$ and $N_{\mathrm{SMB}}$ are lower-bounded by $\eta_{e}$ and $\eta_{n}$, respectively. The notation $\widetilde{\cO}(\cdot)$ suppresses polylogarithmic factors. 

Hence, this work is the first end-to-end algorithm for quantum simulation of chemistry with a complexity sublinear in the grid size $N_\mathrm{grid}$, scaling as $\cO( N_\mathrm{grid}^{2/3} )$. We achieve this end-to-end scaling by combining our exponentially faster method for ISP with the qubitization-based approach to dynamics. Low scaling in $N_\mathrm{grid}$ is critical for first-quantized simulation because $N_\mathrm{grid} \gg \eta$. 

Dropping the subdominant terms, we can write
\begin{align}
    C_{\Toff}(\cA) &= \widetilde{\cO}(\eta \lambda_{H} t/\epsilon) \\
    &= \widetilde{\cO}\big((\eta^2 K_\mathrm{max}^2 + \eta^3 K_\mathrm{max})t/\epsilon\big)\label{eq:cost-eta-Kmax},
\end{align}
where $K_\mathrm{max}=\pi (N_\mathrm{grid}/\Omega)^{1/3}$.
This scaling can be rewritten in terms of the minimal set of input parameters $\{\eta,t,\epsilon \}$. As discussed in \cref{sect:grid}, resolving arbitrary molecular wavefunctions is expected to require $K_{\max}=\cO(\epsilon^{-1/3})$; in that case, we obtain 
\begin{align}
C_{\Toff}(\cA)=\widetilde{\cO}\big(\eta^3t/\epsilon^{4/3}+\eta^2t/\epsilon^{5/3}\big).   \label{eq:complexity_full_alg}
\end{align}
However, if the assumptions given in \cref{sect:grid} are met and the $K_{\max}$ of \cref{eq:k_max_asymp} is sufficient, then $K_{\max}=\cO(\sqrt{\log(1/\epsilon)})$, giving the improved scaling of 
$C_{\Toff}(\cA)= \widetilde{\cO}\left(\eta^3t/\epsilon\right)$.

\begin{figure*}\label{fig:dissociation}
\centering
    \includegraphics[width=\textwidth]{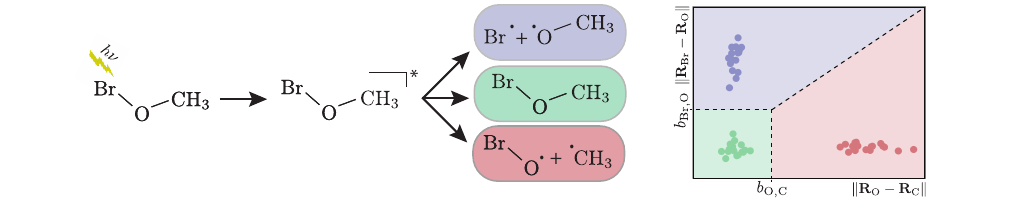}
    \caption{Example of a photodissociation reaction that can be simulated by our algorithm. A molecule of methyl hypobromite, \ce{CH3OBr}, is prepared in a photoexcited state and time-evolved. The nuclear positions $\mathbf{R}$ can be measured to give the bond lengths $\norm{\mathbf{R}_{\mathrm{Br}}-\mathbf{R}_{\mathrm{O}}}$ and $\norm{\mathbf{R}_{\mathrm{O}}-\mathbf{R}_{\mathrm{C}}}$, which determine which of the three color-coded reaction channels the measured nuclear configuration belongs to.}
    \label{fig:reactions}
\end{figure*}

\section{Application to photochemistry}\label{sect:photochemistry}

Fully quantum-dynamical simulations are most important for chemical reactions involving strong correlations between the nuclei and the electrons, a common occurrence in photochemistry. For this reason, we analyze the cost of our approach for simulating photochemical processes, with a focus on the particularly challenging computational problem of photodissociation (or photolysis), i.e., the photoinduced splitting of a molecule into fragments. A typical photochemical process begins by photoexcitation, i.e., a molecule in the electronic ground state absorbing a photon and being promoted to an excited state. The large amount of absorbed energy induces molecular dynamics, which can result in the molecule relaxing to a lower state or dissociating. If the molecule has dissociated, the relevant observable is the photodissociation quantum yield, the probability of a given bond breaking after photon absorption. 

In particular, we focus on atmospheric photochemistry as a  promising application. Most of the Earth's atmospheric chemistry is driven by photochemical processes, including those involving greenhouse gases, ozone depletion, and smog. Understanding their rates, yields, and mechanisms would improve our ability to reduce pollution and construct more accurate climate models~\cite{krol1997implications}. 

Simulating atmospheric photochemistry is likely to be a promising problem for quantum computers to achieve a computational advantage. The widespread breakdown of the BO approximation means that numerically exact, fully quantum calculations of quantum yields on classical computers are prohibitively expensive for all but the smallest molecules~\cite{cigrang2024modeling}. By contrast, on a quantum computer, atmospheric photochemical reactions are ideally suited for our algorithm because they occupy a favorable regime for each of the three parameters that determine the scaling~\cref{eq:complexity_full_alg}. First, atmospheric molecules are small (small $\eta$), while still being too large to be classically tractable in most cases. Second, photoinduced dynamics is ultrafast (small $t$), often occurring within tens to hundreds of femtoseconds~\cite{turro2009principles}. And third, yields and other dynamical observables can usually be computed with relatively high errors (large $\epsilon$, of as much as tens of percentage points)~\cite{sukhodolov2016evaluation} while still being useful for studying photochemistry.

The molecules chosen as examples (\cref{table_resource}) are all involved in important atmospheric processes, but their full quantum simulation is classical intractable. 

Methane (\ce{CH4}) is a major greenhouse gas that is depleted from the atmosphere by photodissociation after absorption of UV radiation in the upper atmosphere~\cite{winterstein2021methane}. The dominant reaction pathway is
\begin{equation}
\ce{CH4^* -> CH3 +H},
\end{equation}
where the star indicates that the molecule is in an excited state. Photodissociation of methane proceeds through a manifold of closely spaced excited states, involving strong nonadiabatic couplings, which makes its full molecular dynamics intractable on classical computers, despite its small size~\cite{wang2024unraveling}.

Three of the molecules photodissociate to form reactive oxygen species that work as cleaning agents in the atmosphere, removing pollutants and greenhouse gases such as methane and carbon monoxide. The Criegee intermediate \ce{CH2OO}~\cite{esposito2021photodissociation,Osborn2015,long2016atmospheric,khan2018criegee} commonly photodissociates by breaking the bond between the oxygen atoms,
\begin{equation}
\ce{CH2OO^* -> H2CO +O},
\end{equation}
The photodissociation of peroxynitric acid (\ce{HNO4})~\cite{wei2014,shchepanovska2021,ma2017,roehl2002} and \ce{C5}-hydroperoxyaldehyde (\ce{C5}-HPALD)~\cite{prlj2020,marsili2022theoretical,wolfe2012photolysis} are important sources of hydroxyl radicals (\ce{OH^.}):
\begin{align}
\ce{HNO4^* &-> NO3^. + OH^.},\\
\ce{C5H8O4^* &-> C5H7O3^. + OH^.},
\end{align}
where the dot indicates a radical. Both \ce{O} and \ce{OH^.} are reactive species that degrade atmospheric pollutants.

Our test set also includes four halogenated compounds. Through photodissociation, these species release halogen radicals that deplete the ozone layer and lead to the formation of long-lived greenhouse gases.
The primary photodissociation channel of trifluoroacetic acid (\ce{CF3CO2H}) is
\begin{equation}
\ce{CF3CO2H^* -> CF3^. + COOH^.},   
\end{equation}
where trifluoromethyl radicals (\ce{CF3^.}) are precursors to greenhouse gasses such as \ce{CF3H} (HFC-23)~\cite{mearns1963photolysis}.
Photodissociation of methyl hypobromite (\ce{CH3OBr})~\cite{stojanovic2015photochemistry},
\begin{equation}
\label{eq:bromine}
\ce{CH3OBr^* -> CH3O^. + Br^.},
\end{equation}
and bromoacetaldehyde ($\ce{BrCH2CHO}$)~\cite{sadiek2021ab},
\begin{equation}
\ce{BrCH2CHO^* -> CH2CHO^. + Br^.},
\end{equation}
can produce bromine radicals that contribute to ozone depletion as well as DNA damage. Similarly damaging are chlorine radicals released through the photodissociation of refrigerants such as HCFC-132b~\cite{rodrigues2016uv},
\begin{equation}
\ce{CF2ClCH2Cl^* -> CF2ClCH2^. + Cl^.}.
\end{equation}

The final molecule is cyclobutanone (\ce{C4H6O}), which can photodissociate into the fragments
\begin{equation}
\ce{C4H6O^* -> C3H6 + CO}.
\end{equation}
Accurately simulating the photodissociation of cyclobutanone, a small, organic molecule, challenges mature classical approaches to chemical dynamics~\cite{mukherjee2024prediction,miller2024ultrafast,peng2024photodissociation,green2025imaging}.

\renewcommand{\arraystretch}{1.2}
\begin{table*}\label{table_resource}
\centering 
    \begin{tabular}{
    p{1.9cm}
    p{0.6cm}
    p{0.8cm}
    p{1.08cm}
    p{1.08cm}
    p{1.10cm}
    p{1.10cm}
    p{1.10cm}
    p{1.2cm}
    p{1.4cm}
    p{1.2cm}
    p{1.4cm} 
    }
    \toprule
    & & & & & \multicolumn{3}{c}{Logical qubits} & \multicolumn{4}{c}{Toffoli gates ($C_\mathrm{Toff}$)} \\
    \cmidrule(r){6-8} \cmidrule(r){9-12}
    Molecule  & $\eta$ & $L$ $(a_0)$ & $\Delta L$ $(10^{-2} a_0)$ &$\lambda_{H}$ $(10^{5} E_\mathrm{h})$ & State storage ($C_\mathrm{data}$)& Ancilla qubits ($C_\mathrm{anc}$) & Total & Nuclear state prep. $(10^{9})$ & Electronic state prep. $(10^{9})$ & Time evolution $(10^{12})$ & Quantum yield using QAE $(10^{12})$ \\
    \midrule
        \ce{CH4} & \num{15} & \num{392} & \num{4.78} & \num{0.85} & \num{585} & \num{3189}  & \num{3774}  & \num{2.79} & \num{0.10}& \num{1.89} &  \num{32.1}\\
    \ce{CH2OO} & \num{29} & \num{298} & \num{3.64} & \num{4.09} & \num{1131} & \num{3385}& \num{4516} & 3.32 & \num{0.42}& \num{18.9} &  \num{322}\\
    \ce{C4H6O} & 49 & \num{595} & 3.63 & \num{7.98} & \num{2058} & \num{3811} &\num{5869}  & \num{7.39} & \num{0.98} & \num{60.1} & \num{1002}\\
    \ce{HNO4} & \num{46} & \num{276}& \num{3.37} & \num{9.60} & \num{1794} & \num{3542}  & \num{5336} & 3.84 & \num{0.56} & \num{64.0} &  \num{1009} \\
    \ce{CF3CO2H} & 64 & 473 & \num{2.89} & \num{20.1} &  \num{2688} & \num{3904} & \num{6592} & \num{4.88} & \num{1.69} & \num{191} & \num{3240}\\
    C\textsubscript{5}-HPALD & 78 & \num{1066} & \num{3.25} & \num{19.9} & \num{3510} &\num{4289} & \num{7799} &  \num{9.87} & \num{2.19} & \num{239} & \num{4070}\\
    HCFC-132b & 74 & \num{520} & 1.59 & \num{64.4} & \num{3330} &\num{4243} & \num{7573} & \num{4.72} & \num{3.43} & \num{742} & \num{12600}\\

    \ce{CH3OBr} & 58 & \num{523} & \num{0.80} & \num{150} & \num{2496} & \num{4512} & \num{7296} & \num{4.12} & \num{1.25} & \num{1350} &  \num{22900} \\

    \ce{BrCH2CHO} & 65 & \num{509} & \num{0.78} & \num{181} & \num{3120} & \num{4578} & \num{7698} & \num{4.10} & \num{1.70} & \num{1810} & \num{30700} \\
    \bottomrule
    \end{tabular}
\caption{Resource estimates for simulating photodissociation yields in molecules whose full quantum dynamics is classically intractable. $\eta$ is the total number of particles (electrons and nuclei), $L$ is the length of one side of the computational cell, $\Delta L = L/N$ is the spacing of the grid points in real space, and $\lambda_H$ is the block-encoding parameter. We give both space and time resources for our algorithm, with the qubit counts divided into qubits to store the state ($C_\mathrm{data}$) and the ancillas ($C_\mathrm{anc}$). The Toffoli gate counts for time evolution indicate one dynamics run for simulation time $t=\SI{30}{fs}$, and the QAE counts give the cost of calculating the quantum yield to accuracy $\epsilon=0.095$ (involving multiple simulations of the dynamics).}
\end{table*}

\subsection{Resource estimation}

To model photoexcitation, we use the Condon approximation, which assumes that the electronic transition from ground to excited state is fast compared to nuclear dynamics~\cite{condon1926theory}. The Condon approximation is typically accurate for small molecules excited by ultraviolet/visible (UV-Vis) light~\cite{tannor2007,smith1968}. Then, the state of the nuclei can be treated as unchanged during the electronic transition and the initial photoexcited state at $t=0$ is the product state
\begin{equation}
    |\Psi(0)\rangle = |\Psi^{(n)}\rangle |\Psi^{(e){*}}\rangle,
\end{equation}
where $|\Psi^{(e){*}}\rangle$ is an electronic excited state. We model $|\Psi^{(e){*}}\rangle$  as the state obtained by a single electron being promoted from the highest occupied molecular orbital (HOMO) to the lowest unoccupied molecular orbital (LUMO). For the nuclear state $|\Psi^{(n)}\rangle$, we assume the molecule is in the ground vibrational state of its electronic ground state (i.e., a multivariate Gaussian) and model it using a VHP. This assumption is valid because most vibrational modes have minimal excited-state populations at atmospherically relevant temperatures.

We estimate the resources needed for calculating the yield $\bra{\Psi(t)}\Pi_{S}\ket{\Psi(t)}$ up to an additive error $\epsilon = 0.095$. This error budget is divided between $\epsilon_{\mathrm{ISP}}$, $\epsilon_{\mathrm{prop}}$ and $\epsilon_\mathrm{QAE}$ as in \cref{eq:error_split}. The errors are allocated as described in \cref{sect:errors}, with the majority allocated to $\epsilon_\mathrm{QAE}=0.0625$ and the majority of the remainder to $\epsilon_{\mathrm{ISP}} = 0.015$, leaving $\epsilon_{\mathrm{prop}}=0.00125$. We allocate more error to $\epsilon_{\mathrm{ISP}}$ than to $\epsilon_{\mathrm{prop}}$ because the cost is more sensitive to $\epsilon_{\mathrm{ISP}}$; in particular, a larger $\epsilon_{\mathrm{ISP}}$ permits a smaller momentum cutoff $K_\mathrm{max}$, which reduces the number of qubits required and the costs of time evolution. Finally, $\epsilon_{\mathrm{ISP}}$ is distributed among the nuclear and electronic errors so that the bounds in \cref{eq:error_isp_sep} are satisfied. 

To gather the input data that specifies $|\Psi^{(n)}\rangle$, we perform a geometry optimization and a Hessian calculation using second-order M\o ller-Plesset perturbation theory~(MP2) with the cc-PVDZ basis set, yielding the equilibrium geometry $\mathbf{R}^{(0)}$ and ground-state Hessian $\mathbf{He}_1$. We use the Psi4 software package~\cite{psi4}. The Hessian is diagonalized to obtain the frequencies $\omega_i$ and the matrix $\mathbf{O}$. This information is used to compute $\mathbf{d}$ and $\mathbf{A}$ in \cref{eq:Qbar_def}, and the frequencies are rescaled by $\omega_i \rightarrow d_{ii}^2\omega_i$. To ensure that the rotational states (see \cref{eq:trans_rot_wvfxn}) only have support over a small region of the simulation grid, we set $\Upsilon_\tau = \max_i \omega_i$ for all $\tau$. For the translational states, we set $\Gamma_\tau = \min_i \omega_i$ for all $\tau$. Because $|\Psi^{(n)}\rangle$ is a multivariate Gaussian, we prepare it using the SSCT approach in \cref{sect:LCT}.   

With the nuclear states known, $L$ is selected as the smallest value that satisfies \cref{eq:eps_shear_bound}. With $L$ decided, $K_{i\mu}^{(n)}$ and $M_{i\mu}$ are solved for with \cref{eq:K_nuc,eq:m_nuc_bound}.

For $|\Psi^{(e){*}}\rangle$, the parameters that specify the molecular orbitals $\chi_a$ (see \cref{sect:input_data_elec}) are found using a restricted Hartree-Fock (RHF) calculation with the 3-21G basis set with the nuclei fixed at $\mathbf{R}^{(0)}$ using the PySCF software package~\cite{pyscf}. The $K^{(e)}_a$ and $M_a$ are determined numerically using the approach from~\cite{huggins2024}, as discussed in \cref{sect:preproc_mo}; in particular, $K^{(e)}_a$ is increased until the error from the resulting approximation to $|\Psi^{(e)*}\rangle$ satisfies the target error.

With $K^{(n)}_{i\mu}$ and $K^{(e)}_a$ determined, $K_{\mathrm{max}}$, $N$, and $n_p$ are computed. We obtain $n_{\mathrm{pad}}$ using \cref{eq:npad_shear}. Finally, we estimate the error $\epsilon_{\mathrm{trim}}$ from trimming the nuclear wavefunction using \cref{eq:MC-error}. Based on $N_{\mathrm{MC}}=10^9$ Monte Carlo samples, $\epsilon_{\mathrm{trim}} \leq 10^{-4}$ with confidence $0.99999$. We therefore take $\epsilon_{\mathrm{trim}}=10^{-4}$.

For consistent comparison, we estimate the cost of calculating each photodissociation yield at $t=\SI{30}{fs}$ after photoexcitation. This $t$ is chosen to exceed the inverse of typical atmospheric photodissociation rates~\cite{prlj2020}. In practice, the required simulation time would not be known a~priori and the simulation would be repeated with longer durations to construct a time series.

To calculate yields, photodissociation reaction channels are characterized as described in \cref{sect:yield_protocol} and illustrated in \cref{fig:dissociation}. In the following, we use \ce{CH3OBr} as the illustrative molecule, but identical methods were applied to the others.
As shown in \cref{fig:dissociation}, three reaction channels are important for \ce{CH3OBr}, each corresponding to a different photodissociation outcome. Formation of \ce{Br^.} according to \cref{eq:bromine} corresponds to the reaction channel characterized by the inequalities
\begin{align}\label{eq:bond_length1}
\norm{\mathbf{R}_{\mathrm{Br}}-\mathbf{R}_{\mathrm{O}}}_2\geq b_{\mathrm{Br}\mathrm{O}}\\
\norm{\mathbf{R}_{\mathrm{O}}-\mathbf{R}_{\mathrm{C}}}_2\leq b_{\mathrm{O}\mathrm{C}},\label{eq:bond_length2}
\end{align}
where $b_{\mathrm{Br}\mathrm{O}}$ and $b_{\mathrm{O}\mathrm{C}}$ are the bond-dissociation cutoff distances, which we set to be \SI{0.4}{\bohr} larger than the corresponding equilibrium bond lengths.

We estimate the cost of measuring the yield $\bra{\Psi(t)}\Pi_{S}\ket{\Psi(t)}$ using the protocol in \cref{sect:yield_protocol}.
For \ce{CH3OBr}, the indicator function $I(\R,S)$ for the desired photodissociation channel is
\begin{multline}
I(\R,S)= \mathbb{B}\left(\norm{\R_{\mathrm{Br}}-\R_{\mathrm{O}}}_2 \geq b_{\mathrm{Br}\mathrm{O}}\right) \\ \wedge \mathbb{B}\left(
\norm{\R_{\mathrm{O}}-\R_{\mathrm{C}}}_2 \leq b_{\mathrm{O}\mathrm{C}}\right).
\end{multline}
We implement the corresponding block-encoding $U_{\Pi_{S}}$ and perform phase estimation on the iterate $\widetilde{W}_{\Pi_{S}}$.

\subsection{Results}

\Cref{table_resource} provides resource estimates for computing the photodissociation yield of the molecules above. Of the molecules considered, the total number of particles $\eta$ (electrons and nuclei) ranges from 29 to 78. The molecules require 3774--7698 logical qubits and \num{3.2e13}--\num{3.07e16} Toffoli gates to calculate the yields. 

As expected, molecules with more particles require more logical qubits, both for state storage ($C_\mathrm{data}$) and for ancillas ($C_\mathrm{anc}$). However, $C_\mathrm{anc}$ is consistently high across all systems, increasing slowly increase with $\eta$. The high $C_\mathrm{anc}$ arise primarily from time evolution, because qubitization requires many ancilla qubits.

To satisfy ISP error bounds, $\Delta L = O(\SI{0.01}{\bohr})$ is typically needed. Features of nuclear wavefunctions are typically much narrower than electronic length scales and requires a smaller $\Delta L$ to resolve. For example, in the ground state of a \ce{C-H} vibration with frequency \SI{3000}{cm^{-1}}, the root-mean-square displacement is \SI{0.15}{\bohr}. In general, grid-based numerical simulations require spacings an order of magnitude smaller than the narrowest features that need to be resolved, which is consistent with our result of $\Delta L = O(\SI{0.01}{\bohr})$. 
\Cref{table_resource} also shows that $\Delta L$ decreases with the mass of the heaviest nucleus, from at least \SI{2.9e-2}{\bohr} for molecules containing only first- and second-row elements to \SI{0.8e-2}{\bohr} for molecules containing bromine. For a given spring constant, heavier nuclei move less than lighter ones, meaning that their wavefunctions are more sharply peaked, requiring a smaller $\Delta L$ to resolve.
By contrast, prior estimates of necessary grid spacings for nuclear dynamics were unrealistically large, including $\Delta L=\SI{0.962}{\bohr}$~\cite{pocrinic2026}.

\section{Discussion} \label{sect:discussion}

We have developed an end-to-end algorithm for simulating full dynamics under the molecular Hamiltonian. The algorithm incorporates three ingredients that have not previously been combined in a single algorithm: bounded-error preparation of general electronic and nuclear initial states in first quantization, including non-separable ones; time evolution under the molecular Hamiltonian; and optimal measurement of dynamical observables, such as photodissociation yields. Because all sources of error are controlled, the final error in the observable can be made arbitrarily small at a cost that scales linearly in the inverse error. These are the first rigorous error bounds for a chemical-dynamics algorithm, whether on classical or quantum computers.

Our algorithm offers an exponential improvement over the best classical methods for chemical dynamics. Our algorithm treats the full dynamics of both electrons and nuclei without the BO approximation or pseudopotentials, block‑encodes the molecular Hamiltonian, and discretizes all degrees of freedom on a common plane‑wave grid. This yields a largely black‑box procedure for real‑time chemical dynamics, in contrast to most classical approaches. In the worst case, state‑of‑the‑art classical wavepacket techniques scale exponentially with system size~\cite{arteaga2024strong,worth2008using,wang2015multilayer,meng2013multilayer,xie2015full,schulze2016multi,cigrang2025roadmap,prlj2025bestpracticesnonadiabaticmolecular}, whereas our quantum algorithm scales polynomially in the number of particles, the time, and the inverse target accuracy~\cref{eq:complexity_full_alg}.

Small system sizes, short timescales, and low accuracy requirements make photochemistry a compelling early application for quantum computing. We show that we can simulate the time evolution of classically intractable reactions using $O(10^{11})$ Toffoli gates for the entire reaction duration. By contrast, da Jornada et al.~\cite{Jornada2025} report costs of $O(10^{11})$ Toffoli gates per femtosecond of time evolution, which amounts to $O(10^{15})$ Toffoli gates for the multi-picosecond catalytic timescales relevant for their system of choice. We also show that we can simulate dynamical observables of interest, namely yields, using $O(10^{13})$--$O(10^{16})$ Toffoli gates and $O(10^4)$ logical qubits. This is comparable to the proposal to run fusion stopping-power calculations using $O(10^{13})$--$O(10^{17})$ Toffoli gates and $O(10^{4})$--$O(10^{5})$ logical qubits~\cite{rubin2024}. 

We anticipate order-of-magnitude reductions in cost from future algorithmic improvements, as has been the experience with quantum algorithms for electronic structure due to improvements in block-encoding techniques, synthesis, and circuit compilation~\cite{lin2020near,low2025}.
For example, we expect that costs could be reduced by allowing different particles to occupy different grids, with finer grids for heavier nuclei and coarser grids for lighter particles, especially electrons. Doing so would significantly reduce $\lambda_H$ and the overall cost. Similarly, our approach of modeling rotational states with narrow Gaussians requires a fine grid, which could be mitigated through a more sophisticated approach.

In addition, we expect that, in practice, the cost of our algorithm will be significantly smaller than the provable error bounds we have given, which rely on worst-case assumptions and overestimate the required resources. Experience from classical dynamics and quantum simulation shows that actual errors at a given discretization are typically far smaller than the worst-case guarantees~\cite{miller2025phase2,hahn2024}. Therefore, in practice, users will tune grid spacings and other parameters until observables are numerically converged, rather than until a worst-case analytic error budget is met.

Our algorithm could extend straightforwardly to bimolecular and more complex processes. Two or more reactants could be initialized in spatially separated regions of the grid, each prepared in its own internal quantum state and with a specified relative momentum. The same approach to time evolution and measurement then applies, including calculating yields.

Our photochemical examples identify a regime of chemical simulation that satisfy three criteria for early quantum advantage: classical hardness, quantum easiness (small molecules, short timescales, and low required error), and broad societal importance. As quantum processors mature and anticipated algorithmic improvements drive costs down, we expect that accurate simulations of chemical dynamics will be among the earliest applications of fault-tolerant quantum computers.

\section*{Acknowledgments}
We thank Ryan Babbush, Maria Kieferov\'a, and Nicholas Rubin for valuable discussions. We were supported by Google Quantum AI, by the Australian Research Council (FT230100653, DP220101602), and by the Sydney Quantum Academy.
DWB worked on this project under a sponsored research agreement with Google Quantum AI.

\bibliography{refs}

\appendix

\onecolumngrid

\section{Galerkin discretization}
\label{app:Galerkin_review}
Here, we review the plane-wave-basis Galerkin discretization of the molecular Hamiltonian $H=T+V$ of \cref{eq:mol-ham}. In the plane-wave basis of \cref{sect:isp},
\begin{align}
T &=\sum_{j=1}^{\eta} \sum_{\mathbf{p}, \mathbf{q}=1}^N T_{\mathbf{p} \mathbf{q}}^{\left(m_{j}\right)}\ket{\mathbf{p}}\bra{\mathbf{q}}_{j} \\ 
V&=\frac{1}{2} \sum_{i \neq j=1}^{\eta} \sum_{\mathbf{p}, \mathbf{q}, \mathbf{r}, \mathbf{s}=1}^N V_{\mathbf{p} \mathbf{q} \mathbf{r} \mathbf{s}}^{\left(\zeta_{i}, \zeta_{j}\right)}\ket{\mathbf{p}}\bra{\mathbf{s}}_{i}\otimes \ket{\mathbf{q}}\bra{\mathbf{r}}_{j} .
\end{align}
The one- and two-electron integrals have convenient closed forms~\cite{martin2004}:
\begingroup
\allowdisplaybreaks
\begin{align}
T_{\mathbf{p}, \mathbf{q}}^{(m)} &= \int d\mathbf{r}\, \phi_{\mathbf{p}}^*(\mathbf{r}) \left(-\frac{\nabla^2}{2 m}\right) \phi_{\mathbf{q}}(\mathbf{r}) \\ 
&= \frac{\left\|\mathbf{k}_{\mathbf{q}}\right\|^2}{2m} \int d\mathbf{r}\, e^{i\left(\mathbf{k}_{\mathbf{q}} - \mathbf{k}_{\mathbf{p}}\right) \cdot \mathbf{r}} 
= \delta_{\mathbf{p}, \mathbf{q}} \frac{\left\|\mathbf{k}_{\mathbf{p}}\right\|^2}{2m},\\
V_{\mathbf{p} \mathbf{q} \mathbf{r} \mathbf{s}}^{(\zeta_{i}\zeta_{j})} &= \int d\mathbf{r}_1\, d\mathbf{r}_2 
\phi_{\mathbf{p}}^*(\mathbf{r}_1)\, \phi_{\mathbf{q}}^*(\mathbf{r}_2)
\left( \frac{\zeta_{i} \zeta_{j}}{\left\|\mathbf{r}_1 - \mathbf{r}_2\right\|} \right) 
\phi_{\mathbf{r}}(\mathbf{r}_2)\, \phi_{\mathbf{s}}(\mathbf{r}_1) \\ 
&= \frac{4 \pi \zeta_{i} \zeta_{j}}{\Omega \left\|\mathbf{k}_{\mathbf{p}} - \mathbf{k}_{\mathbf{s}}\right\|^2}
\int d\mathbf{r}_2 e^{i\left(\mathbf{k}_{\mathbf{r}} - \mathbf{k}_{\mathbf{q}} - (\mathbf{k}_{\mathbf{p}} - \mathbf{k}_{\mathbf{s}})\right) \cdot \mathbf{r}_2} \\ 
&= \frac{4 \pi \zeta_{i} \zeta_{j}}{\Omega} 
\frac{\delta_{\mathbf{k}_{\mathbf{p}} - \mathbf{k}_{\mathbf{s}}, \mathbf{k}_{\mathbf{r}} - \mathbf{k}_{\mathbf{q}}}}{\left\|\mathbf{k}_{\boldsymbol{\nu}}\right\|^2}
\end{align}
\endgroup
where $\mathbf{k}_{\boldsymbol{\nu}} = \mathbf{k}_{\mathbf{p}} - \mathbf{k}_{\mathbf{s}} = \mathbf{k}_{\mathbf{r}} - \mathbf{k}_{\mathbf{q}} \neq 0$.
Combining these, we obtain \cref{eq:Galerkin_Ham}.

\section{Efficient normal-coordinate preprocessing} \label{app:efficient_nc_preproc}
In this appendix, we prove \cref{lem:pwb,lem:mps,lem:twos_complement}, which collectively establish that a given SM can be efficiently projected onto a PWB and encoded as an MPS using a two's-complement representation. In \cref{app:pwb_proj_sm}, we prove \cref{lem:pwb}, which demonstrates that a given SM can be efficiently projected onto a PWB. In \cref{app:mps_encod_sm}, we prove \cref{lem:mps}, which shows that the approximation from \cref{lem:pwb} has an efficient approximation as a polynomial that can be encoded exactly as an MPS. In \cref{app:twos_complement_mps}, we prove \cref{lem:twos_complement}, which shows that the resulting MPS can be efficiently expressed using a two's-complement representation.

Throughout, we use $\sm$ to denote an SM for a vibrational, translational, or rotational degree of freedom. In the case of translational and rotational states, the indices $i\mu$ are substituted with $\tau$. 

\subsection{PWB projection of SMs}
\label{app:pwb_proj_sm}
Here, we prove \cref{lem:pwb}, which demonstrates that $\sm$ can be efficiently projected onto a PWB. The proofs here and in \cref{app:mps_encod_sm} are similar to the one in Appendix~B of \cite{huggins2024}, used to prove \cref{eq:Kelec,eq:M_elec}. The results of \cite{huggins2024} apply to the present case with only minor modifications because the general strategy in \cite{huggins2024} is to first express the Gaussian primitives that constitute an MO in terms of Hermite-Gaussians and then use the properties of these functions to prove bounds, which aligns with our expansion of SMs in terms of Hermite-Gaussian primitives. Therefore, instead of proving \cref{lem:pwb,lem:mps} from scratch, we highlight which parts in Appendix B of \cite{huggins2024} need to be modified to apply to the present case.

\begin{lemma} \label{lem:pwb}
Let $\hatsm$ be an SM defined in \cref{eq:project_pwb},
$\SmDelT\in(0,1)$, and $\SmMom\in\mathbb{R}$ such that
\begin{equation} 
    \SmMom \geq \sqrt{2\omega_i} \sqrt{2\ln((\SmDelT )^{-1}) + \ln\left(\tfrac{1}{\sqrt{2}} + \tfrac{2\sqrt{\pi}}{L \sqrt \omega_i}\right) + (N_{hg}-1)\ln(4(N_{hg} -1)) + 5.4} \, .
\end{equation} 
Then there exists an approximation $\HatSmAprx$ to $\hatsm$ given by replacing $\mathbb{K}$ with $\mathbb{K}_{\mathrm{cut}}$ as defined in \cref{eq:mom-cut-set}, where $\SmMom$ is the momentum cutoff, such that the error $\SmEpsT=D(\hatsm, \HatSmAprx)$ is bounded by $\SmEpsT \leq \SmDelT$.
\end{lemma}
\begin{proof}
We consider an SM 
\begin{equation}
    \sm(Q) = \sum_{\nu=0}^{N_{hg}-1} \SMTensor \frac{1}{\sqrt{2^{\nu} \nu!}}\left(\frac{\omega_i}{\pi }\right)^{1 / 4} e^{-\omega_i Q^2/2} H_{\nu}\left(\sqrt{\omega_i} Q \right),
\end{equation}
whose projection onto a PWB is given in \cref{eq:project_pwb}. We expand $\sm(Q)$ in terms of Hermite-Gaussian functions
\begin{equation}
\psi_\nu(x)=c_\nu e^{-x^2 / 2} H_\nu(x),
\end{equation}
where $H_n(x)$ is the $n$th Hermite polynomial and the normalization constant is 
$c_n = \left(2^n n! \sqrt{\pi} \right)^{-1/2}$,
as
\begin{equation}
    \sm(Q) = \omega_i^{1/4}\sum_{\nu=0}^{N_{hg}-1} \SMTensor \psi_{\nu}(\sqrt{\omega_i}Q).
\end{equation}
To evaluate \cref{eq:project_pwb}, we must compute integrals of the form
\begin{equation}
    \omega_i^{1/4} \sum_{\nu=0}^{N_{hg}-1} \SMTensor \int_{-\infty}^{\infty}du \, \varphi_k^{*}(u)  \psi_{\nu}(\sqrt{\omega_i} u),
\end{equation}
where we use the approximation that $L$ may be taken to infinity. This integral can be evaluated using the change of variables $x = \sqrt{\omega_i}u$ and the fact that Hermite-Gaussian functions are eigenfunctions of the Fourier transform,
\begin{equation}
    \omega_i^{1/4} \sum_{\nu=0}^{N_{hg}-1} \SMTensor \int_{-\infty}^{\infty} \frac{dx}{\sqrt{\omega_i}}\, \varphi_k^{*} \left(\frac{x}{\sqrt{\omega_i}} \right)  \psi_{\nu}(x) = \left(\frac{2 \pi}{L} \right)^{1/2}\frac{1}{\omega_i^{1/4}} \sum_{\nu=0}^{N_{hg}-1} (-i)^{\nu} \SMTensor \, \psi_{\nu}\left(\frac{k}{\sqrt{\omega_i}} \right).
\end{equation} 
We use this result to write \cref{eq:project_pwb} as
\begin{equation} \label{eq:point_of_contact}
    \hatsm(Q) = \frac{\sqrt{2\pi}}{\SmnormAprx \omega_i^{1/4} \sqrt{L}} \sum_{k \in \mathbb{K}} \sum_{\nu=0}^{N_{hg}-1} (-i)^{\nu}  \SMTensor \, \psi_{\nu}\left(\frac{k}{\sqrt{\omega_i}} \right) \phi_k(Q).
\end{equation}
With this expression, we make contact with \cite{huggins2024} by noticing that \cref{eq:point_of_contact} is the same as Eq.~(B6) of \cite{huggins2024} after making the substitutions $\gamma \rightarrow \omega_i/2$ and $\ell \rightarrow N_{hg}-1$. We define the approximation to $\hatsm$ as
\begin{equation} 
    \HatSmAprx(Q) = \frac{\sqrt{2\pi}}{\SmnormAprx \NormT \omega_i^{1/4} \sqrt{L}} \sum_{k \in \mathbb{K}_{\mathrm{cut}}} \sum_{\nu=0}^{N_{hg}-1} (-i)^{\nu}  \SMTensor \, \psi_{\nu}\left(\frac{k}{\sqrt{\omega_i}} \right)\phi_k(Q), \label{eq:SM_approx_pwb}
\end{equation}
where $\NormT$ is a normalization constant that ensures $\Vert \HatSmAprx \Vert_2 = 1$.

The manipulations to be performed to arrive at the desired result are then almost identical to those used to prove Lemma~4 in Appendix~B of~\cite{huggins2024}. Analogous to Eq.~(B45) in~\cite{huggins2024}, we define
\begin{equation}
R=\sum_{k \in \mathbb{K}:|k|>\SmMom}\left|\psi_n\left(\frac{k}{\sqrt{\omega_i}}\right) \psi_m\left(\frac{k}{\sqrt{\omega_i}}\right)\right|.
\end{equation}
We then apply the manipulations in Appendix~B of \cite{huggins2024} to the present case until Eq.~(B58) of \cite{huggins2024}. The analogous expression to Eq.~(B58) of \cite{huggins2024} is 
\begin{equation}
R \leq  \frac{L \big(\frac{2 \pi}{L}+\sqrt{\pi \omega_i/2}\big)}{\pi^{3 / 2}} 
e^{-\left(\SmMom\right)^2/2 \omega_i}\left(\frac{\SmMom}{\sqrt{\omega_i}}\right)^{2 l} \frac{2^{n/2}}{\sqrt{n!}} \frac{2^{m/2}}{\sqrt{m!}}.
\end{equation}
The main difference is that the powers of two do not cancel as they do in Eq.~(B58) of \cite{huggins2024}. Hence, the summation of Eq.~(B61) of~\cite{huggins2024} becomes in the present case:
\begin{equation}
\begin{aligned} \label{eq:bound_sum}
\sum_{n =0}^N \sum_{m =0}^N \frac{2^{n/2}}{\sqrt{n!}} \frac{2^{m/2}}{\sqrt{m!}} =\left(\sum_{n=0}^N \frac{2^{n/2}}{\sqrt{n!}}\right)^2 
 &\leq\left(\sum_{n=0}^{\infty} \frac{2^{n/2}}{\sqrt{n!}}\right)^2  
\leq \frac{97}{2} .
\end{aligned}
\end{equation}
With this result in hand, the remainder of the proof is the same as in \cite{huggins2024} through to the end of of Lemma~4. Following these steps, we find that for a given $\SmDelT$ and $\SmMom$ that satisfies
\begin{equation} \label{eq:k_nuc_full}
    \SmMom \geq \sqrt{2\omega_i} \sqrt{2\ln((\SmDelT)^{-1}) + \ln(97) + \ln\left(\tfrac{1}{\sqrt{2}} + \tfrac{2\sqrt{\pi}}{L \sqrt \omega_i}\right) - \ln( (\SmnormAprx)^2) + (N_{hg}-1)\ln(4(N_{hg} -1))} ,
\end{equation}
the approximation error $\SmEpsT = D(\hatsm,\HatSmAprx)$ is bounded by 
\begin{equation}
    \SmEpsT \leq \SmDelT,
\end{equation}
meaning that SMs can be efficiently projected onto a PWB. We assume that $\SmnormAprx \geq 2/3$ so that we can group $\ln((\SmnormAprx )^2)$ with $\ln(97)$ to give a cleaner bound,
\begin{equation} 
    \SmMom \geq \sqrt{2\omega_i} \sqrt{2\ln((\SmDelT )^{-1}) + \ln\left(\tfrac{1}{\sqrt{2}} + \tfrac{2\sqrt{\pi}}{L \sqrt \omega_i}\right) + (N_{hg}-1)\ln(4(N_{hg} -1)) + 5.4} 
\end{equation} 
because $\ln(97) - \ln((2/3)^2) < 5.4$. The assumption that $\SmnormAprx \geq 2/3$ is satisfied when our assumptions about the integration bounds and $\sm$ decaying to zero sufficiently fast hold, because then $\SmnormAprx \approx 1$.
\end{proof}

\subsection{MPS encoding of SMs}
\label{app:mps_encod_sm}
Here, we prove \cref{lem:mps}, which demonstrate that the PWB approximation to $\sm$ resulting from \cref{lem:pwb} has an efficient approximation as a polynomial that can be encoded exactly as an MPS.

\begin{lemma} \label{lem:mps}
Let $\HatSmAprx$ be given as in \cref{eq:SM_approx_pwb}, $\SmDelP\in(0,1)$, and $M_{i\mu}\in\mathbb{Z}$ satisfy all of the following conditions:
\begingroup
\allowdisplaybreaks
\begin{align} 
\SmPolyDeg & \geq 2, \label{eq:Nuc_M_cond} \\
\SmPolyDeg & \geq \frac{e \SmMom}{\sqrt{2\omega_i}}\left(\frac{e \SmMom}{\sqrt{2\omega_i}}+\sqrt{2 N_{hg}-1}\right), \label{eq:Nuc_M_cond-2}\\
\SmPolyDeg & \geq \frac{1}{\ln 2}\Bigg(4 \ln \left(  (\SmDelP)^{-1} \right) -2 \ln \left(\SmnormAprx\right)-2 \ln \left(\NormT\right)+\ln \Bigg(\frac{\SmMom}{\sqrt{\omega_i}}+\frac{\pi}{L \sqrt{\omega_i}}\Bigg)+\ln (N_{hg})\Bigg)+5. \label{eq:Nuc_M_cond-3}
\end{align}
\endgroup
Then there exists a degree-$(M_{i\mu}-1)$ polynomial approximation
$\SmPolyAprx$ to $\HatSmAprx$,
\begin{equation}
\SmPolyAprx =\sum_{k \in \mathbb{K}_{\mathrm{cut}}} p(k) \varphi_k(x),
\end{equation}
where $\big\Vert \SmPolyAprx \big\Vert^2_2 =1$, $\varphi_k$ is the $k$th PWB basis function, and $p(k)$ is a degree-$(M_{i\mu}-1)$ polynomial. The error of the approximation, $
    \SmEpsP = D(\HatSmAprx,\SmPolyAprx) $,
is bounded by 
\begin{equation} \label{eq:poly_bound}
 \SmEpsP \leq \SmDelP.
\end{equation}
\end{lemma}

\begin{proof}
Our proof of \cref{lem:mps} is analogous to Lemma 5 in Appendix B of \cite{huggins2024}. We start by using Lemma~9 from Appendix~B of \cite{huggins2024}, which states:
``Let $\psi_n(x)$ denote the $n$th Hermite-Gaussian function. Let $C, \lambda_{\rm cheb}\in \mathbb{R}_{>0}$. For any  $\SmPolyDeg \in \mathbb{Z}$  such that $\SmPolyDeg$ satisfies the three conditions
\begin{align}
\SmPolyDeg &\geq \frac{e C}{\sqrt{2}}\left(\frac{e C}{\sqrt{2}}+\sqrt{2 n+1}\right)\\
\SmPolyDeg &\geq \frac{2 \ln \lambda_{\rm cheb}^{-1}}{\ln 2}\\
\SmPolyDeg &\geq 1,
\end{align}
there exists a degree $(\SmPolyDeg-1)$ polynomial $p(x)$ such that $\left|\psi_n(x)-p(x)\right| \leq \lambda_{\rm cheb}$
for all $x \in [-C,C]$.''
Using this lemma, 
\begin{equation}
\left|\psi_\nu\left(\frac{k}{\sqrt{\omega_i}}\right)-p_\nu\left(\frac{k}{\sqrt{ \omega_i}}\right)\right| \leq \lambda_{\rm cheb},
\end{equation}
for any $k \in \big[-\SmMom, \SmMom\big]$ so long as all of these conditions are satisfied: 
\begingroup
\allowdisplaybreaks
\begin{align} \label{eq:point_of_contact_2}
& \SmPolyDeg \geq \frac{e \SmMom}{ \sqrt{2 \omega_i}}\left(\frac{e \SmMom}{ \sqrt{2 \omega_i}}+\sqrt{2 N_{hg}-1}\right) \\
& \SmPolyDeg \geq \frac{2 \ln \lambda_{\rm cheb}^{-1}}{\ln 2} \\
& \SmPolyDeg \geq 2.
\end{align}
\endgroup
These inequalities are the same as Eq.~(B95) of~\cite{huggins2024} after making the substitutions $\gamma \rightarrow \omega_i/2$ and $\ell \rightarrow N_{hg}-1$. Then, we can use the manipulations in Lemma~5 of Appendix~B of \cite{huggins2024} without modification. Following these steps through Eq.~(B114) of~\cite{huggins2024}, we obtain
\begin{equation}
\lambda_{\rm cheb} \leq \frac14 \big(\SmDelP \big)^2 \SmnormAprx \NormT
\left(2\left(\frac{\SmMom}{\sqrt{\omega_i}}+\frac{\pi}{L \sqrt{\omega_i}}\right)N_{hg}\right)^{-1/2}.
\end{equation}
Substituting this expression into the inequalities, we arrive at \cref{eq:Nuc_M_cond,eq:Nuc_M_cond-2,eq:Nuc_M_cond-3}.
\end{proof}

To arrive at \cref{eq:m_nuc_bound} of the main text, we let $\SmDelT = \SmDelP =\SmDelC/2$. By performing similar manipulations to Eqs.~(B20)--(B26) of \cite{huggins2024}, we combine all three conditions on $\SmPolyDeg$ into the more stringent condition,
\begin{equation} \label{eq:stringent_m}
    \SmPolyDeg \geq e^2 \big(\SmMom\big)^2 /\omega_i.
\end{equation}
Thus, if \cref{eq:stringent_m} is satisfied, then all three inequalities in \cref{eq:Nuc_M_cond,eq:Nuc_M_cond-2,eq:Nuc_M_cond-3} are satisfied.

The combined error from approximating $\hatsm$ as a PWB projection and approximating the projection as a polynomial is
\begin{align}
    \SmEpsC &= D(\hatsm,\SmPolyAprx)\leq \SmEpsT + \SmEpsP \leq \SmDelC,
\end{align}
where we used the triangle inequality. The resulting quantum state in the PWB,
\begin{equation} \label{eq:mps_encoding}
\left|\SmPolyAprx\right\rangle= \sum_{k \in \mathbb{K}_{\mathrm{cut}}}p(k)\left|\frac{kL}{2\pi}\right\rangle,
\end{equation}
can be efficiently encoded exactly as an MPS with bond dimension at most $2\SmPolyDeg + 3$ (Corollary~7 of \cite{huggins2024}). 
Because the MPS encoding is exact, the overall error of approximating $\hatsm$ as an MPS is also bounded by $\SmDelC$.

\subsection{Two's complement MPS representation}
\label{app:twos_complement_mps}
Here, we prove \cref{lem:twos_complement}, which establishes that there is an efficient MPS representation of polynomials using a two's-complement representation of integers. This result complements the efficient MPS representation of polynomials using signed-magnitude representation~\cite{huggins2024}.

\begin{lemma}
	\label{lem:twos_complement}
	Let 
    \begin{equation}
		\ket{\psi} = \sum_{k \in G} p(k) \ket{k}_{\mathrm{TC}},
	\end{equation}
    where $p(k)$ is a polynomial of degree $d$ and $G = \left[-2^{n-1}, 2^{n-1}-1\right] \cap \mathbb{Z}$. The integers are represented in the computational basis using a two's-complement representation. Then, $\ket{\psi}$ admits an exact MPS representation with bond dimension at most $M=2d+4$.
\end{lemma}

\begin{proof}
	In an $n$-qubit two's-complement representation, the most significant bit, $s_{n-1}$, acts as the sign bit.
	The integer $k$ represented by the bitstring $s_{n-1} \dots s_0$ is then
	\begin{equation}
		k = -s_{n-1} 2^{n-1} + y,
	\end{equation}
    where $y=\sum_{j=0}^{n-2} s_j 2^j$ is the integer value of the last $n-1$ bits.
	We can partition $\ket{\psi}$ based on the value of $s_{n-1}$,
	\begin{equation}
		\ket{\psi} = \ket{0} \otimes \ket{\psi_{\ge 0}} + \ket{1} \otimes \ket{\psi_{< 0}}.
	\end{equation}
	When $s_{n-1}=0$ (non-negative integers), $k = y$ and $\ket{\psi_{\ge 0}}$ encodes the degree-$d$ polynomial $f(y) = p(y)$ in an unsigned integer representation. When $s_{n-1}=1$ (negative integers), $k = -2^{n-1} + y$ and therefore $\ket{\psi_{< 0}}$ encodes the function $g(y) = p(y - 2^{n-1})$ in an unsigned integer representation, which is also a polynomial of degree $d$ in $y$.

	As explained in~\cite{huggins2024}, the tensorization of a degree-$d$ polynomial on a grid of unsigned integers can be represented by an MPS with bond dimension $d+2$.
	Therefore, both $\ket{\psi_{\ge 0}}$ and $\ket{\psi_{< 0}}$ admit such representations.
	Prepending these matrix product states with a qubit for the sign bit in the state $\ket{0}$ or $\ket{1}$ yields new MPSs with the same maximum bond dimension.
	The full state $\ket{\psi}$ is a direct sum of the two branches, which can be represented by an MPS whose bond dimension is at most the sum of the bond dimensions of each branch.
	Thus, the total bond dimension $M$ satisfies
	$M \le (d+2) + (d+2) = 2d+4$.
\end{proof}

\section{Implementing \texorpdfstring{$\tilde{U}_{\mathbf{A}}^{- \top}$}{UA-T}} \label{app:LCT}
In this appendix, we provide details for implementing $\tilde{U}_{\mathbf{A}}^{- \top}$, and bound the resulting errors when the initial state is a multivariate Gaussian. In \cref{app:implementing_A}, we provide an algorithm to approximate $U_{\mathbf{A}^{-\top}}^{\textsc{lct}}$ along with resource estimates. In \cref{app:nm_algo_grid_pad}, we derive a lower bound on $n_{\mathrm{pad}}$. Finally, in \cref{app:lct_ssct_error}, we bound the error from approximately implementing  $U_{\mathbf{A}^{-\top}}^{\textsc{lct}}$ and $U_{\mathbf{A}^{-\top}}^{\textsc{ssct}}$  when the nuclear state is a Gaussian.
 
\subsection{Implementing \texorpdfstring{$\tilde{U}_{\mathbf{A}^{-\top}}^{\textsc{lct}}$}{UT-1}} 
\label{app:implementing_A}
In this section, we give details and resource estimates for implementing $\tilde{U}_{\mathbf{A}^{-\top}}^{\textsc{lct}}$. Throughout, we use two's complement to represent integers.

The total cost to implement $\tilde{U}_{\mathbf{A}^{-\top}}^{\textsc{lct}}$ is
\begin{align}
    C_{\mathrm{anc}}(\tilde{U}_{\mathbf{A}^{-\top}}^{\textsc{lct}}) &= \mathrm{max} \{C_{\mathrm{anc}}(U_{\mathbf{Y}}), C_{\mathrm{anc}}(\tilde{U}_{2\mathrm{D}\mathbf{S}}), C_{\mathrm{anc}}(U_{\mathbf{J}}), C_{\mathrm{anc}}(U_{\mathbf{S}}) \} \nonumber \\
    & = 4 \bar{n}_{\mathrm{ISP}} -3\\
    C_{\mathrm{Toff}}(\tilde{U}_{\mathbf{A}^{-\top}}^{\textsc{lct}})& = C_{\mathrm{Toff}}(U_{\mathbf{Y}}) + N_G(3 C_{\mathrm{Toff}}(U_{2\mathrm{D} \mathbf{S}}) + C_{\mathrm{Toff}}(U_\mathbf{J})) +  C_{\mathrm{Toff}}(U_{\mathbf{S}})  \nonumber \\
    &  = \tfrac{9}{2}\eta_n^2(8 \bar{n}_{\mathrm{ISP}}^2 + 39\bar{n}_{\mathrm{ISP}} - 8)  - \tfrac{3}{2}\eta_n(8\bar{n}_{\mathrm{ISP}}^2 + 35\bar{n}_{\mathrm{ISP}} -8) - \bar{n}_{\mathrm{ISP}} 
\end{align}
where $U_{2\mathrm{D}\mathbf{S}}$ refers to either  $U_{\mathbf{S}_1(i,j,\varphi_{ij})}$ or $U_{\mathbf{S}_2(i,j,\varphi_{ij})}$, $U_{\mathbf{J}}$ is any $U_{\mathbf{J}(i,j,\theta_{ij})}$, and $U_{\mathbf{S}}$ is any $3\eta_n$-dimensional shear. We derive the costs of the individual unitaries below.

Let $\tilde{U}_{\mathbf{S}}$ be the unitary that implements a $3\eta_n$-dimensional shear, as described in \cref{eq:approx_SL}. $\tilde{U}_{\mathbf{S}}$ can be either a lower or upper shear. We implement $\tilde{U}_{\mathbf{S}}$ by repeating the classically controlled multiplication (CCM) subroutine~\cite{bauer2021}, which performs the transformation
\begin{align} \label{eq:CCM_trans}
    b_{ij}\cdots |n_i\rangle \cdots |n_j\rangle \cdots &\rightarrow b_{ij}\cdots |n_i\rangle_{\mathrm{exten}} \cdots |n_j\rangle_{\mathrm{exten}} \cdots \\ 
    &\rightarrow b_{ij}\cdots |(n_i + b_{ij}n_j) \, \mathrm{cmod}\, E\rangle_{\mathrm{exten}} \cdots |n_j\rangle_{\mathrm{exten}},
\end{align}
where $E = \bar{N}_{\mathrm{ISP}}/2$, $\bar{N}_{\mathrm{ISP}} = 2^{\bar{n}_{\mathrm{ISP}}}$, $\bar{n}_{\mathrm{ISP}} = n_{\mathrm{ISP}} + n_{\mathrm{pad}}$, and $n_{\mathrm{pad}}$ is the number of qubits that $\mathcal{B}_{\mathrm{ISP}}$ is padded by. 
The values $b_{ij}$ are supplied as classical inputs and consist of $\bar{n}_{\mathrm{ISP}}$ bits to represent the integer part and $r$ decimal bits. The registers $|n_i\rangle_{\mathrm{exten}}$ and $|n_j\rangle_{\mathrm{exten}}$ are obtained by extending $|n_i\rangle$ and $|n_j\rangle$ with $r$ additional bits to accommodate fixed-point arithmetic \cite{bauer2021}. After rounding, these additional $r$ bits are uncomputed and discarded. 

Because $b_{ij}$ is truncated to finite precision when loaded onto the quantum computer, $b_{ij} n_j$ can round to a different value than what would be computed using the exact value of $b_{ij}$, introducing an error. By selecting $r = \bar{n}_{\mathrm{ISP}}-1$, we ensure that this error only occurs for grid points at the vertices of $\bar{\mathcal{B}}_{\mathrm{ISP}}$, and that $b_{ij} n_j$ rounds to within $\pm 1$ grid point of the value obtained using the exact representation of $b_{ij}$ \cite{bauer2021}. We use this choice of $r$ throughout.

A single instance of the CCM subroutine loops over $\bar{n}_{\mathrm{ISP}}$ controlled additions of two $\bar{n}_{\mathrm{ISP}} + r$ bit values followed by a loop of depth $r$ in which the $l$th iteration implements controlled addition acting on two $l$-bit values. The cost of one CCM is
\begin{align}
     C_{\mathrm{anc}}(\textsc{ccm}) &=\bar{n}_{\mathrm{ISP}} + 3r  & C_{\mathrm{Toff}}(\textsc{ccm}) &= \bar{n}_{\mathrm{ISP}}^2 + 3\bar{n}_{\mathrm{ISP}} + r \nonumber\\
     &= 4\bar{n}_{\mathrm{ISP}} - 3&
    & = \bar{n}_{\mathrm{ISP}}^2 + 4\bar{n}_{\mathrm{ISP}} -1.
\end{align}

To implement $\tilde{U}_{\mathbf{S}}$, $3\eta_n -i$ CCM operations are applied to the $i$th register, followed by a rounding operation, then $3\eta_n -i$ inverse CCM operations. We perform rounding using a controlled adder \cite{bauer2021}, which uses $2\bar{n}_{\mathrm{ISP}}$ Toffoli gates and $\bar{n}_{\mathrm{ISP}}$ ancillas. Repeating this routine over $3\eta_n -1$ registers yields the following resource costs:
\begin{align}
     C_{\mathrm{anc}}(\tilde{U}_{\mathbf{S}}) &= C_{\mathrm{anc}}(\textsc{ccm}) 
     & C_{\mathrm{Toff}}(\tilde{U}_{\mathbf{S}}) &= \sum_{i=1}^{3\eta_n} (3\eta_n -i)\cdot2C_{\mathrm{Toff}}(\textsc{ccm})  + (3\eta_n -1)\cdot 2\bar{n}_{\mathrm{ISP}} \\
     & = 4 \bar{n}_{\mathrm{ISP}} - 3 &
    &=9\eta_n^2(\bar{n}^2_{\mathrm{ISP}} + 4\bar{n}_{\mathrm{ISP}} -1 ) 
    - 3\eta_n(\bar{n}_{\mathrm{ISP}}^2 + 2\bar{n}_{\mathrm{ISP}} -1) - 2\bar{n}_{\mathrm{ISP}}.
\end{align}

$\tilde{U}_{2 \mathrm{D}\mathbf{S}}$ is given by either $\tilde{U}_{\mathbf{S}_1(i,j,\varphi_{ij})}$ or $\tilde{U}_{\mathbf{S}_2(i,j,\varphi_{ij})}$. These unitaries transform the $i$th and $j$th subregisters by
\begin{align}
    \tilde{U}_{\mathbf{S}_1(i,j,\varphi_{ij})} \cdots |n_i\rangle \cdots |n_j\rangle \cdots &=  \cdots |n_i \oplus R(t_{ij} n_j  \, \mathrm{cmod}\, E)\rangle \cdots |n_j\rangle \cdots\\
    \tilde{U}_{\mathbf{S}_2(i,j,\varphi_{ij})} \cdots |n_i\rangle \cdots |n_j\rangle \cdots &= \, \cdots |n_i\rangle \cdots |n_j \oplus R(-s_{ji} n_i \,\mathrm{cmod}\, E)\rangle \cdots .
\end{align}
Thus, $\tilde{U}_{2\mathrm{D}\mathbf{S}}$ is implemented by first applying a CCM operation, followed by a rounding operation, then an inverse CCM. We perform rounding using a controlled adder that uses $2\bar{n}_{\mathrm{ISP}}$ Toffoli gates and $\bar{n}_{\mathrm{ISP}}$ ancillas \cite{bauer2021}. The $b_{ij}$ in \cref{eq:CCM_trans} are substituted with the appropriate $t_{ij}$ and $s_{ij}$ computed during the Givens decomposition. The cost of a single approximate 2D shearing transformation is then
\begin{align}
    C_{\mathrm{anc}}(\tilde{U}_{2\mathrm{D}\mathbf{S}}) &= C_{\mathrm{anc}}(\textsc{ccm}) & C_{\mathrm{Toff}}(\tilde{U}_{2\mathrm{D}\mathbf{S}}) &= 2C_{\mathrm{Toff}}(\textsc{ccm}) + 2\bar{n}_{\mathrm{ISP}} \nonumber \\
    &= 4\bar{n}_{\mathrm{ISP}} - 3 & &= 2\bar{n}_{\mathrm{ISP}}^2 + 10\bar{n}_{\mathrm{ISP}} - 2.
\end{align}

When $\det(\mathbf{X})= + 1$, $U_{\mathbf{Y}}$ is the identity and no gates are applied. Otherwise, $U_{\mathbf{Y}}$ flips the sign of the first coordinate,
\begin{equation}
    U_{\mathbf{Y}}|n_1\rangle |n_2\rangle \dots |n_{3\eta_n}\rangle = \lvert-n_1\rangle |n_2\rangle \dots |n_{3\eta_n}\rangle.
\end{equation}
$U_{\mathbf{Y}}$ is implemented in two's complement by inverting all $\bar{n}_{\mathrm{ISP}}$ bits of $|n_1\rangle$, adding one, and discarding the carry bit. Inverting the bits can be done with $\bar{n}_{\mathrm{ISP}}$ \textsc{not} gates. Adding two $\bar{n}_{\mathrm{ISP}}$-bit numbers uses $\bar{n}_{\mathrm{ISP}}$ Toffoli gates and $2\bar{n}_{\mathrm{ISP}}$ ancillas~\cite{gidney2018}, being $\bar{n}_{\mathrm{ISP}}$ ancillas to store the number being added and another $\bar{n}_{\mathrm{ISP}}$ to perform addition. The resource requirements for $U_{\mathbf{Y}}$ are therefore
\begin{align}
    C_{\mathrm{anc}}(U_{\mathbf{Y}}) &= 2\bar{n}_{\mathrm{ISP}} &
    C_{\mathrm{Toff}}(U_{\mathbf{Y}}) &= \bar{n}_{\mathrm{ISP}}.
\end{align}

The action of $U_{\mathbf{J}(i,j,\theta_{ij})}$ is
\begin{equation}
   U_{\mathbf{J}(i,j,\theta_{ij})} \cdots |n_i\rangle \cdots |n_j\rangle \cdots = \cdots |\mathrm{sgn}(\theta_{ij})  n_j\rangle \cdots \lvert- \mathrm{sgn}(\theta_{ij})  n_i\rangle \cdots.
\end{equation}
Implementing $U_{\mathbf{J}(i,j,\theta_{ij})}$ consists of swapping $|n_i\rangle$ and $|n_j\rangle$, and flipping the sign of one of the registers. The swap can be performed with $3\bar{n}_{\mathrm{ISP}}$ \textsc{cnot} gates and no ancillas. Inverting the sign uses the same resources as $U_{\mathbf{Y}}$, giving:
\begin{align}
    C_{\mathrm{anc}}(U_{\mathbf{J}(i,j,\theta_{ij})}) &= 2\bar{n}_{\mathrm{ISP}} &
    C_{\mathrm{Toff}}(U_{\mathbf{J}(i,j,\theta_{ij})}) &= \bar{n}_{\mathrm{ISP}}. 
\end{align}

\subsection{Bounding \texorpdfstring{$n_{\mathrm{pad}}$}{npad}} \label{app:nm_algo_grid_pad}
Here, we derive lower bounds on $n_{\mathrm{pad}}$ that ensures zero error from modular addition for both $\tilde{U}_{\mathbf{A}^{-\top}}^{\textsc{lct}}$ and $\tilde{U}_{\mathbf{A}^{-\top}}^{\textsc{ssct}}$.

\subsubsection{Bounding \texorpdfstring{$n_{\mathrm{pad}}$}{npad} for \texorpdfstring{$\tilde{U}_{\mathbf{A}^{-\top}}^{\textsc{lct}}$}{U-LCT}}
Here, we derive a bound on $n_{\mathrm{pad}}$ for $\tilde{U}_{\mathbf{A}^{-\top}}^{\textsc{lct}}$. To do so, we examine how the initial grid $\mathcal{B}_{\mathrm{ISP}} = [- {N_{\mathrm{ISP}}}/{2},{N_{\mathrm{ISP}}}/{2}-1]^{3\eta_n}\cap \mathbb{Z}^{3\eta_n}$ given in \cref{sect:LCT} transforms under $\tilde{U}_{\mathbf{A}^{-\top}}^{\textsc{lct}}$.

To begin, we write $\tilde{U}_{\mathbf{A}^{-\top}}^{\textsc{lct}}$ as
\begin{equation}
   \tilde{U}_{\mathbf{A}^{-\top}}^{\textsc{lct}} = \Bigg( \prod_{l=1}^{\beta+1} \tilde{U}_{\mathbf{F}_l} \Bigg),
\end{equation}
where $\tilde{U}_{\mathbf{F}_l}$ is $l$th sequential unitary of \cref{eq:ortho_seq_shear}, $\mathbf{F}_l$ is the matrix associated with  $\tilde{U}_{\mathbf{F}_l}$, and $\beta = 18 \eta_n^2 - 6\eta_n +1$ is the number of unitaries in \cref{eq:ortho_seq_shear} not including  $\tilde{U}_{\mathbf{L}}$. We also define the
padded grid 
\begin{equation} \label{eq:grid_isp_bar}
    \BBar =  \big[- \tfrac12{\NBar}, \tfrac12 {\NBar} - 1\big]^{3\eta_n}\cap \mathbb{Z}^{3\eta_n},
\end{equation}    
where $\NBar = 2^{\bar{n}_{\mathrm{ISP}}}$ and  $\bar{n}_{\mathrm{ISP}} = n_{\mathrm{ISP}} + n_{\mathrm{pad}}$.

For all $\mathbf{n} \in \mathcal{B}_{\mathrm{ISP}}$ to remain within $\bar{\mathcal{B}}_{\mathrm{ISP}}$ throughout each step of the $\tilde{U}_{\mathbf{A}^{-\top}}^{\textsc{lct}}$, the following must be true:
\begin{equation} \label{eq:npad_condition_1}
  \max_{1 \leq P \leq \beta + 1} \max_{\mathbf{n} \in \mathcal{B}_{\mathrm{ISP}}}\left\Vert R(\mathbf{F}_P \cdots R(\mathbf{F}_{2} \,R(\mathbf{F}_1 \mathbf{n} )) \cdots )  \right\Vert_{\infty} \leq \tfrac12 {\bar{N}_{\mathrm{ISP}}}.
\end{equation}
Letting $\mathbf{a}_P =  R(\mathbf{F}_{P-1} \cdots R(\mathbf{F}_3\, R(\mathbf{F}_{2} \,R(\mathbf{F}_1 \mathbf{n} ))) \cdots )$ allows us to express \cref{eq:npad_condition_1} as
 \begin{equation} \label{eq:npad_condition_2}
  \max_{1 \leq P \leq \beta+1} \max_{\mathbf{n} \in \mathcal{B}_{\mathrm{ISP}}}\left\Vert R(\mathbf{F}_P\mathbf{a}_{P})\right\Vert_{\infty} \leq \tfrac12 {\bar{N}_{\mathrm{ISP}}}.
\end{equation}
We can also write as $\mathbf{F}_{P} \mathbf{a}_P = R(\mathbf{F}_{P} \mathbf{a}_P) + \mathrm{Fr}(\mathbf{F}_{P} \mathbf{a}_P)$, where we define the fractional part $\mathrm{Fr}(\mathbf{v}) = \mathbf{v} - R(\mathbf{v})$, whose components obey $\big|[\mathrm{Fr}(\mathbf{v})]_i\big| \leq \tfrac12$. Then, 
 \begin{align} 
   \max_{1 \leq P \leq \beta + 1} \max_{\mathbf{n} \in \mathcal{B}_{\mathrm{ISP}}}\left\Vert \mathbf{F}_P\mathbf{a}_P - \mathrm{Fr}(\mathbf{F}_P\mathbf{a}_{P})  \right\Vert_{\infty}  &\leq   \max_{1 \leq P \leq \beta + 1} \max_{\mathbf{n} \in \mathcal{B}_{\mathrm{ISP}}} \big(\left\Vert \mathbf{F}_P\mathbf{a}_{P}\Vert_{\infty} +  \Vert \mathrm{Fr}(\mathbf{F}_P\mathbf{a}_{P})  \right\Vert_{\infty} \big) \label{eq:manip1} \\
   &\leq \max_{1 \leq P \leq \beta + 1} \max_{\mathbf{n} \in \mathcal{B}_{\mathrm{ISP}}}\Vert \mathbf{F}_P\mathbf{a}_{P}\Vert_{\infty}  + \tfrac12.\label{eq:manip2}
\end{align}

Next, we substitute in the definition of $\mathbf{a}_P$ and repeat the same set of steps. In particular, 
\begin{align}
    \max_{1 \leq P \leq \beta + 1} \max_{\mathbf{n} \in \mathcal{B}_{\mathrm{ISP}}}\Vert \mathbf{F}_P\mathbf{a}_{P}\Vert_{\infty}  + \tfrac12& = \max_{1 \leq P \leq \beta + 1} \max_{\mathbf{n} \in \mathcal{B}_{\mathrm{ISP}}}\Vert \mathbf{F}_P R(\mathbf{F}_{P-1} \mathbf{a}_{P-1})  \Vert_{\infty}  + \tfrac12 \\
    & = \max_{1 \leq P \leq \beta + 1} \max_{\mathbf{n} \in \mathcal{B}_{\mathrm{ISP}}}\Vert \mathbf{F}_P (\mathbf{F}_{P-1} \mathbf{a}_{P-1} - \mathrm{Fr}(\mathbf{F}_{P-1} \mathbf{a}_{P-1})) \Vert_{\infty}  + \tfrac12 \\
    & \leq \max_{1 \leq P \leq \beta + 1} \max_{\mathbf{n} \in \mathcal{B}_{\mathrm{ISP}}} \big(\Vert \mathbf{F}_P \mathbf{F}_{P-1} \mathbf{a}_{P-1} \Vert_{\infty} + \Vert \mathbf{F}_{P} \mathrm{Fr}(\mathbf{F}_{P-1} \mathbf{a}_{P-1}) \Vert_{\infty} \big) + \tfrac12 \\ 
    & \leq \max_{1 \leq P \leq \beta + 1} \Big(\max_{\mathbf{n} \in \mathcal{B}_{\mathrm{ISP}}}\Vert \mathbf{F}_P \mathbf{F}_{P-1} \mathbf{a}_{P-1} \Vert_{\infty} + \max_{\mathbf{m} \in \mathcal{S}} \Vert \mathbf{F}_{P} \mathbf{m} \Vert_{\infty}\Big) + \tfrac12,
\end{align}
where $\mathcal{S} = [-\tfrac12,\tfrac12]^{3\eta_n}$.
Repeating this procedure $P-1$ times yields the condition
\begin{align}
    \max_{1 \leq P \leq \beta+1} \max_{\mathbf{n} \in \mathcal{B}_{\mathrm{ISP}}}\left\Vert R(\mathbf{F}_P\mathbf{a}_{P})\right\Vert_{\infty} & \leq \max_{1 \leq P \leq \beta + 1}\bigg( \max_{\mathbf{n}\in \mathcal{B}_{\mathrm{ISP}}}  \Vert \mathcal{F}_{P,1} \mathbf{n}\Vert_{\infty}  + \max_{\mathbf{m} \in \mathcal{S}} \sum_{l=2}^{P}\Vert \mathcal{F}_{P,l} \mathbf{m}\Vert_{\infty} \bigg)
    + \tfrac12 \\
    & \leq  \max_{1 \leq P \leq \beta + 1}\max_{\mathbf{n}\in \mathcal{B}_{\mathrm{ISP}}}  \Vert \mathcal{F}_{P,1} \mathbf{n}\Vert_{\infty}  + \max_{\mathbf{m} \in \mathcal{S}} \sum_{l=2}^{\beta + 1}\Vert \mathcal{F}_{\beta+1,l} \mathbf{m}\Vert_{\infty} + \tfrac12 \label{eq:npad_middle_equation} \\ 
    & \leq \frac{N_{\mathrm{ISP}}}{2} \max_{1 \leq P \leq \beta+1} \Vert \mathcal{F}_{P,1} \Vert_{\infty} + \tfrac12 \sum_{l=2}^{\beta+1}\Vert \mathcal{F}_{\beta+1,l}\Vert_{\infty} + \tfrac12 \label{eq:npad_condition_3} 
\end{align}
where $\mathcal{F}_{P,l}$ is the partial product
\begin{equation}
   \mathcal{F}_{P,l} = \mathbf{F}_P \mathbf{F}_{P-1} \cdots \mathbf{F}_{l+1} \mathbf{F}_l.
\end{equation}

Using \cref{eq:npad_condition_2} and solving for $n_{\mathrm{pad}}$, we find
\begin{equation} \label{eq:npad_upper_bound_1}
    n_{\mathrm{pad}} \geq \left\lceil \log \left(N_{\mathrm{ISP}} \max_{1 \leq P \leq \beta + 1} \Vert \mathcal{F}_{P,1} \Vert_{\infty} + \sum_{l=2}^{\beta+1}\Vert \mathcal{F}_{\beta+1,l}\Vert_{\infty} + 1  \right) \right \rceil - n_{\mathrm{ISP}},
\end{equation}
all parts of which can be calculated classically.

A simpler, but looser, bound on $n_{\mathrm{pad}}$ can be obtained from \cref{eq:npad_middle_equation} and bounding the infinity norms in terms of two-norms. The first term in \cref{eq:npad_middle_equation} is then
\begingroup
\allowdisplaybreaks
\begin{align}
   \max_{1 \leq P \leq \beta + 1}\max_{\mathbf{n} \in \mathcal{B}_{\mathrm{ISP}}}\left\Vert \mathcal{F}_{P,1} \mathbf{n} \right\Vert_{\infty} &\leq   \max_{2 \leq P \leq \beta + 1}\max_{\mathbf{n} \in \mathcal{B}_{\mathrm{ISP}}}\left\Vert \mathcal{F}_{P,2} \Vert_{\infty} \Vert\mathbf{L}\Vert_{\infty} \Vert\mathbf{n} \right\Vert_{\infty} \\
   &\leq   \max_{2 \leq P \leq \beta + 1} \sqrt{3\eta_n}  \Vert \mathcal{F}_{P,2} \Vert_{2} \Vert\mathbf{L}\Vert_{\infty} \max_{\mathbf{n} \in \mathcal{B}_{\mathrm{ISP}}}\Vert\mathbf{n} \Vert_{\infty} \\
    & =  \sqrt{3\eta_n} \Vert\mathbf{L}\Vert_{\infty} \frac{N_{\mathrm{ISP}}}{2} \max_{2 \leq P \leq \beta + 1} \Vert \mathcal{F}_{P,2} \Vert_{2},
\end{align}
\endgroup
where we used $\mathbf{F}_1 = \mathbf{L}$. 
To evaluate the maximum over $P$ of $\Vert \mathcal{F}_{P,2} \Vert_{2}$, we observe that, for $P\neq \beta+1$, if $(P-1) \operatorname{mod}4 \in \{0,3\}$, then $\mathcal{F}_{P,2}$ is an orthogonal matrix and $\Vert \mathcal{F}_{P,2} \Vert_{2} = 1$. If $P= \beta+1$, $\mathcal{F}_{P,2}$ is an orthogonal matrix multiplied by a reflection, and $\Vert \mathcal{F}_{P,2} \Vert_{2} = 1$.
If  $(P-1) \operatorname{mod}4 \in \{1,2\}$, then $\mathcal{F}_{P,2}$ is an orthogonal matrix multiplied by one or two shearing transformations, respectively. For $P\operatorname{mod}4 = 1$, 
\begin{equation} \label{eq:FP_2norm}
    \Vert \mathcal{F}_{P,2}\Vert_{2} =  \Vert \mathbf{F}_P\mathcal{F}_{P-1,2}  \Vert_{2} \leq  \Vert \mathbf{F}_{P} \Vert_{2},
\end{equation}
and for $P\, \mathrm{mod}\,4 = 2$, 
\begin{equation} \label{eq:FPFPm1_2norm}
   \Vert \mathcal{F}_{P,2}  \Vert_{2} = \Vert \mathbf{F}_{P}\mathbf{F}_{P-1}\mathcal{F}_{P-2,2} \Vert_{2} \leq \Vert \mathbf{F}_{P} \mathbf{F}_{P-1} \Vert_{2}.
\end{equation}
We can evaluate the two-norms using $\Vert A \Vert_2 = \sqrt{\lambda_{\max}(A^\dagger A)}$,
where $\lambda_{\mathrm{max}}(\cdot)$ denotes the maximum eigenvalue. Then,
\begingroup
\allowdisplaybreaks
\begin{align}
    \lambda_{\mathrm{max}}(\mathbf{F}_P^{\top} \mathbf{F}_P) &= \max_{y\in \{-1,+1 \}} \left\{1 + \tfrac{1}{2}t_{ij}\left(t_{ij}+ y \sqrt{4 + t_{ij}^2} \right) \right\} \\
    &\leq 1 + \tfrac{1}{2}(1 + \sqrt{5} ) \\
    \lambda_{\mathrm{max}}(\mathbf{F}_{P-1}^{\top}\mathbf{F}_{P}^{\top} \mathbf{F}_{P}\mathbf{F}_{P-1}) &= \max_{y\in \{-1, +1 \} } \left\{1 + \tfrac{1}{2}\left( t_{ij}^2 + y  \sqrt{(2 + t_{ij}^2)^2 - 4(c_{ij} + s_{ij} t_{ij})^2} \right)    \right\} \\
    &\leq 1 + \tfrac{1}{2}(1 + \sqrt{5} ), 
\end{align}
\endgroup
where matrix elements $t_{ij}$ and $s_{ij}$ are given in \cref{eq:s1_def,eq:s2_def}, respectively. Thus, for both $P \operatorname{mod}4 \in \{1,2\}$, 
\begin{equation}
    \Vert \mathcal{F}_{P,2}  \Vert_{2} \leq \sqrt{1 + \tfrac{1}{2}(1 + \sqrt{5} )} < 1.619,
\end{equation}
which establishes 
\begin{equation}
    \max_{1 \leq P \leq \beta + 1}\max_{\mathbf{n} \in \mathcal{B}_{\mathrm{ISP}}}\left\Vert \mathcal{F}_{P,1} \mathbf{n} \right\Vert_{\infty} \leq \tfrac12 {1.619\sqrt{3 \eta_n}} N_{\mathrm{ISP}} \Vert \mathbf{L}  \Vert_{\infty}. \label{eq:clean_bound_term1}
\end{equation}

To bound the second term in \cref{eq:npad_middle_equation}, we first write
\begin{align}
    \max_{\mathbf{m} \in \mathcal{S}} \Vert \mathcal{F}_{\beta+1,l} \mathbf{m}\Vert_{\infty} &\leq     \max_{\mathbf{m} \in \mathcal{S}} \Vert \mathcal{F}_{\beta+1,l} \Vert_{\infty} \Vert \mathbf{m}\Vert_{\infty} \\
    &\leq \max_{\mathbf{m} \in \mathcal{S}}  \sqrt{3\eta_n}\, \Vert \mathcal{F}_{\beta+1,l} \Vert_{2} \Vert \mathbf{m}\Vert_{\infty}\\
    &= \tfrac{1}{2}  \sqrt{3\eta_n} \Vert \mathcal{F}_{\beta+1,l} \Vert_{2}.
\end{align}
As was the case for the first term, $\mathcal{F}_{\beta+1, l}$ is either an orthogonal matrix with $\Vert \mathcal{F}_{P,l} \Vert_{2} = 1$, or an orthogonal matrix multiplied by a single shear or a product of shears with $\Vert \mathcal{F}_{P,l} \Vert_{2} \leq 1.619 $. This inequality is easily shown using the same approach as in \cref{eq:FP_2norm,eq:FPFPm1_2norm}. This establishes the bound
\begin{equation}
    \max_{\mathbf{m} \in \mathcal{S}} \sum_{l=2}^{\beta+1}\Vert \mathcal{F}_{\beta+1,l} \mathbf{m}\Vert_{\infty} \leq \tfrac12 \beta 1.619 \sqrt{3 \eta_n}. \label{eq:clean_bound_term2}
\end{equation}

Finally, using \cref{eq:clean_bound_term1,eq:clean_bound_term2} in \cref{eq:npad_condition_3} leads to the condition
\begin{equation}
       1.619 \sqrt{3 \eta_n} ( N_{\mathrm{ISP}}\Vert \mathbf{L}\Vert_{\infty} +\beta)+1  \leq \bar{N}_{\mathrm{ISP}}. 
\end{equation}
Solving for $n_{\mathrm{pad}}$ gives
\begin{equation}
    n_{\mathrm{pad}} \geq \left\lceil \log \left(1.619\sqrt{3 \eta_n} ( N_{\mathrm{ISP}}\Vert \mathbf{L}\Vert_{\infty} +\beta)+1\right) \right \rceil - n_{\mathrm{ISP}}.
\end{equation}

\subsubsection{Bounding \texorpdfstring{$n_{\mathrm{pad}}$}{npad} for \texorpdfstring{$\tilde{U}_{\mathbf{A}^{-\top}}^{\textsc{ssct}}$}{U-SSCT}}
Here, we derive a bound on $n_{\mathrm{pad}}$ for $\tilde{U}_{\mathbf{A}^{-\top}}^{\textsc{ssct}}$. In this case, $n_{\mathrm{pad}}$ is bounded by
\begin{equation}
    \max_{\mathbf{n} \in \mathcal{B}_{\mathrm{ISP}}} \Vert R(\mathbf{L}^{-\top} \mathbf{n}) \Vert_{\infty} \leq \tfrac12 {\bar{N}_{\mathrm{ISP}}}. \label{eq:npad_ssct}
\end{equation}
Using the same approach as in \cref{eq:manip1,eq:manip2}, we find
\begin{align}
    \max_{\mathbf{n} \in \mathcal{B}_{\mathrm{ISP}}} \Vert R(\mathbf{L}^{-\top} \mathbf{n}) \Vert_{\infty}  &\leq \max_{\mathbf{n} \in \mathcal{B}_{\mathrm{ISP}}}  \Vert \mathbf{L}^{-\top} \mathbf{n} \Vert_{\infty} + \tfrac{1}{2} \\
    &\leq \tfrac12 N_{\mathrm{ISP}} \Vert \mathbf{L}^{-\top}\Vert_{\infty} + \tfrac{1}{2}.
\end{align}
Using the above bound with \cref{eq:npad_ssct} gives
\begin{equation}
    n_{\mathrm{pad}} \geq \lceil\log(N_{\mathrm{ISP}} \Vert \mathbf{L}^{-\top}\Vert_{\infty} + 1) \rceil - n_{\mathrm{ISP}}.
\end{equation}

\subsection{Bounding \texorpdfstring{$\tilde{U}_{\mathbf{A}^{-\top}}^{\textsc{lct}}$}{U-LCT} error}
\label{app:lct_ssct_error}
In this section, we bound the error $\epsilon_{\mathrm{LCT}}$ from approximating $U_{\mathbf{A}^{-\top}}^{\textsc{lct}}$ and show that $\epsilon_{\mathrm{LCT}} = \mathcal{O}(\eta_n^2 \Delta)$.
    
Let $|\Xi\rangle$ be a Gaussian initial state defined on $\bar{\mathcal{B}}_{\mathrm{ISP}}$,
\begin{equation} \label{eq:xi_gauss}
    |\Xi\rangle = \frac{1}{\mathcal{N}_{\mathrm{ISP}}} \sum_{\mathbf{n}\in \mathcal{B}_{\mathrm{ISP}}}   e^{-\mathbf{n}^{\top} \boldsymbol{\Sigma}\mathbf{n}/2}|\mathbf{n}\rangle,\quad\, \boldsymbol{\Sigma} = \Delta^2 \boldsymbol{\Sigma}',\quad\, \boldsymbol{\Sigma}' = \mathrm{diag}(\Sigma'_1, \cdots, \Sigma'_{3\eta_n}), \quad \mathrm{and} \quad \Sigma'_i \in \mathbb{R}_{>0}.
\end{equation}
Initially, the nuclear state only has support on $\mathcal{B}_{\mathrm{ISP}}$, i.e., the amplitudes on $\bar{\mathcal{B}}_{\mathrm{ISP}} \setminus \mathcal{B}_{\mathrm{ISP}} $ are zero. The LCT error is
\begin{align}
    \epsilon_{\mathrm{LCT}} &= D(U_{\mathbf{A}^{-\top}}^{\textsc{lct}} |\Xi\rangle, \tilde{U}_{\mathbf{A}^{-\top}}^{\textsc{lct}}|\Xi\rangle  ) \\
    & \leq \epsilon_{\mathrm{shear}} + \epsilon_{\mathrm{ortho}},
\end{align}
where
\begin{align}
    \epsilon_{\mathrm{shear}} &= D(U_{\mathbf{L}}|\Xi\rangle,  \tilde{U}_{\mathbf{L}}|\Xi\rangle) \label{eq:eps_shear_def}\\
    \epsilon_{\mathrm{ortho}} &= D(U_{\mathbf{X}} U_{\mathbf{L}}|\Xi\rangle, \tilde{U}_{\mathbf{X}} U_{\mathbf{L}}|\Xi\rangle) \\
    \tilde{U}_{\mathbf{X}}& = \bigg( \prod_{l=2}^{\beta+1} \tilde{U}_{\mathbf{F}_l} \bigg) \\
     U_{\mathbf{X}} &= \bigg( \prod_{l=2}^{\beta+1} U_{\mathbf{F}_l} \bigg)
\end{align}
and where $\mathbf{F}_{l}$ is the $l$th sequential unitary on the right-hand side of \cref{eq:U_AT_exact}. 

We apply the triangle inequality to further decompose $\epsilon_{\mathrm{ortho}}$ into a sum of errors resulting from 2D shears,
\begin{align} \label{eq:ortho_bound}
    \epsilon_{\mathrm{ortho}} &\leq \sum_{P=2}^{\beta} \epsilon^{(\mathrm{2D-shear})}_P,\\
\label{eq:2d_shear_eps_def}
   \epsilon^{(\mathrm{2D-shear})}_P &= D\big( U_{\mathbf{F}_P} |\Xi_{P-1} \rangle ,\, \tilde{U}_{\mathbf{F}_P}|\Xi_{P-1}\rangle \big)\\
   \label{eq:psi_mult_shear}
    |\Xi_{P}\rangle & = \prod_{l=2}^{P}U_{\mathbf{F}_l} U_{\mathbf{L}}|\Xi\rangle.
\end{align}
Because the reflection and $\pm \pi/2$ rotations are performed exactly, $\epsilon^{(\mathrm{2D-shear})}_P$ is non-zero only when $U_{\mathbf{F}_P}$ and $\tilde{U}_{\mathbf{F}_P}$ correspond to 2D shears. Hence, for $P\geq2$ and $(P-1) \operatorname{mod} 4 = 0$, $ \epsilon_{P}^{(\mathrm{2D-shear})} = 0$. Additionally, $ \epsilon_{\beta+1}^{(\mathrm{2D-shear})} = 0$.

From \cref{lem:shear_error,lem:ortho}, we have
 \begin{align}
        \epsilon_{\mathrm{shear}} & \leq   2^{1/2} \Big(1 -  e^{- \Delta^2 3 \eta_n \lambda_{\mathrm{max}}(\boldsymbol{\Lambda}')  }  \Big)^{1/2}\\
        \epsilon_{P}^{(\mathrm{2D-shear})} & \leq  2^{1/2}\left(1  -  e^{- \Delta^2 [\boldsymbol{\Lambda}_P']_{k_P,k_P}}\right)^{1/2},
    \end{align}
where $\lambda_{\mathrm{max}}(\cdot)$ is the largest eigenvalue of the input matrix, and
\begingroup
\allowdisplaybreaks
\begin{align}
    \boldsymbol{\Lambda}' & = \mathbf{L}^{- \top} \boldsymbol{\Sigma}' \mathbf{L}^{-1}\\
  \boldsymbol{\Lambda}_P' & = (\bar{\mathcal{F}}_P)^{\top} \mathbf{L}^{- \top} \boldsymbol{\Sigma}' \mathbf{L}^{-1} \bar{\mathcal{F}}_P \label{eq:lambda_prime_P}\\
  \bar{\mathcal{F}}_P & = \prod_{l=P}^{2}\mathbf{F}_l^{-1} \\ 
    k_P & = \mathrm{index}(\mathbf{F}_P) =
    \begin{cases}
        i & \mathrm{if}\,\,\mathbf{F}_P = \mathbf{S}_1(i,j,\varphi_{ij})\\
        j & \mathrm{if}\,\,\mathbf{F}_P = \mathbf{S}_2(i,j,\varphi_{ij}).
    \end{cases}
    \label{eq:k_index}
\end{align}
\endgroup
We determine the asymptotic scalings by noting that $(1-e^{-x^2})^{1/2}\leq x$ for all $x\geq 0$. Therefore, 
\begin{align}
        \epsilon_{\mathrm{shear}} &\leq 2^{1/2} \Delta \left(3\eta_n\lambda_{\mathrm{max}}(\boldsymbol{\Lambda}')  \right)^{1/2} = \mathcal{O}( \eta_n^{1/2} \Delta ) \\
        \epsilon_{P}^{\mathrm{2D-shear}} & \leq 2^{1/2} \Delta  \left( [\boldsymbol{\Lambda}_P']_{k_P,k_P} \right)^{1 / 2}= \mathcal{O}( \Delta ).
\end{align}
Because there are $\beta = \mathcal{O}(\eta_n^2)$ 2D shearing terms in $\epsilon_{\mathrm{LCT}}$, 
\begin{equation}
    \epsilon_{\mathrm{LCT}} = \mathcal{O}(\eta_n^2 \Delta).
\end{equation}

\begin{lemma} \label{lem:shear_error}
    Let $\mathbf{S}$ be a shear and $\boldsymbol{\Lambda}' = \mathbf{S}^{- \top} \boldsymbol{\Sigma}' \mathbf{S}^{-1}$, with $\boldsymbol{\Sigma}'$ given in \cref{eq:xi_gauss}. Then, for $|\Xi\rangle$ also in \cref{eq:xi_gauss},
    \begin{equation}
        \epsilon_{\mathrm{shear}} = D( U_{\mathbf{S}}  |\Xi\rangle ,\, \tilde{U}_{\mathbf{S}}  |\Xi\rangle ) \leq   2^{1/2} \Big(1 -  e^{- \Delta^2 3 \eta_n \lambda_{\mathrm{max}}(\boldsymbol{\Lambda}')  }  \Big)^{1/2}.
    \end{equation}
\end{lemma}
\begin{proof}
    We expand the definition of $\epsilon_{\mathrm{shear}}$:
\begingroup
\allowdisplaybreaks
\begin{align} \label{eq:eps_shear}
    \epsilon_{\mathrm{shear}} & =  D( U_{\mathbf{S}}  |\Xi\rangle ,\, \tilde{U}_{\mathbf{S}}  |\Xi\rangle ) \\
    &\leq \big\Vert U_{\mathbf{S}}  |\Xi\rangle - \tilde{U}_{\mathbf{S}}  |\Xi\rangle \big\Vert_2\\
    &= \bigg\Vert \frac{1}{\mathcal{N}_{\mathrm{ISP}}} \sum_{\mathbf{n}\in \bar{\mathcal{B}}_{\mathrm{ISP}}}  e^{-(\mathbf{S}^{-1} \mathbf{n})^{\top} \boldsymbol{\Sigma} \mathbf{S}^{-1}\mathbf{n}/2}| \mathbf{n}\rangle - \frac{1}{\mathcal{N}_{\mathrm{ISP}}} \sum_{\mathbf{n}\in \bar{\mathcal{B}}_{\mathrm{ISP}}}  e^{- (S^{-1}(\mathbf{n}))^{\top} \boldsymbol{\Sigma} S^{-1}(\mathbf{n})/2}  |\mathbf{n} \rangle \bigg\Vert_2\\
    & = \bigg( \frac{1}{\mathcal{N}_{\mathrm{ISP}}^2} \sum_{\mathbf{n}\in \bar{\mathcal{B}}_{\mathrm{ISP}}}  e^{-(\mathbf{S}^{-1} \mathbf{n})^{\top} \boldsymbol{\Sigma} \mathbf{S}^{-1}\mathbf{n}} +   \frac{1}{\mathcal{N}_{\mathrm{ISP}}^2} \sum_{\mathbf{n}\in \bar{\mathcal{B}}_{\mathrm{ISP}}}  e^{- (S^{-1}(\mathbf{n}))^{\top} \boldsymbol{\Sigma} S^{-1}(\mathbf{n})}  - \nonumber\\ 
    & \quad\,\quad\,\quad\,\quad\,\quad\,\quad\,\quad\,\quad\,  -\frac{2}{\mathcal{N}_{\mathrm{ISP}}^2} \sum_{\mathbf{n}\in \bar{\mathcal{B}}_{\mathrm{ISP}}}  e^{-(\mathbf{S}^{-1} \mathbf{n})^{\top} \boldsymbol{\Sigma} \mathbf{S}^{-1}\mathbf{n}/2}  e^{- (S^{-1}(\mathbf{n}))^{\top} \boldsymbol{\Sigma} S^{-1}(\mathbf{n})/2} \bigg)^{1/2}\\
    & = 2^{1/2}\bigg( 1 - \frac{1}{\mathcal{N}_{\mathrm{ISP}}^2} \sum_{\mathbf{n}\in \bar{\mathcal{B}}_{\mathrm{ISP}}}  e^{-(\mathbf{S}^{-1} \mathbf{n})^{\top} \boldsymbol{\Sigma} \mathbf{S}^{-1}\mathbf{n}/2}  e^{- (S^{-1}(\mathbf{n}))^{\top} \boldsymbol{\Sigma} S^{-1}(\mathbf{n})/2}   \bigg)^{1/2} \label{eq:shear_error_last_line},
\end{align}
\endgroup
where $\tilde{U}_{\mathbf{S}}$ is applied as in \cref{eq:inverse_S}. In the final step, we use the fact that shears preserve normalization, meaning that the
first two terms in the penultimate line sum to 1.

Let the difference between $\mathbf{S}^{-1} \mathbf{n}$ and $S^{-1}(\mathbf{n})$ be
\begin{equation} \label{eq:p_def}
    \mathbf{p} = \mathbf{S}^{-1} \mathbf{n} - S^{-1}(\mathbf{n}) \Leftrightarrow  S^{-1}(\mathbf{n}) = \mathbf{S}^{-1} \mathbf{n} -  \mathbf{p}.
\end{equation}
We can gain an intuition for $\mathbf{p}$ by writing out the system of equations given by $\mathbf{m} = S^{-1}(\mathbf{n})$. For a lower shear $\mathbf{L}$,  \cref{eq:n_i_expression} gives:
\begin{align}
    m_1 &= n_1 \nonumber\\
    m_2 & = n_2 - R(b_{21}m_1) &\, &=  n_2 - (b_{21}m_1) + \delta_2  \nonumber\\
    m_3 & = n_3 - R(b_{31}m_1 + b_{32}m_2)  &\, &=  n_3 - (b_{31}m_1 + b_{32}m_2 ) + \delta_3 \nonumber \\
        &\vdotswithin{=}  & \, &  \vdotswithin{=} 
\end{align}
where $\delta_i = (b_{i1}m_1 + b_{i2}m_2 + \ldots + b_{ii-1}m_{i-1}) - R(b_{i1}m_1 + b_{i2}m_2 + \ldots +b_{ii-1}m_{i-1})$, and $|\delta_i| \leq 1/2$. Let $\mathbf{s}_{\mathrm{L}} = \mathbf{L} - I$ and $\boldsymbol{\delta} = (\delta_1,\delta_2, \dots, \delta_{3\eta_n})$. The above system of equations can be expressed as 
\begin{align}
    \mathbf{m} = \mathbf{n} -\mathbf{s}_{\mathrm{L}}\mathbf{m} + \boldsymbol{\delta} &\Rightarrow \mathbf{L}\mathbf{m} = \mathbf{n} + \boldsymbol{\delta} \\
    & \Rightarrow \mathbf{m} = \mathbf{L}^{-1} \mathbf{n} + \mathbf{L}^{-1}\boldsymbol{\delta}\\
    & \Rightarrow S^{-1}(\mathbf{n}) = \mathbf{L}^{-1} \mathbf{n} + \mathbf{L}^{-1}\boldsymbol{\delta}.
\end{align}
Substituting this expression into the definition of $\mathbf{p}$, we find
\begin{align}
    \mathbf{p} &= \mathbf{L}^{-1} \mathbf{n} - S^{-1}(\mathbf{n}) = -\mathbf{L}^{-1}\boldsymbol{\delta} \label{eq:p_error_expression}.
\end{align}
Thus, $\mathbf{p}$ is $-\boldsymbol{\delta}$ transformed by $\mathbf{L}^{-1}$. For an upper shear $\mathbf{U}$, the inverse is described after \cref{eq:n_i_expression}. Repeating the steps above for $\mathbf{U}$, we find $\mathbf{p} = -\mathbf{U}^{-1} \boldsymbol{\delta}$. Thus, for any shear, $\mathbf{p} =- \mathbf{S}^{-1} \boldsymbol{\delta}$. Substituting this into \cref{eq:p_def} gives
\begin{equation}
    S^{-1}(\mathbf{n}) = \mathbf{S}^{-1}( \mathbf{n} +\boldsymbol{\delta}).
    \label{eq:S_inv_simplified}
\end{equation}

Substituting \cref{eq:S_inv_simplified} into \cref{eq:shear_error_last_line} yields
\begingroup
\allowdisplaybreaks
\begin{align}
    \epsilon_{\mathrm{shear}} & = 2^{1/2}\bigg( 1 - \frac{1}{\mathcal{N}_{\mathrm{ISP}}^2} \sum_{\mathbf{n}\in \bar{\mathcal{B}}_{\mathrm{ISP}}}  e^{-(\mathbf{S}^{-1} \mathbf{n})^{\top} \boldsymbol{\Sigma} \mathbf{S}^{-1}\mathbf{n}/2}  e^{- (\mathbf{S}^{-1}\mathbf{n} + \mathbf{S}^{-1} \boldsymbol{\delta})^{\top} \boldsymbol{\Sigma}( \mathbf{S}^{-1}\mathbf{n} + \mathbf{S}^{-1}\boldsymbol{\delta})/2}   \bigg)^{1/2}\\
    & = 2^{1/2}\bigg( 1 - \frac{1}{\mathcal{N}_{\mathrm{ISP}}^2} \sum_{\mathbf{n}\in \bar{\mathcal{B}}_{\mathrm{ISP}}}  e^{- \mathbf{n}^{\top} \boldsymbol{\Lambda} \mathbf{n}/2}  e^{- (\mathbf{n} +  \boldsymbol{\delta})^{\top} \boldsymbol{\Lambda}( \mathbf{n} +  \boldsymbol{\delta})/2}   \bigg)^{1/2}\\
    & \leq 2^{1/2}\bigg( 1 - \frac{1}{\mathcal{N}_{\mathrm{ISP}}^2} \sum_{\mathbf{n}\in \bar{\mathcal{B}}_{\mathrm{ISP}}}  e^{- \mathbf{n}^{\top} \boldsymbol{\Lambda} \mathbf{n}/2}  e^{- (\mathbf{n} + \mathbf{c})^{\top} \boldsymbol{\Lambda}( \mathbf{n} + \mathbf{c})/2}   \bigg)^{1/2},
\end{align}
\endgroup
where $\boldsymbol{\Lambda} = \mathbf{S}^{- \top} \boldsymbol{\Sigma}\mathbf{S}^{-1}$, $\mathbf{c} = 2 \times (c_1, \ldots, c_{3\eta_n})$, and $c_i\in\{-1,1\}$. The sum in the second line is the overlap between two Gaussians, one with a mean of zero and the other with a mean of $-\boldsymbol{\delta}$. The final inequality follows because the overlap strictly decreases with $|\delta_i|$ and $|c_i| > |\delta_i|$. In replacing $\boldsymbol{\delta}$ by $\mathbf{c}$, we assume that $c_i$ are chosen to minimize this overlap. However, we will see in a moment that it is not necessary to specify the signs of these elements.

Completing the square in the exponential yields 
\begin{align} \label{eq:shear_bound_2}
     2^{1/2}\bigg( 1 - \frac{e^{- \mathbf{c}^{\top} \boldsymbol{\Lambda} \mathbf{c}/4} }{\mathcal{N}_{\mathrm{ISP}}^2} \sum_{\mathbf{n}\in \bar{\mathcal{B}}_{\mathrm{ISP}}}   e^{- (\mathbf{n} + {\mathbf{c}}/{2})^{\top} \boldsymbol{\Lambda}( \mathbf{n} + {\mathbf{c}}/{2})}   \bigg)^{1/2} = 2^{1/2}\big( 1 - e^{- \mathbf{c}^{\top} \boldsymbol{\Lambda} \mathbf{c}/4}   \big)^{1/2}.
\end{align}
The last step follows by recognizing that $\mathbf{n}+\mathbf{c}/2$ always aligns with a grid point. Because $\mathbf{c}$ is small, the wavefunction remains far from the edges of $\bar{\mathcal{B}}_{\mathrm{ISP}}$ and the sum
 remains normalized, yielding $\mathcal{N}^2_{\mathrm{ISP}}$. 
 
We can upper bound \cref{eq:shear_bound_2} by noting that
\begin{equation}
     \mathbf{c}^{\top}
\boldsymbol{\Lambda}\mathbf{c} \leq \lambda_{\mathrm{max}}(\boldsymbol{\Lambda}) \Vert \mathbf{c} \Vert_2^2 \leq 4\cdot3\eta_n \lambda_{\mathrm{max}}(\boldsymbol{\Lambda}),
\end{equation}
 which uses the fact that $\boldsymbol{\Lambda}$ is a positive-definite matrix. Finally, 
\begin{align}
    \epsilon_{\mathrm{shear}} &\leq 2^{1/2} \big(1 -  e^{- \mathbf{c}^{\top} \boldsymbol{\Lambda} \mathbf{c} / 4}  \big)^{1/2} \\
    &\leq 2^{1/2} \big(1 -  e^{- 3 \eta_n \lambda_{\mathrm{max}}(\boldsymbol{\Lambda})  }  \big)^{1/2} \\
    & =  2^{1/2} \big(1 -  e^{- \Delta^2 3\eta_n \lambda_{\mathrm{max}}(\boldsymbol{\Lambda}')  }  \big)^{1/2}.
\end{align}
\end{proof}

\begin{lemma} \label{lem:ortho}
Let 
\begin{equation}
    |\Xi_{P}\rangle = \frac{1}{\mathcal{N}_{\mathrm{ISP}}}\sum_{\mathbf{n} \in \bar{\mathcal{B}}_{\mathrm{ISP}}} e^{-  \mathbf{n}^{\top} \boldsymbol{\Lambda}_P \mathbf{n}/2} |\mathbf{n}\rangle,
\end{equation}
where $\boldsymbol{\Lambda}_P = \Delta^2 \boldsymbol{\Lambda}_P'$ and $\boldsymbol{\Lambda}_P'$ is defined in \cref{eq:lambda_prime_P}. Then, the 2D-shearing error obeys the bound
\begin{equation}
    \epsilon_{P}^{(\mathrm{2D-shear})} =D\big( U_{\mathbf{F}_P}  |\Xi_{P-1}\rangle , \tilde{U}_{\mathbf{F}_P}  |\Xi_{P-1}\rangle \big)\leq  2^{1/2}\big(1  -  e^{- \Delta^2 [\boldsymbol{\Lambda}_P']_{k_P,k_P}}\big)^{1/2},\label{eq:lemma5-result}
\end{equation}
where $k_P$ is defined in \cref{eq:k_index}.
\end{lemma}

\begin{proof}
By substituting $|\Xi_{P-1}\rangle$ into the definition $\epsilon_P^{(\mathrm{2D-shear})}$, we have 
\begin{equation} \label{eq:eps_2dshear}
    \epsilon^{(\mathrm{2D-shear})}_P = D\big( U_{\mathbf{F}_P}  |\Xi_{P-1}\rangle , \tilde{U}_{\mathbf{F}_P}  |\Xi_{P-1}\rangle \big) \leq \big\lVert U_{\mathbf{F}_P}|\Xi_{P-1}\rangle - \tilde{U}_{\mathbf{F}_P}|\Xi_{P-1}\rangle\big\rVert_2.
\end{equation}
From here, we follow the same steps as in \cref{lem:shear_error}, except that $|\Xi\rangle$ is replaced with $|\Xi_{P-1}\rangle$ and $U_{\mathbf{S}}$ and $\tilde{U}_{\mathbf{S}}$ are replaced by $U_{\mathbf{F}_{P}}$ and $\tilde{U}_{\mathbf{F}_{P}}$, respectively. For a given $\mathbf{n}$, $\mathbf{F}_{P}^{-1}$ only transforms a single element of the vector whose index is $k_P = \mathrm{index}(\mathbf{F}_P)$. Hence, the analogous expression to \cref{eq:p_error_expression} is $\mathbf{p} = -\mathbf{F}_P^{-1}\boldsymbol{\delta }$ where $\boldsymbol{\delta} = (0,0, \dots, \delta_{k_P}, \dots 0)$, and $|\delta_i| \leq 1/2$ denotes the error of the transformed vector element.

Let $\mathbf{c} = 2(0,0,\ldots, 1_{k_P},\ldots, 0)$. Then, 
\begingroup
\allowdisplaybreaks
\begin{align}
     \epsilon_P^{\mathrm{(2D - shear)}} &\leq  2^{1/2} \bigg(1 -  \frac{1}{\mathcal{N}_{\mathrm{ISP}}^2} \sum_{\mathbf{n} \in \bar{\mathcal{B}}_{\mathrm{ISP} }}  e^{- \mathbf{n}^{\top} \boldsymbol{\Lambda}_{P} \mathbf{n}/2}\,  e^{- (\mathbf{n} + \mathbf{c})^{\top} \boldsymbol{\Lambda}_{P}( \mathbf{n} + \mathbf{c})/2}\bigg)^{1/2} \\
    & = 2^{1/2}\bigg(1 -  \frac{e^{- \mathbf{c}^{\top} \boldsymbol{\Lambda}_{P} \mathbf{c}/4}}{\mathcal{N}_{\mathrm{ISP}}^2} \sum_{\mathbf{n} \in \bar{\mathcal{B}}_{\mathrm{ISP} }}   e^{- (\mathbf{n} + \frac{\mathbf{c}}{2})^{\top} \boldsymbol{\Lambda}_{P}( \mathbf{n} + \frac{\mathbf{c}}{2})}\bigg)^{1/2}\\
    &=  2^{1/2}\big(1 -  e^{- \mathbf{c}^{\top} \boldsymbol{\Lambda}_{P} \mathbf{c}/4} \big)^{1/2}. 
\end{align}
\endgroup
Finally, using $\mathbf{c}^{\top} \boldsymbol{\Lambda}_P \mathbf{c} = 4[\boldsymbol{\Lambda}_P]_{k_P k_P}$, we arrive at \cref{eq:lemma5-result}.
\end{proof}

\section{Implementing \textsc{pk}} \label{app:nm_algo_pkb}
Here, we implement \textsc{pk} using a phase-kickback (PK) subroutine \cite{cleve_1998} and estimate the resources. For a given state
\begin{equation}
    |\hat{\Psi}\rangle = \sum_{\mathbf{n} \in \bar{\mathcal{B}}_{\mathrm{ISP}}} \hat{\Psi}_{\mathbf{n}}|\mathbf{n}\rangle,
\end{equation}
the PK unitary performs the transformation
\begin{equation} \label{eq:pkb}
    \textsc{pk}|\hat{\Psi}\rangle = \sum_{\mathbf{n} \in \bar{\mathcal{B}}_{\mathrm{ISP}}} e^{-i \Delta (\mathbf{R}^{(0)})^{\top} \mathbf{n} } \hat{\Psi}_{\mathbf{n}}|\mathbf{n}\rangle,
\end{equation}
where $\Delta = 2\pi / L$. We implement \textsc{pk} using the decomposition 
\begin{equation}
   \textsc{pk}  = \left(\bigotimes_{i=1}^{3\eta_n} \textsc{pk}_i \right)U_{\mathrm{grad}},
\end{equation}
where $U_{\mathrm{grad}}$ prepares the phase-gradient state $|\psi_{\mathrm{grad}}\rangle$,
\begin{equation}
    U_{\mathrm{grad}} |0\rangle =
    |\psi_{\mathrm{grad}}\rangle = \sum_{l=0}^{N_{\mathrm{grad}}-1} \frac{e^{2 \pi i l /N_{\mathrm{grad}}}}{\sqrt{N_{\mathrm{grad}}}} |l\rangle,
\end{equation}
and $\textsc{pk}_i$ acts on the $i$th register and $|\psi_{\mathrm{grad}}\rangle$ to perform the transformation
\begin{equation}
   \textsc{pk}_i |\hat{\Psi}\rangle |\psi_{\mathrm{grad}}\rangle =\sum_{\mathbf{n} \in \bar{\mathcal{B}}_{\mathrm{ISP}}} e^{-i \Delta n_i R_i^{(0)}} \hat{\Psi}_{\mathbf{n}}|\mathbf{n}\rangle |\psi_{\mathrm{grad}}\rangle.
\end{equation}
Here, $R_i^{(0)}= z_i L /N_{\mathrm{grad}}$ is the $i$th component of $\mathbf{R}^{(0)}$, $z_i \in \left[-{N_{\mathrm{grad}}}/{2}, {N_{\mathrm{grad}}}/{2}-1 \right]\cap \mathbb{Z}$,
$N_{\mathrm{grad}} = 2^{b_{\mathrm{grad}}}$, and $b_{\mathrm{grad}}$ is the number of qubits used to represent $R_{i}^{(0)}$. 
We assume that $b_{\mathrm{grad}}$ has been chosen large enough that the finite-precision value $R_i^{(0)}$ can be expressed exactly.

Each $\textsc{pk}_i$ is implemented as shown in \cref{fig:pkb}. First, we use $U_{z_i}$ to load $z_i$ into an empty ancilla register,
\begin{equation}
    U_{z_i}|\hat{\Psi}\rangle|\psi_{\mathrm{grad}}\rangle |0\rangle = |\hat{\Psi}\rangle|\psi_{\mathrm{grad}}\rangle |z_i\rangle.
\end{equation}
Then, the value $n_i$ in the $i$th register of $|\hat{\Psi}\rangle$ and the $z_i$ are multiplied and the result stored in an empty register, 
\begin{equation}
    \textsc{mult} \sum_{\mathbf{n} \in \bar{\mathcal{B}}_{\mathrm{ISP}}} \hat{\Psi}_{\mathbf{n}}|\mathbf{n} \rangle |\psi_{\mathrm{grad}}\rangle |z_i\rangle|0\rangle  = \sum_{\mathbf{n} \in \bar{\mathcal{B}}_{\mathrm{ISP}}} \hat{\Psi}_{\mathbf{n}}|\mathbf{n} \rangle |\psi_{\mathrm{grad}}\rangle  |z_i\rangle |n_i \cdot z_i\rangle.
\end{equation}
The product $n_i \cdot z_i$ is added to $|\psi\rangle_{\mathrm{grad}}$, which is an eigenstate of addition, resulting in the phase kickback:
\begingroup
\allowdisplaybreaks
\begin{align}
    \textsc{add}\sum_{\mathbf{n} \in \bar{\mathcal{B}}_{\mathrm{ISP}}} \hat{\Psi}_{\mathbf{n}}|\mathbf{n} \rangle \sum_{l=0}^{N_{\mathrm{grad}}-1} \frac{e^{2 \pi i l /N_{\mathrm{grad}}}}{\sqrt{N_{\mathrm{grad}}}}  |l\rangle |z_i\rangle |n_i \cdot z_i\rangle &= \sum_{\mathbf{n} \in \bar{\mathcal{B}}_{\mathrm{ISP}}} \hat{\Psi}_{\mathbf{n}}|\mathbf{n} \rangle \sum_{l=0}^{N_{\mathrm{grad}}-1}\frac{e^{2 \pi i l /N_{\mathrm{grad}}}}{\sqrt{N_{\mathrm{grad}}}}  |l \oplus n_i \cdot z_i\rangle|z_i\rangle |n_i \cdot z_i\rangle \\
    &= \sum_{\mathbf{n} \in \bar{\mathcal{B}}_{\mathrm{ISP}}} \hat{\Psi}_{\mathbf{n}} e^{-2\pi i n_i z_i/ N_{\mathrm{grad}}} |\mathbf{n}\rangle \sum_{l=0}^{N_{\mathrm{grad}}-1} \frac{e^{2 \pi i l /N_{\mathrm{grad}}}}{\sqrt{N_{\mathrm{grad}}}} |l \rangle |z_i\rangle |n_i \cdot z_i\rangle \nonumber \\
    &= \sum_{\mathbf{n} \in \bar{\mathcal{B}}_{\mathrm{ISP}}} \hat{\Psi}_{\mathbf{n}} e^{-i \Delta n_i R_i^{(0)}}|\mathbf{n}\rangle |\psi_{\mathrm{grad}}\rangle |z_i\rangle|n_i \cdot z_i\rangle,
\end{align}
\endgroup
where $\oplus$ is addition $\operatorname{cmod}N_{\mathrm{grad}}/2$.
Finally, $\textsc{mult}^{\dagger}$ and $U_{z_i}^{\dagger}$ are uncomputed, giving \cref{eq:pkb}.

\begin{figure*}
\centering
    \includegraphics[width=\textwidth]{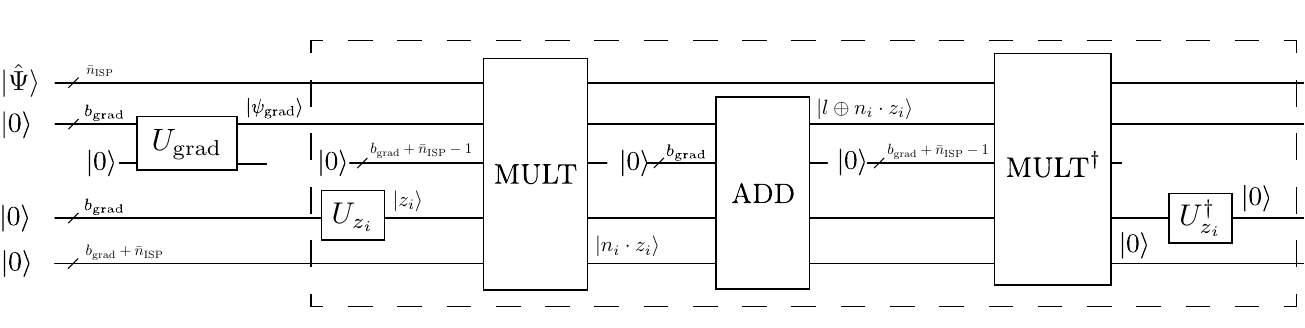}
    \vspace{-5mm}
    \caption{$\text{\textsc{pk}}_i$ circuit. The top wire encodes the $i$th sub-register of $|\hat{\Psi}\rangle$. The second wire stores the phase-gradient state which is prepared by the block $U_{\mathrm{grad}}$. The third wire contains ancilla qubits that are reused throughout the circuit. The fourth wire encodes the phase $z_i$ that we wish to kickback. The block $\textsc{mult}$ multiplies $|z_i\rangle$ by $|n_i\rangle$ and stores the result in the register corresponding to the fifth wire. This result is added $\operatorname{cmod} N_{\mathrm{grad}}/2$ to the phase-gradient state register. Finally, the state in the fifth register is uncomputed by performing $\textsc{mult}^{\dagger}$.}
    \label{fig:pkb}
\end{figure*}

We implement an approximation to \textsc{pk},
\begin{equation}
    \widetilde{\textsc{pk}} = \left(\bigotimes_{i=1}^{3\eta_n} \textsc{pk}_i \right)\tilde{U}_{\mathrm{grad}},
\end{equation}
where $\tilde{U}_{\mathrm{grad}}$ is an approximation to $U_{\mathrm{grad}}$ and all other unitaries are exact. 
We synthesize $\tilde{U}_{\mathrm{grad}}$ by applying Hadamard gates followed by the appropriate $Z$ rotation to each of the $b_{\mathrm{grad}}$ qubits of the phase-gradient register~\cite{nam2020approximate}. Each $Z$ rotation is approximated to within error $\epsilon_{\mathrm{Z}}$ using a repeat-until-success (RUS) algorithm~\cite{bocharov_2015}. 

The resulting error in $\widetilde{\textsc{pk}}$ is 
\begin{align} \label{eq:eps_pk}
    \epsilon_{\textsc{pk}} & = D(\textsc{pk}|\hat{\Psi}\rangle |0\rangle, \widetilde{\textsc{pk}}|\hat{\Psi}\rangle |0\rangle) \\
    &=D(U_{\mathrm{grad}}|\hat{\Psi}\rangle |0\rangle, \tilde{U}_{\mathrm{grad}}|\hat{\Psi}\rangle |0 \rangle) \\
    &\leq  b_{\mathrm{grad}} \epsilon_{Z} .
\end{align}
This bound results from $\tilde{U}_{\mathrm{grad}}$ calling the RUS algorithm $b_{\mathrm{grad}}$ times, introducing an error $\epsilon_Z$ with each call. 
To compile each $Z$ rotation to within error $\epsilon_{\mathrm{PK}} /b_{\mathrm{grad}}$, the expected number of Toffoli gates needed is $(1.149( \log(b_{\mathrm{grad}}) + \log(\epsilon_{\mathrm{PK}}^{-1}) ) + 9.2)/4$, along with a single ancilla qubit~\cite{bocharov_2015}. Therefore, the cost to synthesize $\tilde{U}_{\mathrm{grad}}$ to within $\epsilon_{\mathrm{PK}}$ is 
\begin{align}
    C_{\anc}(\tilde{U}_{\mathrm{grad}}) &= b_{\mathrm{grad}} + 1 \\
    C_{\Toff}(\tilde{U}_{\mathrm{grad}}) &= b_{\mathrm{grad}}(1.149( \log(b_{\mathrm{grad}}) + \log(\epsilon_{\mathrm{PK}}^{-1}) ) + 9.2)/4.
\end{align}

Implementing $U_{z_i}$ can be done with $b_{\mathrm{grad}}$ \textsc{not} gates and no Toffoli gates, as well as
\begin{equation}
    C_{\anc}(U_{z_i}) = b_{\mathrm{grad}}
\end{equation}
ancillas to store the result. Implementing $U_{z_i}^{\dagger}$ uses the same resources.

We can multiply an $\bar{n}_{\mathrm{ISP}}$-bit number with a $b_{\mathrm{grad}}$-bit number using $2 \bar{n}_{\mathrm{ISP}} b_{\mathrm{grad}} - \bar{n}_{\mathrm{ISP}}$ \cite{su2021} Toffoli gates.  We require $b_{\mathrm{grad}} + \bar{n}_{\mathrm{ISP}} -1$ ancilla qubits to perform multiplication and another $b_{\mathrm{grad}} + \bar{n}_{\mathrm{ISP}}$ to store the result. Hence,
 \begin{align}
     C_{\anc}(\textsc{mult}) & = 2b_{\mathrm{grad}} + 2\bar{n}_{\mathrm{ISP}} -1\\
      C_{\Toff}(\textsc{mult}) & = 2 \bar{n}_{\mathrm{ISP}} b_{\mathrm{grad}} - \bar{n}_{\mathrm{ISP}}.
 \end{align}
 The resources needed to perform $\textsc{mult}^{\dagger}$ are the same as for $\textsc{mult}$.
 
Modular addition can be performed using $b_{\mathrm{grad}}$ Toffoli gates and ancillas \cite{gidney2018},
  \begin{align}
     C_{\anc}(\textsc{add}) & = b_{\mathrm{grad}} \\
      C_{\Toff}(\textsc{add}) & = b_{\mathrm{grad}}.
 \end{align}

To implement $\widetilde{\textsc{pk}}$, $|\psi_{\mathrm{grad}}\rangle$ is prepared once and reused for each $\textsc{pk}_i$. PK is performed on $3\eta_n$ registers, each time calling $U_{z_i}$, \textsc{mult}, \textsc{add}, $\textsc{mult}^{\dagger}$, and $U_{z_i}^{\dagger}$. The cost to do so is
 \begin{align}
     C_{\anc}(\widetilde{\textsc{pk}}) &= C_{\anc}(\tilde{U}_{\mathrm{grad}}) + C_{\anc}(U_{z_i}) + C_{\anc}(\textsc{mult})  \\
     & = 4 b_{\mathrm{grad}} + 2\bar{n}_{\mathrm{ISP}} -1 \\
     C_{\Toff}(\widetilde{\textsc{pk}}) &=   C_{\Toff}(\tilde{U}_{\mathrm{grad}}) + 3\eta_n(2C_{\Toff}(\textsc{mult}) + C_{\Toff}(\textsc{add}))  \\
     & = 3\eta_n (4 \bar{n}_{\mathrm{ISP}} b_{\mathrm{grad}} + b_{\mathrm{grad}} - 2\bar{n}_{\mathrm{ISP}}) +  b_{\mathrm{grad}} (1.149 (\log(b_{\mathrm{grad}}\epsilon_{\mathrm{PK}}^{-1})) + 9.2)/4,
 \end{align}
where the ancilla qubits used during $\textsc{mult}$ are reused for $\textsc{add}$.

\section{ISP error bound}\label{app:isp_error_breakdown}
In this appendix, we derive a bound on the total error $\epsilon_{\mathrm{ISP}}$ from ISP. We begin by discussing separable states, where the main challenge is to bound the error from preparing the initial nuclear state. Afterwards, we bound the error from preparing non-separable initial states.  

\subsection{Separable initial states}
\label{app:isp_error_sep}
Here, we bound the error from preparing separable initial states.
  Let $\tilde{U}_{\mathrm{ISP}}^{(\mathrm{sep})}$ be the approximation to $U_{\mathrm{ISP}}^{(\mathrm{sep})}$ implemented on the quantum computer such that $\tilde{U}_{\mathrm{ISP}}^{(\mathrm{sep})}|0\rangle = |\tilde{\hat{\Psi}}\rangle_{\mathrm{MPS}}$. The ISP error is
\begin{align}
    \epsilon_{\mathrm{ISP}}^{(\mathrm{sep})} = D(\Pi_{\mathrm{int}}' U_{\mathrm{ISP}}^{(\mathrm{sep})}|0\rangle, \tilde{U}_{\mathrm{ISP}}^{(\mathrm{sep})}|0\rangle),
\end{align}
where $\Pi_{\mathrm{int}}' = \Pi_{\mathrm{int}}/\mathcal{N}_{\mathrm{int}}$, and $\Pi_{\mathrm{int}}$ and $\mathcal{N}_{\mathrm{int}}$ are defined in \cref{subsubsec:ISP_trimming}.

We can bound $\epsilon_{\mathrm{ISP}}^{(\mathrm{sep})}$ in terms of the electronic and nuclear ISP errors. To see this, we write
\begin{align}
    U_{\mathrm{ISP}}^{(\mathrm{sep})} &=U_{\mathrm{ISP}}^{(e)}U_{\mathrm{ISP}}^{(n)} \quad \mathrm{where} \quad  U_{\mathrm{ISP}}^{(e)}|0\rangle = |\hat{\Psi}^{(e)}\rangle \quad\mathrm{and}  \quad U_{\mathrm{ISP}}^{(n)}|0\rangle = |\hat{\Psi}^{(n)}\rangle\\
    \tilde{U}_{\mathrm{ISP}}^{(\mathrm{sep})} &=\tilde{U}_{\mathrm{ISP}}^{(e)} \tilde{U}_{\mathrm{ISP}}^{(n)} \quad \mathrm{where} \quad  \tilde{U}_{\mathrm{ISP}}^{(e)}|0\rangle = |\tilde{\hat{\Psi}}^{(e)}\rangle_{\mathrm{MPS}} \quad\mathrm{and}  \quad \tilde{U}_{\mathrm{ISP}}^{(n)}|0\rangle = |\tilde{\hat{\Psi}}^{(n)}\rangle_{\mathrm{MPS}}.
\end{align}
Then, by the triangle inequality, 
\begin{align}
    \epsilon_{\mathrm{ISP}}^{(\mathrm{sep})} &\leq D(U_{\mathrm{ISP}}^{(e)}|0\rangle, \tilde{U}_{\mathrm{ISP}}^{(e)}|0\rangle) + D(\Pi_{\mathrm{int}}^{(n)} U_{\mathrm{ISP}}^{(n)}|0\rangle, \tilde{U}_{\mathrm{ISP}}^{(n)}|0\rangle) \\
    &= \epsilon_{\mathrm{ISP}}^{(e)} + \epsilon_{\mathrm{ISP}}^{(n)},
\end{align}
where $\Pi_{\mathrm{int}}^{(n)} = (\sum_{\mathbf{n} \in \mathcal{B}_{\mathrm{int}}} |\mathbf{n}\rangle \langle \mathbf{n|})/\mathcal{N}_{\mathrm{int}}$.

We now derive a bound on $\epsilon_{\mathrm{ISP}}^{(n)}$ because $\epsilon_{\mathrm{ISP}}^{(e)}$ is bounded in \cref{eq:elec_error_bound}. To begin, we expand the ISP unitaries,
\begin{align}
    U_{\mathrm{ISP}}^{(n)} &= \textsc{tc2sm} \cdot \textsc{pk} \cdot U_{\mathbf{A}^{- \top}} \cdot W^{(n)} \cdot \textsc{onb2smb} \cdot \mathrm{SoSlat}^{(n)} \cdot \textsc{asp}^{(n)} \label{eq:U_isp_n} \\
    \tilde{U}_{\mathrm{ISP}}^{(n)} &= \textsc{tc2sm} \cdot\widetilde{\textsc{pk}} \cdot \tilde{U}_{\mathbf{A}^{- \top}} \cdot \widetilde{W}^{(n)} \cdot \textsc{onb2smb} \cdot \mathrm{SoSlat}^{(n)} \cdot \widetilde{\textsc{asp}}^{(n)} \label{eq:U_isp_n_approx} ,
\end{align}
where
\begin{equation}
    W^{(n)} = \bigotimes_{i=1}^{3\eta_n}W_i^{(n)}, \quad \mathrm{where} \quad
    W_i^{(n)} = \prod_{\mu=0}^{N_{\mathrm{SMB}} -1} I - 2|w_{\mu}^{(n)}\rangle \langle w_{\mu}^{(n)}|_i, \quad\, \mathrm{and} \quad |w_{\mu}^{(n)}\rangle_i = \frac{1}{\sqrt{2}}\left(|1\rangle|\mu\rangle - |0\rangle|\hatsmQ\rangle \right).
\end{equation}
All other operators in \cref{eq:U_isp_n,eq:U_isp_n_approx} are defined in \cref{sect:isp_nuc}. 

Our first step is to isolate the trimming error $\epsilon_{\mathrm{trim}}$. Using the triangle inequality, we have
\begin{align}
    \epsilon_{\mathrm{ISP}}^{(n)}& \leq   D \big( \textsc{pk} \cdot  U_{\mathbf{A}^{- \top}} W^{(n)} \textsc{onb2smb} \cdot \mathrm{SoSlat}^{(n)} \cdot \textsc{asp}^{(n)}  |0\rangle , \nonumber \\
     & \quad\,\quad\,\quad\,\quad\,\quad\,\quad\,\quad\,\quad\,\quad\,\quad\,  \widetilde{\textsc{pk}} \cdot \tilde{U}_{\mathbf{A}^{- \top}} \widetilde{W}^{(n)} \textsc{onb2smb} \cdot \mathrm{SoSlat}^{(n)} \cdot \widetilde{\textsc{asp}}^{(n)}|0\rangle  \big)  \nonumber + \epsilon_{\mathrm{trim}},
\end{align}
where
\begin{align}
    \epsilon_{\mathrm{trim}} =  D \big(\Pi_{\mathrm{int}}^{(n)} U_{\mathrm{ISP}}^{(n)}|0\rangle, U_{\mathrm{ISP}}^{(n)}|0\rangle\big),
\end{align}
and we used the fact that $\textsc{tc2sm}$ is exact and does not contribute an error.

Next, we isolate the error resulting from $\widetilde{\textsc{asp}}^{(n)}$,
\begin{equation}
    \epsilon_{\textsc{asp}}^{(n)} = D(\textsc{asp}^{(n)}|0\rangle, \widetilde{\textsc{asp}}^{(n)}|0\rangle  ).
\end{equation}
Using the triangle inequality, 
\begin{align}
    \epsilon_{\mathrm{ISP}}^{(n)}& \leq   D \big(\textsc{pk} \cdot  U_{\mathbf{A}^{- \top}} W^{(n)} \textsc{onb2smb} \cdot \mathrm{SoSlat}^{(n)} \cdot \textsc{asp}^{(n)}  |0\rangle , \nonumber \\
     & \quad\,\quad\,\quad\,\quad\,\quad\,\quad\,\quad\,\quad\,\quad\,\quad\, \widetilde{\textsc{pk}} \cdot \tilde{U}_{\mathbf{A}^{- \top}} \widetilde{W}^{(n)} \textsc{onb2smb} \cdot \mathrm{SoSlat}^{(n)} \cdot \textsc{asp}^{(n)}|0\rangle  \big)  \nonumber \\
      & =  D \big( \textsc{pk} \cdot  U_{\mathbf{A}^{- \top}} W^{(n)} |\Psi^{(n)}_{\bar{\mathbf{Q}}}\rangle_{\mathrm{SMB}}, \,\, \widetilde{\textsc{pk}} \cdot \tilde{U}_{\mathbf{A}^{- \top}} \widetilde{W}^{(n)} |\Psi^{(n)}_{\bar{\mathbf{Q}}}\rangle_{\mathrm{SMB}} \big)  + \epsilon_{\textsc{asp}}^{(n)} + \epsilon_{\mathrm{trim}},
\end{align}
where $ |\Psi^{(n)}_{\bar{\mathbf{Q}}}\rangle_{\mathrm{SMB}} = \textsc{onb2smb}  \cdot \mathrm{SoSlat}^{(n)} \cdot  \textsc{asp}^{(n)} |0\rangle$.

Next, we isolate the errors from preparing the nuclear state in a PWB. Let $X^{(n)}$ be the unitary 
\begin{equation}
    X^{(n)} = \bigotimes_{i=1}^{3\eta_n}X_i^{(n)}, \quad \mathrm{where} \quad
    X_i^{(n)} = \prod_{\mu=0}^{N_{\mathrm{SMB}} -1} I - 2|x_{\mu}^{(n)}\rangle \langle x_{\mu}^{(n)}|_i, \quad\, \mathrm{and} \quad |x_{\mu}^{(n)}\rangle_i = \frac{1}{\sqrt{2}}\left(|1\rangle|\mu\rangle - |0\rangle|\hatsmQAprx\rangle \right).
\end{equation}
We also define $X^{(n,\mathrm{poly})}$ analogously to $X^{(n)}$, except $|\hatsmQAprx\rangle$ is replaced with $|\hatsmQAprx^{(\mathrm{poly})}\rangle$.
The difference between $W^{(n)}$, $X^{(n)}$, $X^{(n,\mathrm{poly})}$, and $\widetilde{W}^{(n)}$ is that they prepare nuclear states in terms of increasingly more approximate SMs: $|\hatsmQ\rangle$, $|\hatsmQAprx\rangle$, $|\hatsmQAprx^{(\mathrm{poly})}\rangle$, and $|\hatsmQAprx\rangle_{\mathrm{MPS}}$, respectively. Using the triangle inequality,
\begin{align}
    \epsilon_{\mathrm{ISP}}^{(n)} & \leq  D\big(\textsc{pk} \cdot U_\mathbf{A^{- \top}} W^{(n)} |\Psi^{(n)}_{\bar{\mathbf{Q}}}\rangle_{\mathrm{SMB}},\, \textsc{pk} \cdot U_{\mathbf{A}^{- \top}} X^{(n)} |\Psi^{(n)}_{\bar{\mathbf{Q}}}\rangle_{\mathrm{SMB}}\big) + \nonumber \\
    & \qquad{} + D \big(\textsc{pk} \cdot U_{\mathbf{A}^{- \top}} X^{(n)} |\Psi^{(n)}_{\bar{\mathbf{Q}}}\rangle_{\mathrm{SMB}}, \, \widetilde{\textsc{pk}} \cdot \tilde{U}_{\mathbf{A}^{- \top}} \widetilde{W}^{(n)} |\Psi^{(n)}_{\bar{\mathbf{Q}}} \rangle _{\mathrm{SMB}}\big) + \epsilon_{\textsc{asp}}^{(n)} + \epsilon_{\mathrm{trim}} \\
    & = \epsilon^{(n,t)} + T + \epsilon_{\textsc{asp}}^{(n)} + \epsilon_{\mathrm{trim}} \label{eq:init_n_bound_1},
\end{align}
where
\begin{align}
    \epsilon^{(n,t)} &= D(X^{(n)}|\Psi^{(n)}_{\bar{\mathbf{Q}}}\rangle_{\mathrm{SMB}},\, W^{(n)}|\Psi^{(n)}_{\bar{\mathbf{Q}}}\rangle_{\mathrm{SMB}} )\\
    T &= D \big( \textsc{pk} \cdot U_{\mathbf{A}^{- \top}} X^{(n)} |\Psi^{(n)}_{\bar{\mathbf{Q}}} \rangle_{\mathrm{SMB}}, \widetilde{\textsc{pk}} \cdot \tilde{U}_{\mathbf{A}^{- \top}}  \widetilde{W}^{(n)} |\Psi^{(n)}_{\bar{\mathbf{Q}}}\rangle_{\mathrm{SMB}}\big).
\end{align}

To bound $\epsilon^{(n,t)}$, we write
\begin{align}
    \epsilon^{(n,t)} & \leq  \big\Vert X^{(n)}- W^{(n)} \big\Vert_{\infty} \\
    & \leq \sum_{i=1}^{3\eta_n} \big\Vert X^{(n)}_i  -  W^{(n)}_i \big\Vert_{\infty}\\
    &\leq 2 \sum_{i=1}^{3\eta_n}\sum_{\mu = 0}^{N_{\mathrm{SMB}}-1} \left\Vert |x_{\mu}^{(n)}\rangle \langle x_{\mu}^{(n)} |_i - |w_{\mu}^{(n)}\rangle \langle w_{\mu}^{(n)} |_i  \right\Vert_{\infty} \\
    & \leq 4 \sum_{i=1}^{3\eta_n}\sum_{\mu = 0}^{N_{\mathrm{SMB}}-1} D\big(|x_{\mu}^{(n)} \rangle_i, |w_{\mu}^{(n)} \rangle_i\big),
\end{align}
where the last line follows because $\lVert x-y \rVert_\infty \le 2 D(x,y)$. We evaluate the last line using the definitions of $|x_{\mu}^{(n)}\rangle_i$ and $|w_{\mu}^{(n)}\rangle_i$,
\begin{equation}
    D\left(|x_{\mu}^{(n)}\rangle_i, |w_{\mu}^{(n)}\rangle_i \right) = \sqrt{1 - \left|  {_i}\langle x_{\mu}^{(n)}|w_{\mu}^{(n)}\rangle_i \right|^2} = \frac{1}{\sqrt{2}} \sqrt{1 - \left|  \langle \tilde{\hat{\Phi}}_{i \mu, \bar{\mathbf{Q}}}^{(n)}|\hat{\Phi}^{(n)}_{i \mu, \bar{\mathbf{Q}}}\rangle \right|^2} = \frac{1}{\sqrt{2}}D\left(|\tilde{\hat{\Phi}}^{(n)}_{i \mu, \bar{\mathbf{Q}}}\rangle, |\hat{\Phi}^{(n)}_{i \mu, \bar{\mathbf{Q}}}\rangle \right),
\end{equation}
giving
\begin{align}
    \epsilon^{(n,t)} &\leq 2^{3/2}\sum_{i=1}^{3\eta_n} \sum_{\mu=0}^{N_{\mathrm{SMB}}-1} D\left(|\hatsmQAprx\rangle ,|\hatsm\rangle \right)\\
     &= 2^{3/2}\sum_{i=1}^{3\eta_n} \sum_{\mu=0}^{N_{\mathrm{SMB}}-1} \epsilon^{(n,t)}_{i \mu} 
\end{align}
where $ \epsilon^{(n,t)}_{i \mu}$ is the truncation error defined in \cref{eq:nuc_trunc_error}. 

The remaining term in \cref{eq:init_n_bound_1} is bounded by 
\begin{multline}
    T \leq  D\big(\textsc{pk} \cdot U_{\mathbf{A}^{- \top}}  X^{(n)} |\Psi^{(n)}_{\bar{\mathbf{Q}}}\rangle_{\mathrm{SMB}}, \,\widetilde{\textsc{pk}} \cdot \tilde{U}_{\mathbf{A}^{- \top}}  X^{(n)}  |\Psi^{(n)}_{\bar{\mathbf{Q}}}\rangle_{\mathrm{SMB}}\big) \\
    +  D \big(\widetilde{\textsc{pk}} \cdot \tilde{U}_{\mathbf{A}^{- \top}}  X^{(n)} |\Psi^{(n)}_{\bar{\mathbf{Q}}}\rangle_{\mathrm{SMB}}, \widetilde{\textsc{pk}} \cdot \tilde{U}_{\mathbf{A}^{- \top}}  \widetilde{W}^{(n)}  |\Psi^{(n)}_{\bar{\mathbf{Q}}}\rangle_{\mathrm{SMB}} \big) \label{eq:big_error_expresion}.
\end{multline}
The second term in \cref{eq:big_error_expresion} is bounded by 
\begin{align}
    D \big(\widetilde{\textsc{pk}} \cdot \tilde{U}_{\mathbf{A}^{- \top}}  X^{(n)} |\Psi^{(n)}_{\bar{\mathbf{Q}}}\rangle_{\mathrm{SMB}}, \widetilde{\textsc{pk}} \cdot \tilde{U}_{\mathbf{A}^{- \top}}  \widetilde{W}^{(n)}  |\Psi^{(n)}_{\bar{\mathbf{Q}}}\rangle_{\mathrm{SMB}} \big)
    & \leq  \epsilon^{(n,p)} + \epsilon^{(n,q)},
\end{align}
where 
\begin{align}
    \epsilon^{(n,p)} &= D\big( X^{(n)}|\Psi^{(n)}_{\bar{\mathbf{Q}}}\rangle_{\mathrm{SMB}},\, X^{(n, \mathrm{poly})}|\Psi^{(n)}_{\bar{\mathbf{Q}}}\rangle_{\mathrm{SMB}} \big), \\
    \epsilon^{(n,q)} &= D\big( X^{(n, \mathrm{poly})}|\Psi^{(n)}_{\bar{\mathbf{Q}}}\rangle_{\mathrm{SMB}},\, \widetilde{W}^{(n)}|\Psi^{(n)}_{\bar{\mathbf{Q}}}\rangle_{\mathrm{SMB}}  \big).
\end{align}

We can bound $\epsilon^{(n,p)}$ and  $\epsilon^{(n,q)}$ using the same approach that we used to bound $\epsilon^{(n,t)}$, finding
\begin{align}
    \epsilon^{(n,p)} &\leq 2^{3/2}\sum_{i=1}^{3\eta_n} \sum_{\mu=0}^{N_{\mathrm{SMB}}-1} \epsilon^{(n,p)}_{i \mu}\quad \text{and} \quad
    \epsilon^{(n,q)} \leq 2^{3/2}\sum_{i=1}^{3\eta_n} \sum_{\mu=0}^{N_{\mathrm{SMB}}-1} \epsilon^{(n,q)}_{i \mu},
\end{align}
where $\epsilon^{(n,p)}_{i \mu}$ is the polynomial approximation error in \cref{eq:nuc_poly_approx} and $\epsilon^{(n,q)}_{i \mu} = D\big(|\hatsmQ\rangle_{\mathrm{MPS}} ,|\hatsmQAprx\rangle_{\mathrm{MPS}} \big)$. As was the case for $\epsilon^{(e,q)}_a$, each $\epsilon^{(n,q)}_{i \mu} $ is bounded by 
\begin{equation}
    \epsilon^{(n,q)}_{i \mu} < 2^{\frac{7}{2}-b_{\mathrm{rot}}
    }\sum_{i=1}^{n_{\mathrm{ISP}}}m_{i\mu j}^{(n)}\log(2 \bar{m}_{i \mu j}^{(n)} ).
\end{equation}

Collecting all error terms into one expression yields
\begin{equation}
    \epsilon_{\mathrm{ISP}}^{(n)} \leq \epsilon_{\textsc{asp}}^{(n)} + \epsilon_{\mathrm{trim}} + \epsilon^{(n,q)} + \epsilon^{(n,c)} + 
    D\big(\textsc{pk} \cdot U_{\mathbf{A}^{- \top}}  |\tilde{\hat{\Psi}}^{(n)}_{\bar{\mathbf{Q}}}\rangle , \, \widetilde{\textsc{pk}} \cdot \tilde{U}_{\mathbf{A}^{- \top}}  |\tilde{\hat{\Psi}}^{(n)}_{\bar{\mathbf{Q}}}\rangle  \big)
\end{equation}
where $\epsilon^{(n,c)} = \epsilon^{(n,t)} + \epsilon^{(n,p)}$, and $|\tilde{\hat{\Psi}}^{(n)}_{\bar{\mathbf{Q}}}\rangle = X^{(n)}|\Psi_{\bar{\mathbf{Q}}}^{(n)} \rangle_{\mathrm{SMB}}$. The remaining sources of error come from transforming from normal coordinates to Cartesian coordinates and trimming the grid of final state. Assuming that the LCT algorithm is used, from \cref{sect:LCT}, we have
\begin{equation}
    U_{\mathbf{A}^{- \top}} =  U_{\bar{\mathbf{O}}^{\top}} U_{\mathbf{S}_{\mathrm{L}}} \quad \mathrm{and}\quad \tilde{U}_{\mathbf{A}^{- \top}} =  \tilde{U}_{\bar{\mathbf{O}}^{\top}} \tilde{U}_{\mathbf{S}_{\mathrm{L}}},
\end{equation}
so that
\begin{align}
    D\big(\textsc{pk} \cdot U_{\bar{\mathbf{O}}^{\top}} U_{\mathbf{S}_{\mathrm{L}}}|\tilde{\hat{\Psi}}^{(n)}_{\bar{\mathbf{Q}}}\rangle,\, \widetilde{\textsc{pk}} \cdot \tilde{U}_{\bar{\mathbf{O}}^{\top}} \tilde{U}_{\mathbf{S}_{\mathrm{L}}}|\tilde{\hat{\Psi}}^{(n)}_{\bar{\mathbf{Q}}}\rangle\big) \leq \epsilon_{\mathrm{shear}} + \epsilon_{\mathrm{ortho}} + \epsilon_{\textsc{pk}},
\end{align}
where we used the triangle inequality to decompose the error into
\begin{align} 
    \epsilon_{\mathrm{shear}} &= D \big( U_{\mathbf{S}_{\mathrm{L}}}|\tilde{\hat{\Psi}}^{(n)}_{\bar{\mathbf{Q}}}\rangle , \tilde{U}_{\mathbf{S}_{\mathrm{L}}}|\tilde{\hat{\Psi}}^{(n)}_{\bar{\mathbf{Q}}}\rangle \big), \label{eq:eps_shear_nuc} \\
    \epsilon_{\mathrm{ortho}} &= D \big( U_{\bar{\mathbf{O}}^{\top}} U_{\mathbf{S}_{\mathrm{L}}}|\tilde{\hat{\Psi}}^{(n)}_{\bar{\mathbf{Q}}}\rangle ,\,  \tilde{U}_{\bar{\mathbf{O}}^{\top}} U_{\mathbf{S}_{\mathrm{L}}}|\tilde{\hat{\Psi}}^{(n)}_{\bar{\mathbf{Q}}}\rangle \big),\label{eq:eps_ortho_nuc} \\
    \epsilon_{\textsc{pk}} & =D \big(\textsc{pk} \cdot U_{\bar{\mathbf{O}}^{\top}}U_{\mathbf{S}_{\mathrm{L}}} 
    |\tilde{\hat{\Psi}}^{(n)}_{\bar{\mathbf{Q}}}\rangle,\,\widetilde{\textsc{pk}} \cdot  U_{\bar{\mathbf{O}}^{\top}} U_{\mathbf{S}_{\mathrm{L}}} |\tilde{\hat{\Psi}}^{(n)}_{\bar{\mathbf{Q}}}\rangle \big), \label{eq:eps_pk_nuc}
\end{align}

In \cref{app:lct_ssct_error}, we showed that, for an orthogonal matrix decomposed into a sequence of 2D shears, 
\begin{equation}
    \epsilon_{\mathrm{ortho}} \leq \sum_{P=2}^{\beta+1} \epsilon^{(\mathrm{2D-shear})}_P,
\end{equation}
where $\epsilon^{(\mathrm{2D-shear})}_P$ is the error from approximating a 2D shear defined in \cref{eq:2d_shear_eps_def}.

Overall, the bound on  $\epsilon_{\mathrm{ISP}}^{(n)}$ in terms of the fundamental errors is
\begin{equation}
    \epsilon_{\mathrm{ISP}}^{(n)} \leq \epsilon_{\textsc{asp}}^{(n)} + 2^{3/2}\sum_{i=1}^{3\eta_n} \sum_{\mu=0}^{N_{\mathrm{SMB}}-1} (\epsilon^{(n,c)}_{i \mu} + \epsilon^{(n,q)}_{i \mu}) + \epsilon_{\mathrm{shear}} + \sum_{P=2}^{\beta+1} \epsilon^{(\mathrm{2D-shear})}_P + \epsilon_{\textsc{pk}} + \epsilon_{\mathrm{trim}} ,
\end{equation}
where $\epsilon^{(n,c)}_{i \mu} = \epsilon^{(n,t)}_{i \mu} + \epsilon^{(n,p)}_{i \mu}$.

\subsection{Non-separable initial states}
In this section, we bound the ISP error from preparing non-separable initial states,
\begin{equation}
    \epsilon_{\mathrm{ISP}} = D\big( \Pi_{\mathrm{int}}' U_{\mathrm{ISP}}|0\rangle ,\widetilde{U}_{\mathrm{ISP}} |0\rangle\big),
\end{equation}
where $\widetilde{U}_{\mathrm{ISP}}$ is defined in \cref{eq:non_sep_uinit_approx}, and 
\begin{equation} 
    U_{\mathrm{ISP}} = (U_{\mathrm{ISP}}^{(e)'} \otimes U_{\mathrm{ISP}}^{(n)'}) \big(\mathrm{SoSlat}^{(en)}  \textsc{asp}^{(en)}\big),
\end{equation}
with
\begin{align}
    U_{\mathrm{ISP}}^{(e)'} &= (W^{(e)})^{\otimes 3\eta_e} \textsc{asym}\; \textsc{onb2mob}\\
    U_{\mathrm{ISP}}^{(n)'} & = \textsc{tc2sm}_{\bar{n}_{\mathrm{ISP}}}^{\otimes 3\eta_n}\, \textsc{pk} \,U_{\mathbf{A}^{-\top}} \, W^{(n)} \textsc{onb2smb}.
\end{align}

We begin by using the triangle inequality to isolate the trimming error $\epsilon_{\mathrm{trim}}$,
\begin{align}
    \epsilon_{\mathrm{ISP}} \leq & D\big(  \big(U_{\mathrm{ISP}}^{(e)'} \otimes U_{\mathrm{ISP}}^{(n)'}\big) \mathrm{SoSlat}^{(en)} \cdot \textsc{asp}^{(en)} |0\rangle ,\big(\widetilde{U}_{\mathrm{ISP}}^{(e)'} \otimes \widetilde{U}_{\mathrm{ISP}}^{(n)'} \big) \mathrm{SoSlat}^{(en)} \cdot \widetilde{\textsc{asp}}^{(en)}  |0\rangle\big) + \epsilon_{\textsc{trim}} 
\end{align}
where
\begin{equation}
    \epsilon_{\mathrm{trim}} = D\big(\Pi_{\mathrm{int}}' \big(U_{\mathrm{ISP}}^{(e)'} \otimes U_{\mathrm{ISP}}^{(n)'}\big) \mathrm{SoSlat}^{(en)} \cdot \textsc{asp}^{(en)} |0\rangle, \big(U_{\mathrm{ISP}}^{(e)'} \otimes U_{\mathrm{ISP}}^{(n)'}\big) \mathrm{SoSlat}^{(en)} \cdot \textsc{asp}^{(en)} |0\rangle \big).
\end{equation}

Next, we isolate the error from $\widetilde{\textsc{asp}}^{(en)}$,
\begin{equation}
    \epsilon_{\mathrm{ASP}}^{(e n)} = D\big(\textsc{asp}^{(e n)}|0\rangle, \widetilde{\textsc{asp}}^{(e n)}|0\rangle\big).
\end{equation}
Using the triangle inequality,
\begin{align}
    \epsilon_{\mathrm{ISP}} 
    \leq & D\big(  \big(U_{\mathrm{ISP}}^{(e)'} \otimes U_{\mathrm{ISP}}^{(n)'}\big) \mathrm{SoSlat}^{(en)} \cdot \textsc{asp}^{(en)} |0\rangle ,\big(\widetilde{U}_{\mathrm{ISP}}^{(e)'} \otimes \widetilde{U}_{\mathrm{ISP}}^{(n)'} \big) \mathrm{SoSlat}^{(en)} \cdot \textsc{asp}^{(en)}  |0\rangle\big) + \epsilon_{\mathrm{trim}} +  \epsilon_{\textsc{asp}}^{(en)}.
\end{align}

To further isolate error terms, we let
\begin{align}
    |\Psi\rangle_{\textsc{mosmb}} &= (\textsc{asym} \otimes I^{(n)} )(\textsc{onb2mob} \otimes \textsc{onb2smb} ) \mathrm{SoSlat}^{(en)}\cdot \textsc{asp}^{(en)} |0\rangle\\
    \mathcal{A} & = I^{(e)} \otimes ( \textsc{pk} \cdot U_{\mathbf{A}^{- \top}} )  \\
    \tilde{\mathcal{A}} & =  I^{(e)} \otimes (\widetilde{\textsc{pk}} \cdot \tilde{U}_{\mathbf{A}^{- \top}} ).
\end{align}
With these definitions, 
\begin{align}
    \epsilon_{\mathrm{ISP}} &\leq  D\big(\mathcal{A} \big( (W^{(e)})^{\otimes \eta_e} \otimes W^{(n)} \big) |\Psi\rangle_{\textsc{mosmb}},\, \tilde{\mathcal{A}} \big( (\widetilde{W}^{(e)})^{\otimes \eta_e} \otimes \widetilde{W}^{(n)} \big) |\Psi\rangle_{\textsc{mosmb}} \big) + \epsilon_{\mathrm{ASP}}^{(e n)} \\ 
    &\leq    T_1 + T_2 + \epsilon_{\mathrm{ASP}}^{(e n)} + \epsilon_{\mathrm{trim}},
\end{align}
where 
\begin{align}
    T_1 & = D\big(\mathcal{A} \big( (W^{(e)})^{\otimes \eta_e} \otimes W^{(n)}  \big) |\Psi\rangle_{\textsc{mosmb}},\, \tilde{\mathcal{A}} \big( (W^{(e)})^{\otimes \eta_e} \otimes X^{(n)}  \big) |\Psi\rangle_{\textsc{mosmb}} \big) \\
    T_2 & = D\big(\widetilde{\mathcal{A}} \big( (W^{(e)})^{\otimes \eta_e} \otimes X^{(n)} \big) |\Psi\rangle_{\textsc{mosmb}},\, \tilde{\mathcal{A}} \big( (\widetilde{W}^{(e)})^{\otimes \eta_e} \otimes \widetilde{W}^{(n)} \big) |\Psi\rangle_{\textsc{mosmb}} \big).
\end{align}

Focusing on $T_2$, we isolate two contributions:
\begin{align}
    T_2 & \leq D\big( \big(W^{(e)})^{\otimes \eta_e} \otimes X^{(n)} \big) |\Psi\rangle_{\textsc{mosmb}},\,  \big( (\widetilde{W}^{(e)})^{\otimes \eta_e} \otimes \widetilde{W}^{(n)}  \big) |\Psi\rangle_{\textsc{mosmb}} \big) \\
    & \leq  T_2^{(e)} + T_{2}^{(n)},
\end{align}
where
\begin{align}
    T_2^{(e)} & = D\big( \big(W^{(e)})^{\otimes \eta_e} \otimes X^{(n)} \big) |\Psi\rangle_{\textsc{mosmb}},\,  \big( (\widetilde{W}^{(e)})^{\otimes \eta_e} \otimes X^{(n)}  \big) |\Psi\rangle_{\textsc{mosmb}} \big) \\
    T_2^{(n)} & = D\big( \big(\widetilde{W}^{(e)})^{\otimes \eta_e} \otimes X^{(n)}  \big) |\Psi\rangle_{\textsc{mosmb}},\,  \big( (\widetilde{W}^{(e)})^{\otimes \eta_e} \otimes \widetilde{W}^{(n)} \big) |\Psi\rangle_{\textsc{mosmb}} \big).
\end{align}
To bound $T_2^{(e)}$, we recognize that \cite{huggins2024}
\begin{align}
     T_2^{(e)} &\leq \big\Vert (W^{(e)})^{\otimes \eta_e} \otimes X^{(n)}   - (\widetilde{W}^{(e)})^{\otimes \eta_e} \otimes X^{(n)} \big\Vert_{\infty} \\
      &\leq \big\Vert (W^{(e)})^{\otimes \eta_e}  - (\widetilde{W}^{(e)})^{\otimes \eta_e} \big\Vert_{\infty}\\
      & \leq 2^{3/2} \eta_e \sum_{a=1}^{N_{\mathrm{MOB}}} \epsilon_a^{(e)}.
\end{align}
For $T_{2}^{(n)}$, we write
\begin{equation}
    T_{2}^{(n)} \leq T_{2}^{(n,p)} + T_{2}^{(n,q)},
\end{equation}
where
\begin{align}
    T_2^{(n,p)} & = D\big( \big(\widetilde{W}^{(e)})^{\otimes \eta_e} \otimes X^{(n)}  \big) |\Psi\rangle_{\textsc{mosmb}},\,  \big( (\widetilde{W}^{(e)})^{\otimes \eta_e} \otimes X^{(n,\mathrm{poly})}  \big) |\Psi\rangle_{\textsc{mosmb}}\big)\\
     T_2^{(n,q)} & = D\big( \big(\widetilde{W}^{(e)})^{\otimes \eta_e} \otimes X^{(n,\mathrm{poly})}  \big) |\Psi\rangle_{\textsc{mosmb}},\,  \big( (\widetilde{W}^{(e)})^{\otimes \eta_e} \otimes \widetilde{W}^{(n)}  \big) |\Psi\rangle_{\textsc{mosmb}} \big).
\end{align}
Using the results derived in \cref{app:isp_error_sep},
\begin{align}
     T_2^{(n,p)} \leq & \big\Vert \big(\widetilde{W}^{(e)})^{\otimes \eta_e} \otimes X^{(n)}  \big)  -  \big( (\widetilde{W}^{(e)})^{\otimes \eta_e} \otimes X^{(n,\mathrm{poly})}  \big) \big\Vert_{\infty} \\
     \leq & \big\Vert   X^{(n)}   -  X^{(n,\mathrm{poly})} \big\Vert_{\infty} \\
     \leq & \epsilon^{(n,p)}\\
     T_2^{(n,q)} \leq & \big\Vert \big(\widetilde{W}^{(e)})^{\otimes \eta_e} \otimes X^{(n,\mathrm{poly})} \big)  -  \big( (\widetilde{W}^{(e)})^{\otimes \eta_e} \otimes \widetilde{W}^{(n)}  \big) \big\Vert_{\infty} \\
     \leq & \big\Vert   X^{(n,\mathrm{poly})}   -  \widetilde{W}^{(n)} \big\Vert_{\infty} \\
     \leq & \epsilon^{(n,q)}.
\end{align}

So far, we have shown that $\epsilon_{\mathrm{ISP}}$ is upper bounded by 
\begin{equation}
     \epsilon_{\mathrm{ISP}} \leq T_1 + \epsilon^{(e)} + \epsilon^{(n,p)} + \epsilon^{(n,q)} + \epsilon_{\mathrm{ASP}}^{(e n)} + \epsilon_{\mathrm{trim}},
\end{equation}
where $\epsilon^{(e)} = 2^{3/2} \eta_e \sum_{a=1}^{N_{\mathrm{MOB}}} \epsilon_a^{(e)}$. To bound $T_1$, we write
\begin{equation}
    T_1 \leq T^{(n,t)}_1 + T^{(\mathcal{A})}_1,
\end{equation}
where
\begin{align}
     T_1^{(n,t)} & = D\big(\mathcal{A} \big( (W^{(e)})^{\otimes \eta_e} \otimes W^{(n)} \big) |\Psi\rangle_{\textsc{mosmb}},\, \mathcal{A} \big( (W^{(e)})^{\otimes \eta_e} \otimes X^{(n)}  \big) |\Psi\rangle_{\textsc{mosmb}}\big)\\
     T_1^{(\mathcal{A})} & = D\big(\mathcal{A} \big( (W^{(e)})^{\otimes \eta_e} \otimes X^{(n)}  \big) |\Psi\rangle_{\textsc{mosmb}},\, \tilde{\mathcal{A}} \big( (W^{(e)})^{\otimes \eta_e} \otimes X^{(n)}  \big) |\Psi\rangle_{\textsc{mosmb}} \big).
\end{align}
For $T_1^{(n,t)}$, we have
\begin{align}
    T_1^{(n,t)} \leq & \big\Vert \big((W^{(e)})^{\otimes \eta_e} \otimes W^{(n)} \big)  -  \big( (W^{(e)})^{\otimes \eta_e} \otimes X^{(n)}  \big) \big\Vert_{\infty} \\
    \leq & \big\Vert W^{(n)} - X^{(n)} \big\Vert_{\infty} \\
    \leq & \epsilon^{(n,t)},
\end{align}
which uses the results from \cref{app:isp_error_sep}. For $T^{(\mathcal{A})}$, we write
\begingroup
\allowdisplaybreaks
\begin{align}
    T_{1}^{(\mathcal{A})} \leq & \big\Vert  (\mathcal{A}  - \tilde{\mathcal{A}})\big((W^{(e)})^{\otimes \eta_e} \otimes X^{(n)}  \big) \big) (\textsc{asym} \otimes I^{(n)} )(\textsc{onb2mob} \otimes \textsc{onb2smb} ) \sum_{I,J} C_{IJ}^{(en)}|I\rangle|J\rangle  \big\Vert_2 \\
    \leq & \sum_{I,J}\big\lvert C_{IJ}^{(en)} \big\rvert \cdot \big\Vert  (\mathcal{A}  - \tilde{\mathcal{A}})\big((W^{(e)})^{\otimes \eta_e} \otimes X^{(n)}  \big) \big) (\textsc{asym} \otimes I^{(n)} )(\textsc{onb2mob} \otimes \textsc{onb2smb} )|I\rangle|J\rangle   \big\Vert_2 \\
    \leq & \sum_{I,J}\big\lvert C_{IJ}^{(en)} \big\rvert \cdot \big\Vert  (\mathcal{A}'  - \tilde{\mathcal{A}}') X^{(n)} \textsc{onb2smb} |J\rangle  \big\Vert_2, 
\end{align}
\endgroup
where
\begin{align}
    \mathcal{A}' & = \textsc{pk} \cdot U_{\mathbf{A}^{- \top}}  \\
    \tilde{\mathcal{A}}' & = \widetilde{\textsc{pk}} \cdot \widetilde{U}_{\mathbf{A}^{- \top}} .
\end{align}
We can bound $T_1^{(\mathcal{A})}$ by defining analogous expressions to \cref{eq:eps_shear_nuc,eq:eps_ortho_nuc,eq:eps_pk_nuc}:
\begingroup
\allowdisplaybreaks
\begin{align} 
    \epsilon_{\mathrm{shear}}' &= \big\Vert U_{\mathbf{S}_{\mathrm{L}}}|\tilde{\hat{\Psi}}^{(n)}_{J,\bar{\mathbf{Q}}}\rangle - \widetilde{U}_{\mathbf{S}_{\mathrm{L}}}|\tilde{\hat{\Psi}}^{(n)}_{J, \bar{\mathbf{Q}}}\rangle \big\Vert_2,  \label{eq:eps_shear_nuc_prime}  \\
    \epsilon_{\mathrm{ortho}}' &= \big\Vert U_{\bar{\mathbf{O}}^{\top}} U_{\mathbf{S}_{\mathrm{L}}}|\tilde{\hat{\Psi}}^{(n)}_{J,\bar{\mathbf{Q}}}\rangle -  \widetilde{U}_{\bar{\mathbf{O}}^{\top}} U_{\mathbf{S}_{\mathrm{L}}}|\tilde{\hat{\Psi}}^{(n)}_{J,\bar{\mathbf{Q}}}\rangle \big\Vert_2,  \label{eq:eps_ortho_nuc_prime} \\
    \epsilon_{\textsc{pk}}' & = \big\Vert \textsc{pk} \cdot U_{\bar{\mathbf{O}}^{\top}}U_{\mathbf{S}_{\mathrm{L}}} 
    |\tilde{\hat{\Psi}}^{(n)}_{J,\bar{\mathbf{Q}}}\rangle - \widetilde{\textsc{pk}} \cdot  U_{\bar{\mathbf{O}}^{\top}} U_{\mathbf{S}_{\mathrm{L}}}|\tilde{\hat{\Psi}}^{(n)}_{J,\bar{\mathbf{Q}}}\rangle \big\Vert_2, \label{eq:eps_pk_nuc_prime},
\end{align}
\endgroup
where $|\tilde{\hat{\Psi}}^{(n)}_{J,\bar{\mathbf{Q}}}\rangle = X^{(n)} \textsc{onb2smb} |J\rangle$ and we assume the LCT algorithm is used. Then,
\begin{equation}
    T_1^{(\mathcal{A})} \leq \sum_{I,J}\big\lvert C_{IJ}^{(en)} \big\rvert \cdot (\epsilon_{\mathrm{shear}}' + \epsilon_{\mathrm{ortho}}' + \epsilon_{\textsc{pk}}' ).
\end{equation}
Crucially, all of the upper bounds that we derive for \cref{eq:eps_shear_nuc,eq:eps_ortho_nuc,eq:eps_pk_nuc} equally apply to \cref{eq:eps_shear_nuc_prime,eq:eps_ortho_nuc_prime,eq:eps_pk_nuc_prime}. This is because our first step in bounding \cref{eq:eps_shear_nuc,eq:eps_ortho_nuc,eq:eps_pk_nuc} is always to bound the trace distance in terms of the two-norm, then derive a bound on the two-norm. 

Finally, the error of non-separable ISP obeys 
\begin{multline}
    \epsilon_{\mathrm{ISP}} \leq \epsilon_{\textsc{asp}}^{(en)} + \epsilon_{\mathrm{trim}} + 2^{3/2} \eta_e \sum_{a=1}^{N_{\mathrm{MOB}}}\epsilon_{a}^{(e)}+  2^{3/2}\sum_{i=1}^{3\eta_n} \sum_{\mu=0}^{N_{\mathrm{SMB}}-1} \big(\epsilon^{(n,c)}_{i \mu} + \epsilon^{(n,q)}_{i \mu}\big) + {}\\
    \sum_{I,J}\big\lvert C_{IJ}^{(en)} \big\rvert \cdot \big( \epsilon_{\mathrm{shear}}' + \sum_{P=2}^{\beta+1} \epsilon'^{(\mathrm{2D-shear})}_P + \epsilon_{\textsc{pk}}' \big),
\end{multline}
where $\epsilon'^{(\mathrm{2D-shear})}_P = \big\Vert U_{\mathbf{F}_P}|\Xi_{P-1}\rangle -  \widetilde{U}_{\mathbf{F}_P}|\Xi_{P-1}\rangle  \big\Vert_2$
and $|\Xi_{P-1}\rangle$ is defined in \cref{eq:psi_mult_shear}.

\section{Coulomb-state preparation}\label{app:Coulomb_state_prep}
Here, we give details on the preparation of the charge state in \cref{eq:charge_state_nuc_elec}, which forms the backbone of the Coulomb state preparation protocol in \cref{subsec:time-evo_PREP_V}. Our approach to preparing the charge state (up to a junk register) mirrors the strategy in \cite{su2021} for preparing a state that only encodes the nuclear charges (see \cref{eq:su-charge-state}). 

To prepare the charge state 
\begin{align}
\ket{\zeta}=\frac{1}{\sqrt{\sum_{j=1}^{\eta}|\zeta_{j}|}}\sum_{j=1}^{\eta}\sqrt{|\zeta_{j}|}\ket{j},
\end{align}
we use one instance of $\QROM$ and an ancillary state 
\begin{align}
\ket{\textrm{unif}}_{\anc}=\UNIF(2\eta_{e},b_{r})\ket{0}_{\anc},
\end{align}
given by a uniform superposition over $2\eta_{e}$ basis states. For a charge neutral system, $2\eta_{e}=\sum_{j=1}^{\eta}|\zeta_{j}|$, giving 
\begin{align}
    \ket{\textrm{unif}}_{\anc}=\frac{1}{\sqrt{\sum_{j=1}^{\eta}|\zeta_{j}|}}\sum_{x=0}^{2\eta_{e}-1}\ket{x}_{\anc}.
\end{align}
We then partition the index set $\cX=\{1,2,\ldots,2\eta_{e}\}$ into subsets of consecutive integers $X_{j}\subset\cX$ for $j\in\{1,2,\ldots,\eta\}$, with $|X_{j}|=|\zeta_{j}|$, where $|X_{j}|$ denotes the cardinality of the set $X_{j}$. Thus, each subset of indices $X_{j}$ is associated with particle $j$ of electric charge $\zeta_{j}$. The state $\ket{\textrm{unif}}_{\anc}$ may then be rewritten as 
\begin{align}
\ket{\textrm{unif}}_{\anc}=\frac{1}{\sqrt{\sum_{i=1}^{\eta}|\zeta_{i}|}}\sum_{j=1}^{\eta}\sqrt{|\zeta_{j}|}\ket{X_{j}},
\end{align}
where 
\begin{align}
\ket{X_{j}}=\frac{1}{\sqrt{|\zeta_{j}|}}\sum_{x\in X_{j}}\ket{x}_{\anc}
\end{align}
is the uniform state over the $j$th index subset.
We then construct a $\QROM_{\anc,m}$ that performs unary iteration on the register of the ancillary state $\ket{\textrm{unif}}_{\anc}$ and performs fanouts onto the $n_{\eta}$-sized register $m$, such that for any index value $x\in X_{j}$,
\begin{align}\label{eq:QROM_Coulomb_def}
\QROM_{\anc,m}\left(\ket{x}_{\anc}\ket{0}_m\right)=\ket{x}_{\anc}\ket{j}_{m},
\end{align}
where register $m$ now contains the exact $n_{\eta}$-bit binary representation of index $j\in\{1,2,\ldots,\eta\}$.
We then apply this $\QROM$ onto 
$\ket{\textrm{unif}}_{\anc}\ket{0}_m$, giving
\begin{align}\QROM_{\anc,m}\left(\ket{\textrm{unif}}_{\anc}\ket{0}_m\right)=
\frac{1}{\sqrt{\sum_{i=1}^{\eta}|\zeta_{i}|}}\sum_{j=1}^{\eta}\sqrt{\zeta_{j}}\,\QROM_{\anc,m}\left(\ket{X_{j}}_{\anc}\ket{0}_{m}\right),
\end{align}
which, using \cref{eq:QROM_Coulomb_def}, yields 
\begin{align}\label{eq:charge_junk_state}
\ket{\zeta_{\junk}}_{m,\anc}=\frac{1}{\sqrt{\sum_{i=1}^{\eta}|\zeta_{i}|}}\sum_{j=1}^{\eta}\sqrt{\zeta_{j}}\,\ket{j}_{m}\ket{X_{j}}_{\anc},
\end{align}
where we have swapped the ordering of the registers. We have thus prepared $\ket{\zeta}$ on register $m$ entangled with some junk on register $\anc$, which is sufficient for our LCU-based block-encoding of $V$ \cite{babbush2018}.

\section{Ancilla costs for time evolution}\label{app:qubit_cost}
Here, we determine the number of ancilla qubits required to prepare the state
\begin{align}
\proxPREP_{H}\ket{0}=\ket{+}_{b}\otimes \ket{\proxL_{m}}_{c}\otimes  \ket{\cL_{w}}_{d} \otimes \ket{\cL_{r,s}}_{e,f} \ket{\tilde{\cL}_{\bnu}}_{g,h}\otimes \ket{\tilde{\cL}_{\zeta\zeta}}_{k,\ell,m}.
\end{align}
The numbers of ancilla qubits used to prepare and represent the component states in $\proxPREP_{H}\ket{0}$ are 
\begin{enumerate}
\item $C_{\anc}(\ket{\proxL_{\theta}}_{a})=1$ (ignoring the $n_{\theta}$ ancillas for the rotation synthesis, which we include in \cref{eq:n_grad}),
\item $C_{\anc}(\ket{+}_{b})=1$,
\item $C_{\anc}(\ket{\proxL_{m}}_{c})=4n_{\eta}+3\mu_{T}+1$,
\item $C_{\anc}(\ket{\cL_{w}}_{d})=4+2$, including the ancillas of the inequality test,
\item$C_{\anc}(\ket{\cL_{r,s}}_{e,f})=2n_{p}$,
\item $C_{\anc}(\ket{\proxL_{\bnu}}_{g,h})=3n_{p}^2+10n_{p}+5n_{\cM}+4n_{\cM}n_{p}+10$,
\item $C_{\anc}(\ket{\proxL_{\doublezeta}}_{k,\ell,m})=6\bits{2\eta_{e}}+3n_{\eta}+4$.
\end{enumerate}

We also need to provide a phase-gradient state that is catalytically used to synthesize the single-qubit rotation $R_{z}(\tilde{\theta})$ that prepares the weighting state $\ket{\proxL_{\theta}}_{a}$ as well as the $\tildegre+1$ single-qubit QSP rotations $R(\phi_{i},\gamma_{i})$. The number of ancilla qubits needed for storing this phase-gradient state is therefore
\begin{align}
n_{\textrm{grad}}= \max\left\{\bits{{\lambda_{H}}/{\epsilon_{\theta}}},  \bits{{1}/{\epsilon_{\rot}}}\right\}.
\label{eq:n_grad}
\end{align}

\subsection{Costing \texorpdfstring{$\cR_{0}^{(\cW)}$}{cR0W} } \label{app:reflection_cost}
To give the cost of $\cR_{0}^{(\cW)}$ in \cref{eq:approx_qubiterate}, we need to account for the size of the register that this reflection acts on. Since $\cR_{0}^{(\cW)}$ acts on the output of $\proxPREP_{H}^{\dagger}$, we only need to consider the qubits used within the routine $\proxPREP_{H}$ that are not erased via measurement-based uncomputation. The number $\out(\PREP_{H}^{\dagger})$ of output qubits then determines both the Toffoli and ancilla costs of $\cR_{0}^{(\cW)}$. We assume the reuse of ancillas from $U_{H}$ (see \cref{eq:block_encoding_U_H}) to perform $\cR_{0}^{(\cW)}$ because it is performed after $\proxPREP_{H}^{\dagger}$.

Hence, the registers that we need to anti-control on to implement $\cR_{0}^{(\cW)}$ are 
\begin{enumerate}
\item register $a$ of size $1$,
\item register $b$ of size $1$,
\item register $c$ of size $n_{\eta}$ and 1 uniform-state-preparation rotation qubit,
\item register $w$ of size $2$ and 1 ancilla being rotated in its preparation,
\item registers $r$ and $s$ of combined size $2n_{p}$ and 1 catalytically used ancilla qubit for the controlled Hadamard gates,
\item some registers from the momentum state, totaling $4n_p+n_{\cM}+3$ qubits,
\item the two uniform-state-preparation registers on the ancillas which are each of size $\bits{2\eta_{e}}$, along with 1 rotation qubits for each.
\end{enumerate}
Thus, the total number of ancillas that are output by $\PREP_{H}^{\dagger}$ and that $\cR_{0}^{(\cW)}$ is anti-controlled on is 
\begin{align}
\out(\PREP_{H}^{\dagger})=n_{\eta}+6n_{p}+n_{\cM}+2\bits{2\eta_{e}}+11.
\end{align}
We implement $\cR_{0}^{(\cW)}$ via a multi-controlled-$Z$ gate that is uncomputed via measurement-based uncomputation. Therefore, the ancilla and Toffoli costs are 
\begin{align}
C_{\anc}\big(\contr{\cR^{(\cW)}_{0}}\big)&=\out(\PREP_{H}^{\dagger})-2\\
C_{\Toff}\big(\contr{\cR^{(\cW)}_{0}}\big)&=\out(\PREP_{H}^{\dagger})-1.
\end{align}

\section{Selection between \texorpdfstring{$T$}{T} and \texorpdfstring{$V$}{V}}\label{app:SEL_H_implement}

Our implementation of $\SEL_{H}$ is an adaptation of the the approach in \cite{su2021}. In \cref{subsec:SEL_H_implement}, we stated that the choice of whether to apply $\SEL_{T}$ or $\SEL_{V}$ is conditioned on the value of qubits flagging success or failure of the state-preparation routines $\ket{\proxL_{\bnu}}=\widetilde{\PREP}_{\boldsymbol{\nu}}\ket{0}$ and $\ket{\proxL_{\doublezeta}}=\widetilde{\PREP}_{\zeta}\ket{0}$ that comprise $\widetilde{\PREP}_{V}$.
We may write the probabilistic preparation of the momentum state as 
\begin{align}
\widetilde{\PREP}_{\boldsymbol{\nu}}\ket{0}_{g}\ket{0}_{h}=\sqrt{p_{\bnu}}\ket{0}_{g}\otimes\PREP_{\bnu}\ket{0}_{h}+\sqrt{1-p_{\bnu}}\ket{1}_{g}\otimes\PREP_{\bnu}^{\perp}\ket{0}_{h},
\end{align}
where success is indicated by the flag qubit $g$ being in the state $\ket{0}_{g}$. Similarly, for the charge state,
\begin{align}
\widetilde{\PREP}_{\zeta}\ket{0}_{k}\ket{0}_{\ell,m}=\sqrt{p_{\zeta}}\ket{0}_{k}\otimes \PREP_{\zeta}\ket{0}_{\ell,m}+\sqrt{1-p_{\zeta}}\ket{1}_{k}\otimes \PREP_{\zeta}^{\perp}\ket{0}_{\ell,m},
\end{align}
where success is indicated by the flag qubit $k$ being in the state $\ket{0}_{k}$.
Together with the state $\ket{\cL_{\theta}}=R_{z}(\theta)\ket{+}_{a}$ on register $a$, the flag registers $g$ and $k$ are used to control whether we apply $\SEL_{T}$ or $\SEL_{V}$. Our implementation of $\SEL_{H}$ is then of the form 
\begin{align}
\SEL_{H}=\sum_{x_{a},x_{g},x_{k}\in\{0,1\}}\ket{x_{a},x_{g},x_{k}}\bra{x_{a},x_{g},x_{k}}_{a,g,k}\otimes \SEL_{(x_{a},x_{g},x_{k})},
\label{eq:SEL_H_appendix}
\end{align}
where $\SEL_{(x_{a}, x_{g}, x_{k})}\in\{\SEL_{T}, \SEL_{V}, I\}$. 

We can follow the two possible strategies identified in~\cite{su2021}---called the $\OR$ and $\AND$ strategies---for assigning each $\SEL_{(x_{a}, x_{g}, x_{k})}$ to an element of $\{\SEL_{T}, \SEL_{V}, I\}$. 
We show below that our approach always succeeds provided we use the $\textrm{OR}$ strategy when $\frac{\lambda_{T}}{\lambda_{T}+\lambda_{V}}\ge 1-p_{\nu}p_{\zeta}$ and the $\textrm{AND}$ strategy otherwise. 

We recall from \cref{eq:block_simple} that the unitary block-encoding of $H$ with LCU norm $\lambda_{H}$ is defined using 
\begin{align}
    H/\lambda_{H}= \bra{0}\PREP_{H}^{\dagger}\SEL_{H}\PREP_{H}\ket{0},  
\end{align}
but that, in practice, it is implemented using the approximate preparation routine
\begin{align}
    \widetilde{\PREP}_{H}=R_{z}(\tilde{\theta})_{a} \otimes \proxPREP_{T}\otimes\widetilde{\PREP}_{V}.
    \label{eq:theta_state_app}
\end{align}
Using an ansatz for $\SEL_{H}$ that only selects $\SEL_{T}$ for register $a$ in state $\ket{0}_{a}$,
\begin{align}
    \SEL_{H}= \dyad{0}{0} \otimes \SEL_{T}\otimes I+ \dyad{1}{1}\otimes I \otimes \SEL_{V},
\end{align}
and performing the block-encoding, leads to a weighted sum of $T$ and $V$ block-encodings,
\begin{align}
\bra{0}\widetilde{\PREP}_{H}^{\dagger}\SEL_{H}\widetilde{\PREP}_{H}\ket{0}&=w(T)\bra{0}\proxPREP_{T}^{\dagger}\SEL_{T}
\proxPREP_{T}\ket{0}+ w(V)\bra{0}\widetilde{\PREP}_{V}^{\dagger}\SEL_{V}\widetilde{\PREP}_{V}\ket{0}\\
& =w(T) \frac{T}{\lambda_{T}}+ w(V)\frac{V}{\lambda_{V}}.
\end{align}
The weights $w(T)$ and $w(V)$ depend on the precise form of $\SEL_{H}$, i.e., on the choices of $\SEL_{(x_{a},x_{g},x_{k})}$. 

We can infer the weights for any set of choices by writing out the states of the flag registers. After suppressing all registers apart from the flags, the possible states are
\begin{align}
\big(\cos \theta\ket{0}_{a}+ \sin \theta\ket{1}_{a}\big)\otimes  \big(\sqrt{p_{\bnu}}\ket{0}_{g}+\sqrt{1-p_{\bnu}}\ket{1}_{g}\big)\otimes \big(\sqrt{p_{\zeta}}\ket{0}_{k}+\sqrt{1-p_{\zeta}}\ket{1}_{k}\big).
\label{eq:flag_amplitudes}
\end{align}
We then write $w(T)$ and $w(V)$ as products of the squares of these coefficients, depending on the flag values. Specifically, we assign weights $w\left(\ket{x_{a}}\right), w\left(\ket{x_{g}}\right), w\left(\ket{x_{k}}\right)$ to values of the flag registers such that  
\begin{align}
    w(\ket{0}_{a}) & =\cos^2 \theta & w(\ket{1}_{a})&= \sin^2\theta\\
    w(\ket{0}_{g})&=p_{\bnu} & w(\ket{1}_{g})& = (1-p_{\bnu})\\
    w(\ket{0}_{k})&= p_{\zeta} & w(\ket{1}_{k})&=(1-p_{\zeta}).
\end{align}
Moreover, we define the sets 
\begin{align}
    S_{T}&=\left\{ \ket{x_{a}}\ket{x_{g}}\ket{x_{k}}\mid\SEL_{(x_{a},x_{g},x_{k})}=\SEL_{T} \right\}\\
    S_{V}&=\left\{ \ket{x_{a}}\ket{x_{g}}\ket{x_{k}}\mid\SEL_{(x_{a},x_{g},x_{k})}=\SEL_{V} \right\}\\
    S_{I}&=\left\{ \ket{x_{a}}\ket{x_{g}}\ket{x_{k}}\mid\SEL_{(x_{a},x_{g},x_{k})}=I \right\},
\end{align}
which denote the collections of flag values for which $\SEL_{T}$, $\SEL_{V}$, or the identity are applied. For each element, we define the weight \begin{align}
 w\left(\ket{x_{a}}\ket{x_{g}}\ket{x_{k}}\right)=w\left(\ket{x_{a}}\right) w\left(\ket{x_{g}}\right) w\left(\ket{x_{k}}\right).
\end{align}
We can then write
\begin{align}\label{eq:weight_defs1}
w(T)&=\sum_{\ket{x_{a}}\ket{x_{g}}\ket{x_{k}}\in S_{T}} w\left(\ket{x_{a}}\right) w\left(\ket{x_{g}}\right) w\left(\ket{x_{k}}\right)\\\label{eq:weight_defs2}
w(V)&=\sum_{\ket{x_{a}}\ket{x_{g}}\ket{x_{k}}\in S_{V}} w\left(\ket{x_{a}}\right) w\left(\ket{x_{g}}\right) w\left(\ket{x_{k}}\right).
\end{align}
Using \cref{eq:weight_defs1,eq:weight_defs2}, we determine $w(T)$ and $w(V)$ for different implementations of $\SEL_{H}$. One condition that must always be met is that $\SEL_{V}$ is applied only if both the momentum state and the charge state were prepared successfully, i.e., when 
$\ket{x_{g}}=\ket{0}_{g}$ and $\ket{x_{k}}=\ket{0}_{k}$.

\subsection{\texorpdfstring{$\OR$}{OR} strategy}
We now  determine the weights $w(T)^{\OR}$ and $w(V)^{\OR}$ for the implementation of $\SEL_{H}$ in the $\OR$ case. 

We apply $\SEL_{V}$ for the set of flag values
\begin{align}
    S_{V}^{\OR}=\{\ket{1}_{a} \ket{0}_{g} \ket{0}_{k}\},
\end{align}
which, in the implementation of $\SEL_{H}$, corresponds to including a term
\begin{align}
    \dyad{1}{1}_{a}\otimes \dyad{0}{0}_{g} \otimes \dyad{0}{0}_{k} \otimes \SEL_{V}.
\end{align}
The corresponding weight is
\begin{align}
    w(V)^{\OR}=p_{\bnu}p_{\zeta} \sin^2\theta.
\end{align}

We apply $\SEL_{T}$ for the complementary set of flag values 
\begin{align}
    S_{T}^{\OR}=\{\ket{0}_{a}\ket{x_{g}}_{g}\ket{x_{k}}_{k},\ket{1}_{a}\ket{1}_{g}\ket{x_{k}}_{k},\ket{1}_{a}\ket{0}_{g}\ket{1}_{k}\},
\end{align}
where $x_{g},x_{k}\in\{0,1\}$. This corresponds to including the terms 
\begin{align}
(\dyad{0}{0}_{a}\otimes I_{g} \otimes I_{k} 
+\dyad{1}{1}_{a}\otimes \dyad{1}{1}_{g} \otimes I_{k}
+\dyad{1}{1}_{a}\otimes \dyad{0}{0}_{g} \otimes \dyad{1}{1}_{k})\otimes \SEL_{T}.
\end{align}
The corresponding weight is
\begin{align}
    w(T)^{\OR}&=\cos^2\theta+\left((1-p_{\bnu})+p_{\bnu}(1-p_{\zeta})\right)\sin^2\theta\\
    &=1-p_{\bnu}p_{\zeta}\sin^{2}\theta,
\end{align}
satisfying $w(T)^{\OR}+w(V)^{\OR}=1$.

To ensure a correct block-encoding, it is necessary that 
\begin{align}
\label{eq:kinetic_weight}
    w(T)^{\OR}&=\frac{\lambda_{T}}{\lambda_{T}+\lambda_{V}}\\
    1-p_{\bnu}p_{\zeta}\sin^{2}(\theta)&= \frac{\lambda_{T}}{\lambda_{T}+\lambda_{V}}.  
\end{align}
A solution for $\theta$ then exists in the $\OR$ case whenever 
\begin{align}
1-p_{\bnu}p_{\zeta}\leq\frac{\lambda_{T}}{\lambda_{T}+\lambda_{V}}.
\end{align}
For future reference, we note that, in the $\OR$ case, $p_{\bnu}p_{\zeta}\geq 1- \frac{\lambda_{T}}{\lambda_{T}+\lambda_{V}} = \frac{\lambda_{V}}{\lambda_{T}+\lambda_{V}}$ and therefore
\begin{align}
    \lambda_{H}^{\OR}=\lambda_{T}+\lambda_{V} \ge \frac{\lambda_{V}}{p_{\bnu}p_{\zeta}}.\label{eq:lambdaH-OR}
\end{align}

\subsection{\texorpdfstring{$\AND$}{AND} strategy} 
We determine $w(T)^{\AND}$ and $w(V)^{\AND}$ analogously. 

We apply $\SEL_{T}$ for the set of flag values
\begin{align}
    S_{T}^{\AND}=\{(\ket{0}_{a}\ket{1}_{g}\ket{x_{k}}_{k}),(\ket{0}_{a}\ket{0}_{g}\ket{1}_{k})\},
\end{align}
where $x_{k}\in\{0,1\}$, which corresponds to the terms
\begin{align}
&(\dyad{0}{0}_{a}\otimes \dyad{1}{1}_{g} \otimes I_{k} 
+\dyad{0}{0}_{a}\otimes \dyad{0}{0}_{g} \otimes \dyad{1}{1}_{k})\otimes \SEL_{T}
\end{align}
and giving
\begin{align}
    w(T)^{\AND}&=\left((1-p_{\nu})+p_{\nu}\left(1-p_{\zeta}\right)\right)\cos^2\theta\\
    &=(1-p_{\bnu}p_{\zeta})\cos^{2}\theta.
\end{align}

We apply $\SEL_{V}$ for the set of flag values 
\begin{align}
S_{V}^{\AND}=\{(\ket{x_{a}}_{a}\ket{0}_{g}\ket{0}_{k})\},
\end{align}
where $x_{a}\in\{0,1\}$, corresponding to
\begin{align}
    I_{a}\otimes \dyad{0}{0}_{g} \otimes \dyad{0}{0}_{k} \otimes \SEL_{V}
\end{align}
and
\begin{align}
    w(V)=p_{\bnu}p_{\zeta}.
\end{align}

For the remaining flag values, we apply the identity. Those are 
\begin{align}
    S_{I}^{\AND}=\{\ket{1}_{a}\ket{1}_{g}\ket{x_{k}}_{k}, \ket{1}_{a}\ket{0}_{g}\ket{1}_{k}\},
\end{align}
where $x_{k}\in\{0,1\}$,meaning that, for the $\AND$ case, $\SEL_{H}$ will include terms
\begin{align}
&(\dyad{1}{1}_{a}\otimes \dyad{1}{1}_{g} \otimes I_{k}
+\dyad{1}{1}_{a}\otimes \dyad{0}{0}_{g} \otimes \dyad{1}{1}_{k})\otimes I.
\end{align}
The additional weight is
\begin{align}
    w(I)^{\AND}&=(1-p_{\bnu}p_{\zeta})\sin^{2}\theta,
\end{align}
which obeys $w(T)^{\AND}+w(V)^{\AND}+w(I)^{\AND}=1$.

These three terms in the implementation of $\SEL_{H}$ now cause the block-encoding of $H$ to be of the form 
\begin{align}
    w(T)^{\AND}\frac{T}{\lambda_{T}}+ w(V)^{\AND}\frac{V}{\lambda_{V}}+w(I)^{\AND}I.
\end{align}

Nevertheless, it is still necessary that the relative weights obey 
\begin{align}
    \frac{w(T)^{\AND}}{w(T)^{\AND}+w(V)^{\AND}} &=\frac{\lambda_{T}}{\lambda_{T}+\lambda_{V}}\\\label{eq:rel_weight_and}
    \frac{(1-p_{\bnu}p_{\zeta})\cos^{2}\theta}{\cos^2\theta+p_{\bnu}p_{\zeta}\sin^2\theta}&=\frac{\lambda_{T}}{\lambda_{T}+\lambda_{V}}.
\end{align}
Because the left-hand side is at most $1-p_{\bnu}p_{\zeta}$, it follows that, in the $\AND$ case, there exists a solution for $\theta$ if
\begin{align}
1-p_{\bnu}p_{\zeta}\geq\frac{\lambda_{T}}{\lambda_{T}+\lambda_{V}},
\end{align}
which is the reverse situation to the $\OR$ case, where $1-p_{\bnu}p_{\zeta}\leq\frac{\lambda_{T}}{\lambda_{T}+\lambda_{V}}$.

To determine $\lambda^{\AND}_{H}$, we use the constraint  
\begin{align}
\frac{\lambda_{V}}{\lambda_{H}^{\AND}}=\frac{w(V)^{\AND}}{w(T)^{\AND}+w(V)^{\AND}+w(I)^{\AND}}=p_{\bnu}p_{\zeta}.
\end{align} 
Therefore,
\begin{align}
\lambda_{H}^{\AND}=\frac{\lambda_{V}}{p_{\bnu}p_{\zeta}}.\label{eq:lambdaHAND1}
\end{align}
Because $1-p_{\bnu}p_{\zeta}\leq\frac{\lambda_{T}}{\lambda_{T}+\lambda_{V}}$, it follows from $1-\frac{\lambda_{T}}{\lambda_{T}+\lambda_{V}}=\frac{\lambda_{V}}{\lambda_{T}+\lambda_{V}}$ that 
\begin{align}
    \lambda_{H}^{\AND}\geq \lambda_{T}+\lambda_{V}.\label{eq:lambdaHAND2}
\end{align}

\subsection{Combined results}

Overall, a valid solution for $\theta$ always exists, either by using the $\OR$ strategy when $1-p_{\bnu}p_{\zeta}\leq\frac{\lambda_{T}}{\lambda_{T}+\lambda_{V}}$ or the $\AND$ strategy otherwise.

We also provide an expression for $\lambda_{H}$ that encompasses both $\AND$ and $\OR$ cases. From \cref{eq:lambdaH-OR,eq:lambdaHAND1,eq:lambdaHAND2} and from accounting for the failure probability of uniform state preparations in $\proxPREP_{H}$, it follows that \begin{align}\label{eq:proxLambda_H_app}
\tilde{\lambda}_{H}=\frac{1}{P_{\eq}}\max\left\{\lambda_{T}+\lambda_{V},\frac{\lambda_{V}}{p_{\bnu}p_{\zeta}}\right\},
\end{align}
where $P_{\eq}$ is the probability of all uniform state preparations in $\proxPREP_{H}$ succeeding. These uniform state preparations occur as subroutines in the preparations of $\ket{\cL_{w}}, \ket{\cL_{m}}$, and $\ket{\cL_{\doublezeta}}$ respectively. Therefore,
\begin{equation}
    P_{\eq}=\Ps(3,8)\Ps(\eta,b_{r})\Ps(\eta,b_{r})^2,
\end{equation}
where $\Ps(n,b_{r})$ is the success probability of preparing a uniform superposition over $n$ basis states using rotations with $b_{r}$-bit precision. This probability is \cite{su2021}
\begin{align}
\Ps(n,b_{r})&=x\big(\big(1+\left(2-4x\right)\sin^2\theta(n,b_r)\big)^2+\sin^2(2\theta(n,b_r))\big)
\label{eq:prob_success_unif}\\
\theta(n,b_{r})&=y^{-1}
\operatorname{round}\!\big(
y \arcsin\!\big((4x)^{-1/2}\big)
\big),
\end{align}
where $x=n\,2^{-\lceil \log n \rceil}$ and $y=2^{b_r}/2\pi$.

\section{Error bounds on time evolution}
\label{app:bound_time_evo}
We bound the error on time evolution in four steps. In \cref{lem:epsilon_prop_bound}, we  first give a bound on the full time evolution error $\epsilon_{\prop}$ in terms of $\epsilon_\mathrm{QSP}$ and $\epsilon_{H}$. Then, to bound the block-encoding error $\epsilon_{H}$, we bound the block-encoding errors $\epsilon_{T}$ and $\epsilon_{V}$ in \cref{lem:epsilon_T_bound,lem:epsilon_V_bound}, respectively. We conclude with
a bound on $\epsilon_{H}$ in \cref{prop:epsilon_H_bound}.

\subsection{Total error}\label{app:QSP_error_bound_proof}
\begin{lemma}  \label{lem:epsilon_prop_bound}
Let $f_{\tilde{d}}$ be the degree-$\tilde{d}$ truncation of the Jacobi-Anger expansion \cite{Abramowitz1964}, such that $\big\lVert e^{-i\tilde{H}t}-\left(\bra{+}_{\anc}\otimes I\right)f_{\tilde{d}}(\proxW)\left(\ket{+}_{\anc}\otimes I\right)\big\rVert_{\infty}\leq \epsilon_{\tilde{d}}$, where $\proxW=e^{-i\arccos(\tilde{H}/\tilde{\lambda}_{H})}$ (see \cref{eq:approx_qubiterate,eq:Ham_blockenc_definition} for definitions). Further, let $\tilde{f}_{\tildegre}$ be an approximation to  $f_{\tildegre}$ implemented via QSP with rotation angles computed up to error $\epsilon_{\phi}$ and $\epsilon_{\gamma}$ and rotations synthesized up to error $\epsilon_{\rot}$. Then, for $\tilde{U}_{\prop}=\left(\bra{+}_{\anc}\otimes I\right)\tilde{f}_{\tildegre}(\proxW)\left(\ket{+}_{\anc}\otimes I\right)$,
\begin{align}
    \big\lVert e^{-iHt}-\tilde{U}_{\prop}\big\rVert_{\infty}\leq \epsilon_{\prop}=t \lVert H-\tilde{H}\rVert_{\infty}+\epsilon_{\tildegre}+(\tildegre+1)\left(\epsilon_{\phi}+\epsilon_{\gamma}+\epsilon_{\rot}\right).
\end{align} 
\end{lemma}
\begin{proof}
In this proof, we abuse notation to write $\tilde{f}_{\tildegre}(\proxW)$ in place of  $\left(\bra{+}_{\anc}\otimes I\right)\tilde{f}_{\tildegre}(\proxW)\left(\ket{+}_{\anc}\otimes I\right)$ and likewise for $f_{\tildegre}$. 
We then upper bound the error on $\big\lVert e^{-iHt}-\tilde{f}_{\tildegre}(\proxW)\big\rVert_{\infty}$ in three steps via the triangle inequality,
\begin{align}
\big\lVert e^{-iHt}- \tilde{f}_{\tildegre}(\proxW) \big\rVert_{\infty}&\leq \big\lVert e^{-iHt}-e^{-i\tilde{H}t}\big\rVert_{\infty}
+ \big\lVert  e^{-i\tilde{H}t}-f_{\tildegre}(\proxW) \big\rVert _{\infty}
+\big\lVert  f_{\tildegre}(\proxW)- \tilf_{\tildegre}(\proxW)  \big\rVert_{\infty}\\
&\leq t \lVert H-\tilde{H}\rVert_{\infty}+\epsilon_{\tildegre}+\tildegre(\epsilon_{\phi}+\epsilon_{\gamma}+\epsilon_{\rot})\\
&=t\epsilon_{H}+\epsilon_{\QSP},
\end{align}
where we have used the definitions in \cref{eq:Ham_blockenc_definition,eq:epsilon_QSP}.
\end{proof}

\subsection{Block-encoding error: Kinetic energy}\label{app: bound_T}
Here, we prove a bound on the error $\lVert T-\tilde{T}\rVert_{\infty}\leq\epsilon_{T}$ for $\tilde{T}=\lambda_{T}\bra{0}\proxblock{T}\ket{0}$.
\begin{lemma}\label{lem: bound_T}
Let $\mu_{T}$ denote the number of bits used to approximate the coefficients of $\ket{\cL_{m}}$. Then, the unitary $\proxblock{T}$ is a $\left(\lambda_{T}, q_{T},\epsilon_{T}\right)$-block-encoding of $T$, with
\begin{align}
    \big\lVert T-\lambda_{T}\bra{0^{q_{T}}}\proxblock{T}\ket{0^{q_{T}}}\big\rVert_{\infty}\leq \epsilon_{T}= \frac{\lambda_{T}}{2^{\mu_{T}}},
\end{align}
where $q_{T}=n_{\eta}+2n_{p}+5$ and  $\lambda_{T}=\frac{6\pi^2}{\Omega^{2/3}} (2^{n_{p}-1}-1)^2 \sum_{j=1}^{\eta} m_{j}^{-1}$.
\label{lem:epsilon_T_bound}
\end{lemma}

\begin{proof}
\Cref{eq:LCU_decomp_kin} gives the LCU decomposition of the kinetic term as 
\begin{align}
T&=\sum_{\ell_{T}}\alpha_{\ell_{T}}H_{\ell_{T}}=\sum_{j=1}^{\eta} \sum_{w\in \{x,y,z\}}\sum_{r,s=0}^{n_{p}-2}\sum_{b\in \{0,1\}} \frac{\pi^2 2^{r+s}}{m_{j}\Omega^{2/3}}H_{\ell_{T}},  
\label{eq:LCU_kin_app}
\end{align}
where $\alpha_{\ell_{T}}=\pi^2 2^{r+s}/m_{j}\Omega^{2/3}$ and where $H_{\ell_{T}}=\sum_{\mathbf{p}\in G} (-1)^{b\cdot\left(p_{w,r}\cdot p_{w,s}\oplus 1\right)}\ket{\mathbf{p}}\bra{\mathbf{p}}_{j}$ is unitary. We let
\begin{align}
 \tilde{T}=\sum_{j=1}^{\eta} \sum_{w\in \{x,y,z\}}\sum_{r,s=0}^{n_{p}-2}\sum_{b\in \{0,1\}} \frac{\pi^2 2^{r+s}}{\Omega^{2/3}\tilde{m}_{j}}H_{\ell_{T}}.  
\end{align}
Then,
\begingroup
\allowdisplaybreaks
\begin{align}
 \lVert T-\tilde{T}\rVert_{\infty}&=\bigg\lVert\sum_{\ell_{T}}\alpha_{\ell_{T}}H_{\ell_{T}}-\sum_{\ell_{T}}\tilde{\alpha}_{\ell_{T}}H_{\ell_{T}}\bigg\rVert_{\infty}\\
    &\leq \bigg| \sum_{\ell_{T}}\left(\alpha_{\ell}-\tilde{\alpha}_{\ell}\right)\bigg| \\
    &\leq 6\pi^{2}\frac{(2^{n_{p}-1}-1)}{\Omega^{2/3}}\Bigg| \sum_{j=1}^{\eta}\left(\frac{1}{m_{j}}-\frac{1}{\tilde{m}_{j}}\right)\Bigg|\\
    &\leq 6\pi^{2}\frac{(2^{n_{p}-1}-1)}{\Omega^{2/3}}\eta\max_{j}\bigg|\frac{1}{m_{j}}-\frac{1}{\tilde{m}_{j}}\bigg|\\
    &\leq 6\pi^{2}\frac{(2^{n_{p}-1}-1)}{\Omega^{2/3}}2^{-\mu_{T}}\sum_{j=1}^{\eta}\frac{1}{m_{j}}\\
    &=\frac{\lambda_{T}}{2^{\mu_{T}}},
\end{align}
\endgroup
where in the first step we are summing over the indices shown in \cref{eq:LCU_kin_app} and 
where the penultimate step follows from the coherent-alias-sampling protocol preparing the mass state $\ket{\cL_{\mass}}$, which guarantees that \cite{babbush2018} 
\begin{align}
\max_{j}\left|\frac{1}{m_{j}}-\frac{1}{\tilde{m}_{j}}\right|\leq \frac{2^{-\mu_{T}}}{\eta}\sum_{j=1}^{\eta}\frac{1}{m_{j}}.
\end{align}    

The number of ancillas $q_{T}$ that are not reset to zero in $\proxblock{T}$ can be obtained from \cref{app:reflection_cost} as  
$q_{T}=n_{\eta}+2n_{p}+5$.
\end{proof}

\subsection{Block-encoding error: Potential energy}\label{app:bound_V}
Here, we prove a bound on the error $\lVert V-\tilde{V}\rVert_{\infty}\leq\epsilon_{V}$ for $\tilde{V}=\lambda_{V}\bra{0}\proxblock{V}\ket{0}$.

\begin{lemma}\label{lem: bound_V}
Let $n_{\cM}$ be the number of bits used to approximate the coefficients of the momentum state $\ket{\cL_{\nu}}$ and let $r_{\bnu}=\frac{4}{\lambda_{\bnu}}(7\times2^{n_{p}+1}-9n_{p}-11-3\times 2^{-n_{p}})$, where $\lambda_{\bnu}=\sum_{\boldsymbol{\nu} \in G_0} \|\boldsymbol{\nu}\|^{-2}_{2}$. Then, the unitary $\proxblock{V}$ is a $\left(\lambda_{V}, q_{V},\epsilon_{V}\right)$-block-encoding of $V$, with
\begin{align}
    \big\lVert V-\lambda_{V}\bra{0^{q_{V}}}\proxblock{V}\ket{0^{q_{V}}}\big\rVert_{\infty}\leq \epsilon_V = \lambda_{V}\frac{r_{\bnu}}{2^{n_{\cM}}} ,
\end{align}
where $q_{V}=4n_{p}+2\bits{2\eta_{e}}+n_{\cM}+5$ and $\lambda_{V}=\lambda_{\bnu}(2\pi \Omega^{1/3})^{-1} {\sum_{i\neq j=1}^{\eta} |\zeta_{i}||\zeta_{j}|} $.
\label{lem:epsilon_V_bound}
\end{lemma}
\begin{proof}

\Cref{eq:LCU_decomp_pot} gives the LCU decomposition of the potential term as 
\begin{align}
V&=\sum_{\ell_{V}}\alpha_{\ell_{V}}H_{\ell_{V}}=\sum_{\bnu\in G_{0}}\sum_{i\neq j=1}^{\eta}\sum_{b\in\{0,1\}}\frac{ |\zeta_{i}||\zeta_{j}|}{2\pi\Omega^{1/3} \norm{\bnu}^{2}_2}H_{(b,i,j,\bnu)},  
\label{eq:LCU_pot_app}
\end{align}
where $\alpha_{\ell_{V}}=|\zeta_{i}||\zeta_{j}|/2\pi\Omega^{1/3} \norm{\bnu}^{2}_2$ and $H_{\ell_{V}}=\sum_{\mathbf{p},\mathbf{q}\in G}\left(-1\right)^{f(\mathbf{p},\mathbf{q},\boldsymbol{\nu},i,j)}\ket{\mathbf{p}+\boldsymbol{\nu}}\bra{\mathbf{p}}_{i} \ket{\mathbf{q}-\boldsymbol{\nu}}\bra{\mathbf{q}}_{j}$ are unitary. We let
\begin{align}
 \tilde{V}=\sum_{\ell_{V}}\alpha_{\ell_{V}}H_{\ell_{V}}=\sum_{\bnu\in G_{0}}\sum_{i\neq j=1}^{\eta}\sum_{b\in\{0,1\}}\frac{\pi |\zeta_{i}||\zeta_{j}|}{2\pi\Omega^{1/3} \norm{\tilde{\bnu}}^{2}_2}H_{(b,i,j,\bnu)}
\end{align}
Then,
\begin{align}
\lVert V-\tilde{V}\rVert_{\infty}&\leq \frac{\sum_{i\neq j=1}^{\eta}|\zeta_{i}||\zeta_{j}|}{2\pi\Omega^{1/3}}\sum_{\bnu \in G_{0}}\left|\frac{1}{\norm{\bnu}_{2}^{2}}-\frac{1}{\norm{\tilde{\bnu}}_{2}^{2}}\right|\\
        &\leq \frac{\sum_{i\neq j=1}^{\eta}|\zeta_{i}||\zeta_{j}|}{2\pi\Omega^{1/3}}\sum_{r=2}^{n_{p}+1}\sum_{\bnu\in B_{\mu}}  \left|\frac{1}{\norm{\bnu}_{2}^{2}}-\frac{1}{\norm{\tilde{\bnu}}_{2}^{2}}\right|\\
        &\leq \frac{\sum_{i\neq j=1}^{\eta}|\zeta_{i}||\zeta_{j}|}{2\pi\Omega^{1/3}}\frac{4}{2^{n_{\cM}}}\left(7\times2^{n_{p}+1}-9n_{p}-11-3\times 2^{-n_{p}}\right)\\
        &=\lambda_{V}\frac{r_{\bnu}}{2^{n_{\cM}}},
\end{align}
where in the penultimate line we used Eq. (113) in \cite{su2021}, and in the last line we used the definitions of $r_{\bnu}$ and $\lambda_{V}$.

The number of ancillas $q_{V}$ that are not reset to zero in $\proxblock{V}$ can be obtained from \cref{app:reflection_cost} as  
$q_{V}=4n_{p}+2\bits{2\eta_{e}}+n_{\cM}+5$.
\end{proof}

\newpage
\begin{lemma}\label{lem:r_bnu_bound}
$r_{\bnu}=\dfrac{4}{\lambda_{\bnu}}(7\times2^{n_{p}+1}-9n_{p}-11-3\times 2^{-n_{p}})\leq 12.$
\label{lemma:bound_r_bnu}
\end{lemma}
\begin{proof}
Let $B_{\mu}=\{\bnu \in G_{0}\mid (\norm{\bnu}_{\infty}< 2^{\mu-1})\wedge(\bigvee_{w\in\{x,y,z\}}(|\bnu_{w}|\geq 2^{\mu-2}))\}$. Then, we have
\begin{align}
\lambda_{\bnu} = \sum_{\boldsymbol{\nu} \in G_0} \frac{1}{\|\boldsymbol{\nu}\|_{2}^{2}}
&= \sum_{\mu=2}^{n_p+1} \sum_{\boldsymbol{\nu} \in B_\mu} 
\frac{1}{|\nu_x|^2 + |\nu_y|^2 + |\nu_z|^2}.
\end{align}
Since for any $\bnu\in B_{\mu}$ it holds that $|\nu_{x}|,|\nu_{y}|,|\nu_{z}|\leq 2^{\mu-1}$, it follows that
\begingroup
\allowdisplaybreaks
\begin{align}
\lambda_{\bnu} &\ge \sum_{\mu=2}^{n_p+1} \sum_{\boldsymbol{\nu} \in B_\mu} 
\frac{1}{3 \cdot 2^{2\mu-2}}  =\frac{4}{3} \sum_{\mu=2}^{n_p+1} \sum_{\boldsymbol{\nu} \in B_\mu} \frac{1}{2^{2\mu}}\\ 
&= \frac{4}{3} \sum_{\mu=2}^{n_p+1} \left( (2^\mu - 1)^3 - (2^{\mu - 1} - 1)^3 \right) \frac{1}{2^{2\mu}} \\
&= \frac{1}{3} \left( 7 \times 2^{n_p + 1}- {9n_p -1 1} - 3 \times 2^{-n_p} \right),
\end{align}
\endgroup
from which our bound follows. 
\end{proof}

\subsection{Error bound on full block-encoding}\label{app:bound_H_blockenc}

We now prove a bound on the error $\lVert H-\tilde{H}\rVert_{\infty}\leq\epsilon_{H}$ for $\tilde{H}=\tilde{\lambda}_{H}\bra{0}\proxblock{H}\ket{0}$.
\begin{proposition}
The unitary $\proxblock{H}$ is a $(\tilde{\lambda}_{H}, q_{H}, \epsilon_{H})$-block encoding of $H$ with 
\begin{align}\label{eq:epsilon_H_expr}
\big\lVert H-\tilde{\lambda}_{H}\bra{0^{q_{H}}}\proxblock{H}\ket{0^{q_{H}}}\big\rVert_{\infty}\leq \epsilon_{H} = \epsilon_{T}+\epsilon_{V}+2\tilde{\lambda}_{H}\Delta\theta,
\end{align}
where $q_{H}=n_{\eta}+6n_{p}+n_{\cM}+2\bits{2\eta_{e}}+11$ and $\Delta \theta = |\theta-\tilde\theta|$, as defined in \cref{eq:delta_theta}.
\label{prop:epsilon_H_bound}
\end{proposition}

\begin{proof}
We note that 
\begin{align}
 \proxblock{H}=\tilde{w}(T)\bra{0}\proxblock{T}\ket{0}+\tilde{w}(V)\bra{0}\proxblock{V}\ket{0}, 
\end{align}
and therefore
\begin{align}
  \tilde{H}&= \tilde{\lambda}_{H} \bra{0}\proxblock{H}\ket{0}\\
  &=\tilde{\lambda}_{H}\Big(\tilde{w}(T)\bra{0}\proxblock{T}\ket{0}+\tilde{w}(V)\bra{0}\proxblock{V}\ket{0}\Big).
\end{align}
We thus obtain by the triangle inequality
\begin{align}
    \lVert H-\tilde{H}\rVert_{\infty}&=\big\lVert (T+V)-\tilde{\lambda}_{H}\big(\tilde{w}(T)\bra{0}\proxblock{T}\ket{0}+\tilde{w}(V)\bra{0}\proxblock{V}\ket{0}\big)\big\rVert_{\infty}\\
    &\leq \big\lVert T-\tilde{\lambda}_{H}\tilde{w}(T)\bra{0}\proxblock{T}\ket{0}\big\rVert_{\infty}+\big\lVert V-\tilde{\lambda}_{H}\tilde{w}(V)\bra{0}\proxblock{V}\ket{0}\big\rVert_{\infty}.\label{eq:first_term}
\end{align}
From Appendix D in \cite{su2021}, it follows that the first term in \cref{eq:first_term} is bounded by $\epsilon_{T}+\tilde{\lambda}_{H}\Delta\theta$ and the second term is bounded by $\epsilon_{V}+\tilde{\lambda}_{H}\Delta\theta$. Therefore, our bound on $\epsilon_{H}$ in \cref{eq:epsilon_H_expr} directly follows. 

The number of ancillas that are not reset to zero in  $\proxblock{H}$ is equal to $\out(\PREP_{H}^{\dagger})$ in \cref{app:reflection_cost},
$q_{H}=n_{\eta}+6n_{p}+n_{\cM}+2\bits{2\eta_{e}}+11$.
\end{proof}

\section{Block-encoding of yield observables}\label{app:indicator_circuit}

Here, we describe the quantum arithmetic circuit that constructs the 
 $(1,1,0)$-block-encoding
 \begin{equation}
    U_{\Pi_{S}} = \sum_{\mathbf{R}} \ket{\mathbf{R}}\bra{\mathbf{R}}_{\data}\otimes X^{I(\mathbf{R},S)\oplus 1}_{\flag}, 
\end{equation}
of the yield observable $\Pi_{S}=\left(\bra{0}\otimes I\right)U_{\Pi_{S}}\left(\ket{0}\otimes I\right)$. Throughout this section, we refer to this circuit as the indicator circuit. We decompose the indicator circuit as 
\begin{align}
    U_{\Pi_{S}}= U_{\bond}^{\dagger}U_{\flag} U_{\bond}(\textsc{tc2sm}_{n_{p}}^{\dagger})^{\otimes 3n_{\nuc}},
\end{align}
where $\textsc{tc2sm}$ converts the $n_{p}$-bit strings in the nuclear data registers from signed magnitude to two's complement to allow for the application of arithmetic circuits. $U_{\bond}$ then computes whether the $B_{j}$ inter-nuclear distances involving $n_{\nuc}$ nuclei that define the reaction channel $S$ satisfy the conditions for the channel and stores the results in $B_{j}$ qubits. If they do, controlled on these $B_{j}$ qubits, $U_{\flag}$ flips the flag qubit from $\ket{1}$ to $\ket{0}$. In particular,
\begin{align}
    U_{\flag}&=\Had^{\otimes B_j}\left(\contr{X}_{\flag}\right) \Had^{\otimes B_j},\\
    U_{\bond}&=\prod_{k \le B_{j}}U_{j,k},  \\
    U_{j,k}&=\COMP|_{2n_{p}}(b_{j,k})\, \SOS^{\alpha_{j,k}\beta_{j,k}}_{n_{p}}\, \SUB^{\alpha_{j,k}\beta_{j,k}}_{n_{p}},
\end{align}
where $\SUB^{\alpha_{j,k}\beta_{j,k}}_{n_{p}}$ performs the subtractions $R_{\alpha_{j,k},w}-R_{\beta_{j,k},w}$ for all three position-vector components. $\SOS^{\alpha_{j,k}\beta_{j,k}}_{n_{p}}$ then sums the squares $\sum_{w}(R_{\alpha_{j,k},w}-R_{\beta_{j,k},w})^2$ of the $n_{p}$-bit differences using Lemma~8 in~\cite{su2021}. Finally, $\COMP|_{2n_{p}}(b_{j,k})$ denotes a comparison test that does not uncompute ancillas at the end. Specifically, $\COMP|_{2n_{p}}(b_{j,k})$ takes the $2n_{p}$-bit output of $\SOS$ and stores the outcome of the comparison $\sum_{w}(R_{\alpha_{j,k},w}-R_{\beta_{j,k},w})^2 \geq b_{j,k}^2$ in an ancilla. 

The costs for the indicator circuit checking  $B_{j}$ inter-nuclear distances involving $n_{\nuc}$ nuclei are then 
\begin{align}
C_{\anc}(\Ind)&=\max\{ C_{\anc}((\textsc{tc2sm}_{n_{p}}^{\dagger})^{\otimes 3n_{\nuc}}),C_{\anc}(U_{\bond}),C_{\anc}(U_{\flag}),C_{\anc}(U_{\bond}^{\dagger}) \}\\
C_{\Toff}(\Ind)&=3n_\mathrm{nuc}C_{\Toff}(\textsc{tc2sm}_{n_{p}}^{\dagger})+C_{\Toff}(U_{\bond})+C_{\Toff}(U_{\flag})+C_{\Toff}(U_{\bond}^{\dagger}),
\end{align}
where 
\begin{align}
C_{\anc}(\textsc{tc2sm}_{n_{p}}^{\dagger})&=n_{p}-2 &
C_{\Toff}(\textsc{tc2sm}_{n_{p}}^{\dagger})&= (n_{p}-2)\\
C_{\anc}(U_{\flag})&=B_j-2&
C_{\Toff}(U_{\flag})&=B_j-1\\
C_{\anc}(U_{\bond})&=B_{j}C_{\anc}(U_{j,k})&
C_{\Toff}(U_{\bond})&=B_jC_{\Toff}(U_{j,k}).
\end{align}
We can break these costs down further as 
\begin{align}
C_{\anc}(U_{j,k})&=\max\{C_{\anc}(\COMP|_{2n_{p}}(b_{j,k})),C_{\anc}(\SOS^{\alpha_{j,k}\beta_{j,k}}_{n_{p}}),C_{\anc}(\SUB^{\alpha_{j,k}\beta_{j,k}}_{n_{p}})\}\label{eq:anc_cost_sqrt_U_jk}\\
C_{\Toff}(U_{j,k})&=C_{\Toff}(\COMP|_{2n_{p}}(b_{j,k}))+
C_{\Toff}(\SOS^{\alpha_{j,k}\beta_{j,k}}_{n_{p}})+C_{\Toff}(\SUB^{\alpha_{j,k}\beta_{j,k}}_{n_{p}}),
\end{align}
\begingroup
\allowdisplaybreaks
Assuming that the 3 subtractions are carried out in parallel, the costs are 
\begin{align}
C_{\anc}(\SUB^{\alpha_{j,k}\beta_{j,k}}_{n_{p}}))&=3n_{p} &
C_{\Toff}(\SUB^{\alpha_{j,k}\beta_{j,k}}_{n_{p}})&=3(n_{p}-1)
\\
C_{\anc}(\SOS^{\alpha_{j,k}\beta_{j,k}}_{n_{p}})&=3n_{p}^2-n_{p}-1&
C_{\Toff}(\SOS^{\alpha_{j,k}\beta_{j,k}}_{n_{p}})&=3n_{p}^{2}-n_{p}-1\\
C_{\anc}(\COMP|_{2n_{p}}(b_{j,k}))&=2n_{p}&
C_{\Toff}(\COMP|_{2n_{p}}(b_{j,k}))&=2n_{p}-2.
\end{align}
\endgroup

We uncompute the carry bits in the inequality test, $(\COMP|_{2n_{p}}(b_{j,k}))^{\dagger}$, at no further Toffoli cost using measurements and phase-fixups. Therefore, $C_{\Toff}(\COMP|_{2n_{p}}(b_{j,k})^{\dagger})=0$ and
\begin{align}
C_{\Toff}(U^{\dagger}_{\bond})&=B_jC_{\Toff}(U^{\dagger}_{j,k})\\
C_{\Toff}(U_{j,k}^{\dagger})&=  C_{\Toff}(\ADD^{\alpha_{j,k}\beta_{j,k}}_{n_{p}})+C_{\Toff}\big((\SOS^{\alpha_{j,k}\beta_{j,k}}_{n_{p}})^{\dagger}\big),
\end{align}
where $\ADD^{\alpha_{j,k}\beta_{j,k}}_{n_{p}}$ uncomputes $\SUB^{\alpha_{j,k}\beta_{j,k}}_{n_{p}}$, with
\begin{align}
C_{\Toff}(\ADD^{\alpha_{j,k}\beta_{j,k}}_{n_{p}})&=3(n_{p}-1)\\
C_{\Toff}\big((\SOS^{\alpha_{j,k}\beta_{j,k}}_{n_{p}})^{\dagger}\big)&=3n_{p}^{2}-n_{p}-1.
\end{align}
Overall, the costs of $U_{\Pi_{S}}$ are
\begin{align}
C_{\anc}(U_{\Pi_{S}})&=1+ B_{j}C_{\anc}(U_{j,k})\\
C_{\mathrm{Toff}}(\Ind)&=3B_j(2 n_p^2 + 2n_p-3)+3n_{\nuc}(n_{p}-2)-1.
\end{align}

\section{Bound on observable-estimation error}\label{app:yield_error}
We bound the error $\epsilon$ on the estimate $\widetilde{\langle \tilde{O}\rangle}$ of the observable expectation value $\bra{\Psi(t)}O\ket{\Psi(t)}$ in two steps. In \cref{lem:state_and_obs_error}, we account for the errors in the state preparation and 
the block-encoding of the observable. Then, in \cref{lem:obs_qpe_error}, we account for the finite precision of phase estimation. Combining the two in \cref{corr:obs_bound_main}, we obtain the desired bound on the error in the estimate of $\bra{\Psi(t)}O\ket{\Psi(t)}$.
\begin{lemma}\label{lem:state_and_obs_error}
Let $\tilde{O}=\lambda_{O}(\bra{0^{k}}\otimes I)\tilde{U}_{O}(\ket{0^{k}}\otimes I)$ such that $\lVert O-\tilde{O}\rVert_{\infty}\leq \epsilon_{O}$ and let $\lVert \ket{\Psi}-\ket{\tilde{\Psi}}\rVert_{2}\leq \epsilon_{\Psi}$. Then, 
\begin{align}
    \big\lvert\bra{\Psi}O\ket{\Psi}-\bra{\tilde{\Psi}}\tilde{O}\ket{\tilde{\Psi}}\big\lvert\leq 2\lambda_{O}\epsilon_{\Psi}+\epsilon_{O}\label{eq:app_obs_error}.
\end{align}
\end{lemma}
\begin{proof}
The proof follows via the triangle inequality. We note that 
\begin{align}
\big\lvert\bra{\Psi}O\ket{\Psi}-\bra{\tilde{\Psi}}O\ket{\tilde{\Psi}}\big\rvert\leq \norm{O}_{\infty}\big\lVert \dyad{\Psi}{\Psi}-\dyad{\tilde{\Psi}}{\tilde{\Psi}}\big\rVert_{1}\leq \lambda_{O}2\epsilon_{\Psi}\label{eq:bound1},
\end{align}
which follows from elementary properties of the trace and Hölder's inequality. Furthermore,
\begin{align}
\big\lvert\bra{\tilde{\Psi}}O\ket{\tilde{\Psi}}-\bra{\tilde{\Psi}}\tilde{O}\ket{\tilde{\Psi}}\big\rvert\leq \lVert O-\tilde{O}\rVert_{\infty}\leq \epsilon_{O}\label{eq:bound2},
\end{align}
which follows from the Cauchy-Schwarz inequality. Combining \cref{eq:bound1,eq:bound2} then gives \cref{eq:app_obs_error}.
\end{proof}

\begin{lemma}\label{lem:obs_qpe_error}
Let $\theta_{\est}$ be an estimate of the phase $\tilde{\theta}=\arccos({\bra{\tilde{\Psi}}\tilde{O}\ket{\tilde{\Psi}}}/{\lambda_{O}})$ obtained via QPE such that $|\theta_{\est}-\tilde{\theta}|\leq\epsilon_{\QAE}/\lambda_{O}$. Then the estimate $\widetilde{\langle \tilde{O}\rangle}$ of $\bra{\tilde{\Psi}}\tilde{O}\ket{\tilde{\Psi}}$ obeys 
\begin{align}
\Big\lvert\widetilde{\langle \tilde{O}\rangle}-\bra{\tilde{\Psi}}\tilde{O}\ket{\tilde{\Psi}}\Big\rvert\leq \epsilon_{\QAE}.
\label{eq:obs_qpe_error}
\end{align}
\end{lemma}
\begin{proof}
We note that 
\begin{align}
\Big\lvert\widetilde{\langle \tilde{O}\rangle}-\bra{\tilde{\Psi}}\tilde{O}\ket{\tilde{\Psi}}\Big\rvert=\big\lvert\lambda_{O}\cos \theta_{\est}-\lambda_{O}\cos \tilde{\theta} \big\rvert. 
\end{align}
Since $\lvert\cos x-\cos y\rvert\leq|x-y|$ and since $|\theta_{\est}-\tilde{\theta}|\leq\epsilon_{\QAE}/\lambda_{O}$, \cref{eq:obs_qpe_error} follows.
\end{proof}

\begin{corollary}\label{corr:obs_bound_main}
Let $\ket{\tilde{\Psi}(t)}=\tilde{U}_{B}\tilde{U}_{\prop}\tilde{U}_{\init}\ket{0}$ and let $\epsilon_{\Psi}=\epsilon_{\init}+\epsilon_{\prop}+\epsilon_{B}$. Then for any observable $O$ with a $(\lambda_{O},k,\epsilon_{O})$-block-encoding, 
\begin{align}
\Big\lvert\widetilde{\langle \tilde{O}\rangle}-\bra{\Psi(t)}O\ket{\Psi(t)}\Big\rvert\leq 2\lambda_{O}(\epsilon_{\init}+\epsilon_{\prop}+\epsilon_{B})+\epsilon_{O}+\epsilon_{\QAE}. 
\end{align}
\end{corollary}
\begin{proof}
This follows from \cref{lem:state_and_obs_error,lem:obs_qpe_error} by the triangle inequality and setting $\ket{\Psi}=\ket{\Psi(t)}$.
\end{proof}

\end{document}